\documentclass[11pt]{article}
\usepackage{amsmath}
\usepackage{amsthm}
\allowdisplaybreaks
\usepackage{amssymb}
\usepackage{algorithm}
\usepackage{subfig}

\usepackage{color}
\usepackage[dvipsnames]{xcolor}
\usepackage{graphicx}
\usepackage{wrapfig,epsfig}
\usepackage{epstopdf}
\usepackage{url}
\usepackage{graphicx}
\usepackage{color}
\usepackage{epstopdf}
\usepackage{algpseudocode}
\usepackage{scrextend}
\usepackage[T1]{fontenc}
\usepackage{bbm}
\usepackage{comment}
\usepackage{authblk}
\usepackage{multicol}
\usepackage{dsfont}
\usepackage{mathtools}
\usepackage{enumitem}
\usepackage{mathrsfs}
\let\C\relax
\usepackage{tikz}
\usepackage{hyperref}  
\hypersetup{colorlinks=true,citecolor=blue,linkcolor=red} 

\usetikzlibrary{arrows}
\usepackage[margin=1in]{geometry}

\graphicspath{{./figs/}}

\definecolor{b2}{RGB}{51,153,255}
\definecolor{mygreen}{RGB}{80,180,0}

\title{Super-resolution and Robust Sparse Continuous Fourier Transform in Any Constant Dimension: Nearly Linear Time and Sample Complexity\footnote{A preliminary version of this paper appears at ACM-SIAM Symposium on Discrete Algorithms (SODA 2023).}}

\author{
 Yaonan Jin\thanks{\texttt{yj2552@columbia.edu}. Columbia University.}
 \quad  
 Daogao Liu\thanks{\texttt{dgliu@uw.com}.  University of Washington.} 
 \quad
 Zhao Song\thanks{\texttt{zsong@adobe.com}. Adobe Research.}
 }

\date{}
\usepackage{bbm}
\usepackage{comment}
\usepackage{authblk}
\usepackage{multicol}
\usepackage{dsfont}
\usepackage{mathtools}
\usepackage{enumitem}
\usepackage{mathrsfs}

\usepackage{amsmath}
\usepackage{amssymb}

\usepackage{algorithm}

\graphicspath{{./figs/}}

\newcommand{\wh}{\widehat}
\newcommand{\wt}{\widetilde}
\newcommand{\ov}{\overline}
\newcommand{\eps}{\varepsilon}
\renewcommand{\epsilon}{\varepsilon}
\renewcommand{\phi}{\varphi}
\newcommand{\B}{\mathcal{B}}
\newcommand{\N}{\mathcal{N}}
\newcommand{\R}{\mathbb{R}}
\newcommand{\Z}{\mathbb{Z}}

\newcommand{\RHS}{\mathrm{RHS}}
\newcommand{\LHS}{\mathrm{LHS}}
\renewcommand{\i}{\mathbf{i}}

\renewcommand{\hat}{\wh}
\renewcommand{\bar}{\ov}
\renewcommand{\d}{\mathrm{d}}

\newcommand{\HashToBins}{\textsc{HashToBins}}
\newcommand{\LocateInner}{\textsc{LocateInner}}
\newcommand{\LocateSignal}{\textsc{LocateSignal}}
\newcommand{\EstimateSignal}{\textsc{EstimateSignal}}
\newcommand{\OneStage}{\textsc{OneStage}}
\newcommand{\MultiStage}{\textsc{MultiStage}}
\newcommand{\MergedStage}{\textsc{MergedStage}}

\newcommand{\RecoveryStage}{\textsc{RecoveryStage}}
\newcommand{\SampleTimePoint}{\textsc{SampleTimePoint}}

\newcommand{\CFT}{CFT}
\newcommand{\DFT}{DFT}
\newcommand{\FFT}{FFT}
\newcommand{\DTFT}{DTFT}

\DeclareMathOperator*{\E}{\mathbf{E}}

\newcommand{\NP}{\mathbb{N}_{\geq 1}}
\DeclareMathOperator*{\C}{\mathbb{C}}

\DeclareMathOperator*{\argmin}{argmin}

\DeclareMathOperator*{\median}{median}

\DeclareMathOperator{\supp}{supp}

\DeclareMathOperator{\poly}{poly}

\DeclareMathOperator{\rect}{rect}
\DeclareMathOperator{\sinc}{sinc}

\DeclareMathOperator{\Dirac}{Delta}

\DeclareMathOperator{\coll}{coll}
\DeclareMathOperator{\off}{off}
\DeclareMathOperator{\unif}{Unif}

\DeclareMathAlphabet{\mathpzc}{OT1}{pzc}{m}{it}

\theoremstyle{plain}
\newtheorem*{setup*}{Set-Up}
\newtheorem{theorem}{Theorem}[section]
\newtheorem{lemma}[theorem]{Lemma}

\newtheorem{fact}[theorem]{Fact}
\newtheorem{claim}[theorem]{Claim}
\newtheorem{corollary}[theorem]{Corollary}

\newtheorem{condition}[theorem]{Condition}

\newtheorem{remark}[theorem]{Remark}

\theoremstyle{definition}
\newtheorem{definition}[theorem]{Definition}

\renewcommand{\epsilon}{\varepsilon}
\renewcommand{\phi}{\varphi}
\renewcommand{\i}{\mathbf{i}}

\renewcommand{\hat}{\wh}
\renewcommand{\bar}{\ov}
\renewcommand{\d}{\mathrm{d}}

\DeclareMathAlphabet{\mathpzc}{OT1}{pzc}{m}{it}

\begin{document}

\begin{titlepage}
  \maketitle
  \begin{abstract}
The ability to resolve detail in the object that is being imaged, named by resolution, is the core parameter of an imaging system.
Super-resolution is a class of techniques that can enhance the resolution of an imaging system and even transcend the diffraction limit of systems.
Despite huge success in the application, super-resolution is not well understood on the theoretical side, especially for any dimension $d \ge 2$.
In particular, in order to recover a $k$-sparse signal, all previous results suffer from either/both $\poly(k)$ samples or running time.

We design robust algorithms for any (constant) dimension under a strong noise model based on developing some new techniques in Sparse Fourier transform (Sparse FT), such as inverting a robust linear system, ``eggshell'' sampling schemes, and partition and voting methods in high dimension.
These algorithms are the first to achieve running time and sample complexity (nearly) linear in the number of source points and logarithmic in bandwidth for any constant dimension,
and we believe the techniques developed in the work can find their further applications on the Super-resolution and Sparse FT problem.

%%%%%%%%%%%%%%
%%%%%%%%%%%%%%

  \end{abstract}
  \thispagestyle{empty}
\end{titlepage}

%\newpage

\newpage
{\hypersetup{linkcolor=black}
\tableofcontents
}

\newpage

%\setcounter{page}{1}

%%% Zhao: Don't move and remove the above table of contents

%%%% Cut-line between first 10 pages and appendix

\section{Introduction}
Since people began to design and study optical systems, the resolution has become the core parameter of an optical system.
Roughly speaking, resolution of an imaging system is defined as its ability to distinguish two points as separate in space and resolve detail in the object being imaged.
Because of the physics of diffraction, there are some fundamental limits on the resolution of an imaging system.
Surprisingly, people find fantastic Super-resolution techniques, and the diffraction limit of systems is transcended.
As an outstanding representative in this field, the Nobel Prize in Chemistry 2014 was awarded jointly to Eric Betzig, Stefan W. Hell, and William E. Moerner
{\it    " for the development of super-resolved fluorescence microscopy."}

We formalize the Super-resolution problem considered in the work here.
Let $x^*(t)= \sum_{i\in[k]}v_i \cdot e^{2\pi\i f_i^\top t}$ be a signal with $k$-point sources in $d$-dimensional space where $f_1,\cdots,f_k \in \R^d$. 
Assume we are able to observe a complex-valued signal function $x(t) = x^*(t) + g(t) \in \C$ over a finite duration $t \in [0, T]^{d}$, where $g(t)$ captures the noise in the measurement and we do not have any assumption on $g(t)$.
To access the signal $x(t) = x^*(t) + g(t)$, an algorithm can only sample $x(\tau_{j})$ at a number of $m$ time points $\{\tau_{j}\}_{j \in [m]}$. These $\{\tau_{j}\}_{j \in [m]}$ can be {\bf arbitrarily} chosen from the duration $t \in [0,T]^d$.\footnote{This assumption is standard in the Continuous Fourier Transform literature, though not standard in the Super-resolution community.}
Like the standard objective in Super-resolution, we hope to design a fast algorithm that can estimate $\{(v_i,f_i)\}_{i\in [k]}$ with few samples.

Moreover, we are more ambitious and want our algorithm should output a $k$-Fourier-sparse recovered signal $x'(t)$ such that, for some approximation ratio $\mathcal{C} > 1$,
\begin{align}\label{eq:intro_formulation:3}
    \text{signal~estimation~error~} := ~ 
       \int_{\mathsf{B}} |x'(t) - x(t)|^{2} \cdot \d t
    ~ \leq ~ \mathcal{C}^{2} \cdot \N^{2},
\end{align}
where $\int_{\mathsf{B}} := \frac{1}{T^d} \int_{[0,T]^d}$, $\N^{2} :=  \int_{ \mathsf{B} } |g(t)|^2\d t + \delta \cdot \sum_{i \in [k]} |v_{i}|^2 > 0$ be the {\em noise level} and $\delta>0$ is some parameter to conclude the noiseless case (i.e. $\int_{\mathsf{B}}|g(t)|^2\d t=0$).
For simplicity, we define $\|g\|_T^2:=\int_{\mathsf{B}}|g(t)|^2\d t$.

To make the problem interesting, we assume a bounded support $\supp(\hat{x}^*) = \{f_{i}\}_{i \in [k]} \subseteq [-F, F]^{d}$ for the frequencies where the parameter $F > 0$ is known, and $\{f_i\}_{i\in [k]}$ have some minimum distance $\eta>0$, i.e. $\min_{i\neq i'}\|f_i-f_j\|_2=\eta$. 
The band-limited assumption on signal and separation assumption on frequencies are standard both in Super-resolution \cite{m15,hk15,cm20} and Continuous Fourier Transform \cite{ps15,ckps16,cp19_colt,cp19_icalp,sswz22}.
As mentioned before, there is no requirement for the noise $g(t)$, which is a strong noise model compared to many previous works. 

% Despite many exciting results in theory, such as \cite{hk15,cfg14,kpro16,cm20}, and huge success in the application, Super-resolution is still not well-understood, especially in the multiple-dimensional case ($d\geq 2$).

In the noise-free case, there are a variety of methods to do Super-resolution \cite{p73,hs90,s93} when duration $T=k/\eta$.
Subsequently, there are some new methods \cite{cf14,cf13} to solve this problem based on assumption that either it is noise-free or $f_j$'s are restricted to be on a grid.
Moreover, there are rich literature in high dimension such as \cite{hk15,cfg14,kpro16}.

We use Table~\ref{tab:comparisiton} to give a rough comparison between previous results and ours. 
Because previous works are under different noise models (we are the strongest, without any restriction on the noise), different settings (some of them can only measure the input signal on grid points, which makes the problem more difficult) or with different focus (e.g., \cite{cm20} focus on the sharp constant $c$ for the minimum duration $c/\eta$ such that one can hope to get polynomial statistical and algorithmic complexity, while \cite{hk15} and ours result lose some logarithmic term on the duration), this is just a high-level comparison.
Some more detailed discussions about related work will be given later.

\begin{table}[h]
    \centering
    \begin{tabular}{|l|l|l|}
    \hline
          {\bf Refs}            &  {\bf \# Samples} & {\bf Running time} \\
    \hline
    \cite{cc13} & $\wt{O}(k)$ & $\wt{O}(k^d)$  \\
     \hline
      \cite{hk15}   & $\wt{O}_{d}(k^2)$ & $\wt{O}_{d}(k^2)$  \\
      \hline
      \cite{cm20}  & $\wt{O}_{d}(k^{2}) \poly(F/\eta) $ & $\wt{O}_{d}(k^6) \poly(F/\eta)$  \\
      \hline
      Ours  & $\wt{O}_{d}(k\log(F/\eta))$ & $\wt{O}_{d}(k\log(F/\eta))$\\
      \hline
    \end{tabular}
    \caption{Rough comparison between previous works and our result for constant dimensions $d$.}
    \label{tab:comparisiton}
\end{table}

One natural idea to deal with high dimension problem is to map it to one dimension.
If we do transformation and project $d$-dimensional signal to one-dimensional, we can directly apply the results in one dimension such as \cite{ps15}, which loses $\poly(k)$ factor in duration, sample and running time complexity.
Another way is to do semi-definite programming (SDP), which is usually based on results of Candes and Fernandez-Granda \cite{cf13,cf14} but the sample complexity and running time can still be very large.

Despite the huge success and developments, Super-resolution in multi-dimensional cases are still not well-understood, and improving efficiency on sampling complexity and computation complexity (running time) is an important and fundamental open problem. 
As described in \cite{hk15},
\begin{center}
    \it  It remains an open problem to reduce the sample complexity ... from $O(k^2)$ to the information theoretical bound $O(k)$, while retaining the polynomial scaling of the computation complexity.
\end{center}

To be even more ambitious, can we achieve nearly linear computation complexity rather than being polynomial with nearly linear sample complexity?
This leads to the following fundamental algorithmic and statistical problem:
\begin{center}
    {\it How efficient a Super-resolution algorithm can be on the running time and sample complexity?}
\end{center}

Our work makes an important step towards solving this problem.

\subsection{Our results}
Roughly speaking, our algorithm {\RecoveryStage} (Algorithm~\ref{alg:recovery_stage}) achieves a constant approximation to the noise level $\N > 0$ in any constant dimension. For the tone estimation, we have the following guarantees.

\begin{theorem}[Informal Tone estimation, see Theorem~\ref{thm:recovery_stage}]
\label{thm:intro_tone}
When {\RecoveryStage} observes the signal $x(t)$ over a duration\footnote{We often denote $f \gtrsim g$ when $f \geq C_{0} \cdot g$ for some universal constant $C_{0} > 0$, and the notation $f \lesssim g$ has a similar meaning. Also, we denote $f \eqsim g$ when both equations $f \gtrsim g$ and $f \lesssim g$ hold.
In this page, these notations hide the dependence on dimension $d$.
}
$
    \text{ $T ~ \gtrsim ~ \eta^{-1}  \cdot \log(k / \delta) $}
$, 
it outputs $k \geq 1$ recovered tones $\{ (v_{i}',f_{i}') \}_{i \in [k]} \subseteq \C \times \R^d$ that approximate the true tones $\{ (v_{i},f_{i}) \}_{i \in [k]}$ up to an error proportional to the noise level $\N > 0$, with high probability.
\iffalse
More concretely, the recovered tones $\{ (v_{i}',f_{i}') \}_{i \in [k]}$ can be reindexed such that: Every ``large'' true tone is recovered well.\footnote{If the magnitude of a true tone is too small $v_{i} \lesssim \N$, then the corresponding signal $x_{i}(t)$ can be erased by a particular noise $g(t) = -x_{i}(t)$, and no estimation guarantee is possible.} Suppose a true tone $i \in [k]$ has a sufficiently large magnitude $| v_{i} | \gtrsim \N$, its magnitude $v_{i} \in \C$ and frequency $f_{i} \in [-F, F]$ are recovered up to\footnote{We denote $f \lesssim_d g$ when $f \leq \poly(d) \cdot g$.}
    $
         \| f_{i}' - f_{i} \|_{2} ~ \lesssim_d ~   \N \cdot T^{-1} \cdot |v_i|^{-1}$, and $ 
         |v_{i}' - v_{i}| ~ \lesssim_d ~  \N
    $. 
\fi
\iffalse
(2) Let $x_{i}'(t) := v_{i}' \cdot e^{2 \pi \i \cdot f_{i}^{'\top} t}$ be the individual recovered signals, for $i \in [k]$. The total tone estimation error is bounded as follows (which is a stronger guarantee than(1) ):\footnote{We pick ${\cal C}$ between $\Omega(1)$ and the {\em signal-to-noise ratio} $\rho \gg 1$ (see Definition~\ref{def:ratio_signal_noise}).}
    \begin{align}\label{eq:intro_result:1}
         \sum_{i \in [k]}   \int_{B} | x_{i}'(t) - x_{i}(t) |^2 \cdot \d t
        ~ \lesssim_d ~ {\cal C}^2 \N^2.
    \end{align}
\fi
The algorithm {\RecoveryStage} takes $k \cdot (\log k)^{d+O(1)} \cdot \log( F / \eta ) \cdot 2^{O(d \log d)}$ samples and time.
\end{theorem}
% Our main theorem directly implies the following results,
% \begin{corollary}
% For $d=2$, we recover all the frequencies and magnitudes in $k \log(F/\eta) \log^3 k$ samples and running time if $T \gtrsim \log(k/\delta) / \eta$. 
% \end{corollary}
% \begin{remark}
% In two dimension $d=2$, our result obtains much better sample complexity and running time (compared to the state-of-art result \cite{cm20}'s $\poly(k,F/\eta)$)  while requires slightly stronger assumption on $\eta$ (compared to \cite{cm20}'s $T\gtrsim 1/\eta$). Note that in sparse Fourier transform literature, the major goal is to get nearly linear in $k$ sample complexity, and $k^{1+\Omega(1)}$ is not allowed.
% Our dependence on $k$ is (nearly) optimal, and our dependence on $F/\eta$ is exponentially better than the one in \cite{cm20}.
% \end{remark}

% \begin{remark}
% If we do a naive transformation, we can project $d$-dimensional signal to one-dimensional and lose $\poly(k)$ factor in duration, and we then can directly apply \cite{ps15}. This will give a result that if $ T \gtrsim \eta^{-1} \cdot \poly(k)$,
% then we can recover signal. One of the major contribution of our result is getting non-trivial $\log (k)$ dependence in duration bound, but not $\poly(k)$ dependence.
% \end{remark}

As for the signal estimation, we have the following guarantee, which to our knowledge is a new guarantee in the Super-resolution literature.

\begin{theorem}[Informal Signal reconstruction, see Theorem~\ref{thm:signal_formal}]
\label{thm:intro_signal}
When {\RecoveryStage} observes the signal $x(t)$ over a duration
$
    \mbox{$T ~ \gtrsim ~ \eta^{-1}   \cdot k^{1 - 1 / d} \cdot \log(k  / \delta)$},
$
the signal estimation error of the $k$-Fourier-sparse recovered signal $x'(t) := \sum_{i \in [k]} x_{i}'(t)$ against the observed signal $x(t) = x^*(t) + g(t)$ is bounded as follows:
\begin{align}\label{eq:intro_result:2}
     \int_{\mathrm{B}} | x'(t) - x(t) |^2 \cdot \d t
    ~ \lesssim ~ \N^2.
\end{align}
\end{theorem}

\begin{remark}
\label{rem:duration}
In one dimension $d = 1$, our algorithm works when the duration $T \gtrsim \eta^{-1} \cdot \log(k / \delta)$, but the state-of-art result \cite{ps15} requires $T \gtrsim \eta^{-1} \cdot \log^2(k / \delta)$ for the signal estimation. (For more details about this improvement, see Section~\ref{sec:our_technique:tone_estimation}.) Indeed, the duration is an equally important optimization goal as the sample complexity and the running time. 
\end{remark}

{\bf Preliminary Discussion:}
For any constant dimensions, we succeed to get an algorithm with both nearly optimal sample complexity and run-time, which is the goal in most of the literature on sparse Fourier transforms.
However, due to the exponential dependence on the dimension in our result, this is not the end of story.

Up to the iterated logarithmic factors, our algorithm {\RecoveryStage} takes $k \cdot (\log k)^{d+O(1)} \cdot \log( F / \eta ) \cdot 2^{O(d \log d)}$ samples/running time.
Merely extending the filter functions into high dimensions requires some very non-trivial efforts, but it already leads to an exponential loss in the dimension.
This is a consequence of our ``{\em precise}'' filter function, seems to be unavoidable using current filtering techniques since even if the one-dimensional filter’s support size is off by a constant factor, it would lead to an exponential loss in the dimension anyways. 
As quoted:
\begin{quote}
     \cite{k16,k17} {\em ``in the discrete settings ...\ the price to pay for the precision of the filter, however, is that each hashing becomes a $\log^{d} k$ factor more costly in terms of sample complexity and running time than in the idealized case ...''}
\end{quote}
To shave the $\log^{d} k$ term in the discrete model, the past works \cite{ik14,k16} randomize the noise by using the ``crude'' filters. However, randomizing the noise does not work in the continuous model, since two noise frequencies $f, f'$ can be arbitrarily close and, no matter how we randomized the noise, the errors can accumulate in the estimation. 
The exponential dependence on dimension seems to be intrinsic to the current
sampling methods, and avoiding it could need completely new methods.

% As mentioned before, in this paper we consider Super-resolution with different focus, and fix the frequency separation $\eta$ to get efficient algorithm on running time and sample complexity, also with the hope to minimize the duration $T$.
% We use the Fast Fourier Transform methods to do Super-resolution with (near) linear running time and sample complexity for constant dimension $d=O(1)$.
% As a trade off, we need to enlarge the duration slightly and sacrifice the dependence on dimension.
% In fact, enlarging the duration (the size of the observation space) is a common method to enhance the resolution.
% For example, astronomers are constantly building astronomical telescopes with larger and larger calibers \cite{nlj+11}.
% Besides, people seldom do Super-resolution in a very high dimension and the constant-dimensional cases when $d=O(1)$ are the most important.

% Though our algorithms sacrifice the dependence on the dimension,
% one of the major contribution of our result is getting non-trivial $\log (k)$ dependence in duration bound, and near linear bounds on sample and running time complexity but not $\poly(k)$ dependence.

%%%%%%%%%%%%%%%%%%%%%%
\subsection{Related works}
%\vspace{-1mm}
\subsubsection{Super-resolution with a different focus}
%\vspace{-1mm}
The previous results \cite{m15,cm20} are focused on finding the minimum possible separations between source points for fixed cutoff frequency (denoted by duration in this paper), such that there exists an algorithm with polynomial running time by using a polynomial number of samples.
As a result, their algorithms are not efficient in running time and sample complexity.

In the following, we compare our work with \cite{cm20} in more detail.
Chen and Moitra~\cite{cm20} investigate a two-dimensional Super-resolution problem which they reduce to the problem of continuous Sparse FT. 
The main difference between their model and our model, is the way how the noise hampers the frequency recovery. Recall that we consider a signal $x(t) = x^*(t) + g(t) \in \C$ over a duration $t \in [0, T]^{d}$, where $x^*(t) \in \C$ is the actual signal that we aim to recover, and the noise $g(t) \in \C$ has a small enough {\em constant-proportional} energy compared to $x^*(t)$, that is,
$\|g\|_{T} \leq 10^{-3} \cdot \|x^*\|_{T}$.\footnote{Recall that the average energy, e.g., of the noise $g(t)$ over duration $t \in [0, T]^{d}$, is defined as $\|g\|_{T} = \frac{1}{T^{d}} \cdot \int_{0}^T |g(t)|^2 \d t$.} In particular, the noise magnitude $|g(t_{0})|$ at a certain time point $t_{0} \in [0, T]^{d}$ has no requirement, and can even be much larger than the signal magnitude $|x^*(t_{0})|$.
In contrast, \cite{cm20} make a stronger assumption on the noise $g(t) \in \C$. At {\em any} time point $t_{0} \in [0, T]^{d}$, they need the noise magnitude $|g(t_{0})|$ is always {\em inverse-polynomially} small, compared to the corresponding average signal energy $\|x^*\|_{T}$.

For their model, Chen and Moitra focus on the two-dimensional case, and their primary emphasis is on refining the duration requirement in the two-dimensional case, i.e., on the exact constant in front of $\eta^{-1}$ for constant $d=2$. For general constant $d$, Chen and Moitra can (via tensor decomposition) get sample complexity $\wt{O}_{d}(k^{2}) \poly(F/\eta) $ and running time $\wt{O}_{d}(k^6) \poly(F/\eta)$,\footnote{The notation $\wt{O}_{d}(f)$ assumes a constant dimension $d \geq 1$ and hides the term $\poly(\log f)$; similar for $\wt{\Omega}_{d}$ and $\wt{\Theta}_{d}$.}
while our running time and sample complexity are $\wt{O}_{d}(k\log(F/\eta)\log^{{O(1)}}(k) )$.
As a trade off, their duration is $T\gtrsim 1/\eta$ while ours is $T\gtrsim \log(k)/\eta$.

Note that in sparse Fourier transform/sparse recovery literature, the major goal is to get nearly linear in $k$ sample complexity, and $k^{1+\Omega(1)}$ is not allowed (see Table 1 in \cite{ns19} and Table 1 in \cite{nsw19}). 
It is well-known that in many cases, $k^2$ or even $k^{1+\Omega(1)}$ samples can make the problem subsequently easier.
Also, it is worth mentioning that, \cite{cm20} considers the ``tone recovery'' problem only, without studying the ``signal recovery'' problem, whereas our paper investigates the both problems.
%\vspace{-2mm}
\subsubsection{Prior works on the sparse FT problem}
\label{sec:previous_techniques}
%\vspace{-2mm}
As our technology originates from Fourier Transform, in this and next sub-subsection, we briefly review several previous works for classic prior works on the discrete FT (DFT) and continuous FT (CFT) separately. For a more detailed overview, the reader can refer to Section~\ref{sec:previous}.

\vspace{.1in}
\noindent
{\bf The discrete model.}
In any dimension $d \geq 1$, the Fourier transform $\hat{x} \in \C^{N}$ is a vector of length $N = n^{d}$. The goal of a sparse DFT algorithm is, given a bunch of samples $x_{i}$ in the time domain and the sparsity parameter $k$, to output a $k$-Fourier-sparse signal $x'$ with the $\ell_{2}/\ell_{2}$-guarantee
\begin{align*}
     \| \hat{x}' - \hat{x} \|_2 ~ \lesssim ~ \min_{k\text{-sparse}~z} \| z - \hat{x} \|_2.
\end{align*}
Following the framework of \cite{gms05,hikp12a,ikp14,ik14,k16}, the idea is to take, multiple times, a set of $\B = B^{d} = \Theta_{d}(k)$\footnote{Here the notation $\Theta_{d}(f)$ assumes a constant dimension $d \geq 1$; similar for $O_{d}$ and $\Omega_{d}$.} linear measurements of the form $u_j := \sum_{i: \mathpzc{h}(i) = j} x_{i} \cdot s(i)$, where $\mathpzc{h}: [N] \mapsto [\mathcal{B}]$ are random hash functions and $s: [N] \mapsto \{\pm 1\}$ are random sign functions. This means ``{\em hashing into $\mathcal{B}$ bins}''. If the linear measurements are ideal, then $O(\log(N / k))$ hashes are enough for sparse recovery and the sample complexity is $O(k \log(N / k))$.

Based on the linear combinations of the samples $x_{i}$, the sparse DFT algorithms will approximate the $u_{j}$'s. That is, we first permute the samples $x_{i}$ via a pseudorandom affine permutation $\mathcal{P}$. Then, the permuted samples $(\mathcal{P} x)_{i}$ are respectively scaled by coefficients $\mathcal{G}(l_{i})$, i.e., the values of a filter function $\mathcal{G}: \R^{d} \mapsto \R$ at a bunch of lattice points $l_{i} \in \R^{d}$. Hence, we use a modified combination
\begin{align}
\label{eq:intro_overview:1}
    \mbox{$u_{j} ~ = ~ \sum_{i: \mathpzc{h}(i) = j} (\mathcal{P} x)_{i} \cdot \mathcal{G}(l_{i})$}.
\end{align}
Different from the {\em binary-valued} sign functions, the filter functions $\mathcal{G}$ shall be ``imperfect'' to reduce the sample complexity. Namely, every coordinate $i \in [n]$ not only contributes $\approx 100\%$ fraction to a target bin, but also ``leak'' a small fraction to each other bin. (And to balance the trade-off between the sample complexity and the running time, the past works like \cite{hikp12a,ik14,k16,k17} use different leakage levels.)

The above approach ``isolates''  most of the head frequencies $\{f_{i}\}_{i \in [k]}$ (i.e., the top-$k$ coordinates of $\hat{x}$). In precise, most $\{f_{i}\}_{i \in [k]}$ are hashed to unique bins, and the ``tail'' frequencies $[n]^{d} \setminus \{f_{i}\}_{i \in [k]}$ contribute very little to those bins. So the algorithm can exactly identify the head frequencies and approximately evaluate the magnitudes $\hat{x}(f_{i})$, producing a $k$-sparse estimation $\hat{x}' \approx \hat{x}$.

Also, notice that the DFT preserves the $\ell_2$-norm of a Fourier spectrum, namely $\|\hat{z}\|_{2} = \|z\|_{2}$ for any $z \in [n]^{d}$, so the $\ell_{2}/\ell_{2}$-guarantees in the frequency/time domains are equivalent.

\vspace{2mm}
\noindent
{\bf The continuous model.}
The one-dimensional sparse CFT problem is introduced by \cite{ps15}, and our formulation is a natural multi-dimensional extension. Different from the discrete model, we cannot recover the exact head frequencies $\{f_{i}\}_{i \in [k]}$ in the continuous model. The current frequencies are off-the-grid, so (i)~any two frequencies $f \neq f'$ can be too close to distinguish \cite{m15}; and (ii)~even if a head frequency $f_{i}$ is well separated from the others, % (Assumption~\ref{ass:frequency_separation}), 
we can only recover it up to some precision that depends on the duration $T > 0$.

As the frequency recovery is not exact, we cannot hope for the best $k$-sparse Fourier spectrum. For this reason, \cite{ps15} considers the tone/signal estimations under the $\ell_{2}/\ell_{2}$-guarantee in the time domain. In addition to the approximation guarantee, sample complexity and running time, we have one more optimization goal -- minimizing the duration $t \in [0, T]^{d}$ for the sampling.

\subsection{Our techniques}
\label{sec:our_technique}

%  In case you need some examples of how to write technique overview, you can read \cite{k17,ckps16}. see \cite{k17}'s page 3,4,5 see \cite{ckps16}'s page 6,7,8
Similar to the previous works, our main task is to recover the head frequencies $\{f_{i}\}_{i \in [k]}$. As if we promise a good approximation $\{f'_{i}\}_{i \in [k]} \approx \{f_{i}\}_{i \in [k]}$, then the magnitudes $\{v_{i}\}_{i \in [k]}$ can be easily recovered. 

To deal with the continuous model, the overall ideas in \cite{ps15} are to translate the hash functions, filter functions and estimation algorithms from the DFT setting to the CFT setting,
and we adopt the similar framework.
However, extension one/two-dimensional (\cite{ps15}/\cite{cm20}) cases to the multi-dimensional continuous case presents a number of challenges, which are addressed in this paper by some interesting techniques. 
% In this section, we outline our novel techniques and approaches.
Among these, there are three most remarkable ones.
\begin{itemize}
    \item Our hashing scheme is specifically designed for the multi-dimensional continuous model, and the ``eggshell'' sampling scheme (for time points) which differs from all the previous ones. %{\em non-uniform}

    \item To learn the frequencies $f_{i}$'s more accurately (while ensuring a {\em logarithmic} algorithm in $F/\eta$ which also means beating $\poly(F/\eta)$ sample/time in \cite{cm20}), we apply (i)~a coarse-grained location procedure, for which we employ technical ingredients from high-dimensional geometry; and then (ii)~a fine-grained location procedure, which is built upon a robust linear-system solver.

    \item %As mentioned (Remark~\ref{rem:duration}), 
    The duration bound required by a recovery algorithm is an equally important optimization goal as the sample complexity and the running time in the super-resolution. To improve the duration bound against the previous algorithm \cite{ps15}, 
    we provide a better analysis by leveraging Parseval's theorem and the convolution theorem in a different manner.
\end{itemize}

\subsubsection{Hashing and sampling}
\label{sec:our_technique:permutation_hashing}

\noindent
{\bf The obstacles.}
As mentioned, we assume the head frequencies (defined by $\textsc{Head} := \{f_{i}\}_{i \in [k]}$) locate within the hypercube $[-F, F]^{d}$ and are separated by $\eta = \min_{i \neq i' \in [k]} \|f_{i} - f_{i'}\|_{2} > 0$, corresponding to the $k$-Fourier-sparse signal $x^*(t)$. The other tail frequencies $\textsc{Tail}$ correspond to the noise $g(t)$.

To recover the head frequencies, a direct attempt is to handle all dimensions $r \in [d]$ separately, through the one-dimensional hashing scheme in \cite{ps15}. Unfortunately, this approach fails to work. For example, suppose two frequencies are equal in the first dimension, i.e., $f_{i, 1} = f_{i', 1}$ for $i \neq i' \in [k]$ (but the overall $\ell_{2}$-distance in the other dimensions is $\geq \eta$). Then regarding the first dimension, no hashing scheme can distinguish these two scenarios: (i)~the desired tones $(f_{i}, v_{i})$ and $(f_{i'}, v_{i'})$; and (ii)~a single tone $(f, v)$ given that $f_{1} = f_{i, 1} = f_{i', 1}$ and $v = v_{i} + v_{i'}$. Thus, an algorithm can {\em miscount} the tones, and recover the top-$(k + 1)$ or even more magnitudes.
% This may lead to non-ignorable errors if the algorithm outputs tones with the largest $k$ magnitudes.
Also, when the miscount happens (in one or more dimensions), an algorithm cannot match the dimension-wise frequencies correctly.
For these reasons, the multi-dimensional model requires a ``not-very-naive'' hashing scheme.
% In the continuous model, the frequency domain is the hypercube $[-F, F]^{d}$ and the magnitudes are given by the Fourier transform $\hat{x}: [-F, F]^{d} \mapsto \C$. For convenience, we partition this domain into the head/tail frequencies $[-F, F]^{d} = \textsc{Head} \sqcup \textsc{Tail}$. The head frequencies  and . 
% The CFT of a head signal $x_{i}^*(t) := v_i \cdot e^{2\pi \i \cdot f_{i}^{\top} t}$ is a shifted Dirac delta function (Section~\ref{sec:intro_formulation}).

% Different from the discrete cases, where the samples w.l.o.g.\ are taken ``{\em on the grid}'', currently we may access a signal sample $x(t)$ at any time point $t \in [0, T]^{d}$ within the duration. For this reason, 

% Recall that the lattice points $l_{i} \in \R^{d}$ determine the scaling coefficients $\mathcal{G}(l_{i})$. Also, the above combination is taken separately for all bins $j \in [\mathcal{B}]$.

% As mentioned, the first ingredient of a sparse FT algorithm is to map the frequency domain into $\mathcal{B} = \Theta_{d}(k)$ many bins, via a pseudorandom permutation $\mathcal{P}$ and a pseudorandom hash function $\mathpzc{h}$.

\vspace{.1in}
\noindent
{\bf Our approach.}
Similar to Eq.~\eqref{eq:intro_overview:1}, we will leverage the measurements
$
u_{j}  =  \sum_{i: \mathpzc{h}(i) = j} \mathcal{P} x(\tau_{i}) ~ \cdot ~ \mathcal{G}(l_{i})
$, 
where $\tau_{i} \in [0, T]^{d}$ are the sampling time points. To define permutation, we introduce three notations : $\Sigma \in \R^{d \times d}$ is scaling frequency domain, and $b \in \R^d$ is shifting frequency domain and $a \in \R^d$ is shifting time domain. We explain how to select them later.
Now, let us present the formal permutation:
\begin{align}
    \mbox{$\hat{\mathcal{P} x}(\mathrm{frac}(\Sigma f - b))
    ~ = ~ \hat{x}(f) \cdot \det(\Sigma)^{-1} \cdot e^{-2 \pi \i \cdot f^{\top} a}$},
    \label{eq:intro_permutation:2}
\end{align}
where the function $\mathrm{frac}: \R^{d} \mapsto [0, 1)^{d}$ computes the coordinate-wise {\em fractional part} of the input. %Below we explain how to select the $d$-to-$d$ random matrix $\Sigma$ and $d$-dimensional random vectors $b, a$ and why we need the $\mathrm{frac}(\cdot)$ function. (For more details, the reader can refer to Section~\ref{sec:HashToBins_multi}.)

There are two requirements for the random matrix $\Sigma$: (i)~it must be invertible; and (ii)~makes any two different head frequencies $\xi_{i} \neq \xi_{i'} \in \textsc{Head}$ hashed into the same bin with probability at most $0.01 \cdot k^{-1}$ (i.e., the collision probability). To these ends, we construct the $\Sigma$ in three steps.
\begin{itemize}
    \item Step~I. We first sample an interim matrix $\Sigma' \sim \unif(\mathbf{SO}(d))$ uniformly at random from the $d$-dimensional {\em rotation group}, leading to a rotation matrix $\Sigma'$ with determinant $|\det(\Sigma')| = 1$. Clearly, such an interim matrix $\Sigma' \in \R^{d \times d}$ is invertible.

    \item Step~II. Let us explain what the bins stand for in the continuous model. Given the transformation $\mathrm{frac}(\Sigma f - b)$ in Eq.~\eqref{eq:intro_permutation:2}, we are interested in the codomain $[0, 1)^{d}$. We partition this unit hypercube into $\mathcal{B} = B^{d} = \Theta_{d}(k)$ isomorphic sub-hypercubes, with the volume $1 / \mathcal{B}$ each. These sub-hypercubes are exactly the bins in the continuous setting.

    \item Step~III. We sample a random scaling factor $\beta \sim \unif[\hat{\beta}, 2 \hat{\beta}]$, where the parameter $\hat{\beta} > 0$ is sufficiently large, and derive the ultimate random matrix by letting $\Sigma := \beta \Sigma'$. Clearly, $\Sigma \in \R^{d \times d}$ is invertible. Below We will explain why this $\Sigma$ gives a small collision probability.
\end{itemize}
According to Eq.~\eqref{eq:intro_permutation:2}, whether two head frequencies $f_{i} \neq f_{i'} \in \textsc{Head}$ collides or not relies on the difference vector $\Sigma (f_{i} - f_{i'}) \in \R^{d}$. Since $\Sigma$ is a random {\em rotation matrix} scaled by $\beta \sim \unif[\hat{\beta}, 2 \hat{\beta}]$, this difference vector is distributed {\em almost uniformly} within the $\ell_{2}$-norm ``eggshell''
\begin{align*}
    \mbox{$\big\{z \in \R^{d}: ~ \hat{\beta} \cdot \|f_{i} - f_{i'}\|_{2} ~ \leq ~ \|z\|_{2} ~ \leq ~ 2 \hat{\beta} \cdot \|f_{i} - f_{i'}\|_{2} \big\}$}.
\end{align*}

\begin{figure}[htbp]
    \centering
    \begin{tabular}[b]{c}
    \includegraphics[width = .25\textwidth]{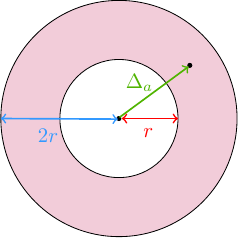}\\
    \small{(a)~Sampling for $\Delta_{a}$}
    \end{tabular} \qquad
    \begin{tabular}[b]{c}
    \includegraphics[width = .25\textwidth]{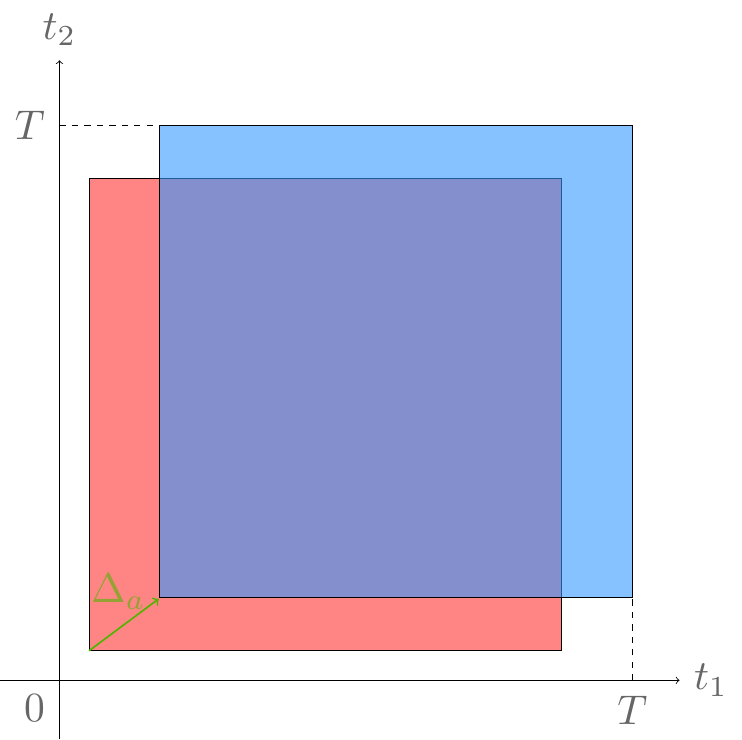}\\
    \small {(b)~Sampling for $a$ and $a'$}
    \end{tabular}
    \caption{
    Demonstration of the sampling scheme. The time difference $\Delta_{a} = (a' - a)$ is sampled from the pink region in Figure~(a), which looks like an ``{\bf eggshell}''. The time points $a, a'$ are sampled respectively from the red region and the blue region in Figure~(b).}
    \label{fig:our_technique:permutation_hashing:sampling}
\end{figure}

% \begin{wrapfigure}{R}{.55\textwidth}
%     \scriptsize
%     \captionsetup{font=small}
%     \vspace{-0.9cm}
%     \centering
%     \subfloat[Sampling for $\Delta_{a}$
%     \label{fig:our_technique:permutation_hashing:sampling_delta}]{
%     \includegraphics[width = .25\textwidth]{sample_delta.pdf}
%     }
%     \hskip .15in
%     \subfloat[Sampling for $a$ and $a'$
%     \label{fig:our_technique:permutation_hashing:sampling_a}]{\includegraphics[width = .25\textwidth]{sample_a.pdf}}
%     \vspace{-0.2cm}
%     \caption{
%     Demonstration of the sampling scheme. The time difference $\Delta_{a} = (a' - a)$ is sampled from the pink region in Figure~\ref{fig:our_technique:permutation_hashing:sampling_delta}, which looks like an ``{\bf eggshell}''. The time points $a, a'$ are sampled respectively from the red region and the blue region in Figure~\ref{fig:our_technique:permutation_hashing:sampling_a}.}
%     \label{fig:our_technique:permutation_hashing:sampling}
%     \vspace{-0.3cm}
% \end{wrapfigure}

\noindent
The concerning frequencies $f_{i} \neq f_{i'}$ have an $\ell_{2}$-distance $\|f_{i} - f_{i'}\|_{2} \geq \eta$ and thus, the above ``eggshell'' is thick enough. That is, the random difference vector $\Sigma (f_{i} - f_{i'})$ is distributed on a large enough support. After rounding, the $\mathrm{frac}(\Sigma (f_{i} - f_{i'}))$ is distributed {\em almost uniformly} within the unit hypercube $[0, 1)^{d}$, and the collision probability roughly equals the volume $1 / \mathcal{B} = \Theta_{d}(k^{-1})$ of a single bin. The parameter $\hat{\beta}$ is set carefully, to ensure a small collision probability $\leq 0.01 \cdot k^{-1}$. Hence, the matrix $\Sigma$ is likely to isolate at least 90\% head frequencies.

The vector $b$ serves as the ``anchor point'' of the hashing scheme $\mathpzc{h}$. Independent of $\Sigma$, we just sample a uniform $b \sim \unif[0, 1)^{d}$ from the unit hypercube. Then due to Eq.~\eqref{eq:intro_permutation:2}, a certain frequency $f \in [-F, F]^{d}$ is equally likely to be hashed into one of the $\mathcal{B} = B^{d} = \Theta_{d}(k)$ bins.

% The sampling scheme for the vector $a \in \R^{d}$ is rather complicated and is elaborated below.

\begin{comment}

\begin{figure}[t]
    \centering
    \begin{minipage}[c]{0.6\textwidth}
    \subfloat[Sampling for $\Delta_{a}$
    \label{fig:our_technique:permutation_hashing:sampling_delta}]{
    \includegraphics[width = .4\textwidth]{sample_delta.pdf}
    }
    \hskip .3in
    \subfloat[Sampling for $a$ and $a'$
    \label{fig:our_technique:permutation_hashing:sampling_a}]{\includegraphics[width = .45\textwidth]{sample_a.pdf}}
    \end{minipage}
    \hskip .25in
    \begin{minipage}[c]{0.325\textwidth}
    \caption{Demonstration of the sampling for $\Delta_{a}$ (the pink region in Figure~\ref{fig:our_technique:permutation_hashing:sampling_delta}), and then $a$ (the red region in Figure~\ref{fig:our_technique:permutation_hashing:sampling_a}) and $a'$ (the blue region in Figure~\ref{fig:our_technique:permutation_hashing:sampling_a}).}
    \label{fig:our_technique:permutation_hashing:sampling}
    \end{minipage}
\end{figure}

\end{comment}

\vspace{.1in}
\noindent
{\bf The ``eggshell'' sampling scheme.}
As Figure~\ref{fig:our_technique:permutation_hashing:sampling} shows, the vector $a \in \R^{d}$ is sampled {\em non-uniformly}, which differs from all the previous sampling schemes \cite{hikp12a,ik14,ps15,k16,k17,ckps16,nsw19}. Recall that this vector $a$ rotates any magnitude $\hat{x}(f) \in \C$ by a certain angle $-2 \pi \cdot f^{\top} a \in \R$ (see Eq.~\eqref{eq:intro_permutation:2}). Let $S = \{f \in \textsc{Tail}: \mathpzc{h}(f) = j\}$ be the tail frequencies hashed into a certain bin $j \in [B]^{d}$, then we hope a small total rotated magnitude
\[
    \mbox{$| \int_{f \in S} \hat{x}(f) \cdot e^{-2 \pi \i \cdot f^{\top} a} \cdot \d f |
     \ll  (\int_{f \in S} | \hat{x}(f) |^{2} \cdot \d f)^{1 / 2}$}.
\]

In the continuous model, the vector $a \in \R^{d}$ represents a sampling time point $t \in [0, T]^{d}$. We must sample this time point {\em almost} (but not exactly) uniformly from a {\em constant proportion} of the duration, such as $a \sim \unif[\frac{0.01}{d} \cdot T, (1 - \frac{0.01}{d}) \cdot T]^{d}$. This is due to the following two reasons.
\begin{itemize}
    \item Recall that the noise level $\N^{2}$ involves the term $\|g\|_{T}^{2} = \frac{1}{T^{d}} \cdot \int_{t \in [0, T]^{d}} |g(t)|^{2} \cdot \d t$, but we have no guarantee on the noise $g(t)$ at a specific time point $t \in [0, T]^{d}$. If the sampling range $A \ni a$ is too small (namely $|A| \ll T^{d}$), the average noise $\frac{1}{|A|} \cdot \int_{t \in A} |g(t)|^{2} \cdot \d t \gg \|g\|_{T}^{2}$ can be intolerably large, and makes the samples $a \in A$ useless.

    \item Unlike the discrete case, where the {\em on-the-grid} frequencies are perfectly separated, two ``continuous'' frequencies $f \neq f' \in [-F, F]^{d}$ can be arbitrarily close (when not both of $f, f'$ are head frequencies). If $\|f - f'\|_{2} \ll 1 / (\sqrt{d} \cdot T)$ and $\hat{x}(f) = \hat{x}(f')$, then over the whole duration $t \in [0, T]^{d}$ (i.e., $\|t\|_{2} \leq \sqrt{d} \cdot T$) the two signals are always close $\hat{x}(f) \cdot e^{-2 \pi \i \cdot f^{\top} t}  \approx  \hat{x}(f') \cdot e^{-2 \pi \i \cdot f^{'\top} t}$. To distinguish the frequencies $f \neq f'$, sampling the $a$ nearly from the whole duration achieves the best we can.
\end{itemize}

We often sample a pair of $a, a' \in [0, T]^{d}$ and consider their difference $\Delta_{a} := (a' - a)$ rather than $a, a'$ themselves. Over the difference vector $\Delta_{a} \in \R^{d}$, a signal with frequency $f \in \R^{d}$ rotates by an angle $2 \pi \cdot f^{\top} \Delta_{a} \in \R$. Denote by $\|\theta\|_{\bigcirc} := \min_{z \in \Z} |\theta + 2 \pi z|$ the ``circular distance''. Our actual observation would be the circular distance $\| 2 \pi \cdot f^{\top} \Delta_{a} \|_{\bigcirc} \in [0, \pi]$.

To distinguish this frequency $f \in \R^{d}$ from the others, and to recover $f \in \R^{d}$ more accurately, we need a largest possible $\ell_{2}$-norm $\|\Delta_{a}\|_{2}$. Moreover, because we do not know the direction of the frequency $f \in \R^{d}$ (or the direction of the difference between $f$ and the interim estimation of it), the sampled $\Delta_{a} \in \R^{d}$ must have a uniformly random direction.

% First, r
% Second, u

% We reserve the {\em boundary} duration $t \notin [\frac{0.01}{d} \cdot T, (1 - \frac{0.01}{d}) \cdot T]^{d}$ for another purpose. Apart from the sample at the time point $t = a$, we access one more sample at a shifted time point $t = a + \Delta^{a}$ (under the same $\Sigma$ and $b$). Using the samples $\hat{\mathcal{P} x}(\mathrm{frac}(\Sigma f - b))$ at both time points, we infer the random phase $2 \pi f^{\top} \Delta^{a} \pmod{2 \pi}$ according to Eq.~\eqref{eq:intro_permutation:2}. This helps us to learn the frequency $f \in [-F, F]^{d}$. (Such approaches for frequency location are first introduced by \cite{hikp12a}, and then developed by \cite{ps15,ckps16}.) In order to leave room for the shift $\Delta^{a} \in \R^{d}$, we must reserve the boundary duration $t \notin [\frac{0.01}{d} \cdot T, (1 - \frac{0.01}{d}) \cdot T]^{d}$ at the beginning.

The above two requirements for the time difference $\Delta_{a} = a' - a$ can violate our previous requirement that, both time points $a, a'$ shall be sampled almost uniformly from a constant proportion of the duration $t \in [0, T]^{d}$. In particular, the dimensionality $d \geq 2$ incurs many technical issues. To overcome these challenges, we sample $a, a'$ in a {\em coupling} fashion. We first determine the time difference $\Delta_{a}$, making it have a uniform random direction. Moreover, the $\ell_{2}$-norm $\|\Delta_{a}\|_{2}$ cannot be too large; otherwise, we cannot ensure that the sampling ranges $A \ni a$ and $A' \ni a'$ are large enough, namely $|A| \eqsim |A'| \eqsim T^{d}$. Both the sampling range of the $\ell_{2}$-norm $\|\Delta_{a}\|_{2}$, and the sampling scheme for $a, a' \in [0, T]^{d}$ (given a specific $\Delta_{a}$) are carefully chosen.

In contrast, suppose we sample two uniform random $a, a' \sim \unif[0, T]^{d}$, then the time difference $\Delta_{a}$ has a {\em non-uniform direction}. So the observed circular distance $\| 2 \pi \cdot f^{\top} \Delta_{a} \|_{\bigcirc}$ will follow a more complicated distribution, being hard to analyze. More importantly, both the true observations $\| 2 \pi \cdot f^{\top} \Delta_{a} \|_{\bigcirc}$ and the ``fake'' observations $\| 2 \pi \cdot f'^{\top} \Delta_{a} \|_{\bigcirc}$ (due to other frequencies $f' \neq f$) may concentrate in a small range like $[0, \frac{\pi}{100}]$. Then, we can't distinguish $f \neq f'$. This issue does not exist in the one-dimensional continuous case or the discrete case:
\begin{itemize}
    \item In the one-dimensional continuous case, $\Delta_{a}$ is just a random number instead of a vector. We need not concern the {\em direction} of $\Delta_{a}$, let alone whether this direction is uniform random.

    \item In the multi-dimensional discrete case, the frequencies are on-the-grid. Thereby, the observed circular distance just has finite possibilities, e.g., $\{0, \frac{1}{N} \cdot \pi, \cdots, \frac{N - 1}{N} \cdot \pi, \pi\}$. It turns out that we can easily distinguish true observations from fake observations.
\end{itemize}
For more details about the sampling scheme, the reader can refer to Section~\ref{sec:locate_inner_time_points}.

% $\{f_{i}\}_{i \in [k]} \subseteq [-F, F]^{d}$

\subsubsection{Sparse recovery}
\label{sec:our_technique:recovery}

\noindent
{\bf The obstacles.}
Using the hash functions and the filters, several kinds of recovery algorithms have been developed in the literature. Again, the main task is to recover the head frequencies $\{f_{i}\}_{i \in [k]}$, and the continuous model is harder since the estimations $f_{i}' \approx f_{i}$ are limited to some precision.

Similar to the past work \cite{ps15}, we use a {\em voting-based} algorithm. Roughly speaking, \cite{ps15} handles the one-dimensional case as follows: twist the frequency domain $[-F, F]$, partition it into $\Theta(k)$ sub-regions, and vote for the {\em probably approximately correct} sub-region(s). Although simple in spirit, generalizing this idea to a higher dimension $d \geq 2$ incurs many new challenges.\footnote{Some of these challenges do not exist (or are less severe) in the discrete model \cite{hikp12a,ik14}, because the twist of the discrete frequency domain $[n]^{d}$, under an appropriate modulo operation, is still itself.} For example, the twist of a hypercube $[-F, F]^{d}$ is complex (but the twist of $[-F, F]$ is just an interval), so a more sophisticated partition scheme is required.
Moreover, since we consider the $\ell_{2}$-distances among $f_{i}$'s %(Assumption~\ref{ass:frequency_separation}) 
but the domain $[-F, F]^{d}$ is a $\ell_{\infty}$-ball, switching between the $\ell_{2}$-/$\ell_{\infty}$-norms raises more technical difficulties. (However, this switch follows automatically in one dimension $d = 1$.)

En route to the final algorithm, we will address some of these challenges.
\iffalse
Even so, we just get a factor-$\poly(d)$ approximation for the frequency/tone estimation. The loss comes from three places: (i)~the coarse-grained location, (ii)~the fine-grained location, and (iii)~the tone estimation. We will also discuss why this $\poly(d)$ factor seems to be the limit of the current methods.
\fi

\vspace{.1in}
\noindent
{\bf Our approach.}
For ease of presentation, we will restrict our attention to a tone $(v_{i}, f_{i}) \in \C \times \R^{d}$ that is isolated by the permutation $\mathcal{P}$ and hashing $\mathpzc{h}$. According to Eq.~\eqref{eq:intro_permutation:2}, a sampling time point $a \in [0, T]^{d}$ gives a measurement $y_{i}(a) \in \C$ such that $y_{i}(a) \approx v_{i} \cdot \det(\Sigma)^{-1} \cdot e^{-2 \pi \i \cdot f_{i}^{\top} a}$.
Here, the ``$\approx$'' notation hides a small error, which stems from the noise frequencies (i.e., $g(t) \in \C$) hashed into the same bin $j := \mathpzc{h}(f_{i}) \in [B]^{d}$.
%Our task is to recover the frequency $f_{i} \in [-F, F]^{d}$, via a bunch of select time points and the resulting measurements.
% in this and both past works \cite{hikp12a,ps15}
To recover the frequency $f_{i}$, the idea is to leverage the difference $\Delta_{a} := (a' - a)$ between two time points $a, a' \in [0, T]^{d}$ and the relative phase
\begin{align}
    \mbox{$\psi_{i}(a, a')
    ~ := ~ \arg(y_{i}(a) / y_{i}(a'))
    ~ \approx ~ \arg(e^{2 \pi \i \cdot f_{i}^{\top} \Delta_{a}})
    ~ = ~ 2 \pi \cdot f_{i}^{\top} \Delta_{a}$}.
    \label{eq:tech_recovery:1}
\end{align}
The above ``$\approx$'' notation hides an error phase of, say, $\pm (2 \pi) / 10^{3}$.

We recover the frequency $f_{i}' \approx f_{i}$ in two steps. First, the ``coarse-grained'' location (Algorithm~\ref{alg:locate_inner}) keeps track of a hypothesis region $\mathcal{H}_{i} \ni f_{i}$ for the frequency (e.g., at the beginning $\mathcal{H}_{i} = [-F, F]^{d}$) and shrinks $\mathcal{H}_{i}$ round by round, and get the rough location of the frequencies in the end.
Second, after receiving the ``coarse-grained''
location $\mathcal{H}_{i}$, the ``fine-grained'' locating (Algorithm~\ref{alg:locater_inner_stronger}) carefully derives $d$ linear equations of the form $2 \pi \cdot f_{i}'^{\top} \Delta_{a}^{r} = \psi_{i}^{r}$ (for all $r \in [d]$) based on $d$ time differences $\Delta_{a}^{r} \in \R^{d}$, and solves these linear equations to find $f_{i}' \approx f_{i}$ within the hypothesis region $\mathcal{H}_{i}$.

%, we find $d$ linear measurements, conditioned on they're all correct and then invert a linear equation.
%we find $d$ linear measurements, conditioned on they're all correct and then invert a linear equation.

\vspace{.1in}
\noindent
{\bf Coarse-grained location via partition and voting in high dimension.}
Suppose that a frequency $f_{i}$ locates in some hypothesis region $\mathcal{H}_{i}$. We carefully divide $\mathcal{H}_{i} = \bigcup_{q \in Q} \mathcal{H}_{i, q}$ into smaller sub-regions and pick a candidate frequency $\xi_{q}$ for each sub-region.
%(E.g., in the one-dimensional case $d = 1$, the work \cite{ps15} divides an interval $\mathcal{H}_{i}$ into shorter intervals, and picks the midpoints of the sub-regions as $\xi_{q}$.)
The frequency $f_{i}$ locates in a unique {\em true} sub-region $\mathcal{H}_{i, q^*}$. Based on the measurements, we can prune some of the {\em wrong} sub-regions $\mathcal{H}_{i, q} \not\ni f_{i}$ and get a {\em smaller} new hypothesis region. As Figure~\ref{fig:our_technique:recovery:frequency_location} shows, the coarse-grained location repeats this pruning process.

\iffalse
\begin{figure}[t]
    \centering
    \begin{minipage}[c]{0.45\textwidth}
    \includegraphics[width = \textwidth]{locate_fre.pdf}
    \end{minipage}
    \hskip .25in
    \begin{minipage}[c]{0.435\textwidth}
    \caption{Demonstration for the {\bf coarse-grained location} in two dimensions $d = 2$. The black points refer to the true frequencies. The blue/green/red circles show that we gradually shrink the hypothesis regions for the frequencies.}
    \label{fig:our_technique:recovery:frequency_location}
    \end{minipage}
\end{figure}
\fi

\begin{figure}
    \centering
    \includegraphics[width = .6\textwidth]{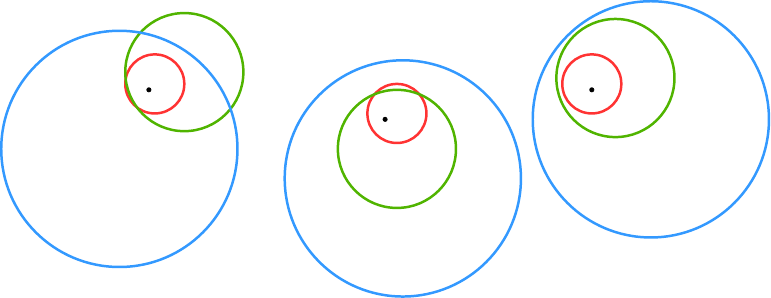}
    \caption{\small Demonstration for the {\bf coarse-grained location} in two dimensions $d = 2$. The black points refer to the true frequencies. The blue/green/red circles show that we gradually shrink the hypothesis regions for the frequencies.}
    \label{fig:our_technique:recovery:frequency_location}
    % \vspace{-2mm}
\end{figure}

% \textcolor{red}{Comment the figure.}
% \begin{wrapfigure}{R}{.45\textwidth}
% \vspace{-4mm}
% \includegraphics[width = .4\textwidth]{locate_fre.pdf}
% \caption{\small Demonstration for the {\bf coarse-grained location} in two dimensions $d = 2$. The black points refer to the true frequencies. The blue/green/red circles show that we gradually shrink the hypothesis regions for the frequencies.}
%     \label{fig:our_technique:recovery:frequency_location}
%     \vspace{-2mm}
% \end{wrapfigure}

Given a pair of sampling time points $a, a' \in [0, T]^{d}$, in view of Eq.~\eqref{eq:tech_recovery:1}, we will vote for every candidates frequency $\xi_{q}$ that satisfies
\begin{align}
    \mbox{$\| 2 \pi \cdot \xi_{q}^{\top} \Delta_{a} - \psi_{i}(a, a') \|_{\bigcirc}
    ~ \leq ~ (2 \pi) / 50$},
    \label{eq:tech_recovery:4}
\end{align}
where the $\RHS$ can be other suitable thresholds. By doing so, (i)~the {\em true} candidate frequency $\xi_{q^*}$ (for which $\mathcal{H}_{i, q^*} \ni f_{i}$) gets a vote with probability $90\%$, since $\xi_{q^*}$ is close enough to $f_{i}$. In contrast, (ii)~if a {\em wrong} candidate frequency $\xi_{q}$ (for which $\mathcal{H}_{i, q} \not\ni f_{i}$) is too far from $f_{i}$, then we hope $\xi_{q}$ to get a vote with probability $< 50\%$. Given Eq.~\eqref{eq:tech_recovery:1} and \eqref{eq:tech_recovery:4}, the wrong candidate frequency $\xi_{q}$ loses a vote when $\| 2 \pi \cdot (\xi_{q} - f_{i})^{\top} \Delta_{a} \|_{\bigcirc} \geq (2 \pi) / 40$. Namely, with probability $> 50\%$, we hope the gap between $(\xi_{q} - f_{i})^{\top} \Delta_{a} \in \R$ and its closest integer to be at least
\begin{align}
    \mbox{$\min_{z \in \Z} | (\xi_{q} - f_{i})^{\top} \Delta_{a} - z|
    ~ \geq ~ 1 / 40$},
    \label{eq:tech_recovery:5}
\end{align}
To this end, the time difference $\Delta_{a} \in \R^{d}$ is sampled to have a uniform random direction and a random $\ell_{2}$-norm $\|\Delta_{a}\|_{2} \sim \unif[w, 2 w]$, for some $w > 0$. In any dimension $d \geq 2$, we have
\begin{align}
    \mbox{$(\xi_{q} - f_{i})^{\top} \Delta_{a}
    ~ = ~ \|\xi_{q} - f_{i}\|_{2} \cdot \|\Delta_{a}\|_{2} \cdot \cos(\gamma)$},
    \label{eq:tech_recovery:2}
\end{align}
where the random angle $\gamma := \langle \xi_{q} - f_{i}, \Delta_{a} \rangle$. Clearly, when a fixed $|\cos(\gamma)| \in [0, 1]$ (namely a fixed direction of $\Delta_{a}$) is not too small, a large enough sampling range for the $\ell_{2}$-norm $\|\Delta_{a}\|_{2} \sim \unif[w, 2 w]$ ensures Eq.~\eqref{eq:tech_recovery:5} with probability $> 50\%$. This is exactly what we desire.

% \noindent
% {\bf Limit of voting process.}
Nonetheless, the coarse-grained location recovers the frequencies by at most $\|\xi_{q^*} - f_{i}\|_{2} \lesssim d / T$ (instead of $\|\xi_{q^*} - f_{i}\|_{2} \lesssim 1/T$). When the difference $\Delta_{a} = (a' - a)$ has a uniform random direction, the angle $\gamma \in [0, \pi]$ concentrates within the range $\pi / 2 \pm \pi / (2 \sqrt{d})$, so with high probability we have $|\cos(\gamma)| \lesssim 1 / \sqrt{d}$. Given Eq.~\eqref{eq:tech_recovery:5} and \eqref{eq:tech_recovery:2}, in order to vote for a wrong candidate frequency $\xi_{q}$ with probability $< 50\%$, we require $\|\xi_{q} - f_{i}\|_{2} \cdot \| \Delta_{a} \|_{2}  \gtrsim  \sqrt{d}$.

Given a specific $\Delta_{a} \in \R^{d}$, the largest possible range from which we sample the two time points $A \ni a, a'$, has the volume $|A| = T^{d} \cdot (1 - \|\Delta_{a}\|_{1} / T )$. As mentioned (Section~\ref{sec:our_technique:permutation_hashing}), this range $A$ must be a constant proportion of the whole duration $t \in [0, T]^{d}$, which requires $\|\Delta_{a}\|_{1} \lesssim T$. However, when $\Delta_{a} \in \R^{d}$ has a uniform random direction, with high probability we have $\| \Delta_{a} \|_{1} \eqsim \sqrt{d} \cdot \| \Delta_{a} \|_{2}$. Thus, it is required that $\|\Delta_{a}\|_{2} \lesssim T / \sqrt{d}$.

Putting the above arguments together gives $\| \xi_q - f_i \|_2 \gtrsim \sqrt{d}/ \| \Delta_a \|_2 \gtrsim d / T$. Namely, we can not recover the frequency $f_{i} \in [-F, F]^{d}$ too well by the coarse-grained location, but it can provide some rough estimations.

% We reserve the {\em boundary} duration $t \notin [\frac{0.01}{d} \cdot T, (1 - \frac{0.01}{d}) \cdot T]^{d}$ for another purpose. Apart from the sample at the time point $t = a$, we access one more sample at a shifted time point $t = a + \Delta^{a}$ (under the same $\Sigma$ and $b$). Using the samples $\hat{\mathcal{P} x}(\mathrm{frac}(\Sigma f - b))$ at both time points, we infer the random phase $2 \pi f^{\top} \Delta^{a} \pmod{2 \pi}$ according to Eq.~\eqref{eq:intro_permutation:2}. This helps us to learn the frequency $f \in [-F, F]^{d}$. (Such approaches for frequency location are first introduced by \cite{hikp12a}, and then developed by \cite{ps15,ckps16}.) In order to leave room for the shift $\Delta^{a} \in \R^{d}$, we must reserve the boundary duration $t \notin [\frac{0.01}{d} \cdot T, (1 - \frac{0.01}{d}) \cdot T]^{d}$ at the beginning.

\begin{figure}[htbp]
    \centering
    \begin{tabular}[b]{c}
    \includegraphics[scale=0.9]{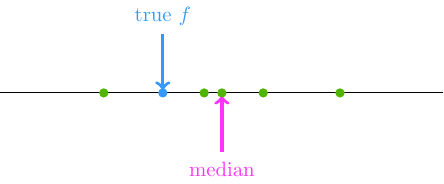}\\
    \small{(a) One dimension}
    \end{tabular}\qquad
    \begin{tabular}[b]{c}
    \includegraphics[scale=1.8]{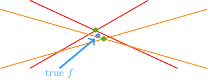}\\
    \small{(b) Two dimensions}
    \end{tabular}
    \caption{\small Demonstration of the {\bf fine-grained location} in one or two dimension(s). When $d = 2$, each observation $\phi^{r} \in \R$ gives a hypothesis line (i.e., a hypothesis one-dimensional hyper-plane) that is close to the true frequency $f$, so a pair of observations/lines determines an estimation $f' \approx f$. Given a bunch of estimations $f'$, we will find a small $\ell_{\infty}$ ball to cover a sufficient amount of estimations. Since this $\ell_{\infty}$ ball is small yet (with high probability) contains the true frequency $f$, its center must be a good enough approximation to $f$. (In one dimension, we just take the median of the estimations $f'$.)}
    \label{fig:our_technique:recovery:fine_grained}
\end{figure}

% \begin{figure}[t]
%     \centering
%     \subfloat[One dimension
%     \label{fig:our_technique:recovery:coarse_grained}]{
%     \includegraphics[scale=0.95]{take_median.pdf}
%     }
%     \hfill
%     \subfloat[Two dimensions
%     \label{fig:our_technique:recovery:fine_grained}]{
%     \includegraphics[scale=1.9]{twodim_fre.pdf}
%     \vspace{-2mm}
%     }
%     \caption{\small Demonstration of the {\bf fine-grained location} in one or two dimension(s). When $d = 2$, each observation $\phi^{r} \in \R$ gives a hypothesis line (i.e., a hypothesis one-dimensional hyperplane) that is close to the true frequency $f$, so a pair of observations/lines determines an estimation $f' \approx f$. Given a bunch of estimations $f'$, we will find a small $\ell_{\infty}$ ball to cover a sufficient amount of estimations. Since this $\ell_{\infty}$ ball is small yet (with high probability) contains the true frequency $f$, its center must be a good enough approximation to $f$. (In one dimension, we just take the median of the estimations $f'$.)}
% \end{figure}

\vspace{.1in}
\noindent
{\bf Fine-grained location via inverting robust linear system.}
The coarse-grained location recovers the frequencies up to an $\ell_{2}$-distance $\lesssim_d 1/ T$.\footnote{We say $a \lesssim_d b$ if $a \leq \poly(d) \cdot b$.} Then the fine-grained location improves this precision to $\lesssim_d  1 / (\rho T)$, where $\rho \gg 1$ is the {\em signal-to-noise ratio} (Definition~\ref{def:ratio_signal_noise}). In one dimension $d = 1$, the past work \cite{ps15} easily achieves so by first deriving a bunch of candidates $\xi_{q^*} \approx f_{i}$ from getting a few of the coarse-grained locations, and then taking the median of these $\xi_{q^*}$'s (as
Figure~\ref{fig:our_technique:recovery:fine_grained}(a) suggests). However, this idea fails in the multi-dimensional case when $d \geq 2$, and the fine-grained location becomes far more complicated.

Roughly speaking, based on a time difference $\Delta_{a}^{r} = (a' - a) \in \R^{d}$, we get an observation $\psi^{r} \in \R$ such that $\E[ |\psi^{r} - f_{i}^{\top} \Delta_{a}^{r} |^{2} ] = 0.1 / \rho^{2}$. Due to Markov's inequality,
\begin{align}
   \mbox{$\Pr[ | \psi^r - f_{i}^{\top} \Delta_a^r | \leq \sqrt{d}/ \rho ] ~ \geq ~ 1 - 0.1 / d$}.
   \label{eq:tech_recovery:6}
\end{align}
%where $f_{i}$ is the true frequency and $\Delta_a^r$ is a random vector.
When $d=1$, we can just take $\psi^r / \Delta_a^r$ as an approximation of $f_{i}$. It suffices to get a good estimation $f_{i}' \approx f_{i}$ via a small number of samples. However, when $d \geq 2$, we cannot extract enough information from the inner product $f_{i}^{\top} \Delta_a^r \in \R$ of the two vectors. To handle this issue, as Figure~\ref{fig:our_technique:recovery:fine_grained}(b) illustrates, we will use $d$ random vectors to form a random matrix $\Delta = [\Delta_{a}^{r}]_{r \in [d]} \in \R^{d \times d}$ that has a bounded spectral norm, \footnote{For a matrix $\Delta$, we use $\| \Delta \|$ to denote the spectral norm of $\Delta$.} and $d$ observations $\psi = (\psi_{r})_{r \in [d]}$. Then Eq.~\eqref{eq:tech_recovery:6} implies that
\begin{align*}
    \mbox{$\Pr[ \| \psi - \Delta f_{i} \|_2 \leq \sqrt{d} \cdot \sqrt{d} / \rho ] ~ \geq ~ 1 - 0.1
    ~ = ~ 0.9$}.
\end{align*}
This gives a good estimation $\Delta^{-1} \psi \approx f_{i}$ with $\| \Delta^{-1} \psi - f_{i} \|_2 \leq \| \Delta^{-1} \| \cdot \| \psi - \Delta f_{i} \|_2 \leq \| \Delta^{-1} \| \cdot (d/\rho)$.  (For a illustration, see Figure~\ref{fig:locate_inner_stronger} in Section~\ref{sec:locate_inner_stronger}.) This approach needs $d$ observations, and the estimation error must be amplified by a $\sqrt{d}$ factor to enable the union bound.

% Also, it is noteworthy that $\E[|\psi^r -  \Delta_a^r \cdot f |]=1/\rho$, but to get a good estimation by combining $d$ samples, we blow up a factor $\sqrt{d}$ to get a good union bound. 

% Given our framework, the $\ell_{2}$-norm $\|\Delta_a^r \|_2$ is at most $O(T/d)$. 
To get a more accurate estimation, our new sampling method discussed before ensures that $\Delta_{a}^{r}$ $\|\Delta_a^r\|_2 \gtrsim T / d$. One additional issue is how to analyze the random matrix $\Delta = [\Delta_{a}^{r}]_{r \in [d]} \in \R^{d \times d}$. Fortunately, one can show the vectors $\Delta_a^r\cdot \sqrt{d} / \|\Delta_a^r\|_2$ are sub-Gaussian isotropic, so we can upper bound the spectral norm $\| \Delta^{-1} \|$.
%after loosing a $d$ factor in the approximation. 
Combining everything and solving the robust linear systems gives $\| f_{i}' - f_{i} \|_2 = \| \Delta^{-1} \psi - f_{i} \|_2 \lesssim_d 1 / (\rho T)$.

\newpage
\section*{Roadmap}
Section~\ref{sec:prelim} provides some basic notations and definitions. Section~\ref{sec:prob} provides a list of probability tools.
Filter, permutation and hashing in one dimension are given in Appendix~\ref{sec:filter_function_single} and ~\ref{sec:HashToBins_single} for completeness, which can be skipped if readers are familiar with them. 
Section~\ref{sec:HashToBins_multi} presents the counterpart filter, permutation and hashing in the multi-dimensional setting. In Section~\ref{sec:locate_inner} and \ref{sec:locate_signal}, we show to how to give accurate estimations of the frequencies. In Section~\ref{sec:sparse_recovery}, we present our sparse recovery algorithm. In Section~\ref{sec:convert}, we show how to obtain the signal estimation by paying a slightly longer duration.
Finally, in Section~\ref{sec:conclusion}, we give a short discussion on some bottlenecks of current methods, and some interesting future directions.

\section{Preliminaries}\label{sec:prelim}

\subsection{Notations}

We denote by $[n]$ the set $\{0, 1, 2, \cdots, n - 1\}$, by $\R$ the set of real numbers, by $\mathbb{Z}$ the set of integers, and by $\mathbb{C}$ the set of complex numbers. Also, $\mathbb{N}_{\geq a}$ refers to the set of integers no less than $a \geq 0$. Let $\supp(f)$ denote the support of a function or vector $f$, and let $\| f\|_0 = | \supp(f) |$ be the cardinality. For a random variable $X$, for convenience we may abuse the notation $\supp(X)$ to denote the support of $X$'s probability density function (PDF).

We use $\max \{a, b\}$ or $\max (a, b)$ (resp.\ $\min \{a, b\}$ or $\min (a, b)$) to denote the maximum (resp.\ the minimum) between $a, b \in \R$. Given any $p \geq 1$, a vector $x = (x_{i})_{i \in [n]} \in \R^{n}$ has the the $\ell_{p}$-norm $\| x \|_{p} := ( \sum_{i \in [n]} |x_i|^p )^{1/p}$; in the case that $p = \infty$, we define $\| x \|_{\infty} := \max_{i \in [n]} | x_{i} |$.

We use the notations $\i := \sqrt{-1}$ and $e^{\i \theta} := \cos ( \theta ) + \i \cdot \sin (\theta)$ for any phase $\arg(e^{\i \theta}) = \theta \in \R$. For a complex number $z = a + \i \cdot b \in \mathbb{C}$, let $a \in \R$ be the real part and let $b \in \R$ be the imaginary part. Also, $\ov{z} := a - \i b \in \C$ denotes the conjugate, and $|z| := \sqrt{z \ov{z}} = \sqrt{a^{2} + b^{2}} \geq 0$ denotes the norm.

\subsection{Fourier transform and convolution}

For convenience, throughout this paper we use the shorthand {\CFT} (the continuous Fourier transform), {\DFT} (the discrete Fourier transform), {\DTFT} (the discrete-time Fourier transform) and {\FFT} (the fast Fourier transform).
\begin{itemize}
    \item In the time domain, we often use the notations $t$ and $\tau$.
    
    \item In the frequency domain, we often use the notations $f$ and $\xi$.
\end{itemize}
Given a $d$-variate function $x(t)$ for $t = (t_{s})_{s \in [d]} \in \R^{d}$, we have the {\CFT} $\hat{x}(f)$ for $f = (f_{r})_{r \in [d]} \in \R^{d}$ and the inverse {\CFT} $x(t)$ for $t \in \R^{d}$:
\begin{align*}
    & \hat{x}(f) ~ := ~ \int_{\tau \in \R^{d}} x(\tau) \cdot e^{-2 \pi \i \cdot f^{\top} \tau} \cdot \d \tau
    && \mbox{and}
    && x(t) ~ := ~ \int_{\xi \in \R^{d}} \hat{x}(\xi) \cdot e^{2 \pi \i \cdot t^{\top} \xi} \cdot \d \xi.
\end{align*}

\begin{definition}[$k$-Fourier-sparse signal]
\label{obs:tones_CFT}
Given any $k$-Fourier-sparse signal $x^*(t)$ with the tones $\{(v_{i}, f_{i})\}_{i \in [k]} \subseteq \C \times \R^{d}$, the corresponding {\CFT} $\hat{x^*}(f)$ is the combination of $k \geq 1$ many (scaled) $d$-dimensional Dirac delta functions, each of which has a point mass (i.e.\ the involved magnitude) $v_{i} \in \C$ at the corresponding frequency $f_{i} \in \supp(\hat{x^*})$. Without ambiguity, we denote $\hat{x^*}[f_{i}] := v_{i} \in \C$ for convenience. Then the $k$-sparse Fourier spectrum $\hat{x^*}(f)$ for $f \in \R^{d}$ can be formulated as
\begin{align*}
    \hat{x^*}(f)
    ~ := ~ \sum_{i \in [k]} v_{i} \cdot \Dirac_{= f_{i}}(f)
    ~ = ~ \sum_{i \in [k]} \hat{x^*}[f_{i}] \cdot \Dirac_{= f_{i}}(f).
\end{align*}
\end{definition}

\begin{definition}[Convolution]
\label{def:convolution}
The convolution $(f * g)(t)$ for $t \in \R^{d}$ of two $d$-variate continuous function $f(t)$ and $g(t)$ is given by
\[
    (f * g)(t) ~ := ~ \int_{\tau \in \R^{d}} f(\tau) \cdot g(t - \tau) \cdot \d \tau,
\]
And the discrete convolution $(f * g)[i]$ for $i \in \Z$ of two same-length vectors $f$ and $g$ is given by\footnote{We define $(f * g)[i] := 0$ in the case that $i \notin \supp(f) = \supp(g)$.}
\[
    (f * g)[i] = \sum_{j \in \Z} f[j] \cdot g[i - j].
\]
\end{definition}

\subsection{An overview of previous techniques}
\label{sec:previous}

The Sparse FT problem falls into the ``{\em sparse recovery}'' paradigm. Among such problems, an exemplar is to learn an {\em approximately $k$-sparse} length-$N$ vector $\hat{y} \in \R^{N}$, by just accessing the length-$N'$ {\em measurements} $y := \Phi \hat{y}$ resulted from an amount of $N'$-to-$N$ {\em sensing matrices} $\Phi \in \R^{N' \times N}$, for some $N' \ll N$. Based on the measurements, an algorithm should output a {\em $k$-sparse} vector $\hat{y}' \in \R^{N}$ that approximates the vector $\hat{y}$. E.g., under the $\ell_{2} / \ell_{2}$ guarantee, we aim at achieving
\begin{align*}
    \mbox{$\| \hat{y}' - \hat{y} \|_2 ~ \lesssim ~ \min_{k-\text{sparse}~z} \| z - \hat{y} \|_2$}.
\end{align*}
Given the flexibility of designing the $\Phi$'s, the above problem is known as {\em compressed sensing}, and the optimization goals are threefold: (i)~to access the fewest measurements, i.e., sample complexity\footnote{Only in the literature on compressed sensing, sample complexity is often called the {\em number of measurements}.}; (ii)~to fast extract the $k$-sparse approximation $\hat{y}' \approx \hat{y}$, i.e., decoding time; and (iii)~to use column-sparsest possible $\Phi$'s, hence a faster encoding time.\footnote{Optimizing encoding time only makes sense when we are allowed to design the sensing matrix, for more details of encoding time, we refer the readers to \cite{ns19}.}

We instead face the (discrete) sparse Fourier transform problem, if the above vector $\hat{y} \in \R^{N}$ is replaced by a length-$N$ Fourier spectrum $\hat{x} \in \C^{N}$ (of any dimension $d \geq 1$) and the measurements $y$ are replaced by the signal samples $x \in \C$. Again, the Fourier spectrum $\hat{x} \in \C^{N}$ is unknown, and we can only leverage the signal samples $x \in \C$. Now our optimization goals are to reduce the sample complexity and the decoding/running time.

\vspace{.1in}
\noindent
{\bf Compressed sensing.}
To leverage the measurements, several past works on compressed sensing \cite{glps10,dipw10,ip11,ipw11,birw16,ns19} first get a bunch of pseudorandom {\em hash functions} $\mathpzc{h}: [N] \mapsto [\mathcal{B}]$, where $\mathcal{B} = \Theta_{d}(k)$ is the number of bins. Such a ``hashing'' is associated with a random sign
function $s: [m] \mapsto \{\pm 1\}$.\footnote{Some previous works use the random Gaussian instead of the random sign functions.} In one hashing, we derive the linear combination of the form 
\begin{align}
\label{eq:previous_tech:1}
    \mbox{$u_j ~ := ~ \sum_{i \in [m]: \mathpzc{h}(i) = j} y_{i} ~ \cdot ~ s(i)$},
\end{align}
for every bin $j \in [\mathcal{B}]$, based on a certain amount of $m = o(N)$ measurements $\{y_{i}\}_{i \in [m]} \subseteq \R^{N'}$. This scheme is known as ``{\em hashing into $\mathcal{B}$ bins}''. Following such ideas, $O(k \log(n / k))$ samples suffice to get a desired $k$-sparse approximation $\hat{y}' \approx \hat{y}$ \cite{glps10,ns19}.
%Besides random binary matrix, there are some works using random gaussian matrix such as \cite{crt06,wai07} and their follow-ups.

\vspace{.1in}
\noindent
{\bf Discrete Fourier transform.}
The very first obstacle to adopting a compressed sensing algorithm to the discrete Sparse FT problem is, how to implement the ``hashing into $\mathcal{B}$ bins'' scheme by using the Fourier samples. Now we observe the signal $x$ in the time domain, but aim to recover its Fourier spectrum $\wh{x} \in \C^{N}$ in the frequency domain.

The approach in the past works \cite{hikp12a,ik14,k16,k17} is to mimic the transformation in Eq.~\eqref{eq:previous_tech:1}. That is, we first permute a bunch of $m = o(N)$ signal samples $\{x_{i}\}_{i \in [m]}$ via a pseudorandom affine permutation $\mathcal{P}$. Then, the permuted samples $\{(\mathcal{P} x)_{i}\}_{i \in [m]}$ are respectively scaled by coefficients $\{\mathcal{G}(l_{i})\}_{i \in [m]}$, i.e., the values of a filter function $\mathcal{G}: \R^{d} \mapsto \R$ at $m = o(N)$ many {\em lattice points} $\{l_{i}\}_{i \in [m]} \subseteq \R^{d}$. Akin to Eq.~\eqref{eq:previous_tech:1}, we use a transformation $u_{j} = \sum_{i \in [m]: \mathpzc{h}(i) = j} (\mathcal{P} x)_{i} \cdot \mathcal{G}(l_{i})$.

The second difficulty is that the hashing is no longer perfect. For compressed sensing, a coordinate $i \in [N]$ contributes $100\%$ to a target bin, and $0\%$ to the other $(\mathcal{B} - 1)$ bins. For the discrete Fourier transform, however, besides the target bin (which still gets $100\%$), any other bin should get a $\delta > 0$ fraction of mass from a coordinate $i \in [N]$. This modification (a.k.a.\ ``leakage'' \cite{ik14}) is to make the ``hashing into $\mathcal{B}$ bins'' efficient. Because of the imperfect hashing, the current sample complexity must involve an extra $\log(1 / \delta)$ factor.

To get a better sense, let us briefly review the techniques in \cite{ik14}. In any dimension $d \geq 1$, the frequency domain $\{\xi_{i}\}_{i \in [n^{d}]} = [n]^{d}$ is ``on-the-grid''. Partition the domain $[n]^{d} = \textsc{Head} \sqcup \textsc{Tail}$ into the head and tail frequencies (i.e., $|\textsc{Head}| = k$ and $|\textsc{Tail}| = n^{d} - k = N - k$) and denote the magnitudes by $\hat{x}[\xi_{i}] \in \C$. Roughly speaking, the permutation by \cite{ik14} works as follows:
\begin{align*}
    \mbox{$\hat{\mathcal{P} x}[\Sigma \xi_{i} - b \pmod{n}]
    ~ = ~ \hat{x}[\xi_{i}] \cdot e^{-\frac{2 \pi \i}{n} \cdot \xi_{i}^{\top} a}$},
\end{align*}
where the modulo operation is taken coordinate-wise, $\Sigma \in [n]^{d \times d}$ is a random matrix, and $b, a \in [n]^{d}$ are random vectors.
%When $d = 1$, this permutation degenerates to the ones in \cite{hikp12a,hikp12b,ikp14,cksz17,k17}.

The matrix $\Sigma \in [n]^{d \times d}$ is sampled uniformly at random among all integer matrices with {\em odd} determinants. So the inverse $\Sigma^{-1} \pmod{n}$ exists, making the permutation {\em one-to-one}. The vector $b \sim \unif[n]^{d}$ is uniform random, i.e., the ``anchor point'' of the permuted frequency domain.

Also, $\Sigma$ and $b$ together determine the hashing $\mathpzc{h}$. Since $\Sigma$ is invertible, the linear transformation $\Sigma \xi_{i} - b \pmod{n}$ forms a bijection from the ``grid'' frequency domain $\{\xi_{i}\}_{i \in [n^{d}]} = [n]^{d}$ to itself. \cite{ik14} partition the codomain $[n]^{d}$ into $\mathcal{B} = B^{d}$ isomorphic Cartesian sub-grid, each of which has $(\frac{n}{B})^{d} = \frac{N}{\mathcal{B}}$ grid points. The sub-grids are exactly the desired bins. For a uniform random ``anchor point'' $b \sim \unif[n]^{d}$, a frequency $\xi_{i} \in [n]^{d}$ is equally likely to fall into one of the bins.

Another crucial observation is that,
%because $n$ is a power of $2$ while $\Sigma$ is sampled uniformly at random among all {\em odd-determinant} matrices,
any two different frequencies $\xi_{i} \neq \xi_{i'} \in [n]^{d}$ fall into the same bin with probability $\leq 0.01 \cdot k^{-1}$ \cite{ik14}. Thus, 90\% head frequencies $\xi_{i} \in \textsc{Head}$ will not collide with other head frequencies, hence being {\em isolated}.

The above permutation samples a uniformly random vector $a \sim \unif[n]^{d}$, and thus rotates a magnitude $\hat{x}[\xi_{i}] \in \C$ by a certain angle $-(2 \pi / n) \cdot \xi_{i}^{\top} a$, i.e., the rotated magnitude $\hat{x}[\xi_{i}] \cdot e^{-\frac{2 \pi \i}{n} \cdot \xi_{i}^{\top} a} \in \C$ has a random phase. This is crucial because,
%Ideally this random phase is uniformly distributed.
given any sufficiently large subset $S \subseteq \textsc{Tail}$ of the tail magnitudes, a uniform random $a \sim \unif[n]^{d}$ makes the total rotated magnitude (over $\xi \in S$) much smaller than the sum of the individual magnitudes.
%I.e., we hope to have
%$
%    | \sum_{\xi \in S} \hat{x}[\xi] \cdot e^{-\frac{2 \pi \i}{n} \cdot \xi^{\top} a} |
%    \ll \sum_{\xi \in S} | \hat{x}[\xi] |.
%$

Let $S = \{\xi \in \textsc{Tail}: \mathpzc{h}(\xi) = j\}$ denote the tail frequencies hashed into a certain bin $j \in [B]^{d}$. Given the above discussions, the total tail magnitude $z_{j} := \sum_{\xi \in S} \hat{x}[\xi] \cdot e^{-\frac{2 \pi \i}{n} \cdot \xi^{\top} a} \in \C$ is small enough such that (i)~$z_{j} \in \C$ will not be identified as a {\em spurious} head frequency, when no head frequency is hashed into the $j$-th bin; and (ii)~$z_{j} \in \C$ will not falsify an isolated head frequency $\xi_{i} \in \textsc{Head}$ too much, when $\xi_{i} \in \textsc{Head}$ is the unique head frequency in the $j$-th bin.

Different from \cite{ik14}, some other works like \cite{k16,k17} use more complicated hash schemes, to improve the sample complexity and/or the running time of the sparse FT algorithm.

\subsection{Technical barriers against a better tone estimation duration}
\label{sec:our_technique:signal_estimation}

The claimed tone estimation guarantee (Theorem~\ref{thm:intro_tone}) requires that $T \gtrsim d^{4.5} / \eta \cdot \log(k d / \delta) \cdot \log d$. Here the $\poly(d)$ term stems from several places.
\begin{enumerate}[label = (\roman*)]
    \item We sample the time points from a large range $|\supp(a)| \eqsim T^{d}$ (Section~\ref{sec:our_technique:permutation_hashing}). Since the vector $a = (a_{r})_{r \in [d]}$ is in $d$ dimension, we need $|\supp(a_{r})| \geq T - \Theta(T / d)$ in any single dimension. The second term $\Theta(T / d)$ (rather than $\Theta(T)$) incurs a factor-$d$ loss in the duration bound.
    
    \item The procedure {\HashToBins} (Algorithm~\ref{alg:HashToBins_multi}) switches the $\ell_{2}$-norm to the $\ell_{\infty}$-norm, and thus incurs another factor-$\sqrt{d}$ loss.
    
    \item How we generate the random matrix $\Sigma \in \R^{d \times d}$ loses a $\sqrt{d}$ factor, to ensure a small collision probability $\Pr[\mathpzc{h}(f_{i}) = \mathpzc{h}(f_{i'})] \leq 0.01 \cdot k^{-1}$ for any two frequencies $f_{i} \neq f_{i'} \in \supp(\hat{x}^*)$.
    
    \item Our filter function $\mathcal{G}$ (see Appendix~\ref{sec:HashToBins_multi}) is modified from the one by \cite{ckps16}, which incurs a factor-$d$ loss in the duration bound $T$. Without the modification, the approximation factor of our algorithm would be $2^{\Theta(d)}$ rather than $\poly(d)$.

    \item To select the $k$ recovered tones $\{(v_{i}, f_{i})\}_{i \in [k]}$ from $k' = \Theta_{d}(k)$ candidate tones (Algorithm~\ref{alg:merged_stage}), we amplify the duration bound by a $(d^{1.5} \log d)$ factor. In particular, we first pay a $(d \log d)$ factor because there are $k' = 2^{\Theta(d \log d)} \cdot k$ candidate tones. Moreover, in the selection process, we cannot afford the running time to query points in $\ell_2$-space (i.e., the memberships regarding some $\ell_{2}$-regions) even with the best data structure. Instead, we will work in the $\ell_{\infty}$-space and choose the gap $\eta' = \eta / \sqrt{d}$, which incurs another factor-$\sqrt{d}$ loss.
\end{enumerate}
To sum up, we need a duration $T \gtrsim \eta'^{-1} \cdot d^3 \cdot \log(k' d / \delta) = \eta^{-1} \cdot d^{4.5} \cdot \log(k d / \delta) \cdot \log d := C_{\text{tone}} \cdot \eta^{-1}$.

% In order to reach tone estimation guarantee, we require $T \geq d^{4.5} \log(kd/\delta) \log d / \eta$. The poly $d$ dependence are from several places, 1) Filter function, we pay $d$ factor, Since we are working on $d$ dimension, in order to fix $2^d$ blow up in approximation ratio, we need to make the sample range between $[1-1/d,1]$.  2) We need to promise sampling vector $\| a\|_2 \eqsim 1$, since the vector lies in $d$ dimension, thus each dimension is roughly $\| a\|_{\infty} \eqsim 1/\sqrt{d}$. 3) In \textsc{HashtoBins}, switching $\| \cdot \|_{\infty}$ and $\| \cdot \|_2$ lossing another $\sqrt{d}$. 4) choosing $\Sigma$ lose a $d$ factor. 5). Reducing $O_d(k)$ tones to exact $k$ tones, this step needs to blow up the duration bound by an extra $d^{1.5} \log d$ factor, the $d\log d$ factor is from we have $k'=2^{d\log d}k$ candidate tones, the $\sqrt{d}$ is from the frequency gap loss $\eta' = \eta/\sqrt{d}$.
%\begin{align*}
%   T \geq \frac{d^3 \log(k'd /\delta)}{\eta'} = \frac{d^{4.5} \log(kd/\delta) \log d }{ \eta }.
%\end{align*}

% \newpage
\section{Probability tools}\label{sec:prob}

In this section, we present a number of classical probability tools to be used in this paper: the Chernoff bound (Lemma~\ref{lem:chernoff}), the Hoeffding bound (Lemma~\ref{lem:hoeffding}) and the Bernstein bound (Lemma~\ref{lem:bernstein}) measure the tail bounds of random scalar variables. 
Further, Lemma~\ref{lem:matrix_bernstein} is a concentration result about random matrices.

We state the classical Chernoff bound below, which is named after Herman Chernoff but is due to Herman Rubin. It gives exponentially decreasing bounds for the tail distributions of the sums of independent random variables.

\begin{lemma}[Chernoff bound \cite{c52}]
\label{lem:chernoff}
Let $\{X_{i}\}_{i \in [n]}$ be $n \geq 1$ independent Bernoulli random variables, such that $X_i = 1$ with probability $p_i \in [0, 1]$ and $X_i = 0$ with probability $1 - p_i$.  Then the following hold for the random sum $X := \sum_{i \in [n]} X_i$ and the expectation $\mu := \E[X] = \sum_{i \in [n]} p_i$.
\begin{description}[labelindent = 1em]
    \item [Part~(a):]
    $ \Pr[ X \geq (1+\delta) \mu ] \leq e^{\delta \cdot \mu} \cdot (1 + \delta)^{-(1 + \delta) \cdot \mu} $ for any $\delta > 0$.
    
    \item [Part~(b):]
    $ \Pr[ X \leq (1-\delta) \mu ] \leq e^{-\delta \cdot \mu} \cdot (1 - \delta)^{-(1 - \delta) \cdot \mu} $ for any $0 < \delta < 1$. 
\end{description}
\end{lemma}

We state the Hoeffding bound below:

\begin{lemma}[Hoeffding bound \cite{h63}]
\label{lem:hoeffding}
Let $\{X_{i}\}_{i \in [n]}$ be $n \geq 1$ independent random variables bounded between $\supp(X_{i}) \subseteq [a_i, b_i]$, for some $a_{i} \leq b_{i} \in \R$. Then the following holds for the random sum $X := \sum_{i \in [n]} X_{i}$ and any $t \geq 0$.
\begin{align*}
    \Pr[ | X - \E[X] | \geq t ]
    ~ \leq ~ 2 \cdot \exp \left( - \frac{2t^2}{ \sum_{i \in [n]} (b_i - a_i)^2 } \right).
\end{align*}
\end{lemma}

We state the Bernstein inequality below:
\begin{lemma}[Bernstein inequality \cite{b24}]
\label{lem:bernstein}
Let $\{X_{i}\}_{i \in [n]}$ be $n \geq 1$ independent zero-mean random variables $\E[X_{i}] = 0$. Suppose that $|X_i| \leq M$ almost surely, for every $i \in [n]$ and some $M \geq 0$. Then the following holds for the random sum $X := \sum_{i \in [n]} X_{i}$ and any $t \geq 0$.
\begin{align*}
\Pr \left[ X > t \right]
~ \leq ~ \exp \left( - \frac{ t^2/2 }{ \sum_{i \in [n]} \E[X_{i}^2]  + M t /3 } \right).
\end{align*}
\end{lemma}

Matrix concentration inequalities have various applications. Below, we state a matrix Bernstein inequality by \cite{t15}, which can be regarded as a matrix version of Lemma~\ref{lem:bernstein}.

\begin{lemma}[{Matrix Bernstein \cite[Theorem~6.1.1]{t15}}]
\label{lem:matrix_bernstein}
Let $\{ X_{i} \}_{i \in [m]} \subseteq \R^{n_{1} \times n_{2}}$ be a set of $m \geq 1$ i.i.d.\ matrices with the expectation $\E[ X_i ] = 0^{n_{1} \times n_{2}}$. For some $M \geq 0$, assume
\begin{align*}
    & \| X_i \| ~ \leq ~ M,
    && \forall i \in [m].
\end{align*}
Let $X = \sum_{i \in [m]} X_i$ be the random sum. Let $\mathrm{Var} [ X ] $ be the matrix variance statistic of the sum:
\begin{align*}
    \mathrm{Var} [X] ~ := ~ \max \left\{ \Big\| \sum_{i \in [m]} \E[ X_i X_i^\top ] \Big\| , ~~ \Big\| \sum_{i \in [m]} \E [ X_i^\top X_i ] \Big\| \right\}.
\end{align*}
Then 
\begin{align*}
    \E[ \| X \| ] ~ \leq ~ \sqrt{2 \cdot \mathrm{Var} [X] \cdot \log (n_1 + n_2)} +  \frac{M}{3} \cdot \log (n_1 + n_2).
\end{align*}
Furthermore, the following holds for any $t \geq 0$.
\begin{align*}
    \Pr[ \| X \| \geq t ] ~ \leq ~ (n_1 + n_2) \cdot \exp \left( - \frac{t^2/2}{ \mathrm{Var} [ X ] + M t /3 }  \right)  .
\end{align*}
\end{lemma}

\begin{lemma}[Sub-gaussian rows {\cite[Theorem~5.39]{ver10}}]
\label{lem:sub-gaussian rows}
Let $A$ be an $N\times n$  matrix whose rows $A_i$ for $i \in [N]$ are
independent sub-gaussian isotropic random vectors in $\R^n$. Then for every $t\geq0$, with probability  at least $1-2\exp(-ct^2)$, we have 
\begin{align*}
    \sqrt{N}-C\sqrt{n}-t\leq s_{\min}(A)\leq s_{\max}(A)\leq \sqrt{N}+C\sqrt{n}+t.
\end{align*}
where $s_{\max}(A)$(resp.\ $s_{\min}(A)$) represents the largest (resp.\ smallest) singular value of matrix $A$, and absolute constants $C=C_K$, $c=c_K$  depend only on the sub-gaussian norm $K=\max_{i\in[N]}\|A_i\|_{\psi_2}$ of the rows.
\end{lemma}
% \newpage
\section{Filter, permutation and hashing in multiple dimensions}
\label{sec:HashToBins_multi}

Different from the previous sections, in this section $t = (t_{s})_{s \in [d]} \in \R^{d}$ and $f = (f_{r})_{r \in [d]} \in \R^{d}$ will respectively denote the $d$-dimensional vectors in the time domain and in the frequency domain, and $i \in \mathbb{N}_{\geq 0}^{d}$ and $j \in \mathbb{N}_{\geq 0}^{d}$ will denote the vector indices.

\subsection{Construction of filter \texorpdfstring{$(\mathcal{G}(t), \hat{\mathcal{G}}(f))$}{}}
\label{subsec:HashToBins_multi:filter}

\begin{definition}[The multi-dimensional filter]
\label{def:shift_filter_function_multi}
Recall the parameters defined in Definition~\ref{def:shift_filter_function_single}:
\begin{itemize}
    \item The number of bins in a single dimension $B = \Theta(d \cdot k^{1 / d})$ is a certain multiple of $d \in \mathbb{N}_{\geq 1}$. Over all the $d \in \mathbb{N}_{\geq 1}$ dimensions, we have $\mathcal{B} = B^{d} = 2^{\Theta(d \log d)} \cdot k$ many bins.
    
    \item The noise level parameter $\delta \in (0, 1)$.
    
    \item $\alpha = \Theta(1 / d)$ is chosen such that $\frac{1}{100 \cdot (d + 1) \cdot \alpha} \in \mathbb{N}_{\geq 1}$ is an integer; clearly $\alpha \leq \frac{1}{100 \cdot (d + 1)} \leq \frac{1}{200}$.
    
    \item $s_1 = \frac{2 B}{\alpha}$ and $s_2 = \frac{1}{B + B / d}$.
    
    \item $\ell = \Theta(\log(k d / \delta))$ is an even integer. We safely assume $\ell \geq 1000$.
\end{itemize}
Further, the {\em width parameter} $W = \Omega(d \cdot \frac{F}{B \eta})$ is chosen to be a sufficiently large integer. Then for any $t = (t_s)_{s \in [d]} \in \R^{d}$ and any $f = (f_{r})_{r \in [d]} \in \R^{d}$, the filter function $(\mathcal{G}(t), \hat{\mathcal{G}}(f))$ is given by
\begin{align*}
    & \mathcal{G}(t) = \prod_{s \in [d]} \mathsf{G}(t_{s})
    && \mbox{and}
    && \hat{\mathcal{G}}(f) = \prod_{r \in [d]} \hat{\mathsf{G}}(f_{r}),
\end{align*}
where the single-dimensional filter $(\mathsf{G}(t_{s}), \hat{\mathsf{G}}(f_{r}))$ is constructed according to Definition~\ref{def:shift_filter_function_single}, under the same parameters $B$, $\delta$, $\alpha$, $s_1$, $s_2$, $\ell$ and $W$.
\end{definition}

\begin{definition}[Hypercube grid]
\label{def:grid}
Define
\[
    \Lambda_{W}(z) := \{f \in \R^{d}: \|f - i\|_{\infty} \leq z \text{ for some vector index } i \in [-W: W]^{d} \}.
\]
This denotes the union of all the hypercubes that (for the chosen $i$'s) have edge length $2 z \geq 0$ and are centered at $i \in [-W: W]^{d}$. Notice that $\Lambda_{W}(z) \supseteq \Lambda_{W}(z')$ for any $z \geq z' \geq 0$.
\end{definition}

\subsection{Properties of filter \texorpdfstring{$(\mathcal{G}(t), \hat{\mathcal{G}}(f))$}{}}

\begin{lemma}[The multi-dimensional filter]
\label{lem:filter_function_multi}
The filter $(\mathcal{G}(t), \hat{\mathcal{G}}(f))[B, \delta, \alpha, \ell, W]$ given in Definition~\ref{def:shift_filter_function_multi} satisfies the following:
\begin{description}[labelindent = 1em]
    \item [Property~I:]
    $e^{-\frac{\delta}{\poly(k, d)}} \cdot  \leq \hat{\mathcal{G}}(f) \leq 1$ for any $f \in \Lambda_{W}(\frac{1 - \alpha}{2 B})$.
    
    \item [Property~II:]
    $\hat{\mathcal{G}}(f) \in [0, 1]$ for any $f \in \Lambda_{W}(\frac{1}{2 B}) \setminus \Lambda_{W}(\frac{1 - \alpha}{2 B})$.
    
    \item [Property~III:]
    $0 \leq \hat{\mathcal{G}}(f) \leq \frac{\delta}{\poly(k, d)}$ for any $f \in \R^{d} \setminus \Lambda_{W}(\frac{1}{2 B})$.
    
    \item [Property~IV:]
    $\supp(\mathcal{G}) \subseteq [-\ell \cdot \frac{B}{\alpha}, \ell \cdot \frac{B}{\alpha}]^{d}$.
    
    \item [Property~V:]
    $\sum_{i \in \mathbb{Z}^{d}} \mathcal{G}(i)^2 \leq e^2 \cdot B^{-d} = e^2 \cdot \mathcal{B}^{-1}$.
\end{description}
\end{lemma}

\begin{figure}
    \centering
    \includegraphics[width = 0.8 \textwidth]{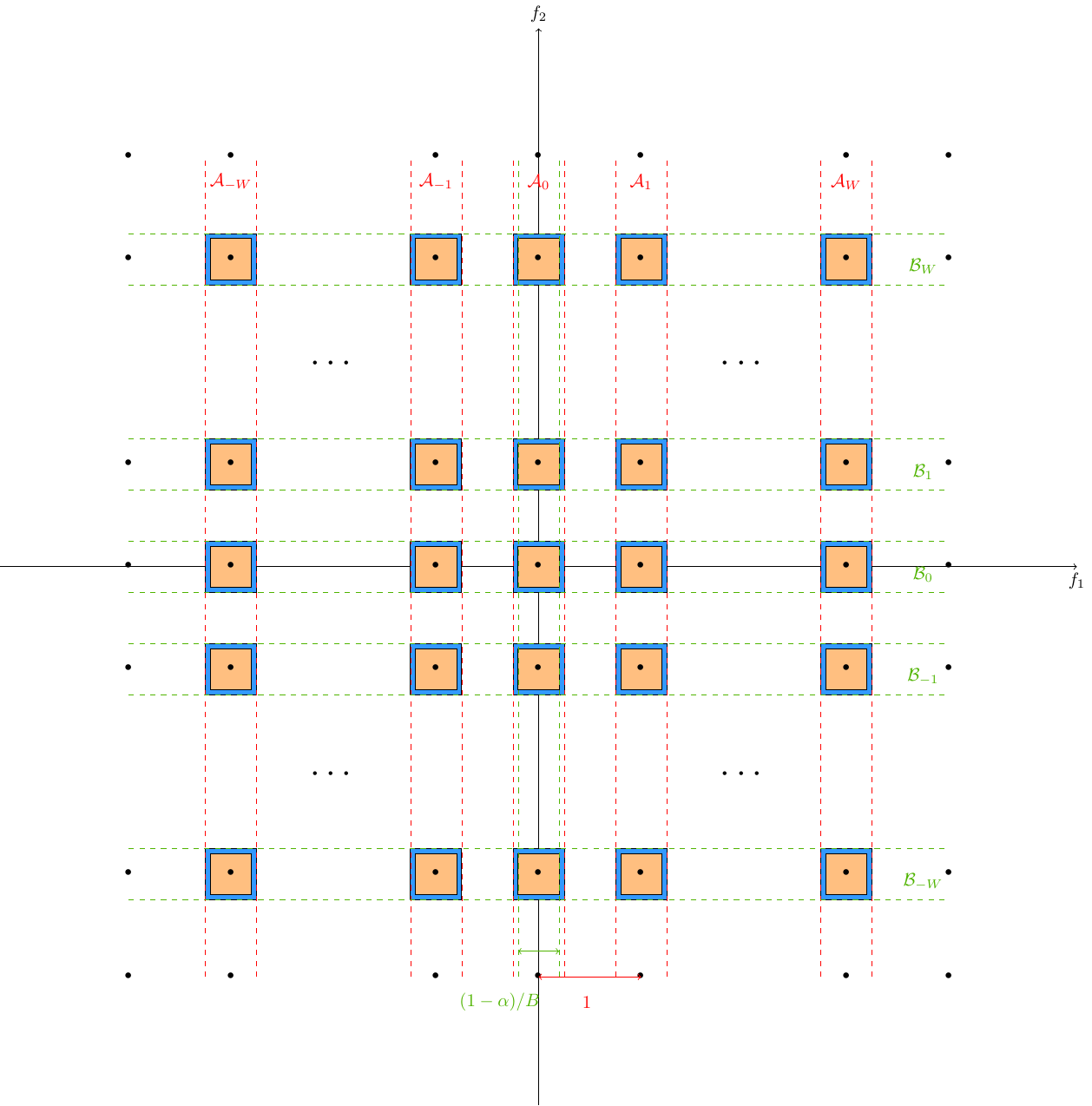}
    \caption{Demonstration for the filter $\hat{\mathcal{G}}(f)$ in two dimension $d = 2$. ``yellow'' refers to Property~I and $f \in \Lambda_{W}(\frac{1 - \alpha}{2 B})$, where $\hat{\mathcal{G}}(f)$'s value is very close to 1; ``blue'' refers to Property~II and $f \in \Lambda_{W}(\frac{1}{2 B}) \setminus \Lambda_{W}(\frac{1 - \alpha}{2 B})$, where the value of $\hat{\mathcal{G}}(f)$ drops sharply, and the other "white" region means Property~III and $f \in \R^{d} \setminus \Lambda_{W}(\frac{1}{2 B})$, where $\hat{\mathcal{G}}(f)$ oscillates near 0. ${\cal A}_i=[i-1/(2B),i+1/(2B)]\times \R$ and ${\cal B}_i=\R\times [i-1/(2B),i+1/(2B)]$.}
    \label{fig:filter_function_multi}
\end{figure}

\subsection{Proof of properties}
\label{subsec:HashToBins_multi:filter_proof}

Below we only present the proofs of Properties~III and V, and the other properties directly follow from the corresponding properties of the single-dimensional filter $(\mathsf{G}(t_{s}), \hat{\mathsf{G}}(f_{r}))$ that are given in Definition~\ref{def:shift_filter_function_single} and Lemma~\ref{lem:shifted_filter_function_single}.

\begin{claim}[Property~III of Lemma~\ref{lem:filter_function_multi}]
\label{cla:lem:filter_function_multi:1}
$0 \leq \hat{\mathcal{G}}(f) \leq \frac{\delta}{\poly(k, d)}$ for any $f \in \R^{d} \setminus \Lambda_{W}(\frac{1}{2 B})$.
\end{claim}

\begin{proof}
We let $r^* \in [d]$ denote (one of) the coordinate that maximizes, over all $r \in [d]$, the distance of $f_{r}$ from the lattice $[-W: W]$. Because $f \in \R^{d} \setminus \Lambda_{W}(\frac{1}{2 B})$, that maximum distance is at least $\frac{1}{2 B}$. Then by construction (see Definition~\ref{def:shift_filter_function_multi}),
\begin{align*}
    \hat{\mathcal{G}}(f)
    = \prod_{r \in [d]} \hat{\mathsf{G}}(f_{r})
    \leq \hat{\mathsf{G}}(f_{r^*})
    \leq \frac{\delta}{\poly(k, d)},
\end{align*}
where the second step follows because $\hat{\mathsf{G}}(f_{r}) \in [0, 1]$ for each coordinate $r \in [d] \setminus \{r^*\}$ (see Properties~II to IV of Lemma~\ref{lem:shifted_filter_function_single}); and the last step follows from Property~III of Lemma~\ref{lem:shifted_filter_function_single}.

This completes the proof of Claim~\ref{cla:lem:filter_function_multi:1}.
\end{proof}

\begin{claim}[Property~V of Lemma~\ref{lem:filter_function_multi}]
\label{cla:lem:filter_function_multi:2}
$\sum_{i \in \mathbb{Z}^{d}} \mathcal{G}(i)^2 \leq e^2 \cdot B^{-d} = e^2 \cdot \mathcal{B}^{-1}$.
\end{claim}

\begin{proof}
Due to Definition~\ref{def:shift_filter_function_multi} that $\mathcal{G}(t) = \prod_{s \in [d]} \mathsf{G}(t_{s})$ for any $t \in \R^{d}$, we have
\begin{eqnarray*}
    \sum_{i \in \mathbb{Z}^{d}} \mathcal{G}(i)^2
    & = & \sum_{i \in \mathbb{Z}^{d}} \Big(\prod_{s \in [d]} \mathsf{G}(i_{s})^{2}\Big) \\
    & = & \prod_{s \in [d]} \Big(\sum_{i_{s} \in \mathbb{Z}} \mathsf{G}(i_{s})^{2}\Big) \\
    & \leq & \prod_{s \in [d]} \Big(\Big(1 + \frac{2}{d}\Big) \cdot B^{-1}\Big) \\
    & = & \Big(1 + \frac{2}{d}\Big)^{d} \cdot B^{-d} \\
    & \leq & e^2 \cdot B^{-d},
\end{eqnarray*}
where the third step follows from Property~VI of Lemma~\ref{lem:shifted_filter_function_single} that $\sum_{i \in \mathbb{Z}} \mathsf{G}(i)^{2} \leq (1 + \frac{2}{d}) \cdot B^{-1}$; and the last step follows because $(1 + \frac{1}{z})^{z} \leq e$ for any $z > 0$.

This completes the proof of Claim~\ref{cla:lem:filter_function_multi:2}.
\end{proof}

\subsection{Construction and properties of standard window \texorpdfstring{$(\mathcal{G}'(t), \hat{\mathcal{G}'}(f))$}{}}
\label{subsec:HashToBins_multi:window}

Now we associate our multi-dimensional filter $(\mathcal{G}(t), \hat{\mathcal{G}}(f))$ given in Definition~\ref{def:shift_filter_function_multi} with another standard window $(\mathcal{G}'(t), \hat{\mathcal{G}'}(f))$ in a similar manner as Lemma~\ref{lem:shifted_window_function_single} and the counterpart results in \cite{hikp12a,hikp12b}, which is more convenient for our later use.

\begin{lemma}[The multi-dimensional standard window]
\label{lem:shifted_window_function_multi}
Consider the filter function $(\mathcal{G}(t), \hat{\mathcal{G}}(f))$ given in Definition~\ref{def:shift_filter_function_multi}, there is another function $(\mathcal{G}'(t), \hat{\mathcal{G}'}(f))$ such that:
\begin{description}[labelindent = 1em]
    \item [Property~I:]
    $\hat{\mathcal{G}'}(f) = 1$ for any $f \in \Lambda_{W}(\frac{1 - \alpha}{2 B})$.
    
    \item [Property~II:]
    $\hat{\mathcal{G}'}(f) \in [0, 1]$ for any $f \in \Lambda_{W}(\frac{1}{2 B}) \setminus \Lambda_{W}(\frac{1 - \alpha}{2 B})$.
    
    \item [Property~III:]
    $\hat{\mathcal{G}'}(f) = 0$ for any $f \in \R^{d} \setminus \Lambda_{W}(\frac{1}{2 B})$.
    
    \item [Property~IV:]
    $\|\hat{\mathcal{G}'} - \hat{\mathcal{G}}\|_{\infty} = \max_{f \in \R^{d}} |\hat{\mathcal{G}'}(f) - \hat{\mathcal{G}}(f)| \leq \frac{\delta}{\poly(k, d)}$.
\end{description}
\end{lemma}

\begin{proof}
We define $\hat{\mathcal{G}'}(f)$ as follows; noticeably, similar to $\hat{\mathcal{G}}(f)$, this is also an even function in every coordinate $r \in [d]$ given that the other $(d - 1)$ coordinates are fixed:
\begin{align*}
    \hat{\mathcal{G}'}(f) = 
    \begin{cases}
    1 & \forall f \in \Lambda_{W}(\frac{1 - \alpha}{2 B}) \\
    \hat{\mathcal{G}}(f) & \forall f \in \Lambda_{W}(\frac{1}{2 B}) \setminus \Lambda_{W}(\frac{1 - \alpha}{2 B}) \\
    0 & \forall f \in \R^{d} \setminus \Lambda_{W}(\frac{1}{2 B})
    \end{cases}.
\end{align*}
Then all the properties above can be inferred from Lemma~\ref{lem:filter_function_multi}.

This completes the proof of Lemma~\ref{lem:shifted_window_function_multi}.
\end{proof}

\subsection{Permutation and hashing}
\label{subsec:HashToBins_multi:prelim}

We adopt the following notations for convenience:
\begin{itemize}
    \item Let $\lfloor z \rfloor \in \mathbb{Z}$ denote the greatest integer that is less than or equal to a real number $z \in \R$. In the case that $z = (z_{r})_{r = 1}^{d} \in \R^{d}$ is a vector, we would abuse the notation $\lfloor z \rfloor = (\lfloor z_{r} \rfloor)_{r = 1}^{d} \in \mathbb{Z}^{d}$.
    
    \item Let $\mathrm{frac}(z) = z - \lfloor z \rfloor \in [0, 1)$ denote the fractional part of a real number $z \in \R$. In the case that $z = (z_{r})_{r = 1}^{d} \in \R^{d}$ is a vector, we would abuse the notation $\mathrm{frac}(z) = (\mathrm{frac}(z_{r}))_{r = 1}^{d} \in [0, 1)^{d}$.
    
    \item Denote the set $[n] = \{0, 1, \cdots, n - 1\}$, for any positive integer $n \in \mathbb{N}_{\geq 1}$.
    
    \item Let $\bar{z} \in \C$ denote the conjugate of a complex number $z \in \C$. Notice that $|z|^2 = z \bar{z}$.
\end{itemize}

\begin{definition}[Setup for permutation and hashing]
\label{def:HashToBins_multi_parameters}
We sample the random matrix $\Sigma \in \R^{d \times d}$ and the random vectors $a, b \in \R^{d}$, and define the parameter $\mathcal{D}$ as follows:
\begin{itemize}
    \item The $d$-to-$d$ random matrix $\Sigma$ is constructed in two steps. First, we sample an interim matrix $\Sigma' \sim \unif(\mathbf{SO}(d))$ uniformly at random from the {\em rotation group}, namely a rotation matrix of determinant $\det(\Sigma') = 1$. Then, we define $\Sigma := \beta \Sigma'$, where the scaling factor $\beta \sim \unif[\frac{2 \sqrt{d}}{B \eta}, \frac{4 \sqrt{d}}{B \eta}]$ is uniform random.
    
    \item The random vector $a \in \R^{d}$ will be specified later in Section~\ref{sec:locate_inner_time_points}. In this section, we only need the property of $a$ given in Conditions~\ref{con:duration} and \ref{con:sampling}, which also will be verified in Section~\ref{sec:locate_inner_time_points}.
    
    \item The random vector $b' = (b'_{r})_{r \in [d]} \sim \unif[0, 1]^{d}$. Then, let $b := \Sigma^{-1} b'$.
    
    \item The parameter $D = \Theta(\ell / \alpha) = \Theta(d \cdot \log(k d / \delta))$ is a sufficiently large integer. Also, let $\mathcal{D} := D^{d}$.
\end{itemize}
\end{definition}

\begin{condition}[Duration requirement]
\label{con:duration}
Given any $i \in [B D]^{d}$ and any choice of the random matrix $\Sigma \in \R^{d \times d}$ according to Definition~\ref{def:HashToBins_multi_parameters}, any choice of $a$ ensures that $\Sigma^{\top} (i + a) \in [0, T]^{d}$ is within the duration.
\end{condition}

\begin{condition}[Sampling requirement]
\label{con:sampling}
Given any $i \in [B D]^{d}$ and any choice of the random matrix $\Sigma \in \R^{d \times d}$ according to Definition~\ref{def:HashToBins_multi_parameters}, the following hold for the random vector $a$:
\begin{eqnarray*}
    \E_{a} \left[ g\big(\Sigma^{\top} (i + a)\big)^2 \right]
    & \lesssim & \frac{1}{T^{d}} \cdot \int_{t \in [0, T]^{d}} |g(t)|^{2} \cdot \d t,
\end{eqnarray*}
\end{condition}

\begin{figure}
    \centering
    \includegraphics[width = .5\textwidth]{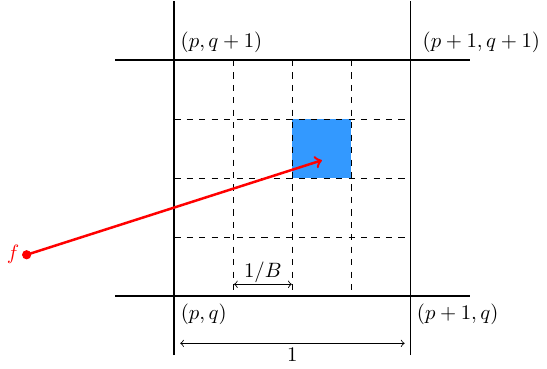}
    \caption{Demonstration for the hashing scheme (Definition~\ref{def:function_hash_multi}) in two dimensions $d = 2$, where $p, q \in \Z$ are integers. The unit square $[p, p + 1) \times [q, q + 1)$ are divided into $\mathcal{B} = B^{2}$ subsquares, and the subsquare in which the frequency $f$ is hashed into, is exactly the index $\mathpzc{h}_{\Sigma, b}(f) \in [B]^{d}$.}
    \label{fig:function_hash_multi}
\end{figure}

\begin{definition}[Hashing]
\label{def:function_hash_multi}
Define the vector-valued function
\begin{align*}
    \mathpzc{h}_{\Sigma, b}(f) = \Big\lfloor B \cdot \mathrm{frac}\Big(\frac{1}{2 B} \cdot \mathbf{1} + \Sigma (f - b)\Big) \Big\rfloor \in [B]^{d}.
\end{align*}
This function ``hashes'' any frequency $f \in [-F, F]^{d}$ into one of the $\mathcal{B} = B^{d} = 2^{\Theta(d \log d)} \cdot k \in \mathbb{N}_{\geq 1}$ bins. When $B = \Theta(d \cdot k^{1 / d})$ is large enough, every bin $j \in [B]^{d}$ is likely to have at most one heavy hitter (namely one tone frequency $f \in \supp(\hat{x^*})$) and if so, we can recover the tone from the hitting bins via an $1$-sparse algorithm. See Figure~\ref{fig:function_hash_multi} for a demonstration.
\end{definition}

\begin{definition}[Offset]
\label{def:function_offset_multi}
Define the vector-valued function
\begin{align*}
    \mathpzc{o}_{\Sigma, b}(f) = \mathrm{frac}\Big(\frac{1}{2 B} \cdot \mathbf{1} + \Sigma (f - b)\Big) - \frac{1}{B} \cdot \mathpzc{h}_{\Sigma, b}(f) - \frac{1}{2 B} \cdot \mathbf{1} \in \Big[-\frac{1}{2 B}, \frac{1}{2 B}\Big)^{d},
\end{align*}
which measures the coordinate-wise distance from the center of the $\mathpzc{h}_{\Sigma, b}(f)$-th bin to $f \in [-F, F]^{d}$.
\end{definition}

\begin{figure}[htbp]
    \centering
    \begin{tabular}[b]{c}
    \includegraphics[width = .6\textwidth]{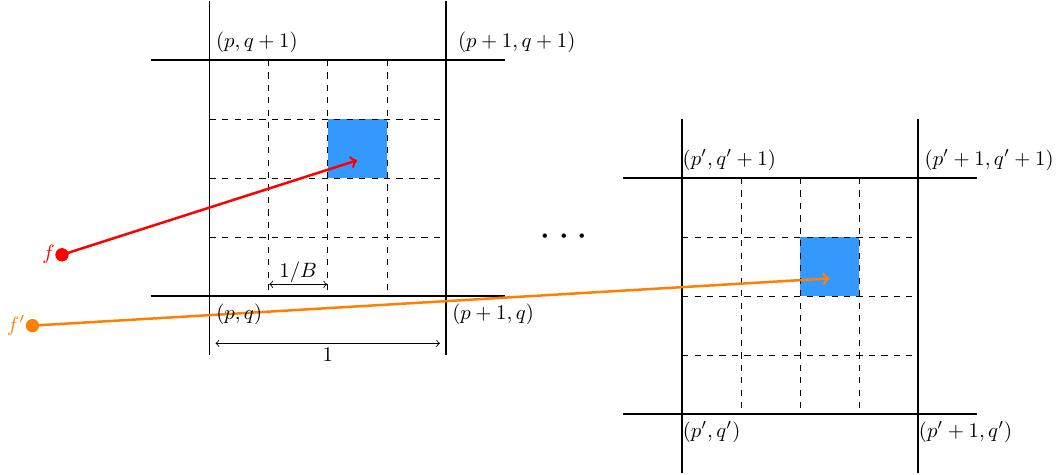}\\
    \small{(a) Collision}
    \end{tabular}\qquad
    \begin{tabular}[b]{c}
    \includegraphics[width = .25\textwidth]{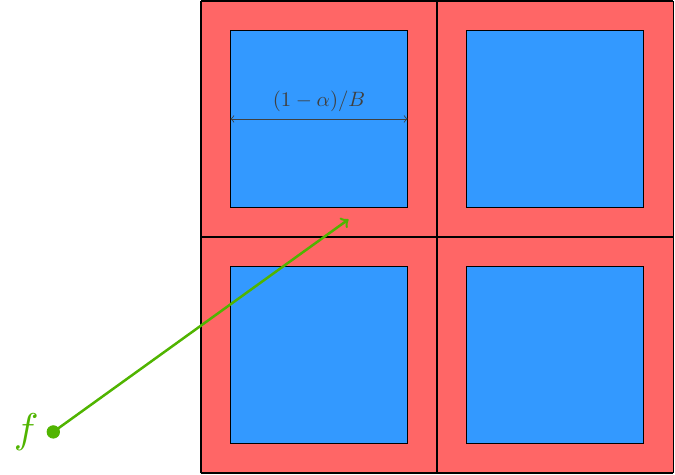}\\
    \small{(b) Large offset}
    \end{tabular}
    \caption{\small Demonstration for the bad events ``collision'' (Definition~\ref{def:event_collision_multi}) and ``large offset'' (Definition~\ref{def:event_large_offset_multi}) in two dimensions $d = 2$. In Figure~6(a), the two frequencies $f \neq f' \in \supp(\hat{x^*})$ may be hashed into two different unit squares (i.e.\ possibly either $p \neq p'$ or $q \neq q'$ or both), but it is always the case that the two subsquares have the same index $\mathpzc{h}_{\Sigma, b}(f) = \mathpzc{h}_{\Sigma, b}(f') \in [B]^{d}$. In Figure~6(b), the red region (that gives a large offset) covers $1 - (1 - \alpha)^{2}$ fractions of the whole plane.}
    \label{fig:bad_events_multi}
\end{figure}

% \begin{figure}
%     \centering
%     \subfloat[Collision
%     \label{fig:event_collision_multi}]{
%     \includegraphics[width = .6\textwidth]{collision.pdf}
%     }
%     \subfloat[Large offset
%     \label{fig:event_large_offset_multi}]{
%     \includegraphics[width = .35\textwidth]{large_off.pdf}
%     }
%     \caption{Demonstration for the bad events ``collision'' (Definition~\ref{def:event_collision_multi}) and ``large offset'' (Definition~\ref{def:event_large_offset_multi}) in two dimensions $d = 2$. In Figure~\ref{fig:event_collision_multi}, the two frequencies $f \neq f' \in \supp(\hat{x^*})$ may be hashed into two different unit squares (i.e.\ possibly either $p \neq p'$ or $q \neq q'$ or both), but it is always the case that the two subsquares have the same index $\mathpzc{h}_{\Sigma, b}(f) = \mathpzc{h}_{\Sigma, b}(f') \in [B]^{d}$. In Figure~\ref{fig:event_large_offset_multi}, the red region (that gives a large offset) covers $1 - (1 - \alpha)^{2}$ fractions of the whole plane.}
%     \label{fig:bad_events_multi}
% \end{figure}

\begin{definition}[Collision]
\label{def:event_collision_multi}
Consider a tone frequency $f \in \supp(\hat{x^*})$, the event $E_{\coll}(f)$ occurs when $\mathpzc{h}_{\Sigma, b}(f') = \mathpzc{h}_{\Sigma, b}(f)$ for some other tone frequency $f' \in \supp(\hat{x^*}) \setminus \{f\}$, namely both $f \neq f' \in \supp(\hat{x^*})$ are hashed into the same bin. In this case, the algorithm cannot recover the two collided tone frequencies $f \neq f'$. See Figure~\ref{fig:bad_events_multi}(a) for a demonstration.
\end{definition}

\begin{definition}[Large offset]
\label{def:event_large_offset_multi}
Consider a tone frequency $f \in \supp(\hat{x^*})$, the event $E_{\off}(f)$ occurs when $\|\mathpzc{o}_{\Sigma, b}(f)\|_{\infty} \geq \frac{1 - \alpha}{2 B}$, i.e.\ the frequency $f \in \supp(\hat{x^*})$ locates on the boundary of the $\mathpzc{h}_{\Sigma, b}(f)$-th bin. In this case, the algorithm also cannot recover the tone frequency $f$. See Figure~\ref{fig:bad_events_multi}(b) for a demonstration.
\end{definition}

Our multi-dimensional permutation scheme is a natural generalization of the single-dimensional one by \cite[Definition~A.5]{ps15}.

 \begin{definition}[Multi-dimensional permutation]
\label{def:permutation_multi}
Let $\mathcal{P}_{\Sigma, b, a} x( t ) = x(\Sigma^{\top} (t + a)) \cdot e^{-2 \pi \i \cdot b^{\top} \Sigma^{\top} t}$ for any $t \in \R^{d}$.
\end{definition}

\begin{lemma}[Identities]
\label{lem:permutation_multi_identity}
The permutation given in Definition~\ref{def:permutation_multi} satisfies that:
\begin{description}[labelindent = 1em]
    \item [Property~I:]
    $\hat{\mathcal{P}_{\Sigma, b, a} x}(\Sigma (f - b)) = \hat{x}(f) \cdot \det(\Sigma)^{-1} \cdot e^{2 \pi \i \cdot a^{\top} \Sigma f}$ for any $f \in \R^{d}$.
    
    \item [Property~II:]
    $\hat{\mathcal{P}_{\Sigma, b, a} x}(t) = \hat{x}(\Sigma^{-1} t + b) \cdot \det(\Sigma)^{-1} \cdot e^{2 \pi \i \cdot a^{\top} (t + \Sigma^{\top} b)}$ for any $t \in \R^{d}$.
\end{description}
\end{lemma}

\begin{proof}
For Property~I, by the definition of the {\CFT}, the $\LHS$ equals
\begin{align*}
    \hat{\mathcal{P}_{\Sigma, b, a} x}\big(\Sigma (f - b)\big)
    & = \int_{t \in \R^{d}} \mathcal{P}_{\Sigma, b, a} x(t) \cdot e^{-2 \pi \i \cdot (f^{\top} - b^{\top}) \Sigma^{\top} t} \cdot \d t \\
    & = \int_{t \in \R^{d}} x\big(\Sigma^{\top} (t + a)\big) \cdot e^{-2 \pi \i \cdot b^{\top} \Sigma^{\top} t} \cdot e^{-2 \pi \i \cdot (f^{\top} - b^{\top}) \Sigma^{\top} t} \cdot \d t \\
    & = e^{2 \pi \i \cdot f^{\top} \Sigma^{\top} a} \cdot \int_{t \in \R^{d}} x\big(\Sigma^{\top} (t + a)\big) \cdot e^{-2 \pi \i \cdot f^{\top} \Sigma^{\top} (t + a)} \cdot \d t \\
    & = e^{2 \pi \i \cdot f^{\top} \Sigma^{\top} a} \cdot \det(\Sigma)^{-1} \cdot \int_{\tau \in \R^{d}} x(\tau) \cdot e^{-2 \pi \i \cdot f^{\top} \tau} \cdot \d \tau \\
    & = e^{2 \pi \i \cdot f^{\top} \Sigma^{\top} a} \cdot \det(\Sigma)^{-1} \cdot \hat{x}(f) \\
    & = e^{2 \pi \i \cdot a^{\top} \Sigma f} \cdot \det(\Sigma)^{-1} \cdot \hat{x}(f),
\end{align*}
where the second step follows from Definition~\ref{def:permutation_multi}; the fourth step is by substitution; and the last step is by the definition of the {\CFT}.

We directly infer Property~II from Property~I by substitution. Lemma~\ref{lem:permutation_multi_identity} follows then.
\end{proof}

\subsection{{\HashToBins}: algorithm}
\label{subsec:HashToBins_multi:alg}

% \begin{algorithm}
% \caption{{\HashToBins} in multiple dimensions}
% \label{alg:HashToBins}
% \SetKwProg{myproc}{Procedure}{}{}
% \myproc{\HashToBins($\Sigma, b, a, D$)}{
% Define $(\mathcal{G}(t), \hat{\mathcal{G}}(f))$ according to Definition~\ref{def:shift_filter_function_single}\label{alg:HashToBins_multi:filter}\;

% Define $\mathcal{P}_{\Sigma, b, a}$ according to Definition~\ref{def:permutation_multi}\;

% Define $y = (y_{j})_{j \in {[B D]}^{d}}$, where $y_{j} = \mathcal{G}(j) \cdot \mathcal{P}_{\Sigma, b, a} x(j)$.
% \label{alg:HashToBins_multi:y}\;

% Define $u = (u_{j})_{j \in [B]^{d}}$, where $u_{j} = \sum_{i \in [D]^{d}} y_{B i + j}$.
% \label{alg:HashToBins_multi:u}\;

% \Return the {\DFT} $\hat{u} = (\hat{u}_{j})_{j \in [B]^{d}}$.
% \label{alg:HashToBins_multi:output}\;
% }\label{alg:HashToBins_multi}
% \end{algorithm}

\begin{algorithm}[!t]
\caption{{\HashToBins} in multiple dimensions}\label{alg:HashToBins}
\begin{algorithmic}[1]
\Procedure{\HashToBins}{$\Sigma, b, a, D$}
\State Define $(\mathcal{G}(t), \hat{\mathcal{G}}(f))$ according to Definition~\ref{def:shift_filter_function_single}.
\label{alg:HashToBins_multi:filter}

\State Define $\mathcal{P}_{\Sigma, b, a}$ according to Definition~\ref{def:permutation_multi}.

\State Define $y = (y_{j})_{j \in {[B D]}^{d}}$, where $y_{j} = \mathcal{G}(j) \cdot \mathcal{P}_{\Sigma, b, a} x(j)$.
\label{alg:HashToBins_multi:y}

\State Define $u = (u_{j})_{j \in [B]^{d}}$, where $u_{j} = \sum_{i \in [D]^{d}} y_{B i + j}$.
\label{alg:HashToBins_multi:u}

\State \Return the {\DFT} $\hat{u} = (\hat{u}_{j})_{j \in [B]^{d}}$.
\label{alg:HashToBins_multi:output}
\EndProcedure
\end{algorithmic}
\label{alg:HashToBins_multi}
\end{algorithm}

\begin{fact}[Identities under {\DFT}/{\DTFT}]
\label{fac:HashToBins_multi_identity}
The following holds for each $j \in [B]^{d}$:
\begin{align*}
    \hat{u}_{j} = \hat{y}_{D j} = \hat{\mathcal{G}} * \hat{\mathcal{P}_{\Sigma, b, a} x}(B^{-1} \cdot j).
\end{align*}
\end{fact}

\begin{proof}
It is noteworthy that $y = (y_{j})_{j \in {[B D]}^{d}}$ is a $(B D)^{d}$-dimensional vector, and $u = (u_{j})_{j \in [B]^{d}}$ is a $B^{d}$-dimensional vector. For the first equality $\hat{u}_{j} = \hat{y}_{j D}$, due to the definition of the {\DFT},
\begin{align*}
    \hat{u}_{j}
    & = \sum_{i \in [B]^{d}} u_{i} \cdot e^{-\frac{2 \pi \i}{B} \cdot j^{\top} i} \\
    & = \sum_{i \in [B]^{d}} \sum_{l \in [D]^{d}} y_{B l + i} \cdot e^{-\frac{2 \pi \i}{B} \cdot j^{\top} i} \\
    & = \sum_{i \in [B]^{d}} \sum_{l \in [D]^{d}} y_{B l + i} \cdot e^{-\frac{2 \pi \i}{B} \cdot j^{\top} (B l + i)} \\
    & = \sum_{i \in [B D]^{d}} y_{i} \cdot e^{-\frac{2 \pi \i}{B} \cdot j^{\top} i} \\
    & = \sum_{i \in [B D]^{d}} y_{i} \cdot e^{-\frac{2 \pi \i}{B D} \cdot (D j)^{\top} i} \\
    & = \hat{y}_{D j},
\end{align*}
where the second step is by Line~\ref{alg:HashToBins_multi:u} of {\HashToBins}; the third step follows since both $j \in [B]^{d}$ and $l \in [D]^{d}$ are $d$-dimensional integer vectors and therefore, $e^{-\frac{2 \pi \i}{B} \cdot j^{\top} (B l)} = e^{-2 \pi \i \cdot j^{\top} l} = 1$; the fourth step is by substitution; and the last step also applies the {\DFT}.

For the second equality, again we know from the definition of the {\DFT} that
\begin{eqnarray*}
    \hat{y}_{D j}
    & = & \sum_{i \in [B D]^{d}} y_{i} \cdot e^{-\frac{2 \pi \i}{B D} \cdot (D j)^{\top} i} \\
    & = & \sum_{i \in [B D]^{d}} \mathcal{G}(i) \cdot \mathcal{P}_{\Sigma, b, a} x(i) \cdot e^{-\frac{2 \pi \i}{B} \cdot j^{\top} i} \\
    & = & \sum_{i \in \mathbb{Z}^d} \mathcal{G}(i) \cdot \mathcal{P}_{\Sigma, b, a} x(i) \cdot e^{-\frac{2 \pi \i}{B} \cdot j^{\top} i} \\
    & = & \hat{\mathcal{G} \cdot \mathcal{P}_{\Sigma, b, a} x}(B^{-1} \cdot j) \\
    & = & \hat{\mathcal{G}} * \hat{\mathcal{P}_{\Sigma, b, a} x}(B^{-1} \cdot j),
\end{eqnarray*}
where the second step is by Line~\ref{alg:HashToBins_multi:y} of {\HashToBins}; the third step follows because $\mathcal{G}(t) = 0$ when $\|t\|_{\infty} \geq \ell \cdot B / \alpha$ (see Lemma~\ref{lem:filter_function_multi}), given a large enough $D = \Theta(d \cdot \log(k d / \delta))$ (see Definition~\ref{def:HashToBins_multi_parameters}); and the fourth step is by the definition of the {\DTFT}.

This completes the proof of Fact~\ref{fac:HashToBins_multi_identity}.
\end{proof}

\begin{fact}[Sample complexity and time complexity]
\label{fac:HashToBins_multi_sample_time}
The procedure {\HashToBins} takes $O(\mathcal{B} \mathcal{D}) = 2^{O(d \cdot \log d)}\cdot k \cdot \log^{d}(k / \delta)$ samples and runs in $O(\mathcal{B} \mathcal{D} + \mathcal{B} \log \mathcal{B}) = 2^{O(d \cdot \log d)} \cdot k \cdot \log^{d}(k / \delta)$ time.
\end{fact}

\begin{proof}

Recall that $\mathcal{B} = B^{d} = 2^{\Theta(d \log d)} \cdot k$ (Definition~\ref{def:shift_filter_function_multi}) and $\mathcal{D} = D^{d} = \log^{d}(k d / \delta)$. The sample complexity is easy to see, because we have exactly $\mathcal{B}$ bins, and each bin $j \in [\mathcal{B}]$ requires exactly $\mathcal{D}$ samples in Line~\ref{alg:HashToBins_multi:u} of {\HashToBins}.

The $\mathcal{B} \mathcal{D}$-dimensional vector $y = (y_{j})_{j \in {[B D]^{d}}}$ can be computed $O(\mathcal{B} \mathcal{D})$ time (assuming $O(1)$-time query oracles to evaluating the filter $(\mathcal{G}(t), \hat{\mathcal{G}}(f))$ and to sampling the signal $x(t)$; see Remark~\ref{rem:shifted_filter_function_single}). Furthermore, the $\mathcal{B}$-dimensional vector $u = (u_{j})_{j \in [B]^{d}}$, where $u_{j} = \sum_{i \in [D]} y_{B i + j}$, can be computed in $O(\mathcal{B} \mathcal{D})$ time. Then we can derive its {\DFT} $\hat{u} = (\hat{u}_{j})_{j \in [B]^{d}}$ through any {\FFT} algorithm in $O(\mathcal{B} \log \mathcal{B})$ time. The claimed time complexity follows as well.
\end{proof}

\subsection{{\HashToBins}: probabilities of bad events}
\label{subsec:HashToBins_multi:probabilities}

\begin{lemma}[Probability of collision]
\label{lem:event_collision_multi}
Consider the random matrix $\Sigma \in \R^{d \times d}$ and the random vector $b \in \R^{d}$ given in Definition~\ref{def:HashToBins_multi_parameters}, for any pair of tone frequencies $f \neq f' \in \supp(\hat{x^*})$, the probability of collision
\begin{eqnarray*}
    \Pr_{\Sigma, b}\big[\mathpzc{h}_{\Sigma, b}(f) = \mathpzc{h}_{\Sigma, b}(f')\big]
    & \leq & 0.01 \cdot k^{-1}.
\end{eqnarray*}
\end{lemma}

\begin{figure}
    \centering
    \includegraphics[width = .8\textwidth]{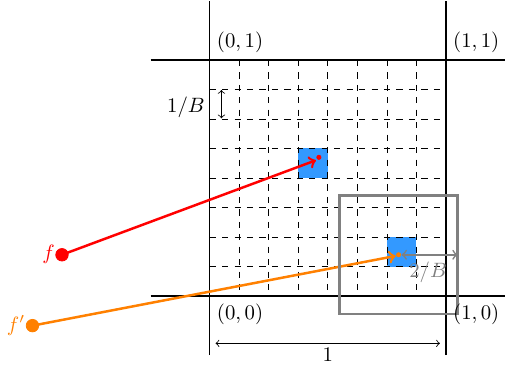}
    \caption{Demonstration for Lemma~\ref{lem:event_collision_multi} in two dimensions $d = 2$, where the hashing of $f$ and the hashing of $f'$ are rounded into the same unit square. A sufficient condition to avoid the collision between $f \neq f' \in \supp(\hat{x^*})$ is that the $\ell_{\infty}$ distance $\|\mathpzc{h}_{\Sigma, b}(f) - \mathpzc{h}_{\Sigma, b}(f')\|_{\infty} \geq \frac{2}{B}$.}
    \label{fig:event_collision_multi_lem}
\end{figure}

\begin{proof}
Recall that $\eta = \min \{\| f - f' \|_{2}: f \neq f' \in \supp(\hat{x^*})\}$ is the minimum $\ell_{2}$-distance between any pair of tone frequencies. Due to Definition~\ref{def:function_hash_multi}, two distinct tone frequencies $f \neq f' \in \supp(\hat{x^*})$ are hashed into bins
\begin{eqnarray*}
    \mathpzc{h}_{\Sigma, b}(f) & = & \Big\lfloor B \cdot \mathrm{frac}\Big(\frac{1}{2 B} \cdot \mathbf{1} + \Sigma (f - b)\Big) \Big\rfloor \\
    \mathpzc{h}_{\Sigma, b}(f') & = & \Big\lfloor B \cdot \mathrm{frac}\Big(\frac{1}{2 B} \cdot \mathbf{1} + \Sigma (f - b)+ \Sigma (f' - f)\Big)  \Big\rfloor.
\end{eqnarray*}
In order to hash $f$ and $f'$ into the same bin, a necessary condition (under any realization of the random vector $b$) is that $\| \Sigma (f' - f) - i \|_{\infty} \leq \frac{1}{B}$; in other words, $\Sigma (f' - f)$ locates in a hypercube that is centered at some integer vector $i = (i_{r})_{r \in [d]} \in \mathbb{Z}^{d}$ and has edge length $\frac{2}{B}$.

Given the above two equations and as Figure~\ref{fig:event_collision_multi_lem} suggests, a necessary condition for the collision $\mathpzc{h}_{\Sigma, b}(f) = \mathpzc{h}_{\Sigma, b}(f')$ is that $\| \Sigma (f' - f) - i \|_{\infty} < \frac{2}{B}$ for some integer vector $i \in \Z^{d}$. Below we upper bound this probability based on case analysis.

Case~(i): when $\eta \leq \| f - f' \|_{2} \leq \frac{B - 2}{4 \sqrt{d}} \cdot \eta$. Recall Definition~\ref{def:HashToBins_multi_parameters} that $\Sigma$ is a {\em rotation matrix} scaled by a random factor $\beta \sim \unif[\frac{2 \sqrt{d}}{B \eta}, \frac{4 \sqrt{d}}{B \eta}]$. Given this, we have $\|\Sigma (f - f')\|_{2} = \beta \cdot \| f - f' \|_{2}$. Further, since $\frac{2 \sqrt{d}}{B \eta} \leq \beta \leq \frac{4 \sqrt{d}}{B \eta}$ and $\eta \leq \| f - f' \|_{2} \leq \frac{B - 2}{4 \sqrt{d}} \cdot \eta$, we have
\begin{align*}
    & \|\Sigma (f - f')\|_{2}
    ~ \geq ~ \frac{2 \sqrt{d}}{B \eta} \cdot \eta
    ~ = ~ \frac{2 \sqrt{d}}{B}, \\
    & \|\Sigma (f - f')\|_{2}
    ~ \leq ~ \frac{4 \sqrt{d}}{B \eta} \cdot \frac{B - 2}{4 \sqrt{d}} \cdot \eta
    ~ = ~ 1 - \frac{2}{B}.
\end{align*}
Given these, we can easily see that $\| \Sigma (f' - f) - i \|_{\infty} \geq 2 / B$ for any integer vector $i \in \Z^{d}$. Namely, the collision never occurs in this case.

Case~(ii): when $\| f - f' \|_{2} > \frac{B - 2}{4 \sqrt{d}} \cdot \eta$.

As for the case (ii), for simplicity, let $r\geq 1-\frac{2\sqrt{d}}{B}$ denote $\|f-f'\|_{2}\cdot (\frac{2\sqrt{d}}{B\eta})$, then we know $\|\Sigma(f-f')\|_{2}$ is distributed uniformly on $[r,2r]$. Let $V_d(r)=\frac{\pi^{d / 2}}{\Gamma(d / 2 + 1)} \cdot (r)^{d}$ represent the volume of $d$-dimensional Euclidean ball of radius $r$ and $S_d(r)= \frac{2\pi^{d/2}}{\Gamma(\frac{d}{2})}r^{d-1}$ represent the surface of area. The probability density function for $\Sigma(f-f')=x$ is $\mathrm{PDF}(x)=\frac{1}{r\cdot S_d(\|x\|_{2})}$ for $r\leq \|x\|_2\leq 2r$. Then the  probability for $\Sigma(f-f')$ falls into any region with volume $V$ is at most 
\begin{align*}
\frac{V}{r\cdot S_d(r)}\lesssim \frac{V\cdot \Gamma(d/2)}{(\sqrt{\pi} r)^d}.
\end{align*}

Then we have 
\begin{eqnarray*}
    \Pr_{\Sigma, b}\big[\mathpzc{h}_{\Sigma, b}(f) = \mathpzc{h}_{\Sigma, b}(f')\big]
    & \leq & (2\cdot\lceil r+\frac{2\sqrt{d}}{B} \rceil)^d\cdot ( 2/B)^d \cdot \frac{ \Gamma(d/2)}{(\sqrt{\pi} r)^d}\\
    &\lesssim&  2^{O(d)} /B^d \cdot (\frac{d}{2})^{d/2} \\ 
    &\leq & 2^{O(d\log d)}\frac{1}{\mathcal{B}}.
\end{eqnarray*}
where the first step is because that  $\Sigma (f' - f)$ must locate in a hypercube that is centered at some integer vector and has edge length $2/B$, and there are at most $(2\cdot\lceil r+\frac{2\sqrt{d}}{B} \rceil)^d$ different hypercubes in the sphere. The second step is by the Stirling's approximation.

This completes the proof of Lemma~\ref{lem:event_collision_multi}.
\end{proof}

\begin{lemma}[Probability of large offset]
\label{lem:event_large_offset_multi}
Consider the random matrix $\Sigma \in \R^{d \times d}$ and the random vector $b \in \R^{d}$ given in Definition~\ref{def:HashToBins_multi_parameters}, for any tone frequency $f \in \supp(\hat{x^*})$, the probability of large offset
\begin{align*}
    \Pr_{\Sigma, b}\big[E_{\off}(f)\big] = \Pr_{\Sigma, b}\left[ \big\| \mathpzc{o}_{\Sigma, b}(f) \big\|_{\infty} \geq \frac{1 - \alpha}{2 B}\right] = 1 - (1 - \alpha)^{d} \leq 0.01.
\end{align*}
\end{lemma}

\begin{proof}
Recall (see Definition~\ref{def:function_hash_multi}) that 
\begin{align*}
\mathpzc{h}_{\Sigma, b}(f) = \lfloor B \cdot \mathrm{frac}(\frac{1}{2 B} \cdot \mathbf{1} + \Sigma (f - b)) \rfloor
\end{align*}and (see Definition~\ref{def:function_offset_multi}) that \begin{align*}
\mathpzc{o}_{\Sigma, b}(f) = \mathrm{frac}(\frac{1}{2 B} \cdot \mathbf{1} + \Sigma (f - b)) - \frac{1}{B} \cdot \mathpzc{h}_{\Sigma, b}(f) - \frac{1}{2 B} \cdot \mathbf{1} \in [-\frac{1}{2 B}, \frac{1}{2 B})^{d}.
\end{align*}
It can be seen that a large offset $|\mathpzc{o}_{\Sigma, b}(f)|_{\infty} \geq \frac{1 - \alpha}{2 B}$ occurs if and only if $\mathrm{frac}(\frac{1}{2 B} \cdot \mathbf{1} + \Sigma (f - b)) \notin S^{d}$, where (in each coordinate) the union of intervals
\begin{align*}
    S = \bigcup_{j \in [B]} \Big(\frac{j}{B} + \frac{\alpha / 2}{B}, \frac{j}{B} + \frac{1 - \alpha / 2}{B}\Big) \subseteq [0, 1).
\end{align*}

Indeed, for any choice of $\Sigma$ according to Definition~\ref{def:HashToBins_multi_parameters}, because $b \in \R^d$ is uniformly random, each $i$-th coordinate of the vector $(\frac{1}{2 B} \cdot \mathbf{1} + \Sigma (f - b))$ is distributed uniformly on an interval of length $|\supp(b_{i})| = 1 \in \mathbb{N}_{\geq 1}$. Accordingly, each $i$-th coordinate of the fractional part $\mathrm{frac}(\frac{1}{2 B} \cdot \mathbf{1} + \Sigma (f - b))$ must be distributed independently and uniformly on $[0, 1)$.

To summarize, the conditional probability $\Pr_{b}[E_{\off}(f) \mid \Sigma]$ always equals the probability that, an {\em coordinate-wise independent} uniform random vector $\wt{b} \sim \unif[0, 1)^{d}$ locates outside the region $S^{d} \subseteq [0, 1)^{d}$. Thus, for any realized $\Sigma$ given by Definition~\ref{def:HashToBins_multi_parameters}, we have
\begin{eqnarray*}
    \Pr_{b}[E_{\off}(f) \mid \Sigma]
    & = & 1 - |S|^{d} \\
    & = & 1 - \left(\left(\frac{1 - \alpha / 2}{B} - \frac{\alpha / 2}{B}\right) \cdot B\right)^{d} \\
    & = & 1 - (1 - \alpha)^{d}.
\end{eqnarray*}
Since $0 < \alpha \leq \frac{1}{100 \cdot (d + 1)} < 1$ (see Definition~\ref{def:shift_filter_function_multi}), we also have $1 - (1 - \alpha)^{d} \leq d \cdot \alpha \leq \frac{d}{100 \cdot (d + 1)} \leq 0.01$.

This completes the proof of Lemma~\ref{lem:event_large_offset_multi}.
\end{proof}

\begin{lemma}[Hashing into same bin]
\label{lem:close_fre_same_bin}
For any tone frequency $f \in \supp(\hat{x^*})$, if the event $E_{\off}(f)$ does not happen, then $\mathpzc{h}_{\Sigma, b}(f) = \mathpzc{h}_{\Sigma, b}(f')$ for any other frequency $f' \in \R^{d}$ that
\begin{align*}
    \| f - f' \|_{2}
    ~ < ~ \frac{\alpha}{8 \sqrt{d}} \cdot \eta
    ~ = ~ \Theta(d^{-1.5} \cdot \eta).
\end{align*}
\end{lemma}

\begin{proof}
Since the event $E_{\off}(f)$ does not happen, the offset $\|\mathpzc{o}_{\Sigma, b}(f)\|_{\infty} < \frac{1 - \alpha}{2 B}$ (see Definitions~\ref{def:event_large_offset_multi} and \ref{def:function_offset_multi}). Given this, by construction (see Definition~\ref{def:function_hash_multi}) a sufficient condition for the function $\mathpzc{h}_{\Sigma, b}$ to hash $f$ and $f'$ into the same bin is $\|\Sigma (f - f') \|_{\infty} < \frac{\alpha}{2B}$. We verify this condition as follows:
\begin{eqnarray*}
    \|\Sigma (f - f') \|_{\infty}
    & \leq & \|\Sigma (f - f') \|_{2} \\
    & \leq & \frac{4 \sqrt{d}}{B \eta} \cdot \| f - f' \|_{2} \\
    & < & \frac{4 \sqrt{d}}{B \eta} \cdot \frac{\alpha}{8 \sqrt{d}} \cdot \eta \\
    & = & \frac{\alpha}{2 B},
\end{eqnarray*}
where the second step follows from Definition~\ref{def:HashToBins_multi_parameters}, i.e.\ $\Sigma$ is a rotation matrix scaled by a random factor $\beta \sim \unif[\frac{2 \sqrt{d}}{B \eta}, \frac{4 \sqrt{d}}{B \eta}]$; and the third step follows from our premise $\| f - f' \|_{2} < \frac{\alpha}{8 \sqrt{d}} \cdot \eta$.

This completes the proof.
\end{proof}

\subsection{{\HashToBins}: error due to noise}
\label{subsec:HashToBins_multi_error:noise}

\begin{lemma}[The error due to noise]
\label{lem:hashtobins_multi_error:noise}
Suppose that Condition~\ref{con:duration} is true for the random vector $a \in \R^{d}$ and that $x^*(t) = 0$ for any $t \in [0, T]^{d}$, then the following holds for any random matrix $\Sigma \in \R^{d \times d}$ given in Definition~\ref{def:HashToBins_multi_parameters}:
\begin{eqnarray*}
    \E_{b, a} \big[\left\| \hat{u} \right\|_{2}^{2}\big]
    & \lesssim & \frac{1}{T^d} \cdot \int_{\tau\in[0,T]^d} \left| g(\tau) \right|^2 \cdot \d \tau.
\end{eqnarray*}
\end{lemma}

\begin{proof}
We know from Parseval's theorem that $\| \hat{u} \|_{2}^{2} = \B \cdot \| u \|_{2}^{2} = \B \cdot \sum_{j \in [B]^{d}} | u_{j} |^{2}$. To see the lemma, let us consider a specific coordinate $| u_{j} |^{2}$. By definition (see Line~\ref{alg:HashToBins_multi:u} of {\HashToBins}),
\begin{align}
    \notag
    \E_{b, a} \left[| u_{j} |^{2}\right]
    & = \E_{b, a} \Big[ \Big| \sum_{i \in [D]^{d}} y_{B i + j} \Big|^{2} \Big] \\
    \notag
    & = \E_{b, a} \Big[ \sum_{i \in [D]^{d}} y_{B i + j} \cdot \sum_{i \in [D]^{d}} \bar{y_{B i + j}} \Big] \\
    \label{eq:lem:hashtobins_multi:error:noise:1}
    & = \sum_{i \in [D]^{d}} \E_{b, a} \left[ y_{B i + j} \cdot  \bar{y_{B i + j}} \right] + \sum_{i \neq i' \in [D]^{d}} \E_{b, a} \left[ y_{B i + j} \cdot  \bar{y_{B i' + j}} \right],
\end{align}
where the second step follows as $|z|^2 = z \bar{z}$ for any complex number $z \in \C$; and the last step follows from the linearity of expectation.

As a premise of the current lemma, the signal $x(t) = x^*(t) + g(t) = g(t)$ for any $t \in [0, T]^{d}$. Due to Line~\ref{alg:HashToBins_multi:y} of {\HashToBins}, for any pair $i \in [D]^{d}$ and any $j \in [B]^{d}$ we have
\begin{eqnarray*}
    y_{B i + j}
    & = & \mathcal{G}(B i + j) \cdot \mathcal{P}_{\Sigma, b, a} g(B i + j) \\
    & = & \mathcal{G}(B i + j) \cdot g\big(\Sigma^{\top} (B i + j + a)\big) \cdot e^{-2 \pi \i \cdot b^{\top} (B i + j)} \\
    & = & S_{i, j} \cdot e^{-2 \pi \i \cdot b^{\top} (B i + j)},
\end{eqnarray*}
where the second step is by Definition~\ref{def:permutation_multi}; and in the last step we denote
\begin{align*}
    S_{i, j} = \mathcal{G}(B i + j) \cdot g\big(\Sigma^{\top} (B i + j + a)\big)
\end{align*}
for ease of notation. Notice that $S_{i, j} \in \R$ is a real number and is determined by the random matrix $\Sigma \in \R^{d \times d}$ and the random vector $a \in \R^{d}$ (see Definition~\ref{def:HashToBins_multi_parameters}).

We can reformulate the first term in Equation~\eqref{eq:lem:hashtobins_multi:error:noise:1} as follows: for each $i \in [D]^{d}$,
\begin{align*}
    \E_{b, a} \left[ y_{B i + j} \cdot  \bar{y_{B i + j}} \right]
    = \E_{a} \left[ S_{i, j} \cdot \bar{S_{i, j}} \right]
    = \E_{a} \left[ S_{i, j}^2 \right]
\end{align*}

Indeed, the second term in Equation~\eqref{eq:lem:hashtobins_multi:error:noise:1} equals zero. Particularly, for any $i \neq i' \in [D]^{d}$ we have
\begin{align}
    \notag
    \E_{b, a} \left[y_{B i + j} \cdot  \bar{y_{B i' + j}}\right]
    & = \E_{b, a} \left[S_{i, j} \cdot \bar{S_{i', j}} \cdot e^{-2 \pi \i B \cdot b^{\top} (i - i')}\right] \\
    \label{eq:lem:hashtobins_multi:error:noise:2}
    & = \E_{a} \left[S_{i, j} \cdot \bar{S_{i', j}}\right] \cdot \E_{b, a} \left[e^{-2 \pi \i B \cdot b^{\top} (i - i')}\right],
\end{align}
where the second step follows since $b \in \R^{d}$ and $S_{i, j}$ are independent (see Definition~\ref{def:HashToBins_multi_parameters}).

Let us investigate the second term in Equation~\eqref{eq:lem:hashtobins_multi:error:noise:2}. We have $b^{\top} (i - i') = \sum_{r \in [d]} b_{r} \cdot (i_{r} - i_{r}')$, in which at least one summand is non-zero (since $i \neq i \in [D]^{d}$). Since both of $B$ and $i_r-i_r'$ are integers and each coordinate $b_{r} \sim \unif[0, 1]$ is independently, the fractional part $\mathrm{frac}(B \cdot b^{\top} (i - i'))$ must follow the distribution $\unif[0, 1)$. Accordingly, the second term in Equation~\eqref{eq:lem:hashtobins_multi:error:noise:2} is equal to
\begin{align*}
    \E_{b} \left[e^{-2 \pi \i B \cdot b^{\top} (i - i')}\right] = 0,
\end{align*}

Applying all of the above arguments to Equation~\eqref{eq:lem:hashtobins_multi:error:noise:1} leads to $\E_{b, a} [| u_{j} |^{2}] = \sum_{i \in [D]^{d}} \E_{a} [ S_{i, j}^2 ]$. Taking all vector indices $j \in [B]^{d}$ into account, we infer that
\begin{eqnarray}
    \E_{b, a} \big[\left\| \hat{u} \right\|_{2}^{2}\big]
    & = & \B \cdot \sum_{j \in [B]^{d}} \E_{b, a} \left[| u_{j} |^{2}\right]
    \notag \\
    & = & \B \cdot \sum_{j \in [B]^{d}} \sum_{i \in [D]^{d}} \E_{a} [ S_{i, j}^2 ]
    \notag  \\
    & = & \B \cdot \sum_{j \in [B]} \sum_{i \in [D]} \E_{a} \left[ \mathcal{G}(B i + j)^2 \cdot g\big(\Sigma^{\top} (B i + j + a)\big)^2 \right]
    \notag  \\
    & = & \B \cdot \sum_{i \in [B D]^{d}} \mathcal{G}(i)^2 \cdot \E_{a} \left[ g\big(\Sigma^{\top} (i + a)\big)^2 \right],
    \label{eq:lem:hashtobins_multi:error:noise:3}
\end{eqnarray}
where the third step is by the definition of $S_{i, j}$; and the last step is by substitution.

Due to Condition~\ref{con:duration}, given any $i \in [B D]^{d}$ and any choice of the random matrix $\Sigma \in \R^{d \times d}$ according to Definition~\ref{def:HashToBins_multi_parameters}, the random vector $a \in \R^{d}$ satisfies that
\begin{eqnarray*}
    \E_{a} \left[ g\big(\Sigma^{\top} (i + a)\big)^2 \right]
    & \lesssim & \frac{1}{T^{d}} \cdot \int_{t \in [0, T]^{d}} |g(t)|^{2} \cdot \d t,
\end{eqnarray*}

Plugging the above equation into Equation~\eqref{eq:lem:hashtobins_multi:error:noise:3} gives
\begin{eqnarray*}
    \E_{b, a} \big[\left\| \hat{u} \right\|_{2}^{2}\big]
    & \lesssim & \frac{1}{T^{d}} \cdot \int_{t \in [0, T]^{d}} |g(t)|^{2} \cdot \d t \cdot \Big( \B \cdot \sum_{i \in [B D]^{d}} \mathcal{G}(i)^2 \Big) \\
    & \lesssim & \frac{1}{T^{d}} \cdot \int_{t \in [0, T]^{d}} |g(t)|^{2} \cdot \d t,
\end{eqnarray*}
where the last step follows from Property~V of Lemma~\ref{lem:filter_function_multi} that $\sum_{i \in \mathbb{Z}^{d}} \mathcal{G}(i)^2 \leq e^2 \cdot \mathcal{B}^{-1}$.

This completes the proof of Lemma~\ref{lem:hashtobins_multi_error:noise}.
\end{proof}

\subsection{{\HashToBins}: error due to bad events}
\label{subsec:HashToBins_multi_error:bad_events}

\begin{lemma}[The error due to bad events]
\label{lem:hashtobins_multi_error:bad_events}
Suppose $g(t) = 0$ for any $t \in [0, T]^{d}$. Given the hash function $\mathpzc{h}_{\Sigma, b}$ under any $\Sigma \in \R^{d \times d}$ and any $b \in \R^{d}$ (according to Definition~\ref{def:HashToBins_multi_parameters}), denote by
\[
    H = \{f \in \supp(\hat{x^*}): \mbox{neither $E_{\coll}(f)$ nor $E_{\off}(f)$ happens}\}
\]
the set of ``good'' tone frequencies. Then the following hold:
\begin{eqnarray*}
    \forall f \in H: ~ \E_{a} \left[\left|\hat{u}_{\mathpzc{h}_{\Sigma, b}(f)} - \hat{x^*}[f] \cdot e^{2 \pi \i \cdot a^{\top} \Sigma f}\right|^{2}\right]
    & \leq & \frac{\delta}{\poly(k, d)} \cdot \|\hat{x^*}\|_{1}^{2}, \\
    \sum_{f \in H} \E_{a} \left[\left|\hat{u}_{\mathpzc{h}_{\Sigma, b}(f)}  - \hat{x^*}[f] \cdot e^{2 \pi \i \cdot a^{\top} \Sigma f}\right|^{2}\right]
    & \leq & \frac{\delta}{\poly(k, d)} \cdot \|\hat{x^*}\|_{1}^{2},
\end{eqnarray*}

\end{lemma}

\begin{proof}
As promised by the lemma, the signal $x(t) = x^*(t) + g(t) = x^*(t)$ for any $t \in [0, T]^{d}$. For a specific frequency $f' \in H$, w.l.o.g.\ we assume that $f'$ is hashed into the $j$-th bin, namely $\mathpzc{h}_{\Sigma, b}(f') = j$ (see Definition~\ref{def:function_hash_multi}). According to Property~II of Fact~\ref{fac:HashToBins_multi_identity},
\begin{align}\label{eq:lem:multi_hashtobins_error:bad_events:1}
    \hat{u}_{j}
    = & ~ \hat{\mathcal{G}} * \hat{\mathcal{P}_{\Sigma, b, a} x^*}(B^{-1} \cdot j)
    \notag \\
    = & ~ \hat{\mathcal{G}'} * \hat{\mathcal{P}_{\Sigma, b, a} x^*}(B^{-1} \cdot j) + (\hat{\mathcal{G}} - \hat{\mathcal{G}'}) * \hat{\mathcal{P}_{\Sigma, b, a} x^*}(B^{-1} \cdot j).
\end{align}

For the second summand in Equation~\eqref{eq:lem:multi_hashtobins_error:bad_events:1}, the corresponding function admits the $\ell_{\infty}$ norm of
\begin{align*}
    \big\|(\hat{\mathcal{G}} - \hat{\mathcal{G}'}) * \hat{\mathcal{P}_{\Sigma, b, a} x^*}\big\|_{\infty}
     \leq  \big\|(\hat{\mathcal{G}} - \hat{\mathcal{G}'})\big\|_{\infty} \cdot \big\|\hat{\mathcal{P}_{\Sigma, b, a} x^*}\big\|_{1} 
     \leq \frac{\delta}{\poly(k, d)} \cdot \big\|\hat{\mathcal{P}_{\Sigma, b, a} x^*}\big\|_{1}
\end{align*}
where the second step is due to Property~IV of Lemma~\ref{lem:shifted_window_function_multi}.

We then have
\begin{align*}
     \ \big\|\hat{\mathcal{P}_{\Sigma, b, a} x^*}\big\|_{1} 
    = & ~  \int_{z \in \R^{d}} \bigg| \hat{x^*}(\Sigma^{-1} z + b) \cdot \det(\Sigma)^{-1} \cdot e^{2 \pi \i \cdot a^{\top} (z + \Sigma b)} \bigg| \cdot \d z  \notag \\
    = & ~ \int_{z \in \R^{d}} \left| \hat{x^*}(\Sigma^{-1} z + b) \right| \cdot \det(\Sigma)^{-1} \cdot \d z \notag \\
    = & ~ \int_{z \in \R^{d}} \big| \hat{x^*}(\xi) \big| \cdot \d \xi \notag \\
    = & ~ \big\|\hat{x^*}\big\|_{1},
\end{align*}
where the first step applies Property~II of Lemma~\ref{lem:permutation_multi_identity}; the second step follows since $|e^{\i \theta}| = 1$ for any $\theta \in \R$ and $\det(\Sigma) \neq 0$ (Definition~\ref{def:HashToBins_multi_parameters}); and the third step is by substitution.

Putting the equations together, we get
\begin{align}\label{eq:lem:multi_hashtobins_error:bad_events:2}
    \big\|(\hat{\mathcal{G}} - \hat{\mathcal{G}'}) * \hat{\mathcal{P}_{\Sigma, b, a} x^*}\big\|_{\infty} \leq \frac{\delta}{\poly(k, d)} \cdot \big\|\hat{x^*}\big\|_{1}.
\end{align}

Moreover, the first summand in Equation~\eqref{eq:lem:multi_hashtobins_error:bad_events:1} equals
\begin{eqnarray}
    \notag
    \hat{\mathcal{G}'} * \hat{\mathcal{P}_{\Sigma, b, a} x^*}(B^{-1} \cdot j)
    & = & \int_{\xi \in \R^{d}} \hat{\mathcal{G}'}(B^{-1} \cdot j - \xi) \cdot \hat{\mathcal{P}_{\Sigma, b, a} x^*}(\xi) \cdot \d \xi \\
    \notag
    & = & \int_{\xi \in \R^{d}} \hat{\mathcal{G}'}(B^{-1} \cdot j - \xi) \cdot \hat{x^*}(\Sigma^{-1} \xi + b) \cdot \det(\Sigma)^{-1} \cdot e^{2 \pi \i \cdot a^{\top} (\xi + \Sigma b)} \cdot \d \xi \\
    \label{eq:lem:multi_hashtobins_error:bad_events:3}
    & = & \int_{\xi \in \R^{d}} \hat{\mathcal{G}'}(B^{-1} \cdot j - \Sigma (\xi - b)) \cdot \hat{x^*}(\xi) \cdot e^{2 \pi \i \cdot a^{\top} \Sigma \xi} \cdot \d \xi,
\end{eqnarray}
where the first step applies the convolution operation; the second step follows from Property~II of Lemma~\ref{lem:permutation_multi_identity}; and the third step is by substitution.

Notably (see Properties~I to III of Lemma~\ref{lem:shifted_window_function_multi} and Definition~\ref{def:grid}), the standard window $\hat{\mathcal{G}'}(\xi)$ is supported within the hypercube grid
\begin{align*}
    \Lambda_{W}\Big(\frac{1}{2 B}\Big) = \Big\{\xi \in \R^{d}: \|\xi - i\|_{\infty} \leq \frac{1}{2 B} \mbox{ for some vector index } i \in [-W: W]^{d} \Big\}.
\end{align*}

Recall that $\Sigma$ is a random rotation matrix scaled by a random factor $\beta \sim \unif[\frac{2 \sqrt{d}}{B \eta}, \frac{4 \sqrt{d}}{B \eta}]$ (see Definition~\ref{def:HashToBins_multi_parameters}). Further, the Fourier spectrum  $\supp(\hat{x^*}) \subseteq [-F, F]^{d}$ is bounded. Under any choice of the random matrix $\Sigma$ of the random vector $b \in \R^{d}$ (see Definition~\ref{def:HashToBins_multi_parameters}), for any $j \in [B]^{d}$ and any $\xi \in [-F, F]^{d}$ we have
\begin{eqnarray*}
    \big\| B^{-1} \cdot j - \Sigma (\xi - b) \big\|_{\infty}
    & \lesssim & \big\| \Sigma (F, F, \cdots, F)^{\top} \big\|_{\infty} \\
    & \leq & \big\| \Sigma (F, F, \cdots, F)^{\top} \big\|_{2} \\
    & \leq & \frac{4 \sqrt{d}}{B \eta} \cdot \big\| (F, F, \cdots, F)^{\top} \big\|_{2} \\
    & \lesssim & d \cdot \frac{F}{B \eta}.
\end{eqnarray*}
Thus, a sufficiently large width parameter $W = \Theta(d \cdot \frac{F}{B\eta})$ (see Definition~\ref{def:shift_filter_function_multi}) guarantees that
\begin{align*}
    \Big\{ \frac{1}{B} \cdot j - \Sigma (\xi - b): \xi \in [-F, F]^{d} \Big\}
    ~ \subseteq ~ [-W, W]^{d},
\end{align*}
for any choice of $\Sigma$ and $b$ (according to Definition~\ref{def:HashToBins_multi_parameters}) and any $j \in [B]^{d}$.

\begin{figure}
    \centering
    \includegraphics[width = .6\textwidth]{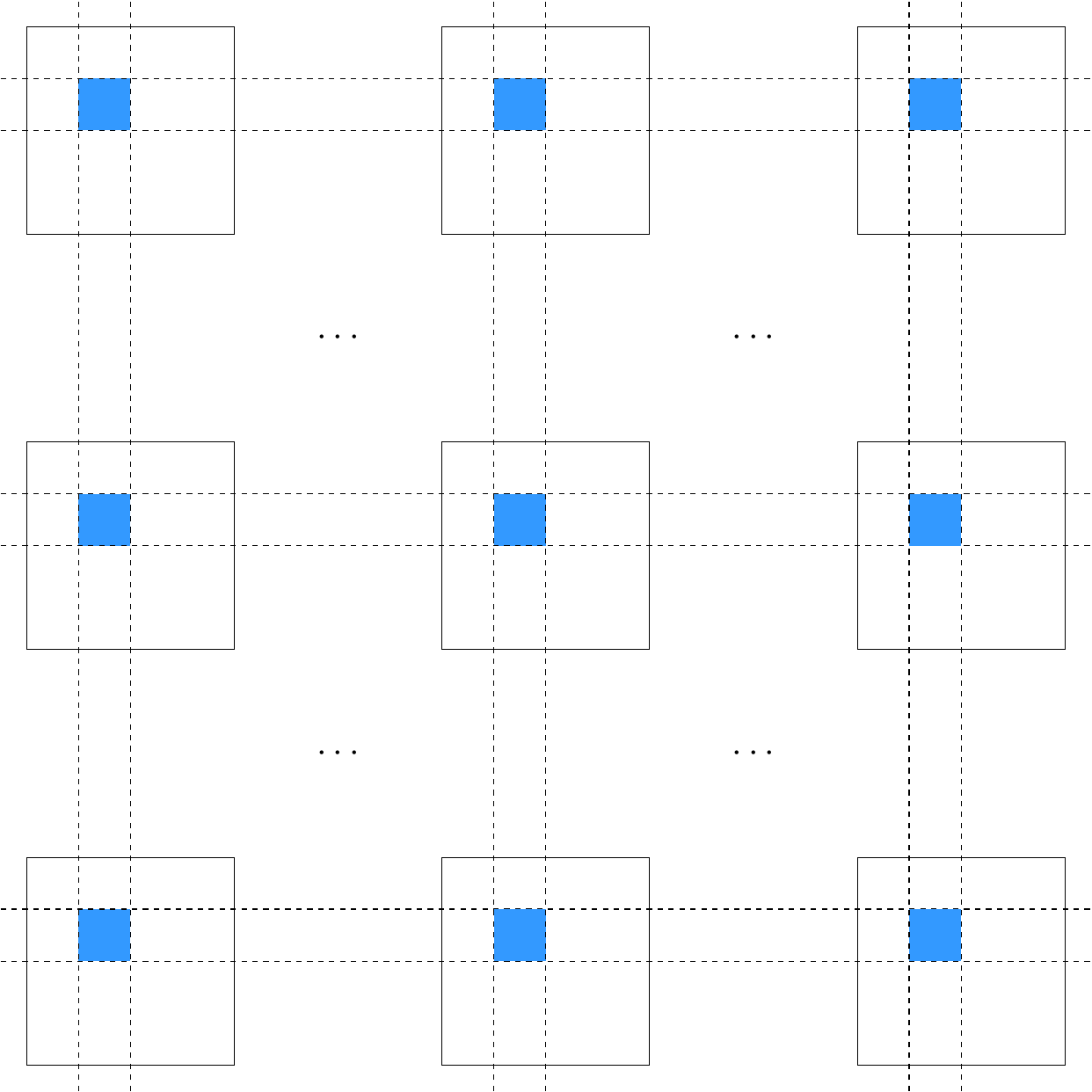}
    \caption{Demonstration for Lemma~\ref{lem:hashtobins_multi_error:bad_events} in two dimensions $d = 2$, where the blue subsquares together denote the region $\{z \in \R^{2}: |B^{-1} \cdot j - z - i|_{\infty} \leq \frac{1}{2 B} \mbox{ for some } i \in \mathbb{Z}^{2}\}$.}
    \label{fig:hashtobins_multi_error:bad_events}
\end{figure}

Given the above arguments and as Figure~\ref{fig:hashtobins_multi_error:bad_events} suggests, Equation~\eqref{eq:lem:multi_hashtobins_error:bad_events:3} suffices to integrate the tone frequencies hashed into the $j$-th bin, namely
\begin{eqnarray*}
    \Phi_{j}
    & = & \left\{\xi \in \supp(\hat{x^*}): \big|B^{-1} \cdot j - \Sigma (\xi - b) - i\big|_{\infty} \leq \frac{1}{2 B} \mbox{ for some } i \in \mathbb{Z}^{d}\right\} \\
    & = & \left\{\xi \in \supp(\hat{x^*}): -\frac{1}{2 B} \cdot \mathbf{1} \preceq B^{-1} \cdot j - \Sigma (\xi - b) - i \preceq \frac{1}{2 B} \cdot \mathbf{1} \mbox{ for some } i \in \mathbb{Z}^{d}\right\} \\
    & = & \left\{\xi \in \supp(\hat{x^*}): \frac{1}{B} \cdot j \preceq \frac{1}{2 B} \cdot \mathbf{1} + \Sigma (\xi - b) + i \preceq \frac{1}{B} \cdot (j + \mathbf{1}) \mbox{ for some } i \in \mathbb{Z}^{d}\right\} \\
    & = & \left\{\xi \in \supp(\hat{x^*}): j \preceq B \cdot \Big(\frac{1}{2 B} \cdot \mathbf{1} + \Sigma (\xi - b) + i\Big) \preceq j + \mathbf{1} \mbox{ for some } i \in \mathbb{Z}^{d}\right\}.
\end{eqnarray*}

In the above condition, $\frac{1}{2 B} \cdot \mathbf{1} + \Sigma (\xi - b) + i$ must be bounded within $[0, 1]^{d}$, as the concerning bin $j \in [B]^{d} = \{0, 1, \cdots, B - 1\}^{d}$. In particular, the case $\|\frac{1}{2 B} + \Sigma (\xi - b) + i\|_{\infty} = 1$ occurs with zero probability, since $b \in \R^d$ follows a {\em continuous} uniform distribution (Definition~\ref{def:HashToBins_multi_parameters}); we safely ignore this case. Given the hash function $\mathpzc{h}_{\Sigma, b}(\xi) \in [B]^{d}$ in Definition~\ref{def:function_hash_multi}, we conclude that
\begin{align*}
    \Phi_{j} = \left\{\xi \in \supp(\hat{x^*}): \mathpzc{h}_{\Sigma, b}(\xi) = j\right\}.
\end{align*}

The concerning tone frequency $f' \in H$ ensures that neither $E_{\coll}(f)$ nor $E_{\off}(f)$ happens:
\begin{itemize}
    \item $E_{\coll}(f')$ does not happen. No other tone frequencies $f \in \supp(\hat{x^*}) \setminus \{f'\}$ collide with $f'$ after the hashing. That is, the $j$-th bin contains $f'$ as the only tone frequency, namely $\Phi_{j} = \{f'\}$.
    
    \item $E_{\off}(f')$ does not happen. The offset $\|\mathpzc{o}_{\Sigma, b}(f')\|_{\infty} < \frac{1 - \alpha}{2 B}$ is small enough, namely the frequency $f'$ lies within the hypercube grid $\Lambda_{W}(\frac{1 - \alpha}{2 B})$. We know from Property~I of Lemma~\ref{lem:shifted_window_function_single} that
    \begin{align*}
        \hat{\mathcal{G}'}(B^{-1} \cdot j - \Sigma (f' - b)) = 1
    \end{align*}
\end{itemize}
Also, recall Observation~\ref{obs:tones_CFT} that $\hat{x^*}(\xi)$ is the combination of $k$ many scaled $d$-dimensional Dirac delta functions (at the tone frequencies $\xi \in \supp(\hat{x^*})$). In precise, for any frequency $\xi \in [-F, F]^{d}$,
\begin{align*}
    \hat{x^*}(\xi) = \sum_{f \in \supp(\hat{x^*})} \hat{x^*}[f] \cdot \Dirac_{= f}(\xi).
\end{align*}

Applying all of the above arguments to Equation~\eqref{eq:lem:multi_hashtobins_error:bad_events:3} results in
\begin{eqnarray*}
    \hat{\mathcal{G}'} * \hat{\mathcal{P}_{\Sigma, b, a} x^*}(j / B)
    & = & \int_{\xi \in \R^{d}} \hat{x^*}[f'] \cdot \Dirac_{= f'}(\xi) \cdot e^{2 \pi \i \cdot a^{\top} \Sigma \xi} \cdot \d \xi \\
    & = & \hat{x^*}[f'] \cdot e^{2 \pi \i \cdot a^{\top} \Sigma f'}.
\end{eqnarray*}
This equation, together with Equation~\eqref{eq:lem:multi_hashtobins_error:bad_events:1} and Equation~\eqref{eq:lem:multi_hashtobins_error:bad_events:2}, implies that
\begin{align*}
    & \left|\hat{u}_{j}  - \hat{x^*}[f'] \cdot e^{2 \pi \i \cdot a^{\top} \Sigma f'}\right|
    \leq \frac{\delta}{\poly(k, d)} \cdot \|\hat{x^*}\|_{1} 
\end{align*}
Taking square on the both sides:
\begin{align*}
     \left|\hat{u}_{j} - \hat{x^*}[f'] \cdot e^{2 \pi \i \cdot a^{\top} \Sigma f'}\right|^{2}
    ~ \leq ~ \frac{\delta^2}{\poly(k, d)} \cdot \|\hat{x^*}\|_{1}^{2}
    ~ \leq ~ \frac{\delta}{\poly(k, d)} \cdot \|\hat{x^*}\|_{1}^{2},
\end{align*}
where the last step is by $0 < \delta < 1$; note that $j = \mathpzc{h}_{\Sigma, b}(f') \in [B]^{d}$.

Finally, we note that $|H| \leq k$, since $H \subseteq \supp(\hat{x^*})$ and there are just $k$ many tone frequencies $f' \in \supp(\hat{x^*})$. Apply the last inequality over all $f' \in H$ and take the expectation over the random vector $a \in \R^{d}$ (see Definition~\ref{def:HashToBins_multi_parameters}), then Lemma~\ref{lem:hashtobins_multi_error:bad_events} follows.
\end{proof}

\subsection{Performance guarantees}

\begin{lemma}[Performance guarantee for {\HashToBins}]
\label{lem:hashtobins_guarantee}
Recall Theorem~\ref{thm:intro_tone} for the $\ell_{2}$-norm noise level
\begin{align*}
    \N^2 ~ := ~ \underbrace{ \frac{1}{T^d} \cdot \int_{t \in [0,T]^d} |g(t)|^2 \cdot \d t }_{ \N_g^2 } ~ + ~ \delta \cdot \underbrace{ \sum_{i \in [k]} |\hat{x^*}[f_{i}]|^2 }_{ \N_v^2 } .
\end{align*}
Sample the matrix $\Sigma \in \R^{d \times d}$ and the vectors $b \in \R^{d}$ according to Definition~\ref{def:HashToBins_multi_parameters}, and suppose that Conditions~\ref{con:duration} and \ref{con:sampling} hold for the random vector $a \in \R^{d}$. Consider the ``good'' frequencies
\begin{align*}
    H ~ := ~ \{f \in \supp(\hat{x^*}): \mbox{neither $E_{\coll}(\xi)$ nor $E_{\off}(\xi)$ happens}\},
\end{align*}
and the bins $I := [\mathcal{B}] \setminus \mathpzc{h}_{\Sigma, b}(\supp(\hat{x^*}))$  with no frequency $\{f_{i}\}_{i \in [k]} = \supp(\hat{x^*})$ hashed into.

Then given any $\Sigma \in \R^{d \times d}$, the following holds for each good frequency $f \in H$:
\begin{eqnarray*}
     \E_{b, a} \Big[ \Big|\hat{u}_{\mathpzc{h}_{\Sigma, b}(f)} - \hat{x}[f] \cdot e^{2 \pi \i \cdot a^{\top} \Sigma f} \Big|^{2}\Big] 
     ~ \lesssim ~ \mathcal{B}^{-1} \cdot \N_g^2 ~ + ~ k^{-1} \cdot \frac{\delta}{\poly(k, d)} \cdot \N_v^2
\end{eqnarray*}
And take all good frequencies $f \in H$ and all bins $j \in I$ into account:
\begin{align*}
    \E_{b, a} \Big[ \sum_{f \in H} \Big|\hat{u}_{\mathpzc{h}_{\Sigma, b}(f)}  - \hat{x^*}[f] \cdot e^{2 \pi \i \cdot a^{\top} \Sigma f} \Big|^{2}
    ~ + ~ \sum_{j \in I} |\hat{u}_i|^2 \Big]
    ~ \lesssim ~ \N^{2}
\end{align*}

\end{lemma}

\begin{proof}
This can be easily seen by combining Lemmas~\ref{lem:hashtobins_multi_error:noise} and \ref{lem:hashtobins_multi_error:bad_events}.
\end{proof}

% \newpage
\section{Locate inner}
\label{sec:locate_inner}

\begin{table}[h]
    \centering
    \begin{tabular}{|l|l|l|l|}
    \hline
        {\bf Statement} & {\bf Section} & {\bf Algorithm} & {\bf Comment} \\ \hline
        Definitions~\ref{def:locate_inner_setup} and \ref{def:sub_hyperball} & Section~\ref{sec:locate_inner_algorithm} & Algorithm~\ref{alg:locate_inner} & Definitions \\ \hline
        Lemma~\ref{lem:locate_inner_sample_time} & Section~\ref{sec:locate_inner_sample_time} & Algorithm~\ref{alg:locate_inner} & Sample complexity and running time \\ \hline
        %Lemma~\ref{lem:locate_inner_duration} & Section~\ref{sec:locate_inner_duration} & Algorithm~\ref{alg:locate_inner} & Duration \\ \hline
        Lemma~\ref{lem:locate_inner_voting} & Section~\ref{sec:locate_inner_voting} & Algorithm~\ref{alg:locate_inner} & Voting process \\\hline
        Lemma~\ref{lem:locate_inner_election} & Section~\ref{sec:locate_inner_election} & Algorithm~\ref{alg:locate_inner} & Election process\\\hline
        Lemma~\ref{cor:locate_inner_guarantees} & Section~\ref{sec:locate_inner_guarantees} & Algorithm~\ref{alg:locate_inner} & Guarantees \\\hline
        Lemmas~\ref{lem:locate_inner_duration_require} and \ref{lem:locate_inner_sampling_require} & Section~\ref{sec:locate_inner_time_points} & Algorithm~\ref{alg:sample_time_point} & Sampling scheme \\ \hline
        Lemma~\ref{lem:locate_inner_stronger} & Section~\ref{sec:locate_inner_stronger} & Algorithm~\ref{alg:locater_inner_stronger} & Stronger guarantees \\ \hline
    \end{tabular}
    \caption{List of Lemmas/Algorithms in locate inner section.}
    \label{tab:list_one_sparse_recovery}
\end{table}

\subsection{Definitions and algorithm}
\label{sec:locate_inner_algorithm}

\begin{definition}[Setup for {\LocateInner}]
\label{def:locate_inner_setup}
%\label{def:setup_LocateInner}
We adopt the following notations:
\begin{itemize}
    \item The guessed approximation ratio $\mathcal{C} \in [120, \rho]$.
    
    \item Let $M = 4 \cdot \lceil 4 \sqrt{d} \cdot \mathcal{C}^{2 / 3} \rceil \in \mathbb{N}_{\geq 1}$.
    
    \item Let $\varpi = \mathcal{C}^{-2 / 3}$; this parameter will be used in the voting scheme (see Definition~\ref{def:voting}).
    
    \item $\phi_{j} = \arg(\hat{u}_{j}) - \arg(\hat{u}_{j}')$ denotes the phase difference between $\hat{u}_{j}$ and $\hat{u}_{j}'$;
    
    \item The number of iterations $\mathcal{R}_{\mathrm{vote}} = \Theta(d \cdot \log (\mathcal{C} \cdot d) + \log \log (F / \eta))$ is sufficiently large.
\end{itemize}
\end{definition}

\begin{figure}
    \centering
    \includegraphics[scale = 0.7]{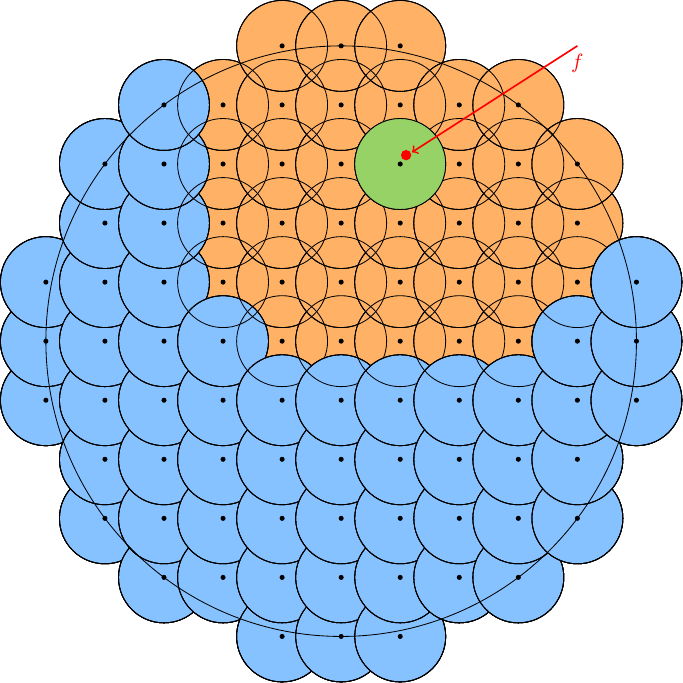}
    \caption{Demonstration for Definition~\ref{def:sub_hyperball} when $d = 2$. The ``red'' point means the true frequency $f$. The ``lime'' region means the {\em true} sub-hyperball $\mathbf{HB}(f_{q^*}^{\mathpzc{grid}[j]}, \frac{1}{M} \cdot L^{\mathpzc{dia}}) \ni f$. The ``blue'' regions mean the {\em wrong} sub-hyperballs, and the remaining ``yellow'' regions (with overlapping parts) represent the {\em intermediate} sub-hyperballs.}
    \label{fig:sub_hyperball}
\end{figure}

\begin{definition}[Hyperball and sub-hyperballs]
\label{def:sub_hyperball}
For any frequency $\mathpzc{List}[j] \in \R^{d}$ and any $L^{\mathpzc{dia}} \geq 0$, $\mathbf{HB}(\mathpzc{List}[j], L^{\mathpzc{dia}})$ denotes the $\ell_{2}$-norm hyperball with center $\mathpzc{List}[j]$ and diameter $L^{\mathpzc{dia}}$:
\begin{eqnarray*}
    \mathbf{HB}(\mathpzc{List}[j], L^{\mathpzc{dia}})
    & := & \left\{\xi \in \R^{d}: \big\| \xi - \mathpzc{List}[j] \big\|_{2} \leq L^{\mathpzc{dia}} / 2\right\}.
\end{eqnarray*}
Let $\bigcup_{q \in Q} \mathbf{HB}(f_{q}^{\mathpzc{grid}[j]}, \frac{1}{M} \cdot L^{\mathpzc{dia}})$ denote a cover of $\mathbf{HB}(\mathpzc{List}[j] L^{\mathpzc{dia}})$, by using a minimum amount of sub-hyperballs that have the diameter $\frac{1}{M} \cdot L^{\mathpzc{dia}}$ each. At most $\mathcal{M} := |Q| = (4 M \cdot \sqrt{d})^{d} = 2^{\Theta(d \cdot \log(\mathcal{C} \cdot d))}$ many sub-hyperballs can be used, namely the {\em external covering number} \cite[Page~337]{sb14}.

Given that the hyperball $\mathbf{HB}(\mathpzc{List}[j], L^{\mathpzc{dia}})$ contains the targeted tone frequency $f \in [-F, F]^{d}$, all the sub-hyperballs can be classified into three groups (as Figure~\ref{fig:sub_hyperball} shows):
\begin{itemize}
    \item The {\em true} sub-hyperball $\mathbf{HB}(f_{q^*}^{\mathpzc{grid}[j]}, \frac{1}{M} \cdot L^{\mathpzc{dia}}) \ni f$, for some index $q^* \in Q$. For convenience, assume that the true sub-hyperball is unique, namely the targeted tone frequency is not on the boundary of two or more sub-hyperballs.\footnote{We make this assumption just to specify {\em the} true sub-hyperball; our proof does not rely on the assumption.}
    
    \item The {\em wrong} sub-hyperballs $q \in Q \setminus \{ q^* \}$ have the $\ell_{2}$-distances
    \[
        \Big\| f_{q^*}^{\mathpzc{grid}[j]} - f_{q}^{\mathpzc{grid}[j]} \Big\|_{2}
        ~ \geq ~ \frac{1}{M} \cdot L^{\mathpzc{dia}} \cdot \big\lceil 4 \sqrt{d} / \varpi \big\rceil
        ~ = ~ \frac{1}{M} \cdot L^{\mathpzc{dia}} \cdot \big\lceil 4 \sqrt{d} \cdot \mathcal{C}^{2 / 3} \big\rceil.
    \]

    \item The remaining sub-hyperballs are called the {\em intermediate} sub-hyperballs.
\end{itemize}
\end{definition}

\begin{figure}
    \centering
    \includegraphics[width = .6\textwidth]{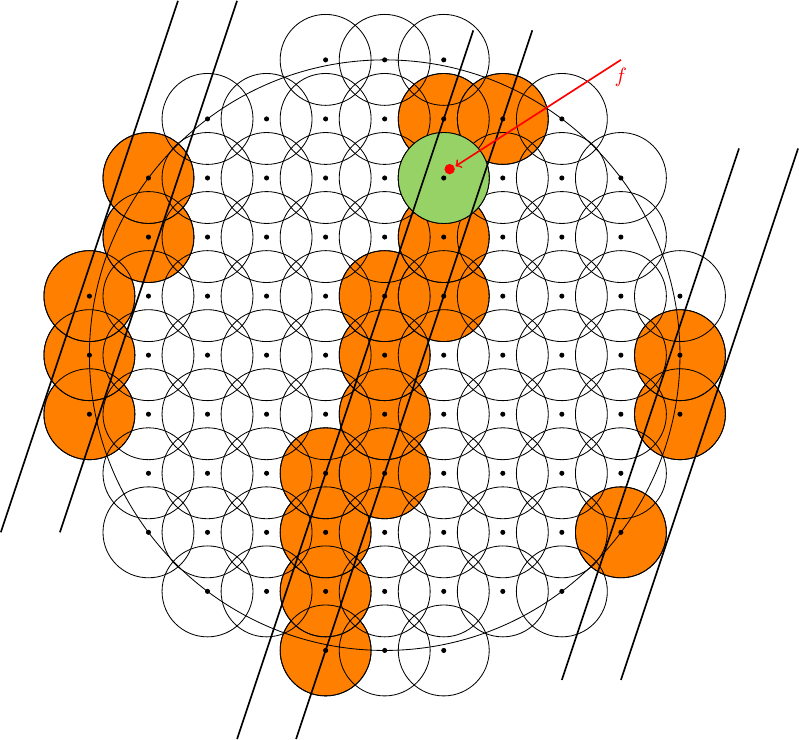}
    \caption{Demonstration for the voting scheme (Definition~\ref{def:voting})}
    \label{fig:voting}
\end{figure}

\begin{definition}[Voting scheme]
\label{def:voting}
Let $\|\theta\|_{\bigcirc} \in [0, \pi]$ denote the ``phase distance'' from $e^{\i 0} = 1$ to any $\theta \in \R$. As Figure~\ref{fig:voting} suggests, given any $\phi_{j} \in \R$ and any $\Delta_{a} := \Delta_{a}^{r}$, let $\mathpzc{Vote}_{j}[q] \gets \mathpzc{Vote}_{j}[q] + 1$ (namely adding a vote to any sub-hyperball $q \in Q$) for which
\begin{align*}
    \Big\|\phi_{j} - 2 \pi \cdot \Delta_{a}^{\top} f_{q}^{\mathpzc{grid}[j]} \Big\|_{\bigcirc}
    ~ \leq ~ \pi \cdot \varpi.
\end{align*}
\end{definition}

\begin{algorithm}[!t]
\caption{{\LocateInner}, Lemmas~\ref{lem:locate_inner_sample_time}, \ref{lem:locate_inner_voting}, \ref{lem:locate_inner_election}, \ref{cor:locate_inner_guarantees}}\label{alg:locate_inner}
\begin{algorithmic}[1]
\Procedure{\LocateInner}{$\Sigma, b, D, \mathpzc{List}, L^{\mathpzc{dia}}, \mathcal{C}, T$} 

\State Define $M \in \mathbb{N}_{\geq 1}$ according to Definition~\ref{def:locate_inner_setup}.
\Comment{Definition~\ref{def:sub_hyperball}}

\For{$j \in [B]^{d}$ that $\mathpzc{List}[j] \neq \mathtt{NIL}$}
    \State Cover $\mathbf{HB}(\mathpzc{List}[j], L^{\mathpzc{dia}})$ via sub-hyperballs $\bigcup_{q \in Q} \mathbf{HB}(f_{q}^{\mathpzc{grid}[j]}, \frac{1}{M} \cdot L^{\mathpzc{dia}})$.
\EndFor

\Statex

\State Initialize $\mathpzc{Vote}_{j}[q] = 0$ for all $q \in Q$ and each $j \in [B]^{d}$.
\Comment{Voting process.~Lemma~\ref{lem:locate_inner_voting}}

\State Define $\mathcal{R}_{\mathrm{vote}} \in \mathbb{N}_{\geq 1}$ according to Definition~\ref{def:locate_inner_setup}.

\For{$r = 1, 2, \cdots, \mathcal{R}_{\mathrm{vote}}$} \label{lin:locate_inner_voting_start}

    \State $(a_{r}, \Delta_{a}^{r}) \gets \SampleTimePoint(M, L^{\mathpzc{dia}}, \mathcal{C}, T)$.
    \Comment{Algorithm~\ref{alg:sample_time_point}}
    \label{lin:locate_inner_voting_sample}
    
    \State $\hat{u} \gets \HashToBins(x, \Sigma, b, a_{r}, D)$.
    \label{lin:locate_inner_voting_a}
    
    \State $\hat{u}' \gets \HashToBins(x, \Sigma, b, a_{r} + \Delta_{a}^{r}, D)$.
    \label{lin:locate_inner_voting_a_Delta}
    
    \State Let $\phi = (\phi_{j})_{j = 1}^{B}$, where $\phi_{j} = \arg(\hat{u}_{j}) - \arg(\hat{u}_{j})$.
    
    \For{$j \in [B]^{d}$ that $\mathpzc{List}[j] \neq \mathtt{NIL}$}
        \State Update $\mathpzc{Vote}_{j}[q]$ according to Definition~\ref{def:voting}.
    \EndFor
    
\EndFor\label{lin:locate_inner_voting_end}

\Statex

\State Initialize $\mathpzc{List}_{\mathrm{new}}[j] = \mathtt{NIL}$ for each $j \in [B]^{d}$.
\Comment{Election process.~Lemma~\ref{lem:locate_inner_election}}

\For{$j \in [B]^{d}$ that $\mathpzc{List}[j] \neq \mathtt{NIL}$}
    \State $\mathpzc{Winner}[j] \gets \bigcup_{q \in Q: \mathpzc{Vote}_{j}[q] \geq \frac{1}{2} \cdot \mathcal{R}_{\mathrm{vote}}} \mathbf{HB}(f_{q}^{\mathpzc{grid}[j]}, \frac{1}{M} \cdot L^{\mathpzc{dia}})$.
    \label{alg:LocateInner:winner}
    
    \If{$\mathpzc{Winner}[j]$ is non-empty}
        \State Let $\mathpzc{List}_{\mathrm{new}}[j]$ be any frequency so that $\mathbf{HB}(\mathpzc{List}_{\mathrm{new}}[j], \frac{1}{2} \cdot L^{\mathpzc{dia}}) \supseteq \mathpzc{Winner}[j]$.
    \EndIf
\EndFor

\State \Return the frequencies $\mathpzc{List}_{\mathrm{new}}$.
\EndProcedure
\end{algorithmic}

\end{algorithm}

\subsection{Sample complexity and running time}
\label{sec:locate_inner_sample_time}

The goal of this section is to prove Lemma~\ref{lem:locate_inner_sample_time}.
\begin{lemma}[Sample complexity and running time of {\LocateInner}]
\label{lem:locate_inner_sample_time}
%Let $\mathcal{D} = D^d$. 
The procedure {\LocateInner} (Algorithm~\ref{alg:locate_inner}) has the following performance guarantees:
\begin{itemize}
    \item The sample complexity is 
    \begin{align*}
        \Theta(\mathcal{R}_{\mathrm{vote}} \cdot \mathcal{B} \mathcal{D}) = 2^{\Theta(d\cdot\log d)} \cdot (\log \mathcal{C} + \log \log (F / \eta)) \cdot k \cdot \mathcal{D}.
    \end{align*}
    \item The running time is
    \begin{align*}
        \Theta(\mathcal{R}_{\mathrm{vote}} \cdot (\mathcal{B} \mathcal{D} + \mathcal{B} \log \mathcal{B}+\mathcal{B}\mathcal{M}))= 2^{\Theta(d\cdot (\log d + \log {\cal C} ))} \cdot \log\log(F/\eta) \cdot k \cdot (\mathcal{D} + \log k)
    \end{align*}
\end{itemize}
\end{lemma}

\begin{proof}
Throughout the procedure {\LocateInner}, the subroutine {\HashToBins} (Algorithm~\ref{alg:HashToBins_multi}) is invoked $2 \cdot \mathcal{R}_{\mathrm{vote}}$ times. Recall Definition~\ref{def:locate_inner_setup} that $\mathcal{R}_{\mathrm{vote}} = \Theta\big(d \cdot \log (\mathcal{C} \cdot d) + \log \log (F / \eta)\big)$.

\vspace{.1in}
{\bf Sample complexity.}
The procedure {\LocateInner} takes samples only by invoking the subroutine {\HashToBins}. Due to Fact~\ref{fac:HashToBins_multi_sample_time}, {\HashToBins} has the sample complexity $O(\mathcal{B} \mathcal{D})$. Recall Definition~\ref{def:shift_filter_function_multi} that $\mathcal{B} = 2^{\Theta(d \cdot \log d)} \cdot k$. Thus, {\LocateInner} has the sample complexity
\begin{eqnarray*}
    \#\mathtt{sample}({\LocateInner})
    & = & \Theta(\mathcal{R}_{\mathrm{vote}} \cdot \mathcal{B} \cdot \mathcal{D})  \\
    & = & \underbrace{ \Theta\big(d \cdot \log (\mathcal{C} \cdot d) + \log \log (F / \eta)\big) }_{ {\cal R}_{\mathrm{vote}} } ~ \cdot ~ \underbrace{ 2^{\Theta(d\cdot\log d)}\cdot k  }_{ \mathcal{B} } ~ \cdot ~ \mathcal{D} \\
    & = & 2^{\Theta(d\cdot\log d)} \cdot (\log \mathcal{C} + \log \log (F / \eta)) \cdot k \cdot \mathcal{D}.
\end{eqnarray*}

\vspace{.1in}
{\bf Running time.}
The running time of {\LocateInner} is dominated by the $\mathcal{R}_{\mathrm{vote}}$ many loops for the voting process (namely the second for loop). Such a loop invokes the subroutine {\HashToBins} twice, and then update $\mathpzc{Vote}_{j}[q]$ for all $q \in Q$ and all $j \in [B]^{d}$. (The subroutine {\SampleTimePoint} runs in $O_{d}(1)$ time; see Algorithm~\ref{alg:sample_time_point}.)

Due to Fact~\ref{fac:HashToBins_multi_sample_time}, {\HashToBins} has the running time $O(\mathcal{B} \mathcal{D} + \mathcal{B} \log \mathcal{B})$ time. Thus, the time that {\LocateInner} spends on hashing is
\begin{eqnarray*}
    \#\mathtt{time}(\mbox{hashing})
    & = & \Theta(\mathcal{R}_{\mathrm{vote}}) ~ \cdot ~ (\mathcal{B} \mathcal{D} + \mathcal{B} \log \mathcal{B}) \\
    & = & \underbrace{ \Theta\big(d \cdot \log (\mathcal{C} \cdot d) + \log \log (F / \eta)\big) }_{ {\cal R}_{\mathrm{vote}} } ~ \cdot ~ \underbrace{2^{\Theta(d\cdot\log d)} \cdot k \cdot (\mathcal{D} + \log k)}_{\mathcal{B} \mathcal{D} + \mathcal{B} \log \mathcal{B}} \\
    & = & 2^{\Theta(d\cdot\log d)} \cdot (\log \mathcal{C} + \log \log (F / \eta)) \cdot k \cdot (\mathcal{D}+\log k)
\end{eqnarray*}

Further, the time that {\LocateInner} spends on voting is 
\begin{eqnarray*}
    \#\mathtt{time}(\mbox{voting})
    & = & \Theta(\mathcal{R}_{\mathrm{vote}} \cdot\mathcal{B} \cdot \mathcal{M}) \\
    & = & \underbrace{ \Theta(d \cdot \log (\mathcal{C} \cdot d) + \log \log (F / \eta)) }_{ {\cal R}_{\mathrm{vote}} }
    ~ \cdot ~ \underbrace{ 2^{\Theta(d \cdot \log d)} \cdot k }_{ \mathcal{B} }
    ~ \cdot ~ \underbrace{ 2^{\Theta(d \cdot \log(\mathcal{C} \cdot d))} }_{ {\cal M} } \\
    & = & 2^{\Theta(d \cdot \log(\mathcal{C} \cdot d))} \cdot k\cdot \log\log(F/\eta).
\end{eqnarray*}

In total, the procedure {\LocateInner} has the running time
\begin{eqnarray*}
    \#\mathtt{time}({\LocateInner})
    & = & \#\mathtt{time}(\mbox{hashing})
    ~ + ~ \#\mathtt{time}(\mbox{voting}) \\
    & = & 2^{\Theta(d \cdot \log(\mathcal{C} \cdot d))} \cdot k \cdot (\mathcal{D} + \log k) \cdot \log\log(F/\eta).
\end{eqnarray*}

This completes the proof.
\end{proof}

\subsection{Voting process}
\label{sec:locate_inner_voting}

The goal of this section is to prove Lemma~\ref{lem:locate_inner_voting}.
\begin{lemma}[The voting process of {\LocateInner}]
\label{lem:locate_inner_voting}
%\label{lem:LocateInner:voting}
Given any realized matrix $\Sigma$ and any realized vector $b$, assume three premises for a particular good tone frequency $f \in H = \{\xi \in \supp(\hat{x^*}): \mbox{neither $E_{\coll}(\xi)$ nor $E_{\off}(\xi)$ happens}\}$:
\begin{itemize}
    \item W.l.o.g.\ the tone frequency $f \in H$ is hashed into the bin $\mathpzc{h}_{\Sigma, b}(f) = j \in [B]^{d}$ (Definition~\ref{def:function_hash_multi}).

    \item The tone frequency $f \in H$ locates within the hyperball $\mathbf{HB}(\mathpzc{List}[j], L^{\mathpzc{dia}})$.
    
    \item Given the guessed approximation ratio $\mathcal{C} \in [120, \rho]$, the following holds for both $a = a_{r}$ and $a = a_{r} + \Delta_{a}^{r}$, in every single iteration $r \in [\mathcal{R}_{\mathrm{vote}}]$ of the procedure {\LocateInner}:
    \begin{eqnarray*}
        \E_{a} \left[\Big| \hat{u}_{j} - \hat{x}[f] \cdot e^{2 \pi \i \cdot a^{\top} \Sigma f} \Big|^{2}\right]
        & \leq & \mathcal{C}^{-2} \cdot \big| \hat{x}[f] \big|^{2}.
    \end{eqnarray*}
\end{itemize}
Then the following hold in every single iteration of procedure {\LocateInner} (Algorithm~\ref{alg:locate_inner}):
\begin{description}[labelindent = 1em]
    \item [Property~I:]
    The (unique) true sub-hyperball gets a vote with probability at least
    \begin{align*}
        1 - \frac{4}{(\mathcal{C} \cdot \varpi)^{2}}
        \;=\; 1 - \frac{4}{\mathcal{C}^{2 / 3}}
        \;>\; \frac{1}{2}.
    \end{align*}
    
    \item [Property~II:]
    Any wrong sub-hyperball gets a vote with probability at most
    \begin{align*}
        8 \varpi + \frac{4}{(\mathcal{C} \cdot \varpi)^{2}}
        \;=\; \frac{12}{\mathcal{C}^{2 / 3}}
        \;<\; \frac{1}{2}.
    \end{align*}
\end{description}
\end{lemma}

\begin{claim}[Property~I of Lemma~\ref{lem:locate_inner_voting}]
\label{cla:LocateInner:1}
The (unique) true sub-hyperball gets a vote with probability at least
\begin{align*}
    1 - \frac{4}{(\mathcal{C} \cdot \varpi)^{2}}
    \;=\; 1 - \frac{4}{\mathcal{C}^{2 / 3}}
    \;>\; \frac{1}{2}.
\end{align*}
\end{claim}

\begin{proof}
For brevity, we rewrite $a_{r}$ and $\Delta_{a}^{r}$ respectively as $a$ and $\Delta_{a}$ in this proof. Note that all of the probabilities and the expectations given below are taken over the random vectors $a$ and $\Delta_{a}$.

Combining the second premise of the lemma and Chebyshev's inequality together, we know that the following holds with probability at least $1 - \frac{2}{(\mathcal{C} \cdot \varpi)^{2}}$:
\begin{align*}
    & \left|\hat{u}_{j} - \hat{x}[f] \cdot e^{2 \pi \i \cdot a^{\top} \Sigma f}\right|
    \leq \big(\varpi / \sqrt{2}\big) \cdot \big|\hat{x}[f]\big|,
\end{align*}
which is equivalent to
\begin{align*}
    \left|\hat{u}_{j} / \hat{x}[f] \cdot e^{ - 2 \pi \i \cdot a^{\top} \Sigma f} - 1\right|
    \leq \varpi / \sqrt{2}.
\end{align*}
I.e., the complex number $\hat{u}_{j} / \hat{x}[f] \cdot e^{- 2 \pi \i \cdot a^{\top} \Sigma f}$ lies in the circle $\{z \in \C: |z - 1| \leq \varpi / \sqrt{2}\}$. Clearly, any complex number in this circle has the phase $\leq \sin^{-1}(\varpi / \sqrt{2})$. In particular,
\begin{align*}
    \Big\|\underbrace{\arg(\hat{u}_{j}) - \arg(\hat{x}[f]) - 2 \pi \cdot a^{\top} \Sigma f}_{A_{1}}\Big\|_{\bigcirc}
    \leq  \sin^{-1}\big(\varpi / \sqrt{2}\big),
\end{align*}
where $\|\theta\|_{\bigcirc} \in [-\pi, \pi)$ denotes the ``phase distance'' from $e^{\i 0} = 1$ to any $\theta \in \R$.

Similarly, when $a$ is replaced with $(a + \Delta_{a})$, with probability $1 - \frac{2}{(\mathcal{C} \cdot \varpi)^{2}}$ we also have
\begin{align*}
    \Big\|\underbrace{\arg(\hat{u}_{j}') - \arg(\hat{x}[f]) - 2 \pi \cdot (a + \Delta_{a})^{\top} \Sigma f}_{A_{2}}\Big\|_{\bigcirc}
    \leq \sin^{-1}\big(\varpi / \sqrt{2}\big),
\end{align*}

Put the above two inequalities together, (by the union bound) the following holds for the phase difference $\phi_{j} = \arg(\hat{u}_{j}) - \arg(\hat{u}_{j}')$ with probability $1 - \frac{4}{(\mathcal{C} \cdot \varpi)^{2}}$:
\begin{align}\label{eq:cla:LocateInner:1:1}
    \Big\|\phi_{j} - 2 \pi \cdot \Delta_{a}^{\top} \Sigma f \Big\|_{\bigcirc}
    = & ~ \big\| A_{1} - A_{2} \big\|_{\bigcirc} \notag \\
    \leq & ~ \big\| A_{1} \big\|_{\bigcirc}+ \big\| A_{2} \big\|_{\bigcirc} \notag \\
    \leq & ~ 2 \sin^{-1}\big(\varpi / \sqrt{2}\big) \notag \\
    \leq & ~  (\pi / 2) \cdot \varpi,
\end{align}
where the second step applies the triangle inequality; and last step follows since for any $z \in (0, 1)$, we have $\sin^{-1}(z / \sqrt{2}) \leq (\pi / 4) \cdot z$.

Let $q^* \in Q$ be the index of the true sub-hyperball $\mathbf{HB}(f_{q^*}^{\mathpzc{grid}[j]}, \frac{1}{M} \cdot L^{\mathpzc{dia}}) \ni f$. Compared with the center frequency $f_{q^*}^{\mathpzc{grid}[j]}$ of this sub-hyperball, the tone frequency $f$ differs by has the $\ell_{2}$-distance
\begin{eqnarray*}
    \Big\| f - f_{q^*}^{\mathpzc{grid}[j]} \Big\|_{2}
    & \leq & \frac{1}{2 M} \cdot L^{\mathpzc{dia}}.
\end{eqnarray*}

Equation~\eqref{eq:cla:LocateInner:1:1} suggests that the phase $\phi_{j} = \arg(\hat{u}_{j}) - \arg(\hat{u}_{j}') \in \R$ is likely to be a good approximation to $2 \pi \cdot \Delta_{a}^{\top} f$. Indeed, that inequality (if true) guarantees a vote for the true sub-hyperball $\mathbf{HB}(f_{q^*}^{\mathpzc{grid}[j]}, \frac{1}{M} \cdot L^{\mathpzc{dia}})$:
\begin{eqnarray}
    \notag
    \Big| 2 \pi \cdot \Delta_{a}^{\top} \Sigma (f - f_{q^*}^{\mathpzc{grid}[j]}) \Big|
    & \leq & 2 \pi \cdot \| \Sigma^{\top} \Delta_{a} \|_{2} \cdot \Big\| f - f_{q^*}^{\mathpzc{grid}[j]} \Big\|_{2} \\
    \notag
    & \leq & 2 \pi \cdot \| \Sigma^{\top} \Delta_{a} \|_{2} \cdot \Big(\frac{1}{2 M} \cdot L^{\mathpzc{dia}}\Big) \\
    \notag
    & \leq & 2 \pi \cdot \Big(\frac{\varpi \cdot M}{2 L^{\mathpzc{dia}}}\Big) \cdot \Big(\frac{1}{2 M} \cdot L^{\mathpzc{dia}}\Big) \\
    \label{eq:cla:LocateInner:1:2}
    & = & (\pi / 2) \cdot \varpi,
\end{eqnarray}
where the second step uses the $\ell_{2}$-distance derived above; and the third step follows because the vector $\Delta_{a}$ is sampled such that $\| \Sigma^{\top} \Delta_{a} \|_{2} \sim \unif[\frac{\varpi \cdot M}{4 L^{\mathpzc{dia}}}, \frac{\varpi \cdot M}{2 L^{\mathpzc{dia}}}]$ (see Algorithm~\ref{alg:sample_time_point}).

Combining everything together, with probability at least $1 - \frac{4}{(\mathcal{C} \cdot \varpi)^{2}}$ we have
\begin{eqnarray*}
    \Big\|\phi_{j} - 2 \pi \cdot \Delta_{a}^{\top} \Sigma f_{q^*}^{\mathpzc{grid}[j]} \Big\|_{\bigcirc}
    & = & \Big\|\phi_{j} - 2 \pi \cdot \Delta_{a}^{\top} \Sigma f \Big\|_{\bigcirc}
    + \Big|2 \pi \cdot \Delta_{a}^{\top} \Sigma (f - f_{q^*}^{\mathpzc{grid}[j]}) \Big| \\
    & \leq &  (\pi / 2) \cdot \varpi ~ + ~ (\pi / 2) \cdot \varpi \\
    & = &  \pi \cdot \varpi,
\end{eqnarray*}
where the first step follows from the triangle inequality; and the second step follows by applying inequalities~\eqref{eq:cla:LocateInner:1:1} and \eqref{eq:cla:LocateInner:1:2}.

Recall Definitions~\ref{def:locate_inner_setup} and \ref{def:voting} that $\varpi = \mathcal{C}^{-2 / 3}$ and $\mathcal{C} \geq 120$. Via elementary calculation, it can be seen that
\[
1 - \frac{4}{(\mathcal{C} \cdot \varpi)^{2}}
\;=\; 1 - \frac{4}{\mathcal{C}^{2 / 3}}
\;\geq\; 1 - \frac{4}{120^{2 / 3}}
\;\approx\; 0.8356
\;>\; \frac{1}{2}.
\]

This completes the proof of Claim~\ref{cla:LocateInner:1}.
\end{proof}

\begin{claim}[Property~II of Lemma~\ref{lem:locate_inner_voting}]
\label{cla:LocateInner:2}
Any wrong sub-hyperball gets a vote with probability at most
\begin{align*}
    8 \varpi + \frac{4}{(\mathcal{C} \cdot \varpi)^{2}}
    \;=\; \frac{12}{\mathcal{C}^{2 / 3}}
    \;<\; \frac{1}{2}.
\end{align*}
\end{claim}

\begin{proof}
Once again, we rewrite $a_{r}$ and $\Delta_{a}^{r}$ respectively as $a$ and $\Delta_{a}$ for simplicity, and all of the probabilities and the expectations in this proof are taken over the random vectors $a$ and $\Delta_{a}$.

Let $q^* \in Q$ be the index of the true sub-hyperball $\mathbf{HB}(f_{q^*}^{\mathpzc{grid}[j]}, \frac{1}{M} \cdot L^{\mathpzc{dia}}) \ni f$. For a specific wrong sub-hyperball $\mathbf{HB}(f_{q}^{\mathpzc{grid}[j]}, \frac{1}{M} \cdot L^{\mathpzc{dia}})$, where $q \in Q \setminus \{q^*\}$, the next inequality turns out to hold with probability at least $1 - 8 \varpi$:
\begin{eqnarray}
    \label{eq:cla:LocateInner:2:2}
    \Big\| 2 \pi \cdot \Delta_{a}^{\top} \Sigma \big( f_{q^*}^{\mathpzc{grid}[j]} - f_{q}^{\mathpzc{grid}[j]} \big) \Big\|_{\bigcirc}
    \geq  2 \pi \cdot \varpi.
\end{eqnarray}
We assume this fact for a while, and will justify this fact in the last part of this proof.

As shown in the proof of Claim~\ref{cla:LocateInner:1}, the following holds with probability at least $1 - 4 \cdot (\mathcal{C} \varpi)^{-2}$:
\begin{eqnarray}
    \label{eq:cla:LocateInner:2:1}
    \Big\|\phi_{j} - 2 \pi \cdot \Delta_{a}^{\top} \Sigma f_{q^*}^{\mathpzc{grid}[j]} \Big\|_{\bigcirc}
    \leq \pi \cdot \varpi.
\end{eqnarray}

Conditioned on both Inequalities~\eqref{eq:cla:LocateInner:2:2} and \eqref{eq:cla:LocateInner:2:1}, we must have
\begin{eqnarray*}
    \Big\| \phi_{j} - 2 \pi \cdot \Delta_{a}^{\top} \Sigma f_{q}^{\mathpzc{grid}[j]} \Big\|_{\bigcirc}
    & = & \Big\| 2 \pi \cdot \Delta_{a}^{\top} \Sigma \big( f_{q^*}^{\mathpzc{grid}[j]} - f_{q}^{\mathpzc{grid}[j]} \big) + \big(\phi_{j} - 2 \pi \cdot \Delta_{a}^{\top} \Sigma f_{q^*}^{\mathpzc{grid}[j]}\big) \Big\|_{\bigcirc} \\
    & \geq &  \pi \cdot \varpi.
\end{eqnarray*}
Given this, we know from Definition~\ref{def:voting} that the $q$-th (wrong) sub-hyperball is guaranteed to lose a vote. And based on the union bound, we derive Claim~\ref{cla:LocateInner:2} as desired:
\begin{eqnarray*}
    \Pr\big[ \mbox{$q$-th sub-hyperball gets a vote} \big]
    & \leq & \Pr\big[ \mbox{Equation~\eqref{eq:cla:LocateInner:2:2} does not hold} \big] \\
    & & + \Pr\big[ \mbox{Equation~\eqref{eq:cla:LocateInner:2:1} does not hold} \big] \\
    & = & 8 \varpi + \frac{4}{(\mathcal{C} \varpi)^{2}} \\
    & = & \frac{12}{\mathcal{C}^{2 / 3}} \\
    & \leq & \frac{12}{120^{2 / 3}} \\
    & \leq & 0.4933 \\
    & \leq & \frac{1}{2}
\end{eqnarray*}
where the third step follows because $\varpi = \mathcal{C}^{-2 / 3}$ (see Definition~\ref{def:voting}); and the fourth step follows because $\mathcal{C} \geq 120$ (see Definition~\ref{def:locate_inner_setup}).

To establish the claim, we are left to justify that Equation~\eqref{eq:cla:LocateInner:2:2} holds with probability at least $1 - 8 \varpi$. Indeed, an equivalent condition of Equation~\eqref{eq:cla:LocateInner:2:2} is that
\begin{quote}
    {\em ``$\Delta_{a}^{\top} \Sigma (f_{q^*}^{\mathpzc{grid}[j]} - f_{q}^{\mathpzc{grid}[j]})$ differs from its closest integer $\lfloor \Delta_{a}^{\top} \Sigma (f_{q^*}^{\mathpzc{grid}[j]} - f_{q}^{\mathpzc{grid}[j]}) + \frac{1}{2} \rfloor$ by at least $\varpi$''},
\end{quote}
because $\| z \|_{\bigcirc} \in [-\pi, \pi)$ denotes the ``phase distance'' from $0 = \arg(e^{\i 0})$ to any $z \in \R$.

We then observe that the vector $\Delta_{a}$ is sampled (see Algorithm~\ref{alg:sample_time_point}) such that $\Sigma^{\top} \Delta_{a}$ has a uniform random direction, and the $\ell_{2}$-norm follows the uniform distribution
\begin{eqnarray*}
    \| \Sigma^{\top} \Delta_{a} \|_{2}
    & \sim & \unif\Big[\frac{\varpi \cdot M}{4 L^{\mathpzc{dia}}}, \frac{\varpi \cdot M}{2 L^{\mathpzc{dia}}}\Big].
\end{eqnarray*}
Given these, we infer (e.g.\ from \cite{cfj13}) that $\Delta_{a}^{\top} \Sigma (f_{q^*}^{\mathpzc{grid}[j]} - f_{q}^{\mathpzc{grid}[j]})$ has the same distribution as the random variable $(\wt{w} \cdot \cos \wt{\theta})$, where
\begin{itemize}
    \item $\wt{\theta} \in [0, \pi]$ is the angle between the vectors $\Sigma^{\top} \Delta_{a}$ and $ (f_{q^*}^{\mathpzc{grid}[j]} - f_{q}^{\mathpzc{grid}[j]})$. It is known that $\wt{\theta}$ has the following probability density function: for all $\wt{\theta} \in [0, \pi]$,
    \begin{align*}
        \mathrm{PDF}(\wt{\theta})
        ~ = ~ \frac{\sin^{d - 2}(\wt{\theta})}{\int_{0}^{\pi} \sin^{d - 2}(z) \cdot \d z}.
    \end{align*}
    
    \item $\wt{w} \sim \unif[w, 2 w]$ with the parameter
    \begin{eqnarray*}
        w & = & \frac{\varpi \cdot M}{4 L^{\mathpzc{dia}}} \cdot \Big\| f_{q^*}^{\mathpzc{grid}[j]} - f_{q}^{\mathpzc{grid}[j]} \Big\|_{2} \\
        & \geq & \frac{\varpi \cdot M}{4 L^{\mathpzc{dia}}} \cdot \frac{1}{M} \cdot L^{\mathpzc{dia}} \cdot \Big\lceil \frac{4 \sqrt{d}}{\varpi} \Big\rceil \\
        & \geq & \sqrt{d},
    \end{eqnarray*}
    where the second step follows from the definition of a {\em wrong} sub-hyperball (see Definition~\ref{def:sub_hyperball}).
\end{itemize}

We conclude from the above that
\begin{eqnarray}
    \label{eq:cla:LocateInner:2:3}
    \underset{\Delta_{a}}{\Pr} \left[ \mbox{Equation~\eqref{eq:cla:LocateInner:2:2} does not hold} \right]
    & = & \underset{\wt{w}, \wt{\theta}}{\Pr} \left[ \Big| \wt{w} \cdot \cos \wt{\theta} - \Big\lfloor \wt{w} \cdot \cos \wt{\theta} + 1/2 \Big\rfloor \Big| \leq \varpi \right]
\end{eqnarray}
It turns out that $\varpi \leq 1/5$ and that $w \geq 1$. Concretely, we know from Definitions~\ref{def:locate_inner_setup} and \ref{def:voting} that
\[
\varpi
\;=\; \mathcal{C}^{-2 / 3}
\;\leq\; 120^{-2 / 3}
\;\approx\; 0.0411
\;<\; \frac{1}{5}.
\]
Further, we have shown that the parameter $w \geq \sqrt{d}$. Given these, Claim~\ref{cla:LocateInner:2:tech_multi} (presented below) is applicable to the $\RHS$ of Equation~\eqref{eq:cla:LocateInner:2:3}. By doing so, we accomplish Claim~\ref{cla:LocateInner:2}.

This completes the proof.
\end{proof}

\begin{claim}[Technical result for Claim~\ref{cla:LocateInner:2}]
\label{cla:LocateInner:2:tech_multi}
Given any $u \in (d/2^d, 1/5 ]$ and any $w \geq \sqrt{d} $, it follows
\begin{eqnarray*}
    \underset{\wt{w}, \wt{\theta}}{\Pr} \left[ \Big| \wt{w} \cdot \cos ( \wt{\theta} ) - \Big\lfloor \wt{w} \cdot \cos ( \wt{\theta} ) + 1/2 \Big\rfloor \Big| \leq u \right]
    & \lesssim & u,
\end{eqnarray*}
where $\wt{w} \sim \unif[w, 2 w]$, and the random phase $\wt{\theta} \in [0, \pi]$ has the probability density function
\begin{align*}
    \mathrm{PDF}(\wt{\theta})
    & ~ = ~ \frac{\sin^{d - 2}(\wt{\theta})}{\int_{0}^{\pi} \sin^{d - 2}(z) \cdot \d z},
    && \forall \wt{\theta} \in [0, \pi].
\end{align*}
\end{claim}

\begin{proof}
Fix $\wt{\theta}$ first.
We know that 
\begin{align*}
    \int_{0}^{\pi} \sin^{d - 2}(z) \cdot \d z &= \pi\cdot\frac{(2(d-2)-1)!!}{(2(d-2))!!}\\
    &\eqsim 1/d,
\end{align*}
where the first step is by induction and the second step follows from the Wallis formula.

We need some asymptotic evaluations:
\begin{itemize}
    \item $|\sin(z)|\eqsim 1-\frac{\cos^2{z}}{2}$ when $|\cos(z)|\ll 1$.
    \item $\sin{x}\eqsim x$ when $|x|\ll 1$.
\end{itemize}

In the next a few paragraphs, we discuss the three cases for $| \cos (\wt{\theta}) |$.
\begin{itemize}
\item Case 1. $ |\cos(\wt{\theta} )| \leq c/\sqrt{d} $
\item Case 2. $ c/\sqrt{d} \leq |\cos(\wt{\theta} )|\leq 1/2 $
\item Case 3. $ |\cos(\wt{\theta} )| > 1/2 $
\end{itemize}

{\bf Case 1.} If $ |\cos(\wt{\theta} )| \leq c / \sqrt{d} $.

First consider the range that $|\cos(\wt{\theta} )|\leq c/\sqrt{d}$ for some small constant $c>1$. Then in this range we know that $|\sin{z}|\eqsim 1-\frac{\cos^2(\wt{\theta} )}{2}\geq 1-\frac{c^2}{d}$, which implies that $|\sin(\wt{\theta})|^{d-2}=\Omega(1)$. In other word, we can treat $\wt{\theta}$ as nearly uniform distributed in this range. By similar arguments in Claim~\ref{cla:LocateInner:2:technical_result}, we can prove in this range
\begin{align*}
     \underset{\wt{\theta}}{\Pr} \left[ \Big| \wt{w} \cdot \cos ( \wt{\theta} ) - \Big\lfloor \wt{w} \cdot \cos ( \wt{\theta} ) + 1/2 \Big\rfloor \Big| \leq u, |\cos(\wt{\theta})|\leq c/\sqrt{d} \right]
     \lesssim u
\end{align*}

{\bf Case 2.} If $c/\sqrt{d} \leq |\cos(\wt{\theta} )| \leq 1/2 $. 

We fix an integer $i >c$.

We have that $|\sin(z)|\geq 1/2$. We can use a straight line to simulate $\cos$ function. For any integer $i\geq c$, we have that 
\begin{align*}
    \cos^{-1}(\frac{i-u}{\wt{w}})-\cos^{-1}(\frac{i+u}{\wt{w}})\lesssim u\cdot (\cos^{-1}(\frac{i-1+u}{\wt{w}})-\cos^{-1}(\frac{i-u}{\wt{w}})).
\end{align*}

We know that $\mathrm{PDF}(\wt{\theta})$ is increasing, then we have that 
\begin{align*}
    u\cdot\underset{\wt{\theta}}{\Pr}[\wt{\theta}\in [\cos^{-1}(\frac{i-u}{\wt{w}}),\cos^{-1}(\frac{i-1+u}{\wt{w}})]\gtrsim  \underset{\wt{\theta}}{\Pr}[\wt{\theta}\in [\cos^{-1}(\frac{i+u}{\wt{w}}),\cos^{-1}(\frac{i-u}{\wt{w}})].
\end{align*}

%In other word,  and $1/2\geq |\cos(\wt{\theta} )|\geq c/\sqrt{d}$, 
We have that
\begin{align*}
    \underset{\wt{\theta}}{\Pr}\left[ \Big| \wt{w} \cdot \cos ( \wt{\theta} ) - i \Big| \leq u,\frac{1}{2}\geq |\cos(\wt{\theta} )|\geq c/\sqrt{d} \right]\lesssim u\cdot \underset{\wt{\theta}}{\Pr}\left[ \Big| \wt{w} \cdot \cos ( \wt{\theta} ) - i \Big| \leq 1 ,\frac{1}{2}\geq |\cos(\wt{\theta} )|\geq c/\sqrt{d}\right].
\end{align*}

Combine this together, we know that
\begin{align*}
    \underset{\wt{\theta}}{\Pr} \left[ \Big| \wt{w} \cdot \cos ( \wt{\theta} ) - \Big\lfloor \wt{w} \cdot \cos ( \wt{\theta} ) + 1/2 \Big\rfloor \Big| \leq u, 1/2\geq |\cos(\wt{\theta})|\geq c/\sqrt{d} \right]
     \lesssim u. 
\end{align*}

{\bf Case 3.} If $|\cos(\wt{\theta})|>1/2$.

 Then we have
\begin{align*}
     \underset{\wt{\theta}}{\Pr} \left[ \Big|  |\cos(\wt{\theta})|\geq 1/2 \right]
     &\leq \frac{\pi\cdot (1/2)^{d-2}}{\int_{0}^{\pi} \sin^{d - 2}(z) \cdot \d z} \\
     &\lesssim \frac{d}{ 2^{d}}\\
     &\leq u.
\end{align*}
The first step is because $|\sin(\wt{\theta})|<1/2$.

{\bf Combine three cases.} Then combine these three cases together, we have
\begin{align*}
    &~ \underset{\wt{\theta}}{\Pr} \left[ \Big| \wt{w} \cdot \cos ( \wt{\theta} ) - \Big\lfloor \wt{w} \cdot \cos ( \wt{\theta} ) + 1/2 \Big\rfloor \Big| \leq u \right]\\
    \leq &~
    \underset{\wt{\theta}}{\Pr} \left[ \Big| \wt{w} \cdot \cos ( \wt{\theta} ) - \Big\lfloor \wt{w} \cdot \cos ( \wt{\theta} ) + 1/2 \Big\rfloor \Big| \leq u, 1/2\geq |\cos(\wt{\theta})|\geq c/\sqrt{d} \right]\\
    + &~ \underset{\wt{\theta}}{\Pr}[|\cos(\wt{\theta})|>1/2]+\underset{\wt{\theta}}{\Pr} \left[ \Big| \wt{w} \cdot \cos ( \wt{\theta} ) - \Big\lfloor \wt{w} \cdot \cos ( \wt{\theta} ) + 1/2 \Big\rfloor \Big| \leq u, |\cos(\wt{\theta})|\leq c/\sqrt{d} \right]\\
     \lesssim &~ u.
\end{align*}

This completes the proof.
\end{proof}

\begin{claim}[Technical result for Claim~\ref{cla:LocateInner:2}]
\label{cla:LocateInner:2:technical_result}
Given any $u \in (0, 1/5 ]$ and any $w \geq 1$, the following holds for the random variables $\wt{w} \sim \unif[w, 2 w]$ and $\wt{\theta} \sim \unif[-\pi, \pi)$:
\begin{eqnarray*}
    \underset{\wt{w}, \wt{\theta}}{\Pr} \left[ \Big| \wt{w} \cdot \sin ( \wt{\theta} ) - \Big\lfloor \wt{w} \cdot \sin ( \wt{\theta} ) + 1/2 \Big\rfloor \Big| \leq u \right]
    & \leq & 8 u.
\end{eqnarray*}
\end{claim}

\begin{figure}
    \centering
    \includegraphics[width = .6\textwidth]{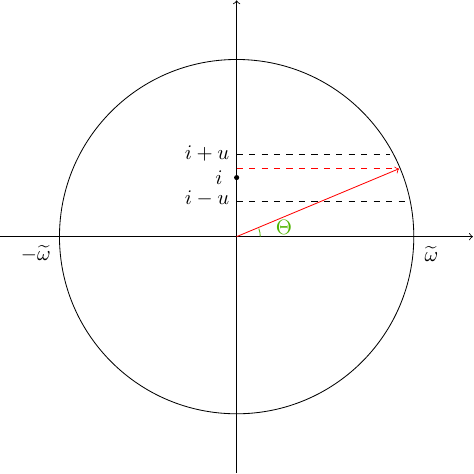}
    \caption{Demonstration for Claim~\ref{cla:LocateInner:2:technical_result}, where $i \in \Z$ is an integer.}
    \label{fig:LocateInner:2:technical_result}
\end{figure}

\begin{proof}
To improve the readability, we provide Figure~\ref{fig:LocateInner:2:technical_result} for demonstration. We denote 
\begin{align*}
\wt{d_{\theta}} = | \wt{w} \cdot \sin (\wt{\theta}) - \lfloor \wt{w} \cdot \sin ( \wt{\theta} ) + 1/2 \rfloor |
\end{align*}
for ease of notation. Notice that $\wt{d_{\theta}}$ represents the distance between $(\wt{w} \cdot \sin ( \wt{\theta} ) )$ and its closest integer. By symmetry, the following random distance $\wt{d_{\psi}}$ has the same distribution as $\wt{d_{\theta}}$:
\begin{align*}
    \wt{d_{\psi}} = \Big| \wt{w} \cdot \sin ( \wt{\psi} ) - \Big\lfloor \wt{w} \cdot \sin( \wt{\psi} ) + 1/2 \Big\rfloor \Big|,
\end{align*}
where the new random phase $\wt{\psi}$ is distributed uniformly on $[0, \frac{\pi}{2}]$ rather than on $[-\pi, \pi)$.

Let us investigate the new random distance $\wt{d_{\psi}}$ via case analysis.
\begin{description}
    \item [Case~(i):] when $(1 - u) / (2 w) \leq \sin ( \wt{\psi} ) \leq 1$.
    
    Notice that this case is non-empty, since $u \in ( 0 , 1/5 ]$ and $w \geq 1$. Suppose the random phase $\wt{\psi}$ is fixed. Because the random variable $\wt{w} \sim \unif[w, 2 w]$, the random closest integer $\lfloor \wt{w} \cdot \sin ( \wt{\psi} ) + 1/2 \rfloor$ admits the following lower and upper bounds:
    \begin{align*}
        \Big\lfloor \wt{w} \cdot \sin ( \wt{\psi} ) + 1 / 2 \Big\rfloor
        & \;\geq\; \wt{w} \cdot \sin ( \wt{\psi} ) - 1 / 2
        \;\geq\; w \cdot \sin ( \wt{\psi} ) - 1 / 2 \\
        \Big\lfloor \wt{w} \cdot \sin ( \wt{\psi} ) + 1/2 \Big\rfloor
        & \;\leq\; \wt{w} \cdot \sin ( \wt{\psi} ) + 1 / 2
        \;\leq\; 2 w \cdot \sin ( \wt{\psi} ) + 1 / 2
    \end{align*}
    Namely, the random closest integer $\lfloor \wt{w} \cdot \sin ( \wt{\psi} ) + 1 / 2 \rfloor$ has at most $(w \cdot \sin ( \wt{\psi} ) + 1)$ many possibilities. Consider the set $\wt{A_{\psi}}$ of all possible $(\wt{w} \cdot \sin ( \wt{\psi} ) )$ such that the random distance $\wt{d_{\psi}} \leq u$:
    \begin{eqnarray*}
        \wt{A_{\psi}} & = & \left\{ \wt{w} \cdot \sin \wt{\psi}: \wt{w} \in [w, 2 w] \mbox{ and } \wt{d_{\psi}} \leq u \right\}.
    \end{eqnarray*}
    Since the closed integer $\lfloor \wt{w} \cdot \sin ( \wt{\psi} ) + 1 / 2 \rfloor$ has at most $(w \cdot \sin ( \wt{\psi} ) + 1)$ many possibilities, the total length of this set $\wt{A_{\psi}}$ is at most
    \begin{eqnarray}\label{eq:cla:LocateInner:2:technical_result:0}
        \big| \wt{A_{\psi}} \big|
        & \leq & 2 u \cdot (w \cdot \sin ( \wt{\psi} ) + 1).
    \end{eqnarray}
    Hence, under any choice of the random phase $\wt{\psi}$, the conditional probability (over the uniform random variable $\wt{w} \sim \unif[w, 2 w]$) below is at most
    \begin{eqnarray*}
        \underset{\wt{w}}{\Pr} \left[ \wt{d_{\psi}} \leq u \;\mid\; \mbox{\bf case~(i)} \right]
        & = & \frac{\big| \wt{A_{\psi}} \big|}{(2 w - w) \cdot \sin ( \wt{\psi} ) } \\
        & \leq & \frac{2 u \cdot (w \cdot \sin ( \wt{\psi} ) + 1)}{(2 w - w) \cdot \sin ( \wt{\psi} ) } \\
        & = & 2 u \cdot \left(1 + (w \cdot \sin ( \wt{\psi} ) )^{-1}\right) \\
        & \leq & 2 u \cdot \left(1 + \frac{2}{1 - u}\right) \\
        & \leq & 7 u,
    \end{eqnarray*}
    where the second step applies Equation~\eqref{eq:cla:LocateInner:2:technical_result:0}; the fourth step follows because (in this case) we assume that $\sin ( \wt{\psi} ) \geq \frac{1 - u}{2 w}$; and the last step is because $u \in (0, 1/5]$.

    \item [Case~(ii):] when $u / w < \sin ( \wt{\psi} ) < (1 - u) / (2 w)$.
    
    Notice that this case is non-empty, since $u < \frac{1 - u}{2}$ for any $u \in (0, 1/5]$. Of course, any realized random variable $\wt{w} \sim \unif[w, 2 w]$ satisfies $\wt{w} \geq w$ and $\wt{w} \leq 2 w$. On the lower-bound part:
    \begin{align*}
        \wt{w} \cdot \sin ( \wt{\psi} )
        \;\geq\; w \cdot \sin ( \wt{\psi} )
        \;>\; w \cdot \frac{u}{w}
        \;=\; u.
    \end{align*}
    Further, on the upper-bound part:
    \begin{align*}
        \wt{w} \cdot \sin ( \wt{\psi} )
        \;\leq\; 2 w \cdot \sin ( \wt{\psi} )
        \;<\; 2 w \cdot \frac{1 - u}{2 w}
        \;=\; 1 - u.
    \end{align*}
    
    Combining both inequalities together, regardless of the realized $\wt{w} \sim \unif[w, 2 w]$, the random variable $(\wt{w} \cdot \sin ( \wt{\psi}) )$ locates between $(u, 1 - u)$ and differs from its closest integer by at least $u$.
    
    From the above arguments, we conclude that the next conditional probability equals zero.
    \begin{eqnarray*}
        \underset{\wt{w}, \wt{\psi}}{\Pr} \left[ \wt{d_{\psi}} \leq u \;\mid\; \mbox{\bf case~(ii)} \right]
        & = & 0.
    \end{eqnarray*}
    
    \item [Case~(iii):] when $0 \leq \sin ( \wt{\psi} ) \leq u / w$.
    
    Of course, the following conditional probability is at most one:
    \begin{eqnarray*}
        \underset{\wt{w}, \wt{\psi}}{\Pr} \left[ \wt{d_{\psi}} \leq u \;\mid\; \mbox{\bf case~(iii)} \right]
        & \leq & 1.
    \end{eqnarray*}

    Observe that $u / w \leq u \leq 1/5$, because $u \in (0, 1/5]$ and $w \geq 1$. Then, since the random phase $\wt{\psi}$ is distributed uniformly on $[0, \frac{\pi}{2}]$, this case happens with probability
    \begin{eqnarray*}
    \underset{\wt{w}, \wt{\psi}}{\Pr} \left[ \mbox{\bf case~(iii)} \right]
    & = & \frac{2}{\pi} \cdot \sin^{-1}(u / w) \\
    & \leq & \frac{2}{\pi} \cdot \sin^{-1}(u) \\
    & \leq & \frac{2}{3} \cdot u,
    \end{eqnarray*}
    where the last step follows since $z \leq \sin(\frac{\pi}{3} \cdot z)$ for any $z \in [0, 1/2]$ and we have $u \in (0, 1/5] \subseteq [0, 1/2]$.
\end{description}

Putting all the three cases together, we conclude that
\begin{eqnarray*}
    \underset{\wt{w}, \wt{\psi}}{\Pr} \left[ \wt{d_{\psi}} \leq u \right]
    & \leq & 7 u \cdot \underset{\wt{w}, \wt{\psi}}{\Pr} \left[ \mbox{\bf case~(i)} \right]
    + 0 \cdot \underset{\wt{w}, \wt{\psi}}{\Pr} \left[ \mbox{\bf case~(ii)} \right]
    + 1 \cdot \underset{\wt{w}, \wt{\psi}}{\Pr} \left[ \mbox{\bf case~(iii)} \right] \\
    & = & 7 u \cdot \underset{\wt{w}, \wt{\psi}}{\Pr} \left[ \mbox{\bf case~(i)} \right]
    + 1 \cdot \underset{\wt{w}, \wt{\psi}}{\Pr} \left[ \mbox{\bf case~(iii)} \right] \\
    & \leq & 7 u \cdot 1
    + 1 \cdot \frac{2}{3} \cdot u \\
    & \leq & 8 u,
\end{eqnarray*}
where the first step applies the bounds on the conditional probabilities derived before.

This completes the proof of Claim~\ref{cla:LocateInner:2:technical_result}.
\end{proof}

\subsection{Election process}
\label{sec:locate_inner_election}

The goal of this section is to prove Lemma~\ref{lem:locate_inner_election}.

\begin{lemma}[The election process of {\LocateInner}]
\label{lem:locate_inner_election}
%\label{lem:LocateInner:regular_iteration}
For any matrix $\Sigma \in \R^{d \times d}$ and any vector $b \in \R^d$, assume three premises for a particular good tone frequency $f \in H = \{\xi \in \supp(\hat{x^*}): \mbox{neither $E_{\coll}(\xi)$ nor $E_{\off}(\xi)$ happens}\}$:
\begin{itemize}
    \item The tone frequency $f \in H$ is hashed into the bin $\mathpzc{h}_{\Sigma, b}(f) = j \in [B]^{d}$ (Definition~\ref{def:function_hash_multi}).

    \item The tone frequency $f \in H$ locates within the hyperball $\mathbf{HB}(\mathpzc{List}[j], L^{\mathpzc{dia}})$.
    
    \item Given the guessed approximation ratio $\mathcal{C} \in [120, \rho]$, the following holds for both $a = a_{r}$ and $a = a_{r} + \Delta_{a}^{r}$, in every single iteration $r \in [\mathcal{R}_{\mathrm{vote}}]$ of the procedure {\LocateInner} (Algorithm~\ref{alg:locate_inner}):
    \begin{eqnarray*}
        \E_{a} \left[\Big| \hat{u}_{j} - \hat{x}[f] \cdot e^{2 \pi \i \cdot a^{\top} \Sigma f} \Big|^{2}\right]
        & \leq & \mathcal{C}^{-2} \cdot \big| \hat{x}[f] \big|^{2}.
    \end{eqnarray*}
\end{itemize}
Then with probability at least $1 - \mathcal{M} \cdot 2^{-\Omega(\mathcal{R}_{\mathrm{vote}})}$, the following hold for the algorithm {\LocateInner}:
\begin{description}[labelindent = 1em]
    \item [Property~I:]
    The tone frequency $f \in H$ locates in one of the winning sub-hyperballs
    \begin{eqnarray*}
        \mathpzc{Winner}[j]
        & = & \bigcup_{q \in Q: \mathpzc{Vote}_{j}[q] \geq \frac{1}{2} \cdot \mathcal{R}_{\mathrm{vote}}} \mathbf{HB}(f_{q}^{\mathpzc{grid}[j]}, M^{-1} \cdot L^{\mathpzc{dia}}).
    \end{eqnarray*}
    
    \item [Property~II:]
    All the winning frequencies $\mathpzc{Winner}[j]$ can be included within the smaller hyperball
    \begin{eqnarray*}
        \mathpzc{Winner}[j] ~ \subseteq ~ \mathbf{HB}(f_{q^*}^{\mathpzc{grid}[j]}, L_{\mathrm{new}}^{\mathpzc{dia}}),
    \end{eqnarray*}
    which is centered at $f_{q^*}^{\mathpzc{grid}[j]}$ and has the new diameter $L_{\mathrm{new}}^{\mathpzc{dia}} := \frac{1}{2} \cdot L^{\mathpzc{dia}}$.
    
    \item [Property~III:]
    The output frequency $\mathpzc{List}_{\mathrm{new}}[j] \in \R^{d}$ makes the tone frequency $f \in H$ locate in a new hyperball that is centered at $\mathpzc{List}_{\mathrm{new}}[j]$ and has a new diameter $L_{\mathrm{new}}^{\mathpzc{dia}} = \frac{1}{2} \cdot L^{\mathpzc{dia}}$:
    \begin{eqnarray*}
        f ~ \in ~
        \mathbf{HB}(\mathpzc{List}_{\mathrm{new}}[j], L_{\mathrm{new}}^{\mathpzc{dia}}).
    \end{eqnarray*}
\end{description}
\end{lemma}

To make Lemma~\ref{lem:locate_inner_election} meaningful, later we will choose a large enough $\mathcal{R}_{\mathrm{vote}} \in \mathbb{N}_{\geq 1}$ such that the failure probability $\mathcal{M} \cdot 2^{-\Omega(\mathcal{R}_{\mathrm{vote}})} \ll 1$. Furthermore, we observe that Property~III of Lemma~\ref{lem:locate_inner_election} (see Figure~\ref{fig:locate_inner_election} for demonstration) is a direct follow-up to Properties~I and II. Below, we would show that Property~I holds with probability $1 - 2^{-\Omega(\mathcal{R}_{\mathrm{vote}})}$, and that Property~II holds with probability $1 - (\mathcal{M} - 1) \cdot 2^{-\Omega(\mathcal{R}_{\mathrm{vote}})}$. Then, all the properties can be inferred via the union bound.

\begin{figure}
    \centering
    \includegraphics[scale=0.7]{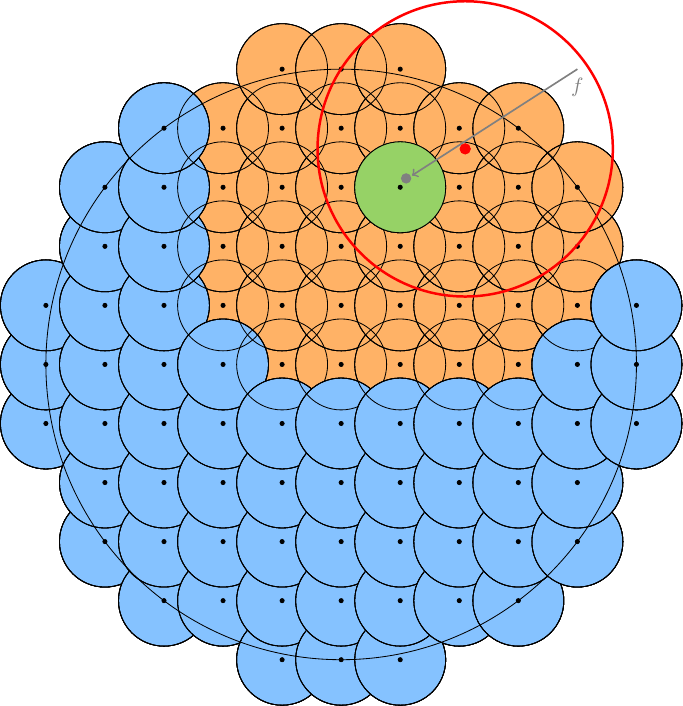}
    \caption{Demonstration for Property~III of Lemma~\ref{lem:locate_inner_election} in two dimensions $d = 2$. The given (black) circle $\mathbf{HB}(\mathpzc{List}[j], L^{\mathpzc{dia}})$ narrows down into a smaller (red) circle $\mathbf{HB}(\mathpzc{List}_{\mathrm{new}}[j], L_{\mathrm{new}}^{\mathpzc{dia}})$.}
    \label{fig:locate_inner_election}
\end{figure}

\begin{claim}[Property~I of Lemma~\ref{lem:locate_inner_election}]
\label{cla:lem:locate_inner_election:1}
With probability at least $1 - 2^{-\Omega(\mathcal{R}_{\mathrm{vote}})}$, the tone frequency $f \in H$ locates in one of the winning sub-hyperballs
\begin{eqnarray*}
    \mathpzc{Winner}[j]
    & = & \bigcup_{q \in Q: \mathpzc{Vote}_{j}[q] \geq \frac{1}{2} \cdot \mathcal{R}_{\mathrm{vote}}} \mathbf{HB}(f_{q}^{\mathpzc{grid}[j]}, M^{-1} \cdot L^{\mathpzc{dia}}).
\end{eqnarray*}
\end{claim}

\begin{proof}
Given the second premise of Lemma~\ref{lem:locate_inner_election} that, the tone frequency $f \in H$ locates within the hyperball $\mathbf{HB}(\mathpzc{List}[j], L^{\mathpzc{dia}})$, there is a unique true sub-hyperball $\mathbf{HB}(f_{q^*}^{\mathpzc{grid}[j]}, \frac{1}{M} \cdot L^{\mathpzc{dia}})$ containing $f \in H$ (see Definition~\ref{def:sub_hyperball} ), for some vector index $q^* \in Q$. Apparently, a necessary condition for $f \in H$ to locate in none of the winning sub-hyperballs is the event
\begin{eqnarray*}
\bar{E_{q^*}}
& = & \big\{\mbox{the $q^*$-th sub-hyperball in total gets less than $\frac{1}{2} \cdot \mathcal{R}_{\mathrm{vote}}$ votes}\big\}.
\end{eqnarray*}

Based on Property~I of Lemma~\ref{lem:locate_inner_voting}, in every iteration $r \in [\mathcal{R}_{\mathrm{vote}}]$, the $q^*$-th sub-hyperball independently loses a vote with probability at most $4 \cdot \mathcal{C}^{-2 / 3} < 1 / 2$. Combining a simple coupling argument together with the Chernoff bound (see Part~(a) of Lemma~\ref{lem:chernoff}), the event $\bar{E_{q^*}}$ happens with probability at most
\begin{eqnarray*}
    \Pr\big[ \bar{E_{q^*}} \big]
    & \leq & \exp\left(-\frac{\mathcal{R}_{\mathrm{vote}}}{2} \cdot \left(\ln\Big(\frac{\mathcal{C}^{2 / 3}}{8}\Big) + \frac{8}{\mathcal{C}^{2 / 3}} - 1\right)\right) \\
    & \leq & \exp\left(-\frac{\mathcal{R}_{\mathrm{vote}}}{2} \cdot \left(\ln\Big(\frac{120^{2 / 3}}{8}\Big) + \frac{8}{120^{2 / 3}} - 1\right)\right) \\
    & = & \exp\big(-\Omega(\mathcal{R}_{\mathrm{vote}})\big)
\end{eqnarray*}
where the second step follows because the formula $\ln z + \frac{1}{z}$ is increasing in $z \in \R_{> 0}$ (and $\mathcal{C} \geq 120$; see Definition~\ref{def:locate_inner_setup}); and the last step follows as $\ln(\frac{120^{2 / 3}}{8}) + \frac{8}{120^{2 / 3}} - 1 \approx 0.1174 = \Omega(1)$.

This completes the proof of Claim~\ref{cla:lem:locate_inner_election:1}.
\end{proof}

\begin{claim}[Property~II of Lemma~\ref{lem:locate_inner_election}]
\label{cla:lem:locate_inner_election:2}
All the winning frequencies $\mathpzc{Winner}[j]$ can be included within the smaller hyperball
\begin{eqnarray*}
    \mathpzc{Winner}[j] ~ \subseteq ~ \mathbf{HB}(f_{q^*}^{\mathpzc{grid}[j]}, L_{\mathrm{new}}^{\mathpzc{dia}}),
\end{eqnarray*}
which is centered at $f_{q^*}^{\mathpzc{grid}[j]}$ and has the new diameter $L_{\mathrm{new}}^{\mathpzc{dia}} := \frac{1}{2} \cdot L^{\mathpzc{dia}}$.
\end{claim}

\begin{proof}
Recall Definition~\ref{def:sub_hyperball} that we cover the hyperball $\mathbf{HB}(\mathpzc{List}[j], L^{\mathpzc{dia}})$ by using $\mathcal{M} = 2^{\Theta(d \cdot \log(\mathcal{C} \cdot d))}$ many sub-hyperballs. Given the second premise of Lemma~\ref{lem:locate_inner_election}, one particular sub-hyperball $q^* \in Q$ is the true sub-hyperball, and there are at most $(\mathcal{M} - 1)$ many wrong sub-hyperballs.

We first demonstrate that, a specific wrong sub-hyperball $q \in Q$ ``wins'' with probability at most $2^{-\Omega(\mathcal{R}_{\mathrm{vote}})}$. By definition (see Line~\ref{alg:LocateInner:winner} of {\LocateInner}), this wrong sub-hyperball ``wins'' if and only if the following event happens:
\begin{eqnarray*}
E_{q}
& = & \big\{\mbox{the $q$-th sub-hyperball in total gets at least $\frac{1}{2} \cdot \mathcal{R}_{\mathrm{vote}}$ votes}\big\}.
\end{eqnarray*}
According to Property~II of Lemma~\ref{lem:locate_inner_voting}, in each iteration $r \in [\mathcal{R}_{\mathrm{vote}}]$, the $q$-th sub-hyperball independently gets a vote with probability at most $12 \cdot \mathcal{C}^{-2 / 3} < 1 / 2$. Combining a simple coupling argument together with the Chernoff bound (see Part~(a) of Lemma~\ref{lem:chernoff}), the event $E_{q}$ happens with probability at most
\begin{eqnarray*}
    \Pr\big[ E_{q} \big]
    & \leq & \exp\left(-\frac{\mathcal{R}_{\mathrm{vote}}}{2} \cdot \left(\ln\Big(\frac{\mathcal{C}^{2 / 3}}{24}\Big) + \frac{24}{\mathcal{C}^{2 / 3}} - 1\right)\right) \\
    & \leq & \exp\left(-\frac{\mathcal{R}_{\mathrm{vote}}}{2} \cdot \left(\ln\Big(\frac{120^{2 / 3}}{24}\Big) + \frac{24}{120^{2 / 3}} - 1\right)\right) \\
    & = & \exp\big(-\Omega(\mathcal{R}_{\mathrm{vote}})\big)
\end{eqnarray*}
where the second step follows because the formula $\ln z + \frac{1}{z}$ is increasing in $z \in \R_{> 0}$ (and $\mathcal{C} \geq 120$; see Definition~\ref{def:locate_inner_setup}); and the last step follows as $\ln(\frac{120^{2 / 3}}{24}) + \frac{24}{120^{2 / 3}} - 1 \approx 9.2161 \times 10^{-5} = \Omega(1)$.

Since there are at most $(\mathcal{M} - 1)$ many wrong sub-hyperballs, we can apply the union bound for all of them. Hence, with probability at least $1 - (\mathcal{M} - 1) \cdot 2^{-\Omega(\mathcal{R}_{\mathrm{vote}})}$, none of the wrong sub-hyperballs ``win'' in the election process.

According to Definition~\ref{def:sub_hyperball}, the diameter of a sub-hyperball is $\frac{1}{M} \cdot L^{\mathpzc{dia}}$, and any {\em intermediate} sub-hyperball $q \in Q$ (or the true sub-hyperball $q^* \in Q$ itself) satisfies that
\[
    \Big\| f_{q^*}^{\mathpzc{grid}[j]} - f_{q}^{\mathpzc{grid}[j]} \Big\|_{2}
    ~ \geq ~ \frac{1}{M} \cdot L^{\mathpzc{dia}} \cdot \big\lceil 4 \sqrt{d} \cdot \mathcal{C}^{2 / 3} \big\rceil.
\]
For these reasons, the $\ell_{2}$-distance between any $f \in \mathpzc{Winner}[j]$ and the center frequency $f_{q^*}^{\mathpzc{grid}[j]}$ of the true sub-hyperball is at most
\begin{eqnarray*}
    \big\| f - f_{q^*}^{\mathpzc{grid}[j]} \big\|_{2}
    & \leq & \left(\lceil 4 \sqrt{d} \cdot \mathcal{C}^{2 / 3} \rceil - 1 + \frac{1}{2}\right) \cdot \frac{1}{M} \cdot L^{\mathpzc{dia}} \\
    & \leq & \frac{1}{4} \cdot L^{\mathpzc{dia}},
\end{eqnarray*}
where the last step is because $M = 4 \cdot \lceil 4 \sqrt{d} \cdot \mathcal{C}^{2 / 3} \rceil$ (see Definition~\ref{def:locate_inner_setup}).

Thus, any frequency $f \in \mathpzc{Winner}[j]$ can be included in a smaller hyperball $\mathbf{HB}(f_{q^*}^{\mathpzc{grid}[j]}, L_{\mathrm{new}}^{\mathpzc{dia}})$ that is centered at $f_{q^*}^{\mathpzc{grid}[j]}$ and has the new diameter $L_{\mathrm{new}}^{\mathpzc{dia}} := \frac{1}{2} \cdot L^{\mathpzc{dia}}$.

This accomplishes the proof of Claim~\ref{cla:lem:locate_inner_election:2}.
\end{proof}

\subsection{Performance guarantees}
\label{sec:locate_inner_guarantees}

The goal of this section is to prove Corollary~\ref{cor:locate_inner_guarantees}.

\begin{corollary}
[The guarantee of {\LocateInner}]
\label{cor:locate_inner_guarantees}
Given $\Sigma$ and $b$ (according to Definition~\ref{def:HashToBins_multi_parameters}), let $H \subseteq \supp(\hat{x^*})$ be a subset of ``good'' tone frequencies:
\begin{align*}
    H = \{\xi \in \supp(\hat{x^*}): \mbox{neither $E_{\off}(\xi)$ nor $E_{\coll}(\xi)$ happens}\}
\end{align*}
Let $j := \mathpzc{h}_{\Sigma,b}(f) \in [B]^{d}$ where a good frequency $f \in H$ is hashed into. Suppose $f \in \mathbf{HB}(\mathpzc{List}[j], L^{\mathpzc{dia}})$ at the beginning, then with failure probability at most $\mathcal{M} \cdot 2^{-\Omega(\mathcal{R}_{\mathrm{vote}})}$, the procedure {\LocateInner} outputs a new frequency $\mathpzc{List}_{\mathrm{new}}[j] \in [-F, F]^{d}$ so that
\begin{eqnarray*}
    f ~ \in ~
    \mathbf{HB}(\mathpzc{List}_{\mathrm{new}}[j], L_{\mathrm{new}}^{\mathpzc{dia}}),
\end{eqnarray*}
where the new diameter $L_{\mathrm{new}}^{\mathpzc{dia}} = L^{\mathpzc{dia}} / 2$.
\end{corollary}

\begin{proof}
This follows immediately from Property~III of Lemma~\ref{lem:locate_inner_election}.
\end{proof}

\subsection{Sampling time points}
\label{sec:locate_inner_time_points}

The procedure {\SampleTimePoint} is given in Algorithm~\ref{alg:sample_time_point}, which is illustrated in Figure~\ref{fig:sample_time_point}.

% \begin{algorithm}[!ht]
% \caption{{\SampleTimePoint}, Lemmas~\ref{lem:locate_inner_duration_require} and \ref{lem:locate_inner_sampling_require}}
% \label{alg:sample_time_point}
% \SetKwProg{myproc}{Procedure}{}{}
% \myproc{\SampleTimePoint($M, L^{\mathpzc{dia}}, \mathcal{C}, T$)}
% {
%  Define $\varpi \in (0, 1)$ according to Definition~\ref{def:locate_inner_setup}\;

%  Sample $\Delta_{a} \in \R^{d}$ such that $\Sigma^{\top} \Delta_{a} \sim \unif\{z \in \R^{d}: \|z\|_{2} = 1\}$\;
% \tcp*{$|\Sigma| \neq 0$; Definition~\ref{def:HashToBins_multi_parameters}}
% \label{alg:sample_time_point:Delta}

% Scale $\Delta_{a}$ by a random factor $\beta \sim \unif[\frac{\varpi \cdot M}{4 L^{\mathpzc{dia}}}, \frac{\varpi \cdot M}{2 L^{\mathpzc{dia}}}]$\;

% Let $A := \{z \in \R^{d}: \{z, z + \Sigma^{\top} \Delta_{a}\} \subseteq [\frac{0.01}{d} \cdot T, \big(1 - \frac{0.01}{d}\big) \cdot T]^{d}\}$\;
% \label{alg:sample_time_point:range}

% Sample $a \in \R^{d}$ such that $\Sigma^{\top} a \sim \unif(A)$\;
 
% \Return $a$ and $\Delta_{a}$\;
% }

% \end{algorithm}

\begin{algorithm}[!ht]
\caption{{\SampleTimePoint}, Lemmas~\ref{lem:locate_inner_duration_require} and \ref{lem:locate_inner_sampling_require}}
\label{alg:sample_time_point}
\begin{algorithmic}[1]
\Procedure{\SampleTimePoint}{$M, L^{\mathpzc{dia}}, \mathcal{C}, T$}

\State Define $\varpi \in (0, 1)$ according to Definition~\ref{def:locate_inner_setup}.

\State Sample $\Delta_{a} \in \R^{d}$ such that $\Sigma^{\top} \Delta_{a} \sim \unif\{z \in \R^{d}: \|z\|_{2} = 1\}$.
\Comment{$|\Sigma| \neq 0$; Definition~\ref{def:HashToBins_multi_parameters}}
\label{alg:sample_time_point:Delta}

\State Scale $\Delta_{a}$ by a random factor $\beta \sim \unif[\frac{\varpi \cdot M}{4 L^{\mathpzc{dia}}}, \frac{\varpi \cdot M}{2 L^{\mathpzc{dia}}}]$.

\State Let $A := \{z \in \R^{d}: \{z, z + \Sigma^{\top} \Delta_{a}\} \subseteq [\frac{0.01}{d} \cdot T, \big(1 - \frac{0.01}{d}\big) \cdot T]^{d}\}$.
\label{alg:sample_time_point:range}

\State Sample $a \in \R^{d}$ such that $\Sigma^{\top} a \sim \unif(A)$.

\State \Return $a$ and $\Delta_{a}$.

\EndProcedure
\end{algorithmic}

\end{algorithm}

\begin{figure}[htbp]
    \centering
    \begin{tabular}[b]{c}
    \includegraphics[width = .35\textwidth]{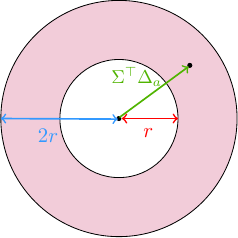}\\
    \small{(a)~Sampling for $\Sigma^{\top} \Delta_{a}$}
    \end{tabular} \qquad
    \begin{tabular}[b]{c}
    \includegraphics[width = .5\textwidth]{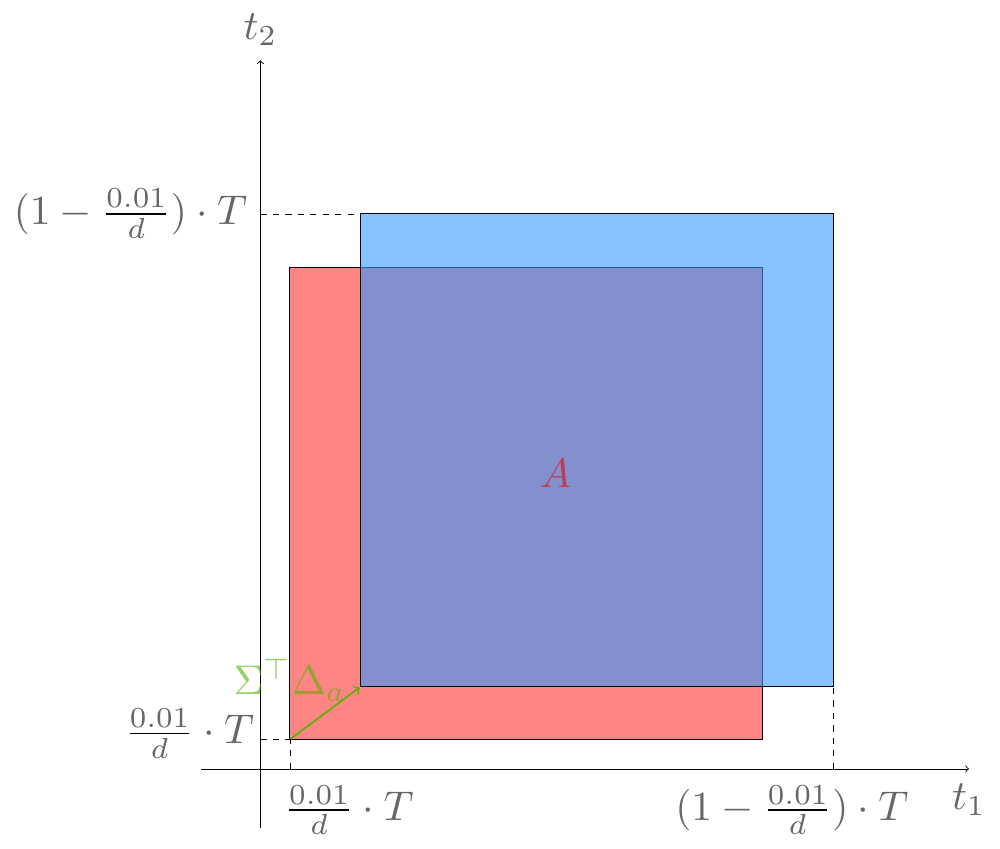}\\
    \small {(b)~Sampling for $\Sigma^{\top} a$ and $\Sigma^{\top} (a + \Delta_{a})$}
    \label{fig:our_technique:permutation_hashing:sampling_a}
    \end{tabular}
    \caption{
    Demonstration of the sampling for $\Delta_{a}$, and then for $a$ and $(a + \Delta_{a})$, where in Figure~13(a) the parameter $r := \frac{\varpi \cdot M}{4 L^{\mathpzc{dia}}}$.}
    \label{fig:sample_time_point}
\end{figure}

% \begin{figure}
%     \centering
%     \subfloat[Sampling for $\Sigma^{\top} \Delta_{a}$
%     \label{fig:sample_time_point_Delta}]
%     {\includegraphics[width = .35\textwidth]{sample_sigma_delta.pdf}}
%     \hfill
%     \subfloat[Sampling for $\Sigma^{\top} a$ and $\Sigma^{\top} (a + \Delta_{a})$
%     \label{fig:sample_time_point_a}]
%     {\includegraphics[width = .5\textwidth]{sample_a_appendix.pdf}}
%     \caption{Demonstration of the sampling for $\Delta_{a}$, and then for $a$ and $(a + \Delta_{a})$, where in Figure~\ref{fig:sample_time_point_Delta} the parameter $r := \frac{\varpi \cdot M}{4 L^{\mathpzc{dia}}}$.}
%     \label{fig:sample_time_point}
% \end{figure}

\subsubsection{Duration requirement}
\label{sec:locate_inner_duration_require}

The goal of this part is to prove Lemma~\ref{lem:locate_inner_duration_require}, and thus to obtain the duration bound required by Condition~\ref{con:duration}.

\begin{lemma}[Duration of {\LocateInner}]
\label{lem:locate_inner_duration_require}
To satisfy Condition~\ref{con:duration}, the sampling duration requirement of the procedure {\LocateInner} (Algorithm~\ref{alg:locate_inner}) is
\begin{align*}
   T ~ = ~ \Omega \big( d^{3} \cdot \eta^{-1} \cdot \log(k d / \delta) \big).
\end{align*}
\end{lemma}

\begin{proof}
The procedure {\LocateInner} uses the samples in the time domain by invoking the subroutine {\HashToBins} (Algorithm~\ref{alg:HashToBins_multi}) with a number of pairs $a' \in \{a, a + \Delta_{a}\}$ output by {\SampleTimePoint} (Algorithm~\ref{alg:sample_time_point}). In particular (see Line~\ref{alg:HashToBins_multi:y} of {\HashToBins}), we take the following sample for all $i \in [BD]^d$ and both $a' \in \{a, a + \Delta_{a}\}$:
\begin{align*}
\mathcal{P}_{\Sigma, b, a'} x(i)
~ = ~x\big(\Sigma^\top (i + a')\big) \cdot e^{-2\pi\i\cdot b^\top i},
\end{align*}
where the equation follows from Definition~\ref{def:permutation_multi}.

To meet Condition~\ref{con:duration}, we shall have
\begin{align}
    \Sigma^\top (i + a') ~ \in ~ [0, T]^{d},
    \label{eq:locate_inner_duration_require:1}
\end{align}
for all $i \in [BD]^d$ and both $a' \in \{a, a + \Delta_{a}\}$, under any choice of the random matrix $\Sigma \in \R^{d \times d}$ (according to Definition~\ref{def:HashToBins_multi_parameters}).

We know from Line~\ref{alg:sample_time_point:range} that both $a' \in \{a, a + \Delta_{a}\}$ satisfy that
\begin{align*}
    \Sigma^\top a' ~ \in ~ \Big[\frac{0.01}{d} \cdot T, \big(1 - \frac{0.01}{d}\big) \cdot T\Big]^{d}.
\end{align*}
Given this, a sufficient condition for Equation~\eqref{eq:locate_inner_duration_require:1} is that
\begin{align*}
    \| \Sigma^\top i \|_{\infty} ~ \leq ~ \frac{0.01}{d} \cdot T,
\end{align*}
for all $i \in [BD]^d$, under any choice of the random matrix $\Sigma \in \R^{d \times d}$.

For the above equation, we deduce that
\begin{eqnarray*}
    \| \Sigma^\top i \|_{\infty}
    & \leq & \| \Sigma^\top i \|_{2} \\
    & \leq & \frac{4 \sqrt{d}}{B \eta} \cdot \| i \|_{2} \\
    & \leq & \frac{4 \sqrt{d}}{B \eta} \cdot \sqrt{d} \cdot B D \\
    & \lesssim & d^{2} \cdot \eta^{-1} \cdot \log(k d / \delta),
\end{eqnarray*}
where the second step follows because $\Sigma \in \R^{d \times d}$ is a {\em rotation matrix} scaled by a random factor $\beta \sim \unif[\frac{2 \sqrt{d}}{B \eta}, \frac{4 \sqrt{d}}{B \eta}]$ (Definition~\ref{def:HashToBins_multi_parameters}); the third step follows since $i \in [BD]^d = \{0, 1, \cdots, B D - 1\}^{d}$; and the last step follows because $D = \Theta(d \cdot \log(k d / \delta))$ (see Definition~\ref{def:HashToBins_multi_parameters}).

Putting the above arguments together, we know that Condition~\ref{con:duration} holds for any sufficiently large $T = \Omega ( d^{3} \cdot \eta^{-1} \cdot \log(k d / \delta) )$.

This completes the proof.
\end{proof}

\subsubsection{Performance guarantees}
\label{sec:locate_inner_sampling_require}

The goal of this part is to prove Lemma~\ref{lem:locate_inner_sampling_require}, and thus to verify Condition~\ref{con:sampling}.

\begin{lemma}[Performance guarantees]
\label{lem:locate_inner_sampling_require}
Suppose that Condition~\ref{con:duration} is true and that $L^{\mathpzc{dia}} \geq \frac{20 d}{T}$ (which will be ensured by Definition~\ref{def:locate_signal_setup}), then Condition~\ref{con:sampling} holds for both $a' \in \{a, a + \Delta_{a}\}$ derived from the procedure {\SampleTimePoint} (Algorithm~\ref{alg:sample_time_point}):
\begin{eqnarray*}
    \E_{a'} \left[ g\big(\Sigma^{\top} (i + a')\big)^2 \right]
    & \lesssim & \frac{1}{T^{d}} \cdot \int_{t \in [0, T]^{d}} |g(t)|^{2} \cdot \d t,
\end{eqnarray*}
for all $i \in [BD]^d$, under any choice of the random matrix $\Sigma \in \R^{d \times d}$ (according to Definition~\ref{def:HashToBins_multi_parameters}).
\end{lemma}

\begin{proof}
Since both time points $a' = a$ and $a' = a + \Delta_{a}$ are constructed in a symmetric fashion (see Line~\ref{alg:sample_time_point:range} of {\SampleTimePoint}), we only need to reason about the time point $a \sim \unif(A)$ given in {\SampleTimePoint}. Denote $T' := (1 - 0.02 / d) \cdot T \geq 0.98 \cdot T$. By construction (see Line~\ref{alg:sample_time_point:Delta}),
\begin{eqnarray*}
    \| \Sigma^{\top} \Delta_{a} \|_{2}
    & \leq & \frac{\varpi \cdot M}{2 L^{\mathpzc{dia}}} \\
    & \leq & \frac{17 \sqrt{d}}{2 L^{\mathpzc{dia}}} \\
    & \leq & \frac{17}{40 \sqrt{d}} \cdot T \\
    & \leq & \frac{1}{2 \sqrt{d}} \cdot T',
\end{eqnarray*}
where the second step follows since $\mathcal{C} \geq 120$ and $\varpi = \mathcal{C}^{-2 / 3}$ and $M = 4 \cdot \lceil 4 \sqrt{d} \cdot \mathcal{C}^{2 / 3} \rceil \leq 17 \sqrt{d} \cdot \mathcal{C}^{2 / 3}$ (see Definition~\ref{def:locate_inner_setup}); the third step follows from the premise that $L^{\mathpzc{dia}} \geq \frac{20 d}{T}$; and the last step holds because $T' \geq 0.98 T$.

%Denote by $\Delta_{a, r} \in \R$ the $r$-th coordinate of $\Delta_{a} \in \R^{d}$, then
We have
\[
    \| \Sigma^{\top} \Delta_{a} \|_{\infty}
    ~ \leq ~ \| \Sigma^{\top} \Delta_{a}\|_{1}
    ~ \leq ~ \sqrt{d} \cdot \| \Sigma^{\top} \Delta_{a}\|_{2}
    ~ \leq ~ \frac{1}{2} \cdot T'.
\]
Let $(\Sigma^{\top} \Delta_{a})_{r}$ denote the $r$-th coordinate of $\Sigma^{\top} \Delta_{a} \in \R^{d}$. For any choice of $\Delta_{a}$ by {\SampleTimePoint}, the volume of the sampling range $\Sigma^{\top} a \sim \unif(A)$ is
\begin{eqnarray*}
    \mathrm{vol}(A)
    & = & \prod_{r \in [d]} \big(T' - | (\Sigma^{\top} \Delta_{a})_{r} |\big) \\
    & = & T'^{d} \cdot \prod_{r \in [d]} \big(1 - | (\Sigma^{\top} \Delta_{a})_{r} | \cdot T'^{-1}\big) \\
    & \geq & T'^{d} \cdot \prod_{r \in [d]} \exp\big(-2 \cdot | (\Sigma^{\top} \Delta_{a})_{r} | \cdot T'^{-1}\big) \\
    & = & T'^{d} \cdot \exp\big(-2 \cdot \|\Sigma^{\top} \Delta_{a}\|_{1} \cdot T'^{-1}\big) \\
    & \geq & T'^{d} \cdot e^{-1} \\
    & \geq & T^{d} \cdot 0.98 \cdot e^{-1},
\end{eqnarray*}
where the first step is by Line~\ref{alg:sample_time_point:range} of {\SampleTimePoint}; the third step follows since $| (\Sigma^{\top} \Delta_{a})_{r} | \leq \| \Sigma^{\top} \Delta_{a} \|_{\infty} \leq \frac{1}{2} \cdot T'$ and $1 - z \geq e^{-2 z}$ when $z \in [0, \frac{1}{2}]$; the fifth step follows since $\|\Sigma^{\top} \Delta_{a}\|_{1} \leq \frac{1}{2} \cdot T'$; and the last step is because $T'^{d} = (1 - 0.02 / d)^{d} \cdot T^{d} \geq 0.98 \cdot T^{d}$.

We conclude from the above that, for any choice of $\Delta_{a}$ by {\SampleTimePoint}, the time point $\Sigma^{\top} a \sim \unif(A)$ is sampled uniformly from a {\em constant} proportion of the duration $t \in [0, T]^{d}$. And because $\Sigma^{\top} (i + a)$ is guaranteed to be within the duration $t \in [0, T]^{d}$, for any choice of $\Sigma \in \R^{d \times d}$ and any $i \in [B D]^{d}$, we have
\begin{eqnarray*}
    \E_{a} \left[ g\big(\Sigma^{\top} (i + a)\big)^2 \right]
    & \lesssim & \frac{1}{T^{d}} \cdot \int_{t \in [0, T]^{d}} |g(t)|^{2} \cdot \d t.
\end{eqnarray*}

This completes the proof.
\end{proof}

\subsection{Stronger Guarantee}
\label{sec:locate_inner_stronger}

% \begin{algorithm}[!ht]
% \caption{A stronger version of {\LocateInner} when search range is small, See Lemma~\ref{lem:locate_inner_stronger}}
% \label{alg:locater_inner_stronger}
% \SetKwProg{myproc}{Procedure}{}{}
% \myproc{$\mathrm{LocateInner}^*(\Sigma, b, D, \mathpzc{List}, L^{\mathpzc{dia}}, \mathcal{C}, T$)}
% {
% \For{$r = 1, 2, \cdots, \mathcal{R}_{\mathrm{reg}}$}
% {
%     $(a_{r}, \Delta_{a}^{r}) \gets \SampleTimePoint(M, L^{\mathpzc{dia}}, \mathcal{C}, T)$\;\tcp*{Algorithm~\ref{alg:sample_time_point}}
    
%      $\hat{u} \gets \HashToBins(x, \Sigma, b, a_{r}, D)$\;
    
%      $\hat{u}' \gets \HashToBins(x, \Sigma, b, a_{r} + \Delta_{a}^{r}, D)$\;
    
%       $\phi_{j,r} = \arg(\hat{u}_{j}) - \arg(\hat{u}_{j})$\;
% }
%  $\mathpzc{List}_{\text{new}}\leftarrow \emptyset$\;
    
%      Form matrix $\Delta^{\top} := [ \Sigma^{\top} \Delta_{a}^{1}, \cdots, \Sigma^{\top} \Delta_{a}^{\mathcal{R}_{\mathrm{reg}}} ] \in \R^{d \times \mathcal{R}_{\mathrm{reg}}}$\;
    
%     \For{$j \in [B]$}
%     { Form vector $\phi_{j} \in \R^{\mathcal{R}_{\mathrm{reg}}} $\;
        
%      $\mathpzc{List}_{\text{new}} \leftarrow \mathpzc{List}_{\text{new}} \cup \{\frac{1}{2\pi} \cdot \Delta^{\dagger} \phi_j\}$\;
%      }
     
%      \Return $\mathpzc{List}_{\text{new}}$\;
% }
% \end{algorithm}

\begin{algorithm}[!ht]\caption{A stronger version of {\LocateInner} when search range is small}
\label{alg:locater_inner_stronger}
\begin{algorithmic}[1]
\Procedure{LocateInner*}{$\Sigma, b, D, \mathpzc{List}, L^{\mathpzc{dia}}, \mathcal{C}, T$} \Comment{Lemma~\ref{lem:locate_inner_stronger}}
    \For{$r = 1, 2, \cdots, \mathcal{R}_{\mathrm{reg}}$}

    \State $(a_{r}, \Delta_{a}^{r}) \gets \SampleTimePoint(M, L^{\mathpzc{dia}}, \mathcal{C}, T)$.
    \Comment{Algorithm~\ref{alg:sample_time_point}}

    \State $\hat{u} \gets \HashToBins(x, \Sigma, b, a_{r}, D)$.

    \State $\hat{u}' \gets \HashToBins(x, \Sigma, b, a_{r} + \Delta_{a}^{r}, D)$.

    \State  $\phi_{j,r} = \arg(\hat{u}_{j}) - \arg(\hat{u}_{j})$.
    
    \EndFor
    
    \State $\mathpzc{List}_{\text{new}}\leftarrow \emptyset$
    
    \State Form matrix $\Delta^{\top} := [ \Sigma^{\top} \Delta_{a}^{1}, \cdots, \Sigma^{\top} \Delta_{a}^{\mathcal{R}_{\mathrm{reg}}} ] \in \R^{d \times \mathcal{R}_{\mathrm{reg}}}$.
    
    \For{$j \in [B]$}
        \State Form vector $\phi_{j} \in \R^{\mathcal{R}_{\mathrm{reg}}} $.
        
        \State $\mathpzc{List}_{\text{new}} \leftarrow \mathpzc{List}_{\text{new}} \cup \{\frac{1}{2\pi} \cdot \Delta^{\dagger} \phi_j\}$.
    \EndFor
    
    \State \Return $\mathpzc{List}_{\text{new}}$.
\EndProcedure
\end{algorithmic}
\end{algorithm}

\begin{lemma}[Stronger guarantees]
\label{lem:locate_inner_stronger}
Let 
\begin{align*}
    \rho^2 = | \hat{x}[f] |^{2}/\E_{a} [| \hat{u}_{j} - \hat{x}[f] \cdot e^{2 \pi \i \cdot a^{\top} \Sigma f} |^{2}]
\end{align*}
Let $C_*= d^2$. Let $\mathcal{R}_{\mathrm{reg}} = C\cdot d$ for some constant $C$. Let $L^{\mathpzc{dia}} = 20d/T$. There is an algorithm (procedure \textsc{LocateInner*} in Algorithm~\ref{alg:locater_inner_stronger}) that output a list of frequencies such that there is a mapping $\pi : [k] \rightarrow [m] $,
\begin{align*}
    \| f_{\pi(i)}' - f_i\|_2 \lesssim \frac{C_*}{\rho T}, ~~~ \forall i \in [k].
\end{align*}
\end{lemma}

\begin{figure}[htbp]
    \centering
    \begin{tabular}[b]{c}
    \includegraphics[width = .425\textwidth]{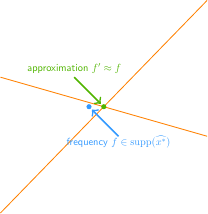}\\
    \small{(a)~Two dimensions}
    \end{tabular} \qquad
    \begin{tabular}[b]{c}
    \includegraphics[width = .425\textwidth]{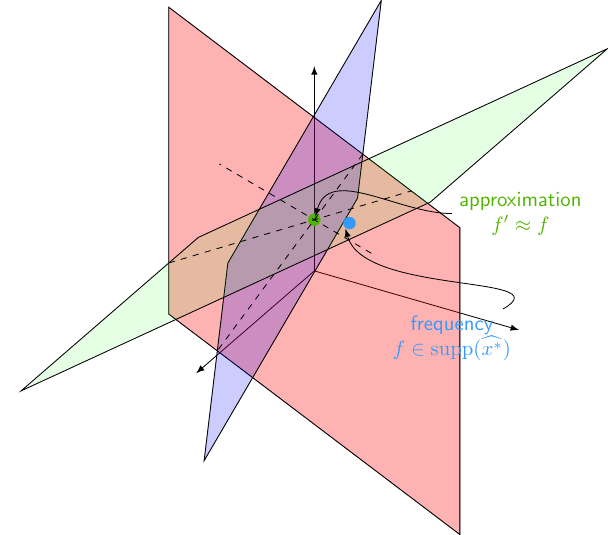}\\
    \small {(b)~Three dimensions}
    \end{tabular}
    \caption{
    Demonstration for Algorithm~\ref{alg:locater_inner_stronger} in two dimensions $d = 2$ and three dimensions $d = 3$.}
    \label{fig:locate_inner_stronger}
\end{figure}

% \begin{figure}
%     \centering
%     \subfloat[Two dimensions
%     \label{fig:locate_inner_stronger_d2}]{
%     \includegraphics[width = .425\textwidth]{twodim_fre_onesample.pdf}
%     }
%     \hfill
%     \subfloat[Three dimensions
%     \label{fig:locate_inner_stronger_d3}]{
%     \includegraphics[width = .525\textwidth]{threedim_fre_onesample.pdf}
%     }
%     \caption{Demonstration for Algorithm~\ref{alg:locater_inner_stronger} in two dimensions $d = 2$ and three dimensions $d = 3$.}
%     \label{fig:locate_inner_stronger}
% \end{figure}

\begin{proof}
We provide Figure~\ref{fig:locate_inner_stronger} for demonstration. Let $d'=\mathcal{R}_{\mathrm{reg}}=Cd$ for simplicity. By Markov inequality, we know that the following holds with probability at least $1 - \frac{1}{10d'}$:
\begin{align*}
    & \left|\hat{u}_{j} - \hat{x}[f] \cdot e^{2 \pi \i \cdot a^{\top} \Sigma f}\right|
    ~ \leq ~ \big|\hat{x}[f]\big|\cdot 10\sqrt{d'}/\rho,
\end{align*}
which is equivalent to
\begin{align*}
    \left|\hat{u}_{j} / \hat{x}[f] \cdot e^{ - 2 \pi \i \cdot a^{\top} \Sigma f} - 1\right|
    ~ \leq ~ 10 \sqrt{d'} / \rho.
\end{align*}
Namely, the complex number $\hat{u}_{j} / \hat{x}[f] \cdot e^{- 2 \pi \i \cdot a^{\top} \Sigma f}$ lies in the circle $\{z \in \C: |z - 1| \leq \varpi / \sqrt{2}\}$. Clearly, any complex number in this circle has the phase less than $ \sin^{-1}(10\sqrt{d'}/\rho.)$. In particular,
\begin{align*}
    \Big\|\underbrace{\arg(\hat{u}_{j}) - \arg(\hat{x}[f]) - 2 \pi \cdot a^{\top} \Sigma f}_{A_{1}}\Big\|_{\bigcirc}
   \leq \sin^{-1}\big(10\sqrt{d'}/\rho\big),
\end{align*}
where $\|\theta\|_{\bigcirc} \in [-\pi, \pi)$ denotes the ``phase distance'' $\min_{z\in \Z}|\theta-2\pi\cdot z|$.

Similarly, when $a$ is replaced with $(a + \Delta_{a})$, with probability $1 - \frac{1}{10d}$ we also have
\begin{align*}
    \Big\|\underbrace{\arg(\hat{u}_{j}') - \arg(\hat{x}[f]) - 2 \pi \cdot (a + \Delta_{a})^{\top} \Sigma f}_{A_{2}}\Big\|_{\bigcirc}
    \leq \sin^{-1}(10\sqrt{d'}/\rho),
\end{align*}

Put the above two inequalities together, (by the union bound) the following holds for the phase difference $\phi_{j,r} = \arg(\hat{u}_{j}) - \arg(\hat{u}_{j}')$ with probability $1 - \frac{2}{10d'}$:
\begin{eqnarray*}
    \notag
    \|\phi_{j,r} - 2 \pi \cdot \Delta_{a}^{r \top} \Sigma f \|_{\bigcirc}
    & = & \| A_{1} - A_{2} \|_{\bigcirc} \\
    \notag
    & \leq & \| A_{1} \|_{\bigcirc} + \| A_{2} \|_{\bigcirc} \\
    \notag
    & \leq & 2 \sin^{-1} ( 10\sqrt{d}/\rho ) \\
    & \leq & 10\sqrt{d'}/\rho.
\end{eqnarray*}
where the second step applies the triangle inequality; and last step follows since for any $z \in (0, 1)$, we have $\sin^{-1}(z / \sqrt{2}) \leq (\pi / 4) \cdot z$.

Then with probability at least 0.8, we have that 
\begin{align*}
    \|\phi - 2\pi \cdot \Delta \cdot f \|_{\infty} \leq 10\sqrt{d'} / \rho,
\end{align*}
where
\begin{align*}
    & \phi ~ := ~
    \begin{bmatrix}
    \phi_{j, 1} \\
    \phi_{j, 2} \\
    \vdots \\
    \phi_{j, d'}
    \end{bmatrix}
    \in \R^{d'}
    && \mbox{and}
    && \Delta ~ := ~
    \begin{bmatrix}
    \Delta_{a}^{1 \top} \Sigma \\
    \Delta_{a}^{2 \top} \Sigma \\
    \vdots \\
    \Delta_{a}^{d \top} \Sigma
    \end{bmatrix}
    \in \R^{d' \times d}.
\end{align*}

We deduce from the above that
\begin{align*}
    \|\phi-2 \pi \cdot \Delta f \|_{2}
    ~ \leq ~ \sqrt{d'} \cdot 10\sqrt{d'} / \rho
    ~ = ~ 10 d' / \rho.
\end{align*}

$\{\Sigma^{\top} \Delta_a^{r}\}_{r \in [d']}$ are uniformly distributed on a sphere. Consider any $r\in[d']$, by Theorem 3.4.6 in \cite{ver18}, we know that $\Sigma^{\top} \Delta_a^{r}$ is sub-guassian. Besides, the value of each coordinate of $\Sigma^{\top}\Delta_a^{r}$ follows a $Beta$-distribution, and $\frac{d}{\|\Sigma^{\top}\Delta_a^r\|_2^2}\cdot \E[\Delta_a^{r\top}\Sigma\Sigma^T\Delta_a^r]=I$. Thus we know that  $\frac{\sqrt{d}}{\|\Sigma^{\top}\Delta_a^r\|_2}\Sigma^{\top} \Delta_a^{r}$ is a sub-gaussian isotropic random vector. By selecting $\|\Sigma^{\top} \Delta_a^{r}\|_2 \eqsim T / d $ and the constant $C=d'/d$ large enough, then by Lemma~\ref{lem:sub-gaussian rows},  we can show that with probability at least $1 - 1 / \poly(d)$, we have $s_{\min}(\Delta)\geq \sqrt{d}\cdot \frac{T}{d^{1.5}}=\frac{T}{d}$. Let $\Delta^\dagger$ represent the Generalized inverse of $\Delta$, let $f_{LS}=\Delta^\dagger \phi$ represent the least squares solution, then we have 
\begin{align*}
   2\pi \| \Delta(f-f_{LS})\|_2\leq 2\|\phi - 2\pi \cdot \Delta \cdot f \|_{2}\lesssim d/\rho.
\end{align*}
%and $\| \Delta^{-1} \| \leq d/T$, where $\| \cdot \|$ denotes the spectral norm of a matrix.

%Since the sampling strategy of $\Delta_a$ and $\Delta$ is an orthogonal matrix, this implies $\| \Delta \|=\|\Delta_a\|_2 \geq T/d$ and $\| \Delta^{-1} \| \leq d/T$. 

Then we have 
\begin{eqnarray*}
    \|f-f_{LS}\|_2 &\leq &\| \Delta(f-f_{LS})\|_2/s_{\min}(\Delta)\\
    &\lesssim& \frac{d}{\rho}\cdot\frac{d}{T}\\
    & = & C_* \cdot \frac{1}{T \rho}.
\end{eqnarray*}

where the first step is by $\|\Delta x\|_2/\|x\|_2\geq s_{\min}(\Delta)$; the second step follows from $s_{\min}(\Delta)\geq T/d$ and the last step follows because we define $C_* := d^{2}$.

This completes the proof.

\end{proof}

\section{Locate signal}
\label{sec:locate_signal}

\begin{table}[!ht]
    \centering
    \begin{tabular}{|l|l|l|l|}
        \hline
        {\bf Statement} & {\bf Section} & {\bf Algorithm} & {\bf Comment}  \\ \hline
        Definition~\ref{def:locate_signal_setup} & Section~\ref{sec:locate_singal_algorithm} & Algorithm~\ref{alg:locate_signal} & Definitions \\ \hline
        Lemma~\ref{lem:locate_signal_sample_time} & Section~\ref{sec:locate_signal_sample_time} & Algorithm~\ref{alg:locate_signal} & Sample complexity and running time \\ \hline
         Lemma~\ref{lem:locate_signal_duration} & Section~\ref{sec:locate_signal_duration} & Algorithm~\ref{alg:locate_signal} & Duration \\ \hline
        Lemma~\ref{lem:locate_signal_guarantees} & Section~\ref{sec:locate_signal_guarantees} &  Algorithm~\ref{alg:locate_signal} & Guarantees, without Alg.~\ref{alg:locater_inner_stronger} \\ \hline
        Lemma~\ref{lem:locate_signal_stronger} & Section~\ref{sec:locate_signal_stronger} & Algorithm~\ref{alg:locate_signal} & Stronger guarantees, with Alg.~\ref{alg:locater_inner_stronger}\\ \hline
    \end{tabular}
    \caption{List of Lemmas/Algorithms in locate signal section.}
    \label{tab:list..}
\end{table}

\subsection{Algorithm}
\label{sec:locate_singal_algorithm}

Denote $H := \{\xi \in \supp(\hat{x^*}): \mbox{neither $E_{\coll}(\xi)$ nor $E_{\off}(\xi)$ happens}\}$. Recall the performance guarantees given in Corollary~\ref{cor:locate_inner_guarantees}:
\begin{quote}
    Assume that a specific ``good'' tone frequency good frequency $f \in H$ locates in a hyperball $\mathbf{HB}(\mathpzc{List}[j], L^{\mathpzc{dia}})$ that is centered at some frequency $\mathpzc{List}[j] \in \R^{d}$ and has the diameter $L^{\mathpzc{dia}} > 0$.
    
    Then with probability at least $1 - \mathcal{M} \cdot 2^{-\Omega(\mathcal{R}_{\mathrm{vote}})}$, then procedure {\LocateInner} (Algorithm~\ref{alg:locate_inner}) outputs $\mathpzc{List}_{\mathrm{new}}[j]$ for which
    \begin{eqnarray*}
        f ~ \in ~
        \mathbf{HB}(\mathpzc{List}_{\mathrm{new}}[j], L_{\mathrm{new}}^{\mathpzc{dia}}).
    \end{eqnarray*}
    where the new diameter $L_{\mathrm{new}}^{\mathpzc{dia}} := \frac{1}{2} \cdot L^{\mathpzc{dia}}$.
\end{quote}
Given this, we would estimate the good frequencies by invoking the procedure {\LocateInner} repeatedly. This idea is implemented as the procedure {\LocateSignal} (Algorithm~\ref{alg:locate_signal}).

\begin{definition}[Setup for {\LocateSignal}]
\label{def:locate_signal_setup}
The procedure {\LocateSignal} keeps track of a number of $\mathcal{B} = 2^{\Theta(d \cdot \log d)} \cdot k$ hyperballs. These hyperballs have
\begin{itemize}
    \item The same initial diameter $L^{\mathpzc{dia}} := 2 \sqrt{d} \cdot F$.
    
    \item The final diameter $L^{\mathpzc{dia}} \in (\frac{20 d}{T}, \frac{40 d}{T}]$ is chosen so that $\log_{2}( \frac{\text{initial~} L^{\mathpzc{dia}}}{ \text{final~} L^{\mathpzc{dia}} })$ is an integer; clearly, this final $L^{\mathpzc{dia}}$ is well defined and is unique.
    
    \item The number of iteration $\mathcal{R}_{\mathrm{search}}
    := \log_{2}( \frac{\text{initial~} L^{\mathpzc{dia}}}{ \text{final~} L^{\mathpzc{dia}} })
    = O(\log(T \cdot F))$.
\end{itemize}
That is, each hyperball is initialized to be $\mathbf{HB}(\mathbf{0}, 2 \sqrt{d} \cdot F) \supseteq [-F, F]^{d}$. Clearly, such a hyperball contains all the ``good'' tone frequencies $f \in H$ at the beginning. Then, the subroutine {\LocateInner} is invoked $\mathcal{R}_{\mathrm{search}}$ times, until the diameter shrinks to the final $L^{\mathpzc{dia}} \in (\frac{20 d}{T}, \frac{40 d}{T}]$.
\end{definition}

% \begin{algorithm}[!ht]
% \caption{{\LocateSignal}, Lemmas~\ref{lem:locate_signal_sample_time}, \ref{lem:locate_signal_duration},  \ref{lem:locate_signal_guarantees}}
% \label{alg:locate_signal}
% \SetKwProg{myproc}{Procedure}{}{}
% \myproc{\LocateSignal($\Sigma, b, D, \mathcal{C}, T$)}
% {

% $\mathpzc{List}[j] \leftarrow \mathbf{0} \in \R^{d}$ for each $j \in [B]^{d}$\; \tcp*{Initialize the center frequency}

% $L^{\mathpzc{dia}} = 2 \sqrt{d} \cdot F$\; \tcp*{Initialize the diameter}

% \For{$r = 1, 2, \cdots, \mathcal{R}_{\mathrm{search}}$}
% {
%  $\mathpzc{List}_{\mathrm{new}} \gets \LocateInner(\Sigma, b, D, \mathpzc{List}, L^{\mathpzc{dia}}, \mathcal{C}, T)$\; \tcp*{Algorithm~\ref{alg:locate_inner}}
    
%  $\mathpzc{List} \gets \mathpzc{List}_{\mathrm{new}}$\;

% $L^{\mathpzc{dia}} \gets \frac{1}{2} \cdot L^{\mathpzc{dia}}$\;
% }

% \tcp{ $L^{\mathpzc{dia}} = \Theta(d/T)$ }

% {\color{blue}$\mathpzc{List} \gets \textsc{LocateInner*}(\Sigma, b, D, \mathpzc{List}, L^{\mathpzc{dia}}, \mathcal{C}, T)$}\; \tcp*{Algorithm~\ref{alg:locater_inner_stronger}}

% $\mathpzc{List}^* \leftarrow \mathpzc{List}$ (after removing the $\mathtt{NIL}$'s)\; \tcp{the frequencies}

% \Return $\mathpzc{List}^*$.
% }
% \end{algorithm}

\begin{algorithm}[!ht]
\caption{{\LocateSignal}, Lemmas~\ref{lem:locate_signal_sample_time}, \ref{lem:locate_signal_duration},  \ref{lem:locate_signal_guarantees}}
\label{alg:locate_signal}
\begin{algorithmic}[1]
\Procedure{\LocateSignal}{$\Sigma, b, D, \mathcal{C}, T$} 

\State  $\mathpzc{List}[j] \leftarrow \mathbf{0} \in \R^{d}$ for each $j \in [B]^{d}$. \Comment{Initialize the center frequency}

\State  $L^{\mathpzc{dia}} = 2 \sqrt{d} \cdot F$. \Comment{Initialize the diameter}

\For{$r = 1, 2, \cdots, \mathcal{R}_{\mathrm{search}}$}
    
    \State $\mathpzc{List}_{\mathrm{new}} \gets \LocateInner(\Sigma, b, D, \mathpzc{List}, L^{\mathpzc{dia}}, \mathcal{C}, T)$. \Comment{Algorithm~\ref{alg:locate_inner}}
    
    \State $\mathpzc{List} \gets \mathpzc{List}_{\mathrm{new}}$.
    
    \State $L^{\mathpzc{dia}} \gets \frac{1}{2} \cdot L^{\mathpzc{dia}}$.
\EndFor

\State \Comment{ $L^{\mathpzc{dia}} = \Theta(d/T)$ }

\State {\color{blue}$\mathpzc{List} \gets \textsc{LocateInner*}(\Sigma, b, D, \mathpzc{List}, L^{\mathpzc{dia}}, \mathcal{C}, T)$} \Comment{Algorithm~\ref{alg:locater_inner_stronger}}

\State $\mathpzc{List}^* \leftarrow \mathpzc{List}$ (after removing the $\mathtt{NIL}$'s) \Comment{the frequencies}

\State \Return $\mathpzc{List}^*$.
\EndProcedure
\end{algorithmic}

\end{algorithm}

\subsection{Sample complexity and running time}
\label{sec:locate_signal_sample_time}

The goal of this section is to prove Lemma~\ref{lem:locate_signal_sample_time}.

\begin{lemma}[Sample complexity and running time of {\LocateSignal}]
\label{lem:locate_signal_sample_time}
The procedure {\LocateSignal} (Algorithm~\ref{alg:locate_signal}) has the following performance guarantees:
\begin{itemize}
    \item The sample complexity is $2^{\Theta(d \cdot \log d)} \cdot k \cdot \mathcal{D} \cdot (\log \mathcal{C} + \log \log (F / \eta)) \cdot \log(T \cdot F)$.
    
    \item The running time is $2^{\Theta(d \cdot \log(\mathcal{C} \cdot d))} \cdot k \cdot (\mathcal{D} + \log k) \cdot \log\log(F/\eta) \cdot \log( T \cdot F)$.
    
    \item The output $\mathpzc{List}^*$ contains at most $O(\mathcal{B}) = O(B^{d}) = 2^{O(d \cdot \log d)} \cdot k$ many candidate frequencies.
\end{itemize}
\end{lemma}

\begin{figure}
    \centering
    \includegraphics[width = .5\textwidth]{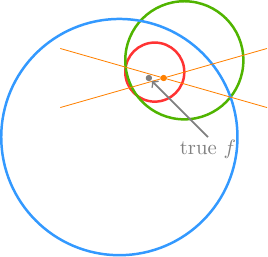}
    %\hfill
    %\includegraphics[width = .4\textwidth]{twodim_fre.pdf}
    \caption{Demonstration of Algorithm~\ref{alg:locate_signal} for a single frequency $f \in \supp(\hat{x^*})$ in two dimensions $d = 2$. The blue/green/red circles refer to the coarse-grained location, and the ``orange'' lines refer to the fine-grained location.}
    \label{fig:locate_signal_sample_time}
\end{figure}

\begin{proof}
How many frequencies the output $\mathpzc{List}^*$ contains is easy to see, since $\mathpzc{List}^*$ is indexed by the bins $j \in [B]^{d}$. Below we quantify the sample complexity and the running time.

\vspace{.1in}
{\bf Sample complexity.}
The procedure {\LocateSignal} invokes the subroutine {\LocateInner} $\mathcal{R}_{\mathrm{search}}$ times. Due to Lemma~\ref{lem:locate_inner_sample_time}, the subroutine {\LocateInner} has the sample complexity
\begin{align*}
    \#\mathtt{sample}({\LocateInner})
    ~ = ~ 2^{\Theta(d\cdot\log d)} \cdot k \cdot \mathcal{D} \cdot (\log \mathcal{C} + \log \log (F / \eta)).
\end{align*}
Thus, {\LocateSignal} has the sample complexity
\begin{eqnarray*}
    \#\mathtt{sample}({\LocateSignal})
    & = & \#\mathtt{sample}({\LocateInner}) ~ \cdot ~ \mathcal{R}_{\mathrm{search}} \\
    & = & 2^{\Theta(d\cdot\log d)} \cdot k \cdot \mathcal{D} \cdot (\log \mathcal{C} + \log \log (F / \eta)) \cdot \log(T \cdot F).
\end{eqnarray*}

\vspace{.1in}
{\bf Running time.}
Due to to Lemma~\ref{lem:locate_inner_sample_time}, the subroutine {\LocateInner} has the running time
\begin{eqnarray*}
    \#\mathtt{time}({\LocateInner})
    & = & 2^{\Theta(d \cdot \log(\mathcal{C} \cdot d))} \cdot k \cdot (\mathcal{D} + \log k) \cdot \log\log(F/\eta).
\end{eqnarray*}
Thus, {\LocateSignal} has the running
\begin{eqnarray*}
    \#\mathtt{time}({\LocateSignal})
    & = & \#\mathtt{time}({\LocateInner}) ~ \cdot ~ \mathcal{R}_{\mathrm{search}} \\
    & = & 2^{\Theta(d \cdot \log(\mathcal{C} \cdot d))} \cdot k \cdot (\mathcal{D} + \log k) \cdot \log\log(F/\eta) \cdot \log( T \cdot F).
\end{eqnarray*}

This completes the proof of Lemma~\ref{lem:locate_signal_sample_time}.
\end{proof}

\subsection{Duration requirement}
\label{sec:locate_signal_duration}

The goal of this section is to prove Lemma~\ref{lem:locate_signal_duration}.

\begin{lemma}[Duration of {\LocateSignal}]
\label{lem:locate_signal_duration}
The sampling duration requirement of the procedure {\LocateSignal} (Algorithm~\ref{alg:locate_signal}) is
\begin{align*}
   T ~ = ~ \Omega \big( d^{3} \cdot \eta^{-1} \cdot \log(k d / \delta) \big).
\end{align*}
\end{lemma}

\begin{proof}
This follows immediately from Lemma~\ref{lem:locate_inner_duration_require}.
\end{proof}

\subsection{Performance guarantees}
\label{sec:locate_signal_guarantees}

The goal of this section is to prove Lemma~\ref{lem:locate_signal_guarantees}.

\begin{lemma}[Guarantees of {\LocateSignal}]
\label{lem:locate_signal_guarantees}
Given $\Sigma \in \R^{d \times d}$ and $b \in \R^d$, the output list $\mathpzc{List}^*$ of procedure {\LocateSignal} (Algorithm~\ref{alg:locate_signal}) contains at most $\mathcal{B} = 2^{O(d \cdot \log d)} \cdot k$ many frequencies with minimum separation $\Omega(\eta)$. Let $H \subseteq \supp(\hat{x^*})$ be a subset of ``good'' tone frequencies:
\begin{align*}
    H = \{\xi \in \supp(\hat{x^*}): \mbox{neither $E_{\off}(\xi)$ nor $E_{\coll}(\xi)$ happens}\}
\end{align*}
For any good frequency $f \in H$, suppose that its signal-to-noise ratio $\rho(i) \geq \mathcal{C}$ (see Definition~\ref{def:ratio_signal_noise}), then with probability at least $99\%$, there exists an output frequency $f' \in \mathpzc{List}^*$ such that 
\begin{align*}
\| f - f' \|_{2} ~ \lesssim ~ \frac{d}{ T}.
\end{align*}
\end{lemma}

\begin{proof}
The concerning frequency $f \in H$ w.l.o.g.\ is hashed into the bin $j := \mathpzc{h}_{\Sigma,b}(f) \in [B]^{d}$. Recall Definition~\ref{def:locate_signal_setup} that the procedure {\LocateSignal} keeps track of a number of $\mathcal{B} = 2^{\Theta(d \cdot \log d)} \cdot k$ hyperballs. The $j$-th hyperball is initialized to be $\mathbf{HB}(\mathbf{0}, 2 \sqrt{d} \cdot F) \supseteq [-F, F]^{d}$, and thus contains the frequency $f \in H$. Then, the procedure {\LocateSignal} invokes the subroutine {\LocateInner} $\mathcal{R}_{\mathrm{search}}$ times, each of which shrinks the diameter of the $j$-th hyperball by half, until the diameter drops down to the final $L^{\mathpzc{dia}} \in (\frac{20 d}{T}, \frac{40 d}{T}]$.

\vspace{.1in}
{\bf Failure probability.}
For the concerning frequency $f \in H$, we know from Corollary~\ref{cor:locate_inner_guarantees} that each invocation of {\LocateInner} fails with probability at most $\mathcal{M} \cdot 2^{-\Omega(\mathcal{R}_{\mathrm{vote}})}$. By the union bound, the failure probability of {\LocateSignal} is at most
\begin{eqnarray*}
    \mathcal{R}_{\mathrm{search}} \cdot \mathcal{M} \cdot 2^{-\Omega(\mathcal{R}_{\mathrm{vote}})}
    & = & \underbrace{O(\log(F \cdot T))}_{\mathcal{R}_{\mathrm{search}}}
    ~ \cdot ~ \underbrace{2^{\Theta(d \cdot \log (\mathcal{C} \cdot d))}}_{\mathcal{M}}
    ~ \cdot ~ 2^{-\Omega(\mathcal{R}_{\mathrm{vote}})} \\
    & \leq & 1\%,
\end{eqnarray*}
where the first step follows because the parameters $\mathcal{R}_{\mathrm{search}} = O(\log(T \cdot F))$ and $\mathcal{M} = 2^{\Theta(d \cdot \log (\mathcal{C} \cdot d))}$ (see Definitions~\ref{def:locate_signal_setup} and \ref{def:locate_inner_setup}); and the last step holds since we choose in Definition~\ref{def:locate_inner_setup} a sufficiently large $\mathcal{R}_{\mathrm{vote}} = \Theta\big(d \cdot \log (\mathcal{C} \cdot d) + \log \log(F / \eta))$.

\vspace{.1in}
{\bf Performance guarantee.}
At the beginning, the initial $j$-th hyperball $\mathbf{HB}(\mathbf{0}, 2 F) = [-F, F]^{d}$ contains the concerning frequency $f \in H$. If the procedure {\LocateSignal} succeeds in all of the first $r \in [\mathcal{R}_{\mathrm{search}}]$ iterations, then (Corollary~\ref{cor:locate_inner_guarantees}) we locate $f \in H$ within a hyperball that is centered at some frequency $\mathpzc{List}_{\mathrm{new}}[j] \in [-F, F]^{d}$ and has the diameter $L_{\mathrm{new}}^{\mathpzc{dia}} = 2 \sqrt{d} \cdot F \cdot 2^{-r}$. Formally, we have
\begin{eqnarray*}
    f ~ \in ~
    \mathbf{HB}(\mathpzc{List}_{\mathrm{new}}[j], L_{\mathrm{new}}^{\mathpzc{dia}}).
\end{eqnarray*}

In particular, if all of the $\mathcal{R}_{\mathrm{search}}$ iterations succeed, the diameter drops down to the final $L^{\mathpzc{dia}} \in (\frac{20 d}{T}, \frac{40 d}{T}]$. As a consequence, we have
\begin{eqnarray*}
    f ~ \in ~
    \mathbf{HB}(\mathpzc{List}^*[j], 40 d / T).
\end{eqnarray*}
That is, the $\ell_{2}$-distance between the concerning tone frequency $f \in H$ and the output frequency $\mathpzc{List}^*[j]$ is at most $\frac{1}{2} \cdot \frac{40 d}{T} = \frac{20 d}{T}$.

This completes the proof of Lemma~\ref{lem:locate_signal_guarantees}.
\end{proof}

\subsection{Stronger guarantees}\label{sec:locate_signal_stronger}
The goal of this section is to improve Lemma~\ref{lem:locate_signal_guarantees}.
\begin{lemma}[Stronger guarantees, compared to Lemma~\ref{lem:locate_signal_guarantees}]
\label{lem:locate_signal_stronger}
Given $\Sigma \in \R^{d \times d}$ and $b \in \R^d$ (according to Definition~\ref{def:HashToBins_multi_parameters}), the output list $\mathpzc{List}^*$ of procedure {\LocateSignal} (Algorithm~\ref{alg:locate_signal}) contains at most $\mathcal{B} = 2^{O(d \cdot \log d)} \cdot k$ many frequencies with minimum separation $\Omega(\eta)$. Let $H \subseteq \supp(\hat{x^*})$ be a subset of ``good'' tone frequencies:
\begin{align*}
    H = \{\xi \in \supp(\hat{x^*}): \mbox{neither $E_{\off}(\xi)$ nor $E_{\coll}(\xi)$ happens}\}
\end{align*}
For any good frequency $f \in H$, suppose that its signal-to-noise ratio $\rho(i) \geq \mathcal{C}$ (see Definition~\ref{def:ratio_signal_noise}), then with probability at least $99\%$, there exists an output frequency $f' \in \mathpzc{List}^*$ such that 
\begin{align*}
\| f - f' \|_{2} ~ \lesssim ~ C_* \cdot \frac{1}{ \rho T }.
\end{align*}
\end{lemma}
\begin{proof}
The proof follows from Lemma~\ref{lem:locate_inner_stronger}. 
\end{proof}

 %%% Locate signal
% \newpage
\section{Sparse recovery}
\label{sec:sparse_recovery}

% For ease of presentation, throughout this section we would shift the sampling time domain from $t \in [0, T]^{d}$ to $t \in [-T / 2, T / 2]^{d}$.

\begin{table}[h]
    \centering
    \begin{tabular}{|l|l|l|l|} \hline
       {\bf Statement} & {\bf Section} & {\bf Algorithm} & {\bf Comment} \\ \hline
       Definition~\ref{def:ratio_signal_noise} & Section~\ref{sec:sparse_recovery_definition} & None & Definitions and facts \\  \hline
       Lemma~\ref{lem:estimate_signal} & Section~\ref{sec:estimate_signal} & Algorithm~\ref{alg:estimate_signal} & Estimate signal \\ \hline
       Lemma~\ref{lem:one_stage_sample_time} & Section~\ref{sec:one_stage_sample_time} & Algorithm~\ref{alg:one_stage} & One stage, sample complexity and running time  \\ \hline
       Lemma~\ref{lem:one_stage_guarantees} & Section~\ref{sec:one_stage_guarantees} & Algorithm~\ref{alg:one_stage} & One stage, guarantees \\ \hline
       Lemma~\ref{lem:multi_stage} & Section~\ref{sec:multi_stage} & Algorithm~\ref{alg:multi_stage} & Multi sage \\ \hline
       Lemma~\ref{lem:merged_stage_time} & Section~\ref{sec:merged_stage_time} & Algorithm~\ref{alg:merged_stage} & Merged stage, running time \\ \hline
       Lemma~\ref{lem:merged_stage_guarantees} & Section~\ref{sec:merged_stage_guarantees} & Algorithm~\ref{alg:merged_stage} & Merged stage, guarantees \\ \hline
       Lemma~\ref{lem:merged_stage_twice} & Section~\ref{sec:merged_stage_twice} & Algorithm~\ref{alg:recovery_stage} & Running merged stage twice \\ \hline
        Theorem~\ref{thm:recovery_stage} & Section~\ref{sec:recovery_stage} & Algorithm~\ref{alg:recovery_stage} & Recovery stage \\ \hline
    \end{tabular}
    \caption{List of Lemmas/Algorithms in sparse recovery section}
    \label{tab:my_label}
\end{table}

\subsection{Definitions and facts}
\label{sec:sparse_recovery_definition}

\begin{definition}[Signal-to-noise ratio]
\label{def:ratio_signal_noise}
For the $i$-th tone $(v_{i}, f_{i})$, define the signal-to-noise ratio $\rho_{i} := |v_{i}| / \mu_{i} \geq 0$, where the noise $\mu_{i} \geq 0$ is given by
\[
    \mu_{i}^2 ~ = ~ \E_{\Sigma, b, a} [ |\hat{u}_{j}' \cdot e^{-2 \pi \i \cdot a^{\top} f_{i}} - v_{i}|^{2} ],
\]
where $j = \mathpzc{h}_{\Sigma, b}(f_{i}) \in [B]^{d}$ is the bin that the tone frequency $f_{i} \in [-F, F]^{d}$ is hashed into according to Definition~\ref{def:function_hash_multi}, and $\hat{u}_{j}' = \hat{u}_{j} \cdot e^{-(\pi \i / B) \cdot \|j\|_{1}}$.
\end{definition}

\begin{definition}[Hypercube]
\label{def:sub_hypercube}
For any frequency $f \in \R^{d}$ and any $L^{\mathpzc{edge}} \geq 0$, we denote by $\mathbf{HC}(f, L^{\mathpzc{edge}})$ the $\ell_{\infty}$-norm hypercube with center $f \in \R^{d}$ and the edge length $L^{\mathpzc{edge}}$:
\begin{eqnarray*}
    \mathbf{HC}(f, L^{\mathpzc{edge}})
    & := & \left\{\xi \in \R^{d}: \| \xi - f \|_{\infty} \leq L^{\mathpzc{edge}} / 2\right\}.
\end{eqnarray*}
\end{definition}

\subsection{{\EstimateSignal}}
\label{sec:estimate_signal}

The goal of this section is to prove Lemma~\ref{lem:estimate_signal}.

\begin{algorithm}\caption{{\EstimateSignal}}\label{alg:estimate_signal}
\begin{algorithmic}[1]
\Procedure{EstimateSignal}{$\Sigma, b, a, D, T, \mathpzc{List}$}

    \State Sample $a \in \mathbb{R}^{d}$ according to Definition~\ref{def:locate_inner_setup}.
    
    \State Let $\hat{u} \gets \HashToBins(\Sigma, b, a, D)$.
    \Comment{Algorithm~\ref{alg:HashToBins_multi}}
    
    \State Let $v'(\xi) = \hat{u}_{h_{\Sigma,b}(\xi)}\cdot e^{-(\pi \i / B) \cdot \|\mathpzc{h}_{\Sigma, b}(\xi)\|_{1}} \cdot e^{-2 \pi \i \cdot a^{\top} \xi}$ for $\xi \in \mathpzc{List}$.
    
    \State \Return $\{v'(\xi)\}_{\xi \in \mathpzc{List}}$.
\EndProcedure
\end{algorithmic}
\end{algorithm}

\begin{lemma}[{\EstimateSignal}]\label{lem:estimate_signal}
The procedure {\EstimateSignal} (Algorithm~\ref{alg:estimate_signal}) satisfies that:
\begin{itemize}
    \item The sample complexity is upper bounded by the sample complexity of the procedure {\LocateSignal} (Algorithm~\ref{alg:locate_signal}).
    
    \item The running time is upper bounded by the running time of the procedure {\LocateSignal} (Algorithm~\ref{alg:locate_signal}).
\end{itemize}
Denote by $H \subseteq [k]$ the indices of a subset of true tones $\{(v_{i}, f_{i})\}_{i \in [k]}$ for which neither $E_{\coll}(f_{i})$ nor $E_{\off}(f_{i})$ happens.
There is a subset $S \subseteq H$ and an injection $\pi : S \mapsto [k]$ such that
\begin{description}[labelindent = 1em]
    \item [Property I:]
    For the tones in set $S$, the (partial) tone estimation error
    \begin{align*}
        \sum_{i\in S} \E_{\Sigma, b} \left[ \frac{1}{T^d} \cdot \int_{\tau \in [0,T]^d} \big| v_{i}' \cdot e^{2 \pi \i \cdot f_{i}'^\top \tau} - v_{\pi(i)} \cdot e^{2 \pi \i \cdot f_{\pi(i)}^\top \tau} \big|^2 \cdot \d \tau\right]
        ~ \lesssim ~ ({\cal C}^2 + d C_*^2) \cdot \N^2.
    \end{align*}

    \item [Property~III:]
    For each tone $i \in S$, the (single) tone estimation error
    \begin{align*}
        | v_{i}' - v_{\pi(i)} | ~ \lesssim ~ ({\cal C} + \sqrt{d} C_*) \cdot \N.
    \end{align*} 
\end{description}

\end{lemma}
\begin{proof}
The bounds on the sample complexity and the running time are direct follow-ups to the previous lemmas. Also, Properties~I and II will be proved soon after in Lemma~\ref{lem:one_stage_guarantees}; particularly, we will specify the subset $S \subseteq H$ therein.
\end{proof}

\subsection{{\OneStage}: algorithm, sample complexity and running time}
\label{sec:one_stage_sample_time}

The goal of this section is to prove Lemma~\ref{lem:one_stage_sample_time}.

% \begin{algorithm}
% \caption{{\OneStage}, Lemma~\ref{lem:one_stage_sample_time}, \ref{lem:one_stage_guarantees} }\label{alg:one_stage}
% \SetKwProg{myproc}{Procedure}{}{}
% \myproc{\OneStage($x,\Sigma, b, D, \mathcal{C}, T$)}
% {
% $\mathpzc{List} \gets {\LocateSignal}(\Sigma, b, D, M, \mathcal{C}, T)$\;
%     \tcp*{Algorithm~\ref{alg:locate_signal}}
    
%     \For{$\xi \in \mathpzc{List}$}
%     {
%         \If{either $E_{\coll}(\xi)$ or $E_{\off}(\xi)$ or both happen}
%          { Remove $\xi$ from $\mathpzc{List}$.
%         }
%     }
    
%      $\{v'(\xi)\}_{\xi \in \mathpzc{List}} \gets {\EstimateSignal}(\Sigma, b, a, D, T, \mathpzc{List})$\; \tcp*{Algorithm~\ref{alg:estimate_signal}}
     
%      Add an supplementary list $\mathpzc{List}_{sup}= \{(0,\xi_i)\}_{i=1}^{k}$, for which $\min_{\xi,\xi'\in \mathpzc{List}_{sup}}\|\xi-\xi'\|_2\geq \eta$, and $\min_{\xi\in\mathpzc{List},\xi'\in\mathpzc{List}_{Sup}}\|\xi-\xi'\|_2\geq \eta$.\label{lin:supplementary_list}
%     \tcp*{Used in Claim~\ref{cla:one_stage_guarantees:2}}
    
%  \Return $\{ (v'(\xi),\xi ) \}_{\xi \in \mathpzc{List}}\cup \{(0,\xi)\}_{\xi\in\mathpzc{List}_{Sup}}$.
% }
% \end{algorithm}

\begin{algorithm}
\caption{{\OneStage}, Lemma~\ref{lem:one_stage_sample_time}, \ref{lem:one_stage_guarantees} }\label{alg:one_stage}
\begin{algorithmic}[1]
 \Procedure{\OneStage}{$x,\Sigma, b, D, \mathcal{C}, T$}
    \State $\mathpzc{List} \gets {\LocateSignal}(\Sigma, b, D, M, \mathcal{C}, T)$.
    \Comment{Algorithm~\ref{alg:locate_signal}}
    
    \For{$\xi \in \mathpzc{List}$}
    
        \If{either $E_{\coll}(\xi)$ or $E_{\off}(\xi)$ or both happen}
            \State Remove $\xi$ from $\mathpzc{List}$.
        \EndIf
    
    \EndFor
    
    \State $\{v'(\xi)\}_{\xi \in \mathpzc{List}} \gets {\EstimateSignal}(\Sigma, b, a, D, T, \mathpzc{List})$ \Comment{Algorithm~\ref{alg:estimate_signal}}
    
    \State Add an supplementary list $\mathpzc{List}_{sup}= \{(0,\xi_i)\}_{i=1}^{k}$, for which $\min_{\xi,\xi'\in \mathpzc{List}_{sup}}\|\xi-\xi'\|_2\geq \eta$, and $\min_{\xi\in\mathpzc{List},\xi'\in\mathpzc{List}_{Sup}}\|\xi-\xi'\|_2\geq \eta$.\label{lin:supplementary_list}
    \Comment{Used in Claim~\ref{cla:one_stage_guarantees:2}}
    
\State \Return $\{ (v'(\xi),\xi ) \}_{\xi \in \mathpzc{List}}\cup \{(0,\xi)\}_{\xi\in\mathpzc{List}_{Sup}}$.
\EndProcedure
\end{algorithmic}
\end{algorithm}

\begin{lemma}[Sample complexity and running time of {\OneStage}]
\label{lem:one_stage_sample_time}
The procedure {\OneStage} (Algorithm~\ref{alg:one_stage}) has the following performance guarantees:
\begin{itemize}
	\item The sample complexity is $2^{\Theta(d\cdot\log d)} \cdot (\log \mathcal{C} + \log \log (F / \eta)) \cdot k \cdot \log (F \cdot T) \cdot \mathcal{D}$.
	
	\item The running time is $O(2^{\Theta(d\cdot (\log d + \log {\cal C} ))} \cdot \log(F \cdot T)\cdot  \log\log(F/\eta) \cdot k \cdot ({\cal D} + \log k))$.
	
	\item The output $\{ (v'(\xi),\xi ) \}_{\xi \in \mathpzc{List}}$ contains at most $O(\mathcal{B}) = 2^{O(d \cdot \log d)} \cdot k$ many candidate tones.
\end{itemize}
\end{lemma}

\begin{proof}
It follows directly from previous Lemma.
\end{proof}

%%%%%%%%%%%%%%%%%%%%%%%%%%%%%%%%%%%%%%%%%%%%%%%%%%%%%%%%%%%%%%%%%%%%%%%%%%%%%%%%%%%%%%%%%%%%%%%%%%%%%%%%%%%%%%%%%%%%%%%%%%%%%%%%%%%%%%%%%%%%%%%%%%%%%%%%%%%%%%%%%%%%%%%%%%%%%%%%%%%%%%%%%%%%%%%%%%%%%%%%%%%%%%%%%%%%%%%%%%%%%%%%%%%%%%%%%%%%%%%%%%%%%%%%%%%%%%%%%%%%%%%%%%%%%%%%%%%%%%%%%%%%%%%%%%%%%%%%%%%%%%%%%%%%%%%%%%%%%%%%%%%%%%%%%%%%%%%%%%%%%%%%%%%%%%%%%%%%%%%%%%%%%%%%%%%%%%%%%%%%%%%%%%%%%%%%%%%%%%%%%%%%%%%%%%%%%%%%%%%%%%%%%%%%%%%%%%%%%%%%%%%%%%%%%%%%%%%%%%%%%%%%%%%%%%%%%%%%%%%%%%%%%%%%%%%%%%%%%%%%%%%%%%%%%%%%%%%%%%%%

\subsection{{\OneStage}: performance guarantees}
\label{sec:one_stage_guarantees}

The goal of this section is to prove Lemma~\ref{lem:one_stage_guarantees}.

\begin{lemma}[Guarantees of {\OneStage}]
\label{lem:one_stage_guarantees}
The procedure {\OneStage} (Algorithm~\ref{alg:one_stage}) has the following performance guarantees.  For each true tone $(v_i,f_i)$, it ``succeeds'' in {\LocateSignal} (Algorithm~\ref{alg:locate_signal}) with probability at least $0.99$. More specifically, let $S \subseteq H$ denote the set of successful tones in {\LocateSignal} (Algorithm~\ref{alg:locate_signal}). There exists an injection $\pi : S \mapsto [k]$ such that
\begin{description}[labelindent = 1em]
    \item [Property I:]
    Each true tone $(v_i,f_i)$ is estimated well with probability $\Pr[i \in S] \geq 0.9$ and if so, those tones whose signal-to-noise ratio $\rho(i) \geq \mathcal{C}$ (see Definition~\ref{def:ratio_signal_noise}) has the estimation error
    \begin{align*}
        \| f_{i}' - f_{\pi(i)} \|_2 ~ \lesssim ~ C_* \frac{1}{\rho_{\pi(i)} \cdot T}.
    \end{align*}
    %%%% adding C_*

    \item [Property~II:]
    For all the successfully recovered tones $S$, the (partial) tone estimation error
    \begin{align*}
        \sum_{i\in S} \E_{\Sigma, b} \left[ \frac{1}{T^d} \cdot \int_{\tau \in [0,T]^d} \big| v_{i}' \cdot e^{2 \pi \i \cdot f_{i}'^\top \tau} - v_{\pi(i)} \cdot e^{2 \pi \i \cdot f_{\pi(i)}^\top \tau} \big|^2 \cdot \d \tau\right]
        ~ \lesssim ~ ({\cal C}^2 + d C_*^2) \cdot \N^2.
    \end{align*}

    \item [Property~III:]
    For each successfully recovered tone $i \in S$, if its signal-to-noise ratio $\rho(i) \geq \mathcal{C}$, the (single) tone estimation error
    \begin{align*}
        | v_{i}' - v_{\pi(i)} | ~ \lesssim ~ ({\cal C} + \sqrt{d} C_*) \cdot \N.
    \end{align*}
\end{description}
\end{lemma}

%\begin{proof}
The performance guarantees on the sample complexity, the duration, the success probability, and the running time are controlled by the counterpart performance guarantees of the subroutine {\LocateSignal}. For ease of presentation, here we omit the formal proofs of these performance guarantees.

\begin{claim}[Property I of Lemma~\ref{lem:one_stage_guarantees}]
\label{cla:one_stage_guarantees:1}
For each successfully recovered tone $i \in S$ with large enough signal-to-noise ratio, the frequency estimation error
\begin{align*}
    \| f_{i}' - f_{\pi(i)} \|_2 ~ \lesssim ~ C_* \frac{1}{\rho_{\pi(i)} \cdot T}.
\end{align*}
\end{claim}
\begin{proof}
This follows directly from Lemma~\ref{lem:locate_signal_stronger}.
\end{proof}

%%%%%%%%%%%%%%%%%%%%%%%%%%%%%%%%%%%%%%%%%%%%%%%%%%%%%%%%%%%%%%%%%%%%%%%%%%%%%%%%%%%%%%%%%%%%%%%%%%%%%%%%%%%%%%%%%%%%%%%%%%%%%%%%%%%%%%%%%%%%%%%%%%%%%%%%%%%%%%%%%%%%%%%%%%%%%%%%%%%%%%%%%%%%%%%%%%%%%%%%%%%%%%%%%%%%%%%%%%%%%%%%%%%%%%%%%%%%%%%%%%%%%%%%%%%%%%%%%%%%%%%%%%%%%%%%%%%%%%%%%%%%%%%%%%%%%%%%%%%%%%%%%%%%%%%%%%%%%%%%%%%%%%%%%%%%%%%%%%%%%%%%%%%%%%%%%%%%%%%%%%%%%%%%%%%%%%%%%%%%%%%%%%%%%%%%%%%%%%%%%%%%%%%%%%%%%%%%%%%%%%%%%%%%%%%%%%%%%%%%%%%%%%%%%%%%%%%%%%%%%%%%%%%%%%%%%%%%%%%%%%%%%%%%%%%%%%%%%%%%%%%%%%%%%%%%%%%%%%%%

\begin{claim}[Property II of Lemma~\ref{lem:one_stage_guarantees}]
\label{cla:one_stage_guarantees:2}
For all the successfully recovered tones $S$, the (partial) tone estimation error
\begin{align*}
    \sum_{i\in S} \E_{\Sigma, b} \left[ \frac{1}{T^d} \cdot \int_{\tau \in [0,T]^d} \big| v_{i}' \cdot e^{2 \pi \i \cdot f_{i}'^\top \tau} - v_{\pi(i)} \cdot e^{2 \pi \i \cdot f_{\pi(i)}^\top \tau} \big|^2 \cdot \d \tau\right]
    ~ \lesssim ~ ({\cal C}^2 + d C_*^2) \cdot \N^2.
\end{align*}
\end{claim}

\begin{proof}
Consider a specific true tone $(v_{\pi(i)}, f_{\pi(i)})$ that $i \in S \subseteq H$, for which neither $E_{\coll}(f_{\pi(i)})$ nor $E_{\off}(f_{\pi(i)})$ happens, and this tone ``succeeds'' in {\LocateSignal} (Algorithm~\ref{alg:locate_signal}). Assume w.l.o.g.\ that the tone frequency is hashed into the bin $j = \mathpzc{h}_{\Sigma, b}(f_{\pi(i)}) \in [B]^{d}$ (according to Definition~\ref{def:function_hash_multi}). For simplicity, we adopt the following notations in this proof:
\begin{itemize}
    \item $j' = \mathpzc{h}_{\Sigma, b}(f_{i}') \in [B]^{d}$ is the bin where the estimation frequency $f_{i}'$ hashed into;
    
    \item $v_{i}' = \hat{u}'_{j'}   \cdot e^{-2 \pi \i \cdot a^{\top} f_{i}'} \in \C$ is the estimation magnitude returned by the procedure {\EstimateSignal} (Algorithm~\ref{alg:estimate_signal});
    
    \item $v_{i}'' = \hat{u}'_{j'}  \cdot e^{-2 \pi \i \cdot a^{\top} f_{i}} \in \C$; and
    
    \item $\mu_{\pi(i)}^2 = \E_{\Sigma, b, a} [ |\hat{u}_{j}' \cdot e^{-2 \pi \i \cdot a^{\top} f_{\pi(i)}} - v_{\pi(i)}|^{2} ] \geq 0$ according to Definition~\ref{def:ratio_signal_noise}.
\end{itemize}

We discuss two cases for the signal-to-noise ratio $\rho_{\pi(i)} = |v_{i}| / \mu_{\pi(i)}$.

{\bf Case~(i):} when the signal-to-noise ratio $\rho_{\pi(i)} = |v_{i}| / \mu_{\pi(i)} \geq {\cal C}$, i.e.\ when the premise for Lemma~\ref{lem:locate_inner_voting} holds.

According to Markov inequality, the equation below holds with probability at least $1 - {\cal C}^{-2} \geq 0.9999$ (given that $\mathcal{C} \geq 120$; see Definition~\ref{def:locate_inner_setup}). In what follows, We assume that this equation holds.
\begin{align}
    \label{eq:cla:one_stage_guarantees:2:1}
    \left|\hat{u}'_{j} \cdot e^{-2 \pi \i \cdot a^{\top} f_{\pi(i)}} - v_{\pi(i)}\right|^2
    ~ \leq ~ {\cal C}^2 \cdot \mu_{i}^2.
\end{align}

 The procedure {\LocateSignal} and the follow-up procedures have the desired performance guarantees. In particular, we derive a ``good'' frequency estimation $f_{i}' \in [-F, F]^{d}$ from the procedure {\LocateSignal}. Also, it follows from Lemma~\ref{lem:locate_signal_stronger} that
\begin{align}
    \label{eq:cla:one_stage_guarantees:2:2}
    \| f_{i}' - f_{\pi(i)} \|_2 ~ \lesssim ~ C_* \frac{1}{\rho_{\pi(i)} \cdot T}.
\end{align}

Since we choose a large enough duration $T = \Omega(\frac{d^{4.5}\log(dk/\delta)\log d}{\eta})$ finally according to Theorem~\ref{thm:intro_tone} and $\rho_{\pi(i)}\geq {\cal C}\geq 120$, the $\ell_{2}$-norm frequency error $\| f_{i}' - f_{\pi(i)} \|_2$ is sufficiently small and thus Lemma~\ref{lem:close_fre_same_bin} is applicable. That is, the tone frequency $f_{\pi(i)}$ and the estimation frequency $f_{i}'$ are hashed into the same bin $j = j' \in [B]^{d}$ (see Lemma~\ref{lem:close_fre_same_bin}).
Combining the above arguments together, we have
\begin{eqnarray}
    & & \frac{1}{T^d} \cdot \int_{\tau \in [0, T]^d} \big| v_{i}' \cdot e^{2 \pi \i \cdot f_{i}'^\top \tau} - v_{\pi(i)} \cdot e^{2 \pi \i \cdot f_{\pi(i)}^\top \tau} \big|^2 \cdot \d \tau
    \notag \\
    & = & \frac{1}{T^d} \cdot \int_{\tau \in [0, T]^d} \big| \hat{u}'_{j'} \cdot e^{-2 \pi \i \cdot a^{\top} f_{i}'} \cdot e^{2 \pi \i \cdot f_{i}'^\top \tau} - v_{\pi(i)} \cdot e^{2 \pi \i \cdot f_{\pi(i)}^\top \tau} \big|^2 \cdot \d \tau
    \notag \\
    & = & \frac{1}{T^d} \cdot \int_{\tau \in [0, T]^d} \big| \hat{u}'_{j}  \cdot e^{-2 \pi \i \cdot a^{\top} f_{i}'} \cdot e^{2 \pi \i \cdot f_{i}'^\top \tau} - v_{\pi(i)} \cdot e^{2 \pi \i \cdot f_{\pi(i)}^\top \tau} \big|^2 \cdot \d \tau
    \notag \\
    & \leq & \underbrace{\frac{1}{T^d} \cdot \int_{\tau \in [0, T]^d} \big| \hat{u}_{j}' \cdot e^{-2 \pi \i \cdot a^{\top} f_{i}'} \cdot e^{2 \pi \i \cdot f_{i}'^\top \tau} - v_{\pi(i)} \cdot e^{2 \pi \i \cdot f_{i}'^\top \tau} \big|^2}_{A_{1}(\tau)} \cdot \d \tau
    \notag \\
    & & ~ + ~ \frac{1}{T^d} \cdot \int_{\tau \in [0, T]^d} \big| v_{\pi(i)} \cdot e^{2 \pi \i \cdot f_{i}'^\top \tau} - v_{\pi(i)} \cdot e^{2 \pi \i \cdot f_{\pi(i)}^\top \tau} \big|^2 \cdot \d \tau
    \notag \\
    & = & {A_{1}(\tau)} ~ + ~ 2 \cdot |v_{\pi(i)}|^2 \cdot (1 - \sinc_{T}(f_{i}' - f_{\pi(i)}))
    \notag \\
     & \leq & A_1(\tau)
    ~ + ~ 2 \cdot |v_{\pi(i)}|^2 \cdot (\frac{\pi^2}{6} \cdot T^2 \cdot \| f_{i}' - f_{\pi(i)} \|_{2}^2)
    \notag \\
    & \lesssim & A_1(\tau)
    ~ + ~ C_*^2 \cdot |v_{\pi(i)}|^2 / \rho_{\pi(i)}^2
    \label{eq:cla:one_stage_guarantees:2:3_1}
\end{eqnarray} %%%add C_*
where the first step is by the definition of $v_{i}'$ (see Algorithm~\ref{alg:estimate_signal}); the second step follows because the tone frequency $f_{\pi(i)}$ and the estimation frequency $f_{i}'$ are hashed into the same bin $j = j' \in [B]^{d}$; the third step applies the triangle inequality;   the forth step applies Property~II of Lemma~\ref{lem:tone_distance} to the second summand; the fifth step applies Part~(e) of Fact~\ref{fac:sinc_function_multi}; the last step applies Equation~\eqref{eq:cla:one_stage_guarantees:2:2}.
%and the ratio $\mathcal{C} \geq 120$ (see Definition~\ref{def:locate_inner_setup}).

Then we consider $A_1(\tau)$. We have 
\begin{eqnarray}
    & & A_1(\tau) \notag \\
    & = & \frac{1}{T^d} \cdot \int_{\tau \in [0, T]^d} \big| \hat{u}'_{j} \cdot e^{-2 \pi \i \cdot a^{\top} f_{i}'} \cdot e^{2 \pi \i \cdot f_{i}'^\top \tau} - v_{\pi(i)} \cdot e^{2 \pi \i \cdot f_{i}'^\top \tau} \big|^2 \cdot \d \tau
    \notag \\
    &\leq & \frac{1}{T^d}\cdot \int_{\tau \in [0, T]^d} \big|\hat{u}'_{j}  \cdot e^{-2 \pi \i \cdot a^{\top} f_{i}'} \cdot e^{2 \pi \i \cdot f_{i}'^\top \tau}-v_{\pi(i)} \cdot e^{2 \pi \i \cdot f_{i}'^\top \tau}\cdot e^{-2\pi\i\cdot a^{\top}(f_i'-f_{\pi(i)})} \big|^2\cdot\d\tau \notag\\
    & & ~ + ~ \frac{1}{T^d}\int_{\tau\in [0, T]^d}\big| v_{\pi(i)} \cdot e^{2 \pi \i \cdot f_{i}'^\top \tau}\cdot e^{-2\pi\i\cdot a^{\top}(f_i'-f_{\pi(i)})}-  v_{\pi(i)} \cdot e^{2 \pi \i \cdot f_{i}'^\top \tau} \big|^2 \cdot\d \tau \notag\\
    &=& \left|\hat{u}_{j}  \cdot e^{-2 \pi \i \cdot a^{\top} f_{\pi(i)}} - v_{\pi(i)}\right|^2\cdot \frac{1}{T^d}\int_{\tau\in [0, T]^d}\big|e^{2\pi\i\cdot f_i'^{\top}\tau}\big|^2\cdot\d\tau \notag \\ 
    & & ~+~ \frac{1}{T^d}\int_{\tau\in [0, T]^d}\big| v_{\pi(i)} \cdot e^{2 \pi \i \cdot f_{i}'^\top \tau}\cdot e^{-2\pi\i\cdot a^{\top}(f_i'-f_{\pi(i)})}-  v_{\pi(i)} \cdot e^{2 \pi \i \cdot f_{i}'^\top \tau} \big|^2 \cdot\d \tau \notag\\
    & \leq & {\cal C}^2\cdot\mu_{i}^2 + \frac{1}{T^d}\int_{\tau\in [0, T]^d}\big| v_{\pi(i)} \cdot e^{2 \pi \i \cdot f_{i}'^\top \tau}\cdot e^{-2\pi\i\cdot a^{\top}(f_i'-f_{\pi(i)})}-  v_{\pi(i)} \cdot e^{2 \pi \i \cdot f_{i}'^\top \tau} \big|^2 \cdot\d \tau \notag\\
    & = & {\cal C}^2\cdot\mu_{i}^2+ |v_{\pi(i)}|^2\cdot|e^{-2\pi\i\cdot a^{\top}(f_i'-f_{\pi(i)})}-1|^2\notag \\
    & \lesssim & {\cal C}^2\cdot\mu_{i}^2+ |v_{\pi(i)}|^2\cdot|a^{\top}(f_i'-f_{\pi(i)})|^2 \notag\\
    & \lesssim & {\cal C}^2\cdot\mu_{i}^2+ |v_{\pi(i)}|^2\cdot \| a \|_2 \cdot \| f_i'-f_{\pi(i)} \|_2 \notag \\
    & \lesssim & {\cal C}^2\cdot\mu_{i}^2 + d C_*^2 |v_{\pi(i)}|^2/\rho_{\pi(i)}^2
    \label{eq:cla:one_stage_guarantees:2:3_2}
\end{eqnarray}
where the second step is by triangle inequality; the fourth step is by the integral equals to 1 and follows from Equation~\eqref{eq:cla:one_stage_guarantees:2:1}; the sixth step is because $|e^{2\pi\i x}-1|\leq |x|$; the last step is because $\|a\|_{2}\leq \sqrt{d} T$ (see Algorithm~\ref{alg:sample_time_point}) and $\|f_i'-f_{\pi(i)}\|_2\lesssim C_*/(T\rho_{\pi(i)})$, which implies $|a^{\top}(f_i'-f_{\pi_{\pi(i)}})|\lesssim \sqrt{d} C_* / \rho_{\pi(i)}$.

Combine Equation~\eqref{eq:cla:one_stage_guarantees:2:3_1} and Equation~\eqref{eq:cla:one_stage_guarantees:2:3_2} together, we prove that 
\begin{align*}
    \frac{1}{T^d} \cdot \int_{\tau \in [0, T]^d} \big| v_{i}' \cdot e^{2 \pi \i \cdot f_{i}'^\top \tau} - v_{\pi(i)} \cdot e^{2 \pi \i \cdot f_{\pi(i)}^\top \tau} \big|^2 \cdot \d \tau
    &\lesssim {\cal C}^2\cdot\mu_i^2+ d C_*^2 |v_{\pi(i)}|^2/\rho_{\pi(i)}^2\\
    &\lesssim ({\cal C}^2 + d C_*^2) \cdot\mu_i^2
\end{align*}
The last step follows because the signal-to-noise ratio $\rho_{\pi(i)} = |v_{\pi(i)}| / \mu_{\pi(i)}$ (see Definition~\ref{def:ratio_signal_noise}).
This accomplishes Case~(i).

{\bf Case~(ii):} when the signal-to-noise ratio $\rho_{\pi(i)} = |v_{\pi(i)}| / \mu_{\pi(i)} \leq {\cal C}$, i.e.\ when the premise for Lemma~\ref{lem:locate_inner_voting} does not hold. Under this case, we can use $(0,f_i')$ to recover the true tone $(v_i,f_i)$, as
\begin{align*}
     \frac{1}{T^d} \cdot \int_{\tau \in [0, T]^d} \big| v_{i}' \cdot e^{2 \pi \i \cdot f_{i}'^\top \tau} - v_{\pi(i)} \cdot e^{2 \pi \i \cdot f_{\pi(i)}^\top \tau} \big|^2 \cdot \d \tau 
     = |v_{\pi(i)}|^2 
    \leq {\cal C}^2\cdot\mu_i^2.
\end{align*}

Recall that we add a supplementary list $\mathpzc{List}_{sup}$ in the Line~\ref{lin:supplementary_list} of {\OneStage}, and there are enough candidates with zero magnitude and minimum separation in the list. We can let $S$ include some tones $(0,\xi)$ from $\mathpzc{List}_{sup}$ when needed.

This completes the proof.
\end{proof}

%%%%%%%%%%%%%%%%%%%%%%%%%%%%%%%%%%%%%%%%%%%%%%%%%%%%%%%%%%%%%%%%%%%%%%%%%%%%%%%%%%%%%%%%%%%%%%%%%%%%%%%%%%%%%%%%%%%%%%%%%%%%%%%%%%%%%%%%%%%%%%%%%%%%%%%%%%%%%%%%%%%%%%%%%%%%%%%%%%%%%%%%%%%%%%%%%%%%%%%%%%%%%%%%%%%%%%%%%%%%%%%%%%%%%%%%%%%%%%%%%%%%%%%%%%%%%%%%%%%%%%%%%%%%%%%%%%%%%%%%%%%%%%%%%%%%%%%%%%%%%%%%%%%%%%%%%%%%%%%%%%%%%%%%%%%%%%%%%%%%%%%%%%%%%%%%%%%%%%%%%%%%%%%%%%%%%%%%%%%%%%%%%%%%%%%%%%%%%%%%%%%%%%%%%%%%%%%%%%%%%%%%%%%%%%%%%%%%%%%%%%%%%%%%%%%%%%%%%%%%%%%%%%%%%%%%%%%%%%%%%%%%%%%%%%%%%%%%%%%%%%%%%%%%%%%%%%%%%%%%

\begin{claim}[Property III of Lemma~\ref{lem:one_stage_guarantees}]
For each successfully recovered tone $i \in S$, the (single) tone estimation error
\begin{align*}
    | v_{i}' - v_{\pi(i)} | ~ \lesssim ~ ({\cal C} + \sqrt{d} C_*) \cdot \N.
\end{align*}
\end{claim}

\begin{proof}
This can be directly inferred from Claim~\ref{cla:one_stage_guarantees:2}.
\end{proof}

%%%%%%%%%%%%%%%%%%%%%%%%%%%%%%%%%%%%%%%%%%%%%%%%%%%%%%%%%%%%%%%%%%%%%%%%%%%%%%%%%%%%%%%%%%%%%%%%%%%%%%%%%%%%%%%%%%%%%%%%%%%%%%%%%%%%%%%%%%%%%%%%%%%%%%%%%%%%%%%%%%%%%%%%%%%%%%%%%%%%%%%%%%%%%%%%%%%%%%%%%%%%%%%%%%%%%%%%%%%%%%%%%%%%%%%%%%%%%%%%%%%%%%%%%%%%%%%%%%%%%%%%%%%%%%%%%%%%%%%%%%%%%%%%%%%%%%%%%%%%%%%%%%%%%%%%%%%%%%%%%%%%%%%%%%%%%%%%%%%%%%%%%%%%%%%%%%%%%%%%%%%%%%%%%%%%%%%%%%%%%%%%%%%%%%%%%%%%%%%%%%%%%%%%%%%%%%%%%%%%%%%%%%%%%%%%%%%%%%%%%%%%%%

\subsection{{\MultiStage}}
\label{sec:multi_stage}

The goal of this section is to prove Lemma~\ref{lem:multi_stage}.
\begin{lemma}[{\MultiStage}]
\label{lem:multi_stage}
The procedure {\MultiStage} (Algorithm~\ref{alg:multi_stage}) satisfies the following:
\begin{itemize}
    \item The sample complexity is ${\cal R}_{\mathrm{merge}}$ times the sample complexity of {\OneStage} (Algorithm~\ref{alg:one_stage}).
    
    \item The running time is ${\cal R}_{\mathrm{merge}}$ times the running time of {\OneStage} (Algorithm~\ref{alg:one_stage}).
    
	%\item The output $\mathpzc{List}^*$ contains at most $O(\mathcal{B}) = 2^{O(d \cdot \log d)} \cdot k \cdot \log k$ many candidate tones.
\end{itemize}
\end{lemma}

\begin{proof}
All these properties can be easily inferred from Lemma~\ref{lem:one_stage_sample_time}.
\end{proof}

% \begin{algorithm}
% \caption{{\MultiStage}, See Lemma~\ref{lem:multi_stage}}
% \label{alg:multi_stage}
% \SetKwProg{myproc}{Procedure}{}{}
% \myproc{$\mathrm{MultiStage}(x, D, \mathcal{C}, T, {\cal R}_{{\rm merge}})$}
% {
%  Let $\mathpzc{List}^* \gets \emptyset$\;

% \For{$r = 1, 2, \cdots, {\cal R}_{{\rm merge}}$}
%     {
%      Sample $\Sigma \in \R^{d \times d}$ and $b \in \R^{d}$ according to Definition~\ref{def:HashToBins_multi_parameters}\;
    
%      $\mathpzc{List}_{\mathrm{new}} \gets {\OneStage}(x,\Sigma, b, D, \mathcal{C}, T)$\;
%     \tcp*{Algorithm~\ref{alg:one_stage}}
    
%      $\mathpzc{List}^* \gets \mathpzc{List}^* \cup \mathpzc{List}_{\mathrm{new}}$\;
% }
%  \Return the tones $\mathpzc{List}^*$
% }

% \end{algorithm}

\begin{algorithm}
\caption{{\MultiStage}}
\label{alg:multi_stage}
\begin{algorithmic}[1]
\Procedure{MultiStage}{$x, D, \mathcal{C}, T, {\cal R}_{{\rm merge}}$} \Comment{Lemma~\ref{lem:multi_stage}}

\State Let $\mathpzc{List}^* \gets \emptyset$.

\For{$r = 1, 2, \cdots, {\cal R}_{{\rm merge}}$}
    
    \State Sample $\Sigma \in \R^{d \times d}$ and $b \in \R^{d}$ according to Definition~\ref{def:HashToBins_multi_parameters}.
    
    \State $\mathpzc{List}_{\mathrm{new}} \gets {\OneStage}(x,\Sigma, b, D, \mathcal{C}, T)$.
    \Comment{Algorithm~\ref{alg:one_stage}}
    
    \State $\mathpzc{List}^* \gets \mathpzc{List}^* \cup \mathpzc{List}_{\mathrm{new}}$.
    
\EndFor

\State \Return the tones $\mathpzc{List}^*$.
\EndProcedure
\end{algorithmic}
\end{algorithm}

\begin{lemma}[Guarantees of \textsc{MultiStage}]
\label{lem:multi_stage_guarantees}
The procedure {\MultiStage} (Algorithm~\ref{alg:multi_stage}) repeat the procedure {\OneStage} (Algorithm~\ref{alg:one_stage}) for $\mathcal{R}_{\mathrm{merge}} = \Theta(d \cdot \log d \cdot \log k)$ times, and returns a set $\mathpzc{List} = \{(v_{i}', f_{i}')\}_{i \in [m]}$ of $m = |\mathpzc{List}| = 2^{O(d \cdot \log d)} \cdot k \cdot \log k \in \mathbb{N}_{\geq 1}$ many candidate tones. With probability at least $1-1/\poly(k)$, there are at least $\mathcal{R}_{\mathrm{merge}}$ different disjoint subsets $\{S_r\}$ of $\mathpzc{List}$, where for each $r$ we have that $S_r\subset \mathpzc{List}$ and $|S_r|\leq k$, and for $r\neq r'$ we have that $ S_r\cap S_{r'}=\emptyset$. For each $S_r$, there is a injective projection $\pi_r: S_r\rightarrow [k]$ and has the following properties:
\begin{description}

\item [Property I:]
    For each true tone $(v_i,f_i)$, $\Pr[i\in S_r]\geq 0.9$ and if the signal-to-noise ratio is large enough, the frequency estimation error
    \begin{align*}
        \| f_{i}' - f_{\pi_r(i)} \|_2 ~ \lesssim ~ C_* \frac{1}{\rho_{\pi_r(i)} \cdot T}.
    \end{align*}

    \item [Property~II:]
    For all the successfully recovered tones $S_r$, the (partial) tone estimation error
    \begin{align*}
        \sum_{i\in S_r} \E_{\Sigma, b} \left[ \frac{1}{T^d} \cdot \int_{\tau \in [0,T]^d} \big| v_{i}' \cdot e^{2 \pi \i \cdot f_{i}'^\top \tau} - v_{\pi_r(i)} \cdot e^{2 \pi \i \cdot f_{\pi_r(i)}^\top \tau} \big|^2 \cdot \d \tau\right]
        ~ \lesssim ~ ({\cal C}^2 + d C_*^2) \cdot \N^2.
    \end{align*}

    \item [Property~III:]
    For each successfully recovered tone $i \in S_r$, if its signal-to-noise ratio is large enough, then we can bound the (single) tone estimation error
    \begin{align*}
        | v_{i}' - v_{\pi_r(i)} | ~ \lesssim ~ ({\cal C} + \sqrt{d} C_*) \cdot \N.
    \end{align*}

\end{description}
\end{lemma}
\begin{proof}
This lemma can be proved directly by Lemma~\ref{lem:one_stage_guarantees}.
\end{proof}

\subsection{{\MergedStage}: algorithm and running time}
\label{sec:merged_stage_time}

The goal of this section is to prove Lemma~\ref{lem:merged_stage_time}. The concerning Algorithm~\ref{alg:merged_stage} ({\MergedStage}) is demonstrated in Figure~\ref{fig:merged_stage}.

\begin{lemma}[{\MergedStage}, Input size and Running time]
\label{lem:merged_stage_time}
The procedure {\MergedStage} (Algorithm~\ref{alg:merged_stage}) has the following properties:
\begin{itemize}
    \item The input $\mathpzc{List} = \{(v_{j}, f_{j})\}_{j \in [m]}$ is a multi-set of $m = |\mathpzc{List}| = 2^{O(d \cdot \log d)} \cdot k \cdot {\cal R}_{\rm merge}$ many candidate tones, where ${\cal R}_{{\rm merge}} = \Theta(d \cdot \log d \cdot \log k)$ is sufficiently large.
    
    \item The running time is $2^{O(d \cdot \log d)} \cdot k \cdot {\cal R}_{\rm merge} \cdot \log^{d}(k \cdot {\cal R}_{\rm merge}) = 2^{O(d \cdot \log d)} \cdot k \cdot \log^{O(d)} k$.
\end{itemize}
\end{lemma}

\begin{proof}
The first property about the input $\mathpzc{List}$ is guaranteed by Lemma~\ref{lem:multi_stage}.

The second property is proved by using a well-known data-structure.
We use a textbook $d$-dimensional range tree data-structure (see section 5 in \cite{vsbo00}).

\begin{theorem}[Theorem 5.11 in \cite{vsbo00}]
\label{thm:range_tree}
Let $P$ be a set of $n$ points in $d$-dimensional space, with $d \geq 2$. A layered range tree for $P$ uses $O(n \log^{d-1} n )$ storage and it can be constructed in $O(n \log^{d-1} n)$ time. With this range tree one can report the points in $P$ that lie in a rectangular query range in $O(\log^{d-1} n + q)$ time, where $q$ is the number of reported points.
\end{theorem}

The above theorem works for $\ell_{\infty}$-norm. %We first rotate our axis, this operation will reduce $\ell_{1}$-norm query problem to a $\ell_{\infty}$-norm query problem. 
By choosing 
\begin{align*}
    n = m = | \mathpzc{List} |
\end{align*}
we complete the proof of running time.
\end{proof}

\begin{figure}[htbp]
    \centering
    \begin{tabular}[b]{c}
    \includegraphics[width=.8\textwidth]{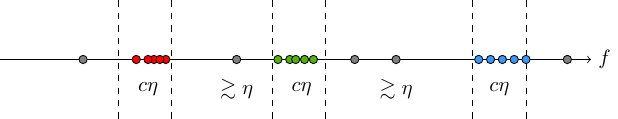}\\
    \small{(a)~{\MergedStage} in one dimension}
    \end{tabular} \qquad
    \begin{tabular}[b]{c}
    \includegraphics[width=.8\textwidth]{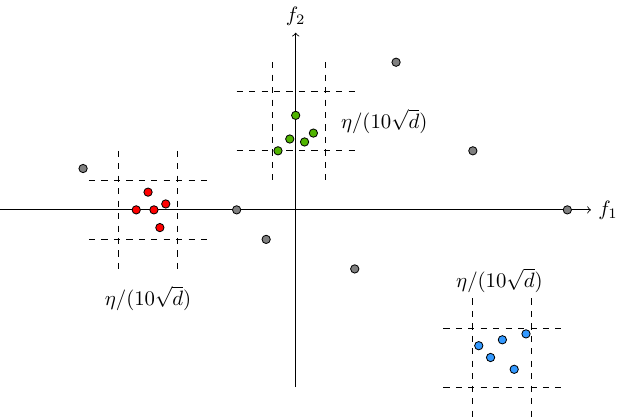}\\
    \small {(b)~{\MergedStage} in two dimensions}
    \end{tabular}
    \caption{
    Demonstration for Algorithm~\ref{alg:merged_stage} in one dimension ($d = 1$) and two dimensions ($d = 2$).}
    \label{fig:merged_stage}
\end{figure}

\begin{algorithm}
\caption{{\MergedStage}, Lemmas~\ref{lem:merged_stage_time} and \ref{lem:merged_stage_guarantees}}
\label{alg:merged_stage}
\begin{algorithmic}[1]
\Procedure{MergedStage}{$\mathpzc{List}, \mathcal{R}_{\mathrm{merge}}$}

\State Denote $\mathpzc{List} = \{(v_{j}, f_{j})\}_{j \in [m]}$ for $m = |\mathpzc{List}|$.

% \State \Comment{$c_{1}$ is small constant like $0.05$ and and $c_{2}$ is a constant like $0.9$.}

\State Build a $d$-dimensional segment tree $\textsc{Tree}$ on the frequencies $\{f_{j}\}_{j \in [m]} \subseteq [-F, F]^{d}$.

\State All these frequencies $\{f_{j}\}_{j \in [m]}$ are {\em unmarked}.

\State $\mathpzc{List}^* \leftarrow \emptyset$.

\While {$\textsc{Tree}$ has at least one {\em unmarked} frequency}
    \State Choose an arbitrary {\em unmarked} frequency $\xi_{i}$ from $\textsc{Tree}$.
    
     \If {$\textsc{Tree}.\textsc{count}(\mathbf{HC}(\xi_{i}, \eta / d^3)) ~ \geq ~ 8 / 10 \cdot \mathcal{R}_{\mathrm{merge}}$}\label{lin:judge_density}\Comment{Theorem~\ref{thm:range_tree}}
        \State $f^* \gets \xi_i$.\label{lin:take_median_frequency}
        
        \State $v^* \gets \median \{v_{j}: j \in [m] \mbox{ and } f_{j} \in \mathbf{HC}(\xi_{i}, \eta/ d^3) \}$.\label{lin:take_median_mag}
        
        \State $\mathpzc{List}^* \gets \mathpzc{List}^* \cup (v^*, f^*)$.
        
        \State Delete $\{f_{j}: j \in [m] \mbox{ and } f_{j} \in \mathbf{HC}(\xi_{i}, \eta/(10\sqrt{d})) \}$ from $\textsc{Tree}$.
        
        \State Delete $\{(v_{i}, f_{j}): j \in [m] \mbox{ and } f_{j} \in \mathbf{HC}(\xi_{i}, \eta /(10\sqrt{d})) \}$ from $\mathpzc{List}$.
        
    \Else 
        \State {\em Mark} the chosen frequency $\xi_{i}$ in  $\textsc{Tree}$.
    \EndIf
\EndWhile

\State \Return the tones $\mathpzc{List}^*$.

\EndProcedure
\end{algorithmic}
\end{algorithm}

\subsection{{\MergedStage}: performance guarantees}
\label{sec:merged_stage_guarantees}

The goal of this section is to prove Lemma~\ref{lem:merged_stage_guarantees}.

\begin{claim}[Approximate formula for tone-wise error]
\label{cla:property_median_approximation}
The following holds for any pair of tones $(v, f) \in \C \times \R^{d}$ and $(v^*, f^*) \in \C \times \R^{d}$:
\begin{eqnarray*}
    \mathrm{err}
    & := & \frac{1}{T^d} \cdot \int_{\tau \in [0, T]^d} \big| v \cdot e^{2 \pi \i \cdot f^{\top} \tau} - v^* \cdot e^{2 \pi \i \cdot f^{*\top} \tau} \big|^2 \cdot \d \tau \\
    & \eqsim & | v - v^* |^{2} ~ + ~ | v^* |^{2} \cdot \big(1 - \sinc_{T}(f - f^*)\big) \\
    & \eqsim & | v - v^* |^{2} ~ + ~ | v^* |^{2} \cdot \min \{ 1, ~ T^{2} \cdot \|f - f^*\|_{2}^{2} \} .
\end{eqnarray*}
\end{claim}

\begin{proof}
We first prove the second part of the claim, which is equivalent to
\begin{align}
    1 - \sinc_{T}(f - f^*) ~ \eqsim ~ \min \{ 1, ~ T^{2} \cdot \|f - f^*\|_{2}^{2} \}.
    \label{eq:property_median_approximation:1}
\end{align}
Indeed, when $T^{2} \cdot \|f - f^*\|_{2}^{2} \geq (\frac{2.05}{\pi})^{2} \eqsim 1$, we know from Part~(d) of Fact~\ref{fac:sinc_function_multi} that
\[
    1 - \sinc_{T}(f - f^*) ~ = ~ 1 \pm \frac{1}{2} ~ \eqsim ~ 1.
\]
And when $T^{2} \cdot \|f - f^*\|_{2}^{2} < (\frac{2.05}{\pi})^{2} \eqsim 1$, we know from Part~(c) of Fact~\ref{fac:sinc_function_multi} that
\begin{eqnarray*}
    1 - \sinc_{T}(f - f^*)
    & \geq & 1 - \exp\big(-(\pi^{2} / 6) \cdot T^{2} \cdot \|f - f^*\|_{2}^{2}\big) \\
    & \gtrsim & T^{2} \cdot \|f - f^*\|_{2}^{2},
\end{eqnarray*}
and that
\begin{eqnarray*}
    1 - \sinc_{T}(f - f^*)
    & \leq & 1 - \exp\big(-(\pi^{2} / 5) \cdot T^{2} \cdot \|f - f^*\|_{2}^{2}\big) \\
    & \lesssim & T^{2} \cdot \|f - f^*\|_{2}^{2},
\end{eqnarray*}
Combining the above arguments together implies Equation~\eqref{eq:property_median_approximation:1}.

In what follows, we prove the first part of the claim that
\begin{eqnarray}
    \mathrm{err}
    & \eqsim & | v - v^* |^{2} ~ + ~ | v^* |^{2} \cdot \big(1 - \sinc_{T}(f - f^*)\big).
    \label{eq:property_median_approximation:2}
\end{eqnarray}
We know Property~II of Lemma~\ref{lem:tone_distance} that
\begin{eqnarray*}
    \mathrm{err}
    & = & | v |^{2} + | v^* |^{2} ~ - ~ \big( v \cdot \bar{v^*} + \bar{v} \cdot v^* \big) \cdot \sinc_{T}(f - f^*).
\end{eqnarray*}
For brevity, we denote $w_{1} \cdot e^{\i \cdot \theta} = v / v^*$ for some norm $w_{1} \geq 0$ and some phase $\theta \in [0, 2 \pi)$, and denote $w_{2} = \sinc_{T}(f - f^*) \in [-\frac{1}{4}, 1]$ (see Part~(e) of Fact~\ref{fac:sinc_function_multi}). We notice that the formula
\[
    | v - v^* |^{2} ~ + ~ | v^* |^{2} \cdot \big(1 - \sinc_{T}(f - f^*)\big)
\]
is non-negative. As a consequence, to verify Equation~\eqref{eq:property_median_approximation:2}, it suffices to show that the following function $L(w_{1}, w_{2}, \theta) \eqsim 1$, for any $w_{1} \geq 0$, any $w_{2} \in [-\frac{1}{4}, 1]$ and any $\theta \in [0, 2 \pi)$:
\begin{eqnarray*}
    L(w_{1}, w_{2}, \theta)
    & := & \frac{\mathrm{err}}{\RHS \mbox{ of } \eqref{eq:property_median_approximation:2}} \\
    & = & \frac{| v^* \cdot w_{1} \cdot e^{\i \cdot \theta} |^{2} + |v^*|^{2} - 2 \cdot |v^* \cdot w_{1} \cdot e^{\i \cdot \theta}| \cdot |v^*| \cdot \cos(\theta) \cdot w_{2}}{| v^* \cdot w_{1} \cdot e^{\i \cdot \theta} - v^* |^{2} + |v^*|^{2} \cdot (1 - w_{2})} \\
    & = & \frac{w_{1}^{2} + 1 - 2 \cdot w_{1} \cdot w_{2} \cdot \cos(\theta)}{| w_{1} \cdot e^{\i \cdot \theta} - 1 |^{2} + 1 - w_{2}} \\
    & = & \frac{w_{1}^{2} + 1 - 2 \cdot w_{1} \cdot w_{2} \cdot \cos(\theta)}{w_{1}^{2} - w_{2} + 2 - 2 \cdot w_{1} \cos(\theta)},
\end{eqnarray*}
where the second step is by the definition of $w_{1}$, $w_{2}$ and $\theta$; the third step divides both the numerator and the denominator by $|v_{i}'|^{2}$; and the last step can be seen via elementary calculation.

Let us investigate the partial derivative $\frac{\partial L}{\partial \theta}$ in $\theta \in [0, 2 \pi)$:
\begin{eqnarray*}
    \frac{\partial L}{\partial \theta}
    & = & \frac{2 \cdot w_{1} \cdot w_{2} \cdot \sin(\theta)}{w_{1}^{2} - w_{2} + 2 - 2 \cdot w_{1} \cos(\theta)}
    - \frac{w_{1}^{2} + 1 - 2 \cdot w_{1} \cdot w_{2} \cdot \cos(\theta)}{\big(w_{1}^{2} - w_{2} + 2 - 2 \cdot w_{1} \cos(\theta)\big)^{2}} \cdot 2 \cdot w_{1} \cdot \sin(\theta) \\
    & = & -\sin(\theta) \cdot \underbrace{2 \cdot w_{1} \cdot \frac{w_{1}^{2} \cdot (1 - w_{2}) + (1 - w_{2})^{2}}{\big(w_{1}^{2} - w_{2} + 2 - 2 \cdot w_{1} \cos(\theta)\big)^{2}}}_{A_{4}},
\end{eqnarray*}
where the second step can be seen via elementary calculation.

Because $w_{1} \geq 0$ and $w_{2} \in [-\frac{1}{4}, 1]$, we must have $A_{4} \geq 0$. As a result, for any fixed $w_{1}$ and $w_{2}$, the function $L(w_{1}, w_{2}, \theta)$ is non-increasing when $\theta \in [0, \pi]$, and is non-decreasing when $\theta \in [\pi, 2 \pi)$. The functions $L_{\min}(w_{1}, w_{2}) := \min_{\theta \in [0, 2 \pi)} L(w_{1}, w_{2}, \theta)$ and $L_{\max}(w_{1}, w_{2}) := \max_{\theta \in [0, 2 \pi)} L(w_{1}, w_{2}, \theta)$ for any $w_{1} \geq 0$ and any $w_{2} \in [-\frac{1}{4}, 1]$ are given by
\begin{eqnarray*}
    L_{\min}(w_{1}, w_{2})
    & = & L(w_{1}, w_{2}, \pi)
    ~ = ~ \frac{A_{5}(w_{1}, w_{2})}{A_{6}(w_{1}, w_{2})}, \\
    A_{5}(w_{1}, w_{2})
    & := & w_{1}^{2} + 1 + 2 \cdot w_{1} \cdot w_{2}, \\
    A_{6}(w_{1}, w_{2})
    & := & w_{1}^{2} + 2 \cdot w_{1} - w_{2} + 2,
\end{eqnarray*}
and
\begin{eqnarray*}
    L_{\max}(w_{1}, w_{2})
    & = & L(w_{1}, w_{2}, 0)
    ~ = ~ \frac{A_{7}(w_{1}, w_{2})}{A_{8}(w_{1}, w_{2})}, \\
    A_{7}(w_{1}, w_{2})
    & := & w_{1}^{2} + 1 - 2 \cdot w_{1} \cdot w_{2}, \\
    A_{8}(w_{1}, w_{2})
    & := & w_{1}^{2} - 2 \cdot w_{1} - w_{2} + 2.
\end{eqnarray*}

We now justify the lower-bound part of Equation~\eqref{eq:property_median_approximation:2} by exploring the function $L_{\min}(w_{1}, w_{2})$. For any fixed $w_{1} \geq 1$, the numerator $A_{5}(w_{1}, w_{2})$ is a non-decreasing function in $w_{2} \in [-\frac{1}{4}, 1]$, while the denominator $A_{6}(w_{1}, w_{2})$ is a non-increasing {\em non-negative} function in $w_{2} \in [-\frac{1}{4}, 1]$. Given these, we can infer the lower-bound part of Equation~\eqref{eq:property_median_approximation:2} as follows:
\begin{eqnarray*}
    L(w_{1}, w_{2}, \theta)
    & \geq & \min_{w_{1} \in [0, 1]}
    \min_{w_{2} \in [-\frac{1}{4}, 1]}
    L_{\min}(w_{1}, w_{2}) \\
    & = & \min_{w_{1} \in [0, 1]}
    L_{\min}(w_{1}, -1 / 4) \\
    & = & \min_{w_{1} \in [0, 1]}
    \frac{w_{1}^{2} + 1 - (1 / 2) \cdot w_{1}}{w_{1}^{2} + 2 \cdot w_{1} - (1 / 4) + 2} \\
    & \approx & 0.3107,
\end{eqnarray*}
where the last step can be seen via numeric calculation.

We next show the upper-bound part of Equation~\eqref{eq:property_median_approximation:2} by exploring the function $L_{\max}(w_{1}, w_{2})$. For any fixed $w_{1} \geq 1$, both of the numerator $A_{7}(w_{1}, w_{2})$ and the denominator $A_{8}(w_{1}, w_{2})$ are {\em linear} functions in $w_{2} \in [-\frac{1}{4}, 1]$. Accordingly, $L_{\max}(w_{1}, w_{2})$ itself is a monotone function in $w_{2} \in [-\frac{1}{4}, 1]$. We can infer the upper-bound part of Equation~\eqref{eq:property_median_approximation:2} as follows:
\begin{eqnarray*}
    \max_{w_{1} \in [0, 1]} \max_{w_{2} \in [-\frac{1}{4}, 1]}
    L_{\max}(w_{1}, w_{2})
    & = & \max_{w_{1} \in [0, 1]} \max \big\{L_{\max}(w_{1}, -1 / 4), ~ ~ L_{\max}(w_{1}, 1) \big\} \\
    & = & \max_{w_{1} \in [0, 1]} \max \left\{\frac{w_{1}^{2} + 1 + (1 / 2) \cdot w_{1}}{w_{1}^{2} - 2 \cdot w_{1} + (9 / 4)}, ~ ~ 1 \right\} \\
    & \approx & 2.7247,
\end{eqnarray*}
where the last step can be seen via numeric calculation.

This completes the proof.
\end{proof}

\begin{lemma}[Guarantees for {\MergedStage}]
\label{lem:merged_stage_guarantees}
The procedures {\MergedStage} (Algorithm~\ref{alg:merged_stage}) returns a set $\mathpzc{List} = \{(v_{i}', f_{i}')\}_{i \in [m]}$ of $m = |\mathpzc{List}| = 2^{O(d \cdot \log d)} \cdot k \in \mathbb{N}_{\geq 1}$ many candidate tones. With probability at least $1 - 1 / \poly(k)$, the outputs $\mathpzc{List} = \{(v_{i}', f_{i}')\}_{i \in [m]}$ satisfies the following:
\begin{description}
    \item [Property~I:]
    The set size $m = 2^{O(d\cdot\log d)}\cdot k$, and the frequency separation 
    \begin{align*}
    \min_{i, j \in [m]} \|f_{i}' - f_{j}'\|_{2} \gtrsim \eta/\sqrt{d}.
    \end{align*}

    \item [Property~II:]
    For the true tones $\{(v_{i}, f_{i})\}_{i \in [k]}$, there is an injection $\pi: [k] \mapsto [m]$ such that
    \begin{align*}
        \sum_{i \in [k]} \frac{1}{T^d} \cdot \int_{\tau \in [0, T]^d} \Big| v_{\pi(i)}' \cdot e^{2 \pi \i \cdot f_{\pi(i)}'^\top \tau} - v_{i} \cdot e^{2 \pi \i \cdot f_{i}^\top \tau} \Big|^2 \cdot \d \tau
        ~ \leq ~ ({\cal C}^2 + d C_*^2) \cdot \N^2 .
    \end{align*}
\end{description}
\end{lemma}

\begin{proof}

The size of output is straightforward from Algorithm~\ref{alg:merged_stage}. If we add one candidate tone into the output, we will delete at least $8/10\cdot  \mathcal{R}_{\mathrm{merge}}$ tones.

{\bf Property I:}
    The set size can be induced from proof of the Property II. As for the frequency separation, it comes from that if we choose to take the median of $\mathbf{HC}(\xi_i, \eta/d^3)$ for $\xi_i$, we will clear a larger region $\mathbf{HC}(\xi_i, \eta / (10\sqrt{d}))$. 
    
    It is safe to clear the larger region, as we  have an assumption that  $\min_{i\neq j}\|f_i-f_j\|_2\geq \eta$, which implies that $\min_{i\neq j}\|f_i-f_j\|_{\infty}\geq \eta/\sqrt{d}$ for true tones $\{(v_i,f_i)\}$. Suppose $\xi_i$ is a successful recovery of true tone $f_i$. Then if we find a cluster of successful recovered tones $\mathbf{HC}(\xi_i, \eta/d^3)$, for all other successful recovered tones $\xi_j$ where $j\neq i$, we have that
    \begin{align*}
        \| \xi_i - \xi_j \|_{\infty} 
        = &~ \|\xi_i -f_i+f_i-f_j+f_j-\xi_j\|_{\infty}\\
        \geq & ~ \|f_i-f_j\|_{\infty} - \|\xi_i-f_i\|_{\infty} -\|f_j-\xi_j\|_{\infty}\\
        \geq & ~ \|f_i-f_j\|_{2}/\sqrt{d} - \|\xi_i-f_i\|_{2} - \|f_j-\xi_j\|_{2}\\
        \gtrsim & ~  \eta/\sqrt{d} - 2 C_*/(\rho T)\\
        \gtrsim & ~  \eta/\sqrt{d} - 2 C_*/ T\\
        \gtrsim & ~ \eta/\sqrt{d}
    \end{align*}
    where the second step follows from triangle inequality, the third step follows from $ \| \cdot \|_2 / \sqrt{ d } \leq \| \cdot \|_{\infty} \leq \| \cdot \|_2 $ , the last step follows from $T \geq C_* \sqrt{d} / \eta$.

    This means that $\xi_j \notin \mathbf{HC}(\xi_i, \eta /(10\sqrt{d}))$ and proves the safety of the operation.

{\bf Property II:}

For each true tone $(v_i,f_i)$, by Lemma~\ref{lem:multi_stage_guarantees}, with probability at least $1-1/\poly(k)$, there are at least $0.8\mathcal{R}_{\mathrm{merge}} $ a  successful recovery $\{(v_i',f_i')\}$ of it, where $\|f_i'-f_i\|_2\lesssim C_*/ (\rho T )$. By the choice of duration $T  =  \Omega \big( d^{3} \cdot \eta^{-1} \cdot \log(k d / \delta) \big) $ by Lemma~\ref{lem:locate_inner_duration_require}, we know that $\|f_i'-f_i\|_{\infty}\leq \|f_i'-f_i\|_2\ll\eta/d^3$. And let $\mu^2(f_i)$ denote the expected error of successful recovery $(v_i',f_i')$ :
\begin{align*}
    \mu^2(f_i) &= \E_{\Sigma,b,v_i',f_i'} \left[ \frac{1}{T^d}\int_{t\in [0,T]^d}|v_i'\cdot e^{2\pi\i\cdot f_i'^\top t}-v_i\cdot e^{2\pi\i f_i^\top t}|^2\cdot \d t \right]
\end{align*}

Then by Markov Inequality, we know that 
\begin{align}\label{eq:concentration_successful_recovery_single_tone}
    \Pr\left[\int_{t\in [0,T]^d}|v_i'\cdot e^{2\pi\i\cdot f_i'^\top t}-v_i\cdot e^{2\pi\i f_i^\top t}|^2\cdot \d t\geq 10\mu^2(f_i) \right]\leq 1/10.
\end{align}

By Lemma~\ref{lem:one_stage_guarantees}, we can bound the summation of expected errors of successful recovery:
\begin{align}
    \label{eq:bound_summation_error}
    \sum_{i\in[k]}\mu^2(f_i)& \lesssim ({\cal C}^2 + d C_*^2)\cdot\N^2.
\end{align}

As a summary, for each true tone $(v_i,f_i)$, we have shown that there are at least $0.8\mathcal{R}_{\mathrm{merge}} $ successful recovery $\{(v_i',f_i')\}$ of it, ie. $\textsc{Tree}.\textsc{count}(\mathbf{HC}(\xi_i, \eta / d^3))  \geq  8 / 10 \cdot \mathcal{R}_{\mathrm{merge}}$. Then we will take the any frequency $f_i^*$ in $\mathbf{HC}(f_i, \eta / d^3)$ in Line~\ref{lin:take_median_frequency} and coordinate-wise median of magnitude $v_i^*$ of successful recovery in $\mathbf{HC}(f_i, \eta / d^3)$ in Line~\ref{lin:take_median_mag}.

Among the successful recovery $\{(v_i',\xi_i)|f_i'\in \mathbf{HC}(\xi_i, \eta / d^3)\}$ of $(v_i,f_i)$, with probability $1-1/\poly(k)$, at least half of them will have error less than $10\mu^2(f_i)$. Note that $f^* = \xi_i$. \footnote{Note that we only need to take coordinate wise median for $v$, for frequency $f$, using $\xi_i$ is good enough. Since $\xi_i$ is close to the true $f$.} To be more specific, with probability at least $1-1/\poly(k)$,
%\begin{align*}
%    \median_{f_i'\in \mathbf{HC}(f_i, \eta / d^3)}\frac{1}{T^d} \cdot \int_{t \in [0, T]^d} \big| v_i' \cdot e^{2 \pi \i \cdot f_i'^{\top} \tau} - v_i \cdot e^{2 \pi \i \cdot f_i^{\top} \tau} \big|^2 \cdot \d t &\leq 10\mu^2(f_i).
%\end{align*}

\begin{align*}
     \frac{1}{T^d} \cdot \int_{t \in [0, T]^d} \big|v_i^*\cdot e^{2\pi\i\cdot f_i^{*\top}}-v_i \cdot e^{2 \pi \i \cdot f_i^{\top}t}\big|^2\dot \d t
    \lesssim \mu^2(f_i),
\end{align*}

Then we have
\begin{align*}
     \sum_{i \in [k]} \frac{1}{T^d} \cdot \int_{\tau \in [0, T]^d} \Big| v_{i}' \cdot e^{2 \pi \i \cdot f_{i}'^\top \tau} - v_{i} \cdot e^{2 \pi \i \cdot f_{i}^\top \tau} \Big|^2 \cdot \d \tau
    \lesssim  \sum_{i\in[k]} \mu^2(f_i)
    \leq  ({\cal C}^2 + d C_*^2) \cdot \N^2 .
\end{align*}

 This completes the proof.
\end{proof}

\subsection{Running {\MergedStage} twice}
\label{sec:merged_stage_twice}

The goal of this section is to prove Lemma~\ref{lem:merged_stage_twice}.

\begin{definition}[Setup for {\RecoveryStage}]
\label{def:merged_stage_twice}
Given two sets
\begin{align*}
    & \mathpzc{List}_{1}^* = \{(v_{i}', f_{i}')\}_{i \in [k']}
    && \mbox{ and }
    && \mathpzc{List}_{2}^* = \{(v_{i}'', f_{i}'')\}_{i \in [k'']}
\end{align*}
of sizes $k', k''= 2^{O(d \cdot \log d)} \cdot k \in \mathbb{N}_{\geq 1}$, output each pair $(v_{i}'', f_{i}'')$ in the second set (for $i \in [k'']$) that has a small frequency distance $\|f_{i}'' - f_{j}'\|_{2} \leq c/ T$, against some frequency $f_{j}'$ in the first set (for $j \in [k']$). Denote the resulting set by $\{(v_{i}'', f_{i}'')\}_{i \in S} \subseteq \{(v_{i}'', f_{i}'')\}_{i \in [k'']}$ of size $|S| = k^{*} \leq k'' = 2^{O(d \cdot \log d)} \cdot k$.
\end{definition}

\begin{algorithm}
\caption{{\RecoveryStage}}
\label{alg:recovery_stage}
\begin{algorithmic}[1]
\Procedure{RecoveryStage}{$x, D, {\cal C}, T$} \Comment{Theorem~\ref{thm:recovery_stage}}
    \State ${\cal R}_{{\rm merge}} \gets \Theta(d \cdot \log d \cdot \log k)$.
    
    \State $\mathpzc{List}' \leftarrow {\MultiStage}(x, D, \mathcal{C}, T, {\cal R}_{{\rm merge}})$.
    \Comment{Algorithm~\ref{alg:multi_stage}}
    
    \State $\mathpzc{List}_{1}^* \gets \MergedStage(\mathpzc{List}', {\cal R}_{{\rm merge}})$.
    \Comment{Algorithm~\ref{alg:merged_stage}} 
    
    \State $\mathpzc{List}'' \leftarrow {\MultiStage}(x, D, \mathcal{C}, T, {\cal R}_{{\rm merge}})$.
    \Comment{Algorithm~\ref{alg:multi_stage}}
    
    \State $\mathpzc{List}_{2}^* \gets \MergedStage(\mathpzc{List}'', {\cal R}_{{\rm merge}})$.
    \Comment{Algorithm~\ref{alg:merged_stage}}
    
    \State Derive $\mathpzc{List}^*$ from $\mathpzc{List}_{1}^*$ and $\mathpzc{List}_{2}^*$ according to Definition~\ref{def:merged_stage_twice}.
    \Comment{Lemma~\ref{lem:merged_stage_twice}}
    
    \State Sort $\mathpzc{List}^* = \{(v_{i}^*, f_{i}^*)\}_{i = 1}^{|\mathpzc{List}^*|}$ in decreasing order of magnitudes $|v_{i}^*|$.
    
    \State $\mathpzc{List}^*_{[k]} \leftarrow $ the top-$k$ tones $\{(v_{i}^*, f_{i}^*)\}_{i = 1}^{k}$ in $\mathpzc{List}^*$
    
    \State \Return $\mathpzc{List}^*_{[k]}$.
\EndProcedure
\end{algorithmic}

\end{algorithm}

\begin{lemma}[Running {\MergedStage} twice]
\label{lem:merged_stage_twice}
Given two sets $\{(v_{i}', f_{i}')\}_{i \in [k']}$ and $\{(v_{i}'', f_{i}'')\}_{i \in [k'']}$ of sizes $k', k''= 2^{O(d  \log d)} \cdot k \in \mathbb{N}_{\geq 1}$, assume w.l.o.g.\ that Definition~\ref{def:merged_stage_twice} selects $k^{*} \leq k''$ pairs $\{(v_{i}'', f_{i}'')\}_{i \in [k^{*}]}$ of the second set, then these $k^* = 2^{O(d  \log d)} \cdot k$ pairs can be reindexed such that
\begin{align*}
    \sum_{i \in [k]} \frac{1}{T^{d}} \cdot \int_{\tau \in [0, T]^{d}} \Big| v_{i} \cdot e^{2 \pi \i \cdot f_{i}^{\top} \tau} - v_{i}' \cdot e^{2 \pi \i \cdot f_{i}^{'\top} \tau} \Big|^{2} \cdot \d \tau
    ~ + ~ \sum_{i \in [k^*] \setminus [k]} |v_{i}'|^{2}
    ~ \lesssim ~ \mathcal{C}^{2} \cdot \N^{2}.
\end{align*}
\end{lemma}

\begin{proof}

By Claim~\ref{cla:property_median_approximation}, the following holds for any pair of tones $(v, f) \in \C \times \R^{d}$ and $(v^*, f^*) \in \C \times \R^{d}$:
\begin{eqnarray*}
    \mathrm{err} ((v,f),(v^*,f^*)) &= &
     \frac{1}{T^d} \cdot \int_{\tau \in [0, T]^d} \big| v \cdot e^{2 \pi \i \cdot f^{\top} \tau} - v^* \cdot e^{2 \pi \i \cdot f^{*\top} \tau} \big|^2 \cdot \d \tau \\
      & \eqsim & | v - v^* |^{2} ~ + ~ (| v^* |^{2}+|v|^2) \cdot \min \{ 1, ~ T^{2} \cdot \|f - f^*\|_{2}^{2} \}
\end{eqnarray*}

Then by Lemma~\ref{lem:merged_stage_guarantees}, with probability at least $1-1/\poly(k)$, there is a permutation of the output of the first run $\{(v_i',f_i')\}_{i\in [k']}$ and an injective projection $\pi:[k]\rightarrow [k]$, subject to
\begin{align*}
    \sum_{i=1}^{k}\big( (|v_i'|^2+|v_{\pi(i)}|^2)\cdot \min \{ 1,T^2\cdot \|f_i' - f_{\pi(i)}\|_{2}^{2} \} +|v_i'-v_{\pi(i)}|^2  \big)\lesssim ({\cal C}^2 + d C_*^2) \N^2
\end{align*}

If $\|f_i'-f_{\pi(i)}\|_2>1/T$, 
then $\mathrm{err}((0,f_i'),(v_{\pi(i)},f_{\pi(i)}))\leq \mathrm{err}((v_i',f_i'),(v_{\pi(i)},f_{\pi(i)}))$. Let $S = \{i \in [k] : \|f_i'-f_{\pi(i)}\|_2\leq c/T \}$ 
for any $c=O(1)$. We can rewrite the result:
\begin{eqnarray*}
    & & \sum_{i\in S}\big( (|v_i'|^2+|v_{\pi(i)}|^2)\cdot \min \{ 1,T^2\cdot \|f_i' - f_{\pi(i)}\|_2^2 \} +|v_i'-v_{\pi(i)}|^2  \big)+ \sum_{i\in [k]/S}(|v_{\pi(i)}|^2+|v_i'|^2) \\
    & \lesssim & ({\cal C}^2 + d C_*^2) \N^2 .
\end{eqnarray*}

If we can know the set $S$ and the right permutation of the output of the first run, we can output a set of tones that meet this lemma easily. But the problem is that we do not have the information. This is why we run the {\MergedStage} twice. Recall that the signal $x^*$ we want to recover is defined by $\{v_i,f_i\}_{i=1}^{k}$, then it is equivalent to define $x^*$ by $\{v_i,f_i\}_{i=1}^{k}\cup \{0,f_{i}'\}_{i=1}^{k'}$, where $f_i'$ is the output of the first run of {\MergedStage}. Then the number of frequencies is $2^{O(d\log d)}k$ and the separation gap is $\Omega(\eta/\sqrt{d})$, then Lemma~\ref{lem:merged_stage_guarantees} applies again. 

Define 
\begin{align*}
S' = \{i \in [k''] :\exists j \in [k'] , \|f_j'-f_i''\|_2\leq 1/T\},
\end{align*}
and we can reindex $\{v_i'',f_i''\}$ such that:

\begin{align*}
    ({\cal C}^2 + d C_*^2) \N^2 
    \gtrsim & ~ \sum_{S'\cap [k]} \big((| v_i'' |^2+|v_{i}|^2) \min \{ 1, T^2 \cdot \| f_i'' - f_{i} \|_2^2 \} +(|v_i''|^2+|v_{\pi(i)}|^2) \big)\\
    & ~ + \sum_{i\in S'\setminus[k]}((0^2+| v_i'' |^2) \cdot \min \{ 1, T^2 \| f_i'' - f_i' \|_2^2 \} +|v_i''-0|^2)\\
    & ~ + \sum_{i\in [k]\setminus S'}(| v_i |^2+|v_i''|^2)+\sum_{i\in [k']\setminus(S'\cup [k])}|v_i''|^2 \\
    \geq & ~ \sum_{S'\cap [k]} \big((| v_i'' |^2+|v_{i}|^2) \min \{ 1, T^2 \cdot \| f_i''-f_{i} \|_2^2 \} +(|v_i''|^2+|v_{\pi(i)}|^2) \big)\\
    & ~ + \sum_{i\in S'\setminus[k]}| v_i'' |^2+\sum_{i\in [k]\setminus S'}|v_{i}|^2
\end{align*}

This is exactly the summation of error of $\{(v_i'',f_i'')\}_{i\in S'}$, and $|S'|=k^*\leq k''=2^{O(d\log d)} k$ which complete the proof.

%%%%%Very similar to \cite[Lemma~3.10]{ps15}.

\end{proof}

\subsection{{\RecoveryStage}}
\label{sec:recovery_stage}

The goal of this section is to prove Theorem~\ref{thm:recovery_stage}. Before the proof of the main result, we need the following lemma:

\begin{lemma}
\label{lem:recovery_stage_local}
The following holds for any three tones $(v_{\pi(i)}, f_{\pi(i)}) \in \C \times \R^{d}$ and $(v_{i}^*, f_{i}^*) \in \C \times \R^{d}$ and $(v_{i}', f_{i}') \in \C \times \R^{d}$ that $|v_{i}'| \gtrsim |v_{i}^*| \eqsim |v_{\pi(i)}|$:
\begin{align*}
    & \underbrace{\frac{1}{T^{d}} \cdot \int_{\tau \in [0, T]^{d}} \Big| v_{i}' \cdot e^{2 \pi \i \cdot f_{i}^{'\top} \tau} - v_{\pi(i)} \cdot e^{2 \pi \i \cdot f_{\pi(i)}^{\top} \tau} \Big|^{2} \cdot \d \tau}_{A_{1}} \\
    ~ \lesssim ~ & \underbrace{\frac{1}{T^{d}} \cdot \int_{\tau \in [0, T]^{d}} \Big| v_{i}^* \cdot e^{2 \pi \i \cdot f_{i}^{*\top} \tau} - v_{\pi(i)} \cdot e^{2 \pi \i \cdot f_{\pi(i)}^{\top} \tau} \Big|^{2} \cdot \d \tau}_{A_{2}}
    ~ + ~ |v_{i}'|^2.
\end{align*}
\end{lemma}

\begin{proof}

We will show in Property~II of  Lemma~\ref{lem:tone_distance} (see Section~\ref{sec:convert_split}) that
\begin{eqnarray*}
    A_{1}
    & = & | v_{i}' |^{2} ~ + ~ | v_{\pi(i)} |^{2} ~ - ~ \big( v_{i}' \cdot \bar{v_{\pi(i)}} + \bar{v_{i}'} \cdot v_{\pi(i)} \big) \cdot \sinc_{T}(f_{i}' - f_{\pi(i)}), \\
    A_{2}
    & = & | v_{i}^* |^{2} ~ + ~ | v_{\pi(i)} |^{2} ~ - ~ \big( v_{i}^* \cdot \bar{v_{\pi(i)}} + \bar{v_{i}^*} \cdot v_{\pi(i)} \big) \cdot \sinc_{T}(f_{i}^* - f_{\pi(i)}).
\end{eqnarray*}

Because $|v_{i}'| \gtrsim |v_{i}^*| \eqsim |v_{\pi(i)}|$, we can easily verify the lemma by elementary calculation (notice that $|\sinc_{T}(f_{i}' - f_{\pi(i)})| \leq 1$ and $|\sinc_{T}(f_{i}^* - f_{\pi(i)})| \leq 1$ for any $f_{\pi(i)}, f_{i}^*, f_{i}' \in \R^{d}$).
\end{proof}

\begin{theorem}[{\RecoveryStage}, formal of Theorem~\ref{thm:intro_tone}]
\label{thm:recovery_stage}
Let 
\begin{align*}
T \geq \frac{d^{4.5}\log(kd/\delta) \log d}{\eta}.
\end{align*}
Let ${\cal C}$ be some universal constant and $C_*=d^2$.
The procedure {\RecoveryStage} (Algorithm~\ref{alg:recovery_stage}) takes 
\begin{align*}
2^{\Theta(d\log d)}\cdot k\cdot \log^{d+1}(k/\delta) \cdot  \log(F/\eta) \cdot \log\log(F/\eta)
\end{align*}
samples over $[0,T]$, runs in
\begin{align*}
 2^{O(d\cdot\log d)}\cdot k \cdot \log^{O(d)}(k/\delta) \cdot \log(F/\eta) \cdot \log\log(F/\eta).
\end{align*}
time and outputs a set $\{(v_{i}', f_{i}')\}_{i \in [k]} \subset \C \times \R^{d}$ of size $k \in \mathbb{N}_{\geq 1}$ such that the following hold with probability $1 - 1/ \poly(k)$
\begin{description}
        \item [Property I] Magnitude estimation 
        \begin{align*}
          |v_{i} - v_{i}'| \leq ({\cal C} + \sqrt{d} C_*) \cdot \N , \forall i \in [k] .
        \end{align*}
        \item [Property II] Frequency estimation
        \begin{align*}
            \| f_{i} - f_{i}' \|_2 \leq C_* \frac{1}{ \rho \cdot T } , \forall i \in[k] .
        \end{align*}
        \item [Property III] Tone estimation (Total)
	    \begin{align*}
	        \sum_{i \in [k]} \frac{1}{T^{d}} \cdot \int_{\tau \in [0, T]^{d}} \Big| v_{i} \cdot e^{2 \pi \i \cdot f_{i}^{\top} \tau} - v_{i}' \cdot e^{2 \pi \i \cdot f_{i}^{'\top} \tau} \Big|^{2} \cdot \d \tau
	        ~ \lesssim ~ ({\cal C}^2 + d C_*^2) \cdot \N^{2}.
	    \end{align*}
	    \item [Property IV] The frequency separation of output frequencies 
    	\begin{align*}
    	    \min_{i \neq j \in [k]} \| f_{i}' - f_{j}' \|_2 \geq \eta/2.
    	\end{align*}
	\end{description}
\end{theorem}

Recall that $C_*=d^2$ is stated in the statement in Lemma~\ref{lem:locate_inner_stronger}.
\begin{claim}[Sample complexity, running time and duration of Theorem~\ref{thm:recovery_stage}]

\end{claim}
\begin{proof}

{\bf Sample complexity.}

\begin{align*}
&~ {\cal R}_{\rm merge}\cdot  2^{\Theta(d\cdot\log d)} \cdot (\log \mathcal{C} + \log \log (F / \eta)) \cdot k \cdot \log (F \cdot T) \cdot \mathcal{D} \\
= &~ 2^{\Theta(d\log d)}\cdot k\cdot \log^{d+1}(k/\delta)\cdot  \log\log(F/\eta)\cdot\log(F/\eta)
\end{align*}

{\bf Running time.}

\begin{align*}
& ~ {\cal R}_{\rm merge}( 2^{O(d \cdot \log d)} \cdot k \cdot \log^{d}(k \cdot {\cal R}_{\rm merge}) + 2^{\Theta(d\cdot (\log d + \log {\cal C} ))} \cdot \log(F \cdot T)\cdot  \log\log(F/\eta) \cdot k \cdot ({\cal D} + \log k)) \\
= &~ 2^{O(d \cdot \log d)} \cdot k \cdot \log^{O(d)} k + 2^{\Theta(d\log d)}\cdot\log(F/\eta)\cdot k\cdot\log\log(F/\eta)\cdot\log^{d+1}(k/\delta)\\
=& ~ 2^{O(d\cdot\log d)}\cdot k \cdot \log^{O(d)}(k/\delta)\cdot\log(F/\eta)\cdot\log\log(F/\eta).
\end{align*}

{\bf Duration.}
As we run {\MergedStage} twice, the first run for $k$-sparsity signal and the second run for $k'=2^{O(d\log d)}k$-sparsity signal, $\eta'=\frac{\eta}{\sqrt{d}}$, by Lemma~\ref{lem:locate_signal_duration}, the duration is
\begin{align*}
T = \Omega \Big( \frac{d^{3}\log(dk'/\delta)}{\eta'} \Big) = \Omega \Big( \frac{d^{4.5}\log(dk/\delta) \log d}{\eta} \Big).
\end{align*}

\end{proof}

\begin{claim}[Property I of Lemma~\ref{thm:recovery_stage}]
\begin{align*}
    |v_{i} - v_{i}'| \leq ({\cal C} + \sqrt{d} C_*) \cdot \N, \forall i \in [k]
\end{align*}
\end{claim}
\begin{proof}
This proof is a direct application of previous Lemma.
\end{proof}

\begin{claim}[Property II of Lemma~\ref{thm:recovery_stage}]
\label{clarecovery_stage:property_2}
\begin{align*}
        \| f_{i} - f_{i}' \|_2 \leq C_* \frac{1}{ \rho \cdot T }, \forall i \in[k]
\end{align*}
\end{claim}
\begin{proof}
The proof is a direct application of previous Lemma.
\end{proof}

\begin{claim}[Property III of Lemma~\ref{thm:recovery_stage}]
\begin{align*}
	  \sum_{i \in [k]} \frac{1}{T^{d}} \cdot \int_{\tau \in [0, T]^{d}} \Big| v_{i} \cdot e^{2 \pi \i \cdot f_{i}^{\top} \tau} - v_{i}' \cdot e^{2 \pi \i \cdot f_{i}^{'\top} \tau} \Big|^{2} \cdot \d \tau
	  ~ \lesssim ~ ({\cal C}^2 + d C_*^2) \cdot \N^{2}.
\end{align*}
\end{claim}
\begin{proof}
We denote by $\{(v_{i}^*, f_{i}^*)\}_{i \in [k'']}$ the set of tones derived according to Definition~\ref{def:merged_stage_twice}, where $k'' = 2^{O(d \cdot \log d)} \cdot k$ (and we safely assume $k'' \geq k$ in view of Lemma~\ref{lem:merged_stage_guarantees}). We assume w.l.o.g.\ that each $(v_{i}^*, f_{i}^*) \in \C \times \R^{d}$ of the top-$k$ largest-magnitude tones (for each $i \in [k]$) is mapped to a {\em true} tone $(v_{\pi(i)}, f_{\pi(i)}) \in \C \times \R^{d}$ according to Lemma~\ref{lem:merged_stage_twice}.

Let $\{(v_{i}', f_{i}')\}_{i \in [k]}$ be a subset of the recovered tones $\{(v_{i}^*, f_{i}^*)\}_{i \in [k'']}$ that have the top-$k$ largest magnitudes; these $k \in \mathbb{N}_{\geq 1}$ tones together form the output of the procedure {\RecoveryStage} (Algorithm~\ref{alg:recovery_stage}). Upon reindexing, we safely assume that $|v_{i}'| \geq |v_{i}^*|$ for each $i \in [k]$. Also, we know from Lemma~\ref{lem:merged_stage_twice} that $\min_{i \neq j \in [k]} \|f_{i}' - f_{j}'\| \gtrsim \eta$.

For each $i \in [k]$, let us consider these three tones $(v_{\pi(i)}, f_{\pi(i)})$ and $(v_{i}^*, f_{i}^*)$ and $(v_{i}', f_{i}')$. In the case $i \in S$ for which $(v_{i}^*, f_{i}^*) \neq (v_{i}', f_{i}')$, it follows from Lemma~\ref{lem:recovery_stage_local} that
\begin{align*}
    & \frac{1}{T^{d}} \cdot \int_{\tau \in [0, T]^{d}} \Big| v_{i}' \cdot e^{2 \pi \i \cdot f_{i}^{'\top} \tau} - v_{\pi(i)} \cdot e^{2 \pi \i \cdot f_{\pi(i)}^{\top} \tau} \Big|^{2} \cdot \d \tau \\
    ~ \lesssim ~ & \frac{1}{T^{d}} \cdot \int_{\tau \in [0, T]^{d}} \Big| v_{i}^* \cdot e^{2 \pi \i \cdot f_{i}^{*\top} \tau} - v_{\pi(i)} \cdot e^{2 \pi \i \cdot f_{\pi(i)}^{\top} \tau} \Big|^{2} \cdot \d \tau + |v_{i}'|^2.
\end{align*}

And in the other case $i \in S \subseteq [k]$ for which $(v_{i}^*, f_{i}^*) = (v_{i}', f_{i}')$, of course we have
\begin{align*}
    & \frac{1}{T^{d}} \cdot \int_{\tau \in [0, T]^{d}} \Big| v_{i}' \cdot e^{2 \pi \i \cdot f_{i}^{'\top} \tau} - v_{\pi(i)} \cdot e^{2 \pi \i \cdot f_{\pi(i)}^{\top} \tau} \Big|^{2} \cdot \d \tau \\
    ~ = ~ & \frac{1}{T^{d}} \cdot \int_{\tau \in [0, T]^{d}} \Big| v_{i}^* \cdot e^{2 \pi \i \cdot f_{i}^{*\top} \tau} - v_{\pi(i)} \cdot e^{2 \pi \i \cdot f_{\pi(i)}^{\top} \tau} \Big|^{2} \cdot \d \tau.
\end{align*}

Taking all the indices $i \in [k]$ into account, we know from the above two equations that
\begin{align*}
    & \sum_{i \in [k]} \frac{1}{T^{d}} \cdot \int_{\tau \in [0, T]^{d}} \Big| v_{i}' \cdot e^{2 \pi \i \cdot f_{i}^{'\top} \tau} - v_{\pi(i)} \cdot e^{2 \pi \i \cdot f_{\pi(i)}^{\top} \tau} \Big|^{2} \cdot \d \tau \\
    ~ \lesssim ~ & \sum_{i \in [k]} \frac{1}{T^{d}} \cdot \int_{\tau \in [0, T]^{d}} \Big| v_{i}^* \cdot e^{2 \pi \i \cdot f_{i}^{*\top} \tau} - v_{\pi(i)} \cdot e^{2 \pi \i \cdot f_{\pi(i)}^{\top} \tau} \Big|^{2} \cdot \d \tau + \sum_{i \in S} |v_{i}'|^2.
\end{align*}

Because $\{(v_{i}', f_{i}')\}_{i \in [k]}$ are chosen to be the top-$k$ largest-magnitude recovered tones among $\{(v_{i}^*, f_{i}^*)\}_{i \in [k'']}$, the set $S$ involved in the summation $\sum_{i \in S} |v_{i}'|^2$ only includes those small-magnitude recovered tones. We thus conclude that
\begin{align*}
    \sum_{i \in [k]} \frac{1}{T^{d}} \cdot \int_{\tau \in [0, T]^{d}} \Big| v_{i}' \cdot e^{2 \pi \i \cdot f_{i}^{'\top} \tau} - v_{\pi(i)} \cdot e^{2 \pi \i \cdot f_{\pi(i)}^{\top} \tau} \Big|^{2} \cdot \d \tau
    ~ \lesssim ~ ({\cal C}^2 + d C_*^2) \cdot \N^2.
\end{align*}

This completes the proof of Property III of Theorem~\ref{thm:recovery_stage}.
\end{proof}

\begin{claim}[Property IV of Lemma~\ref{thm:recovery_stage}]
\begin{align*}
    	    \min_{i \neq j \in [k]} \| f_{i}' - f_{j}' \|_2 \geq \eta/2.
\end{align*}
\end{claim}
\begin{proof}
Since $\min_{i\neq j \in [k]} \| f_i -f_j\|_2 \geq \eta $ and for all $i \in [k]$, $\| f_i - f_i' \|_2 \leq  \min(C_* \frac{1}{ \rho \cdot T },d/T) \leq \eta /10$ (given Lemma~\ref{lem:locate_signal_guarantees}, Claim~\ref{clarecovery_stage:property_2} and the duration 
\begin{align*}
T = \Omega \Big( \frac{d^{4.5}\log(kd/\delta) \log d}{\eta} \Big),
\end{align*} 
we can infer the current claim.
\end{proof}

%%% Sparse recovery
% \newpage
\section{Converting tone estimation into signal estimation}
\label{sec:convert}

This section is structured in the following way:
\begin{itemize}
    \item Section~\ref{sec:our_technique:tone_estimation} briefly discusses the high-level idea of the proof.
    \item Section~\ref{sec:convert_definitions} provides some basic definitions and mathematical facts.
    \item Section~\ref{sec:convert_split} splits the signal estimation error into the tone-wise errors (which we call the diagonal terms) and the cross-tone errors (which we call the off-diagonal terms).
    \item Section~\ref{sec:convert_off_diagonal} provides an upper bound for the cross-tone errors (i.e.\ the off-diagonal terms) via some advanced analytic tools.
    \item Section~\ref{sec:convert_combine} combines everything together, converting the tone estimation error into the signal estimation error as desired.
    \item Section~\ref{sec:convert_geometry} states several geometry properties.
    \item Section~\ref{sec:convert_main} presents our main result.
\end{itemize}
For ease of presentation, throughout this section we would shift the sampling time domain from $t \in [0, T]^{d}$ to $t \in [-T / 2, T / 2]^{d}$.

% upon \texorpdfstring{\cite{ps15}}{}}
\subsection{Improvement of signal estimation duration}\label{sec:our_technique:tone_estimation}

The claimed tone estimation guarantee holds when the duration $T \gtrsim \eta^{-1} \cdot C_{\text{tone}}$ (Theorem~\ref{thm:intro_tone}; see Section~\ref{sec:our_technique:signal_estimation} for more discussions), for $C_{\text{tone}} := d^{4.5} \cdot \log(k d / \delta)\cdot \log (d)$. To further get the signal estimation guarantee (Theorem~\ref{thm:intro_signal}), we adopt the proof framework of \cite{ps15} but provide a better analysis. Particularly, we will show that the signal estimation holds when
\begin{align*}
    \mbox{$T ~ \gtrsim ~ \eta^{-1} \cdot ( C_{\text{tone}} + d^{1.5} \cdot k^{1-1/d} \cdot \log k )$}.
\end{align*}
In one dimension $d=1$, this bound is $\eta^{-1} \cdot \log(k / \delta)$, which improves the $\eta^{-1} \cdot \log^2(k / \delta)$ bound by \cite{ps15} and answers an open question in the thesis \cite{s19}.
% Indeed, the duration $T$ is an equally important optimization goal as the sample complexity and running time. As quoted from \cite{m15}:
%\begin{quote}
%    ``The basic question is, how large does the cutoff frequency (a.k.a.\ duration in this paper) need to be in order to recover the signals ...\ a broad set of applications including medical imaging, microscopy, astronomy, radar, geophysics and spectroscopy \textcolor{red}{(see \cite{} and references therein)}. This is also a fundamental algorithmic and statistical problem in its own right.''
%\end{quote}

Following \cite{ps15}, we rewrite the signal estimation error as $\LHS$ of Eq.~\eqref{eq:intro_result:2} = $\sum_{i ,j \in [k]} \mathrm{err}_{i,j}$, where
\begin{align*}
    \mbox{$\mathrm{err}_{i, j}
    ~ := ~ \frac{1}{T^{d}} \cdot \int_{t \in [0, T]^{d}} (x_{i}'(t) - x_{i}^*(t)) \cdot \bar{(x_{j}'(t) - x_{j}^*(t))} \cdot \d t$}.
\end{align*}
%The total {\em tone-wise} error $\sum_{i \in [k]} \mathrm{err}_{i, i}$ is exactly the tone estimation error (i.e., $\LHS$ of Eq.~\eqref{eq:intro_result:1}). 
To get the signal estimation, \cite{ps15} proves that the {\em cross-tone} errors $|\mathrm{err}_{i, j}|$ for $i \neq j \in [k]$ converge to zero at the rate
% \footnote{In precise, the work \cite{ps15} proves this convergence rate for the one-dimensional case $d = 1$. But following the arguments therein, we can easily derive this claimed convergence rate in the general case $d \geq 1$.}
\begin{align}
    \mbox{$|\mathrm{err}_{i, j}|
    ~ \lesssim ~ \sqrt{\mathrm{err}_{i, i} \cdot \mathrm{err}_{j, j}} \cdot \sqrt{d} \cdot \frac{\log(1 + \|f_{i} - f_{j}\|_{2} \cdot T)}{\|f_{i} - f_{j}\|_{2} \cdot T}
    ~ = ~ O(\sqrt{d} \cdot T^{-1} \cdot \log T)$}.
    \label{eq:intro_duration:2}
\end{align}
%for $i \neq j \in [k]$. In contrast, b
%\begin{align}
%    \mbox{$|\mathrm{err}_{i, j}|
%    ~ \lesssim ~ \sqrt{d} \cdot \big(\|f_{i} - f_{j}\|_{2} \cdot T\big)^{-1} \cdot \sqrt{\mathrm{err}_{i} \cdot \mathrm{err}_{j}}.{\color{white}~\log(1 + \|f_{i} - f_{j}\|_{2} \cdot T)}$}
%\end{align}

Based on a new application of Parseval's theorem and the convolution theorem, we will prove that the $O(\log T)$ term in Eq.~\eqref{eq:intro_duration:2} can be removed. More concretely, Parseval's theorem gives the analytic formulas of the errors $|\mathrm{err}_{i, i}|$ and $|\mathrm{err}_{i, j}|$, in the case of an {\em infinite} duration $t \in \R^{d}$. Inspired by this, we access the proof details of Parseval's theorem. Following the involved arguments and the convolution theorem, we get the counterpart formulas (Lemma~\ref{lem:tone_distance}) in the case of a {\em finite} duration $t \in [0, T]^{d}$. These analytic formulas and other arguments together give the faster convergence rate $|\mathrm{err}_{i, j}| = O(\sqrt{d} \cdot T^{-1})$. In contrast, \cite{ps15} just uses the approximate formulas of the errors $|\mathrm{err}_{i, i}|$ and $|\mathrm{err}_{i, j}|$, which incurs the factor-$\log(k / \delta)$ loss in their signal estimation duration.

%By leveraging the faster convergence rate in Eq.~\eqref{eq:intro_duration:2}, we improve the duration bound for signal estimation (see Theorem~\ref{thm:intro_signal}) and thus answer an open question asked in the thesis \cite{s19}. 

Indeed, we have a concrete example for which $|\mathrm{err}_{i, j}| = \Omega(\sqrt{d} \cdot T^{-1})$, matching our new convergence rate in the duration $T$ and the dimension $d \geq 1$. We believe that this tight convergence rate, as well as the analytic formulas in Lemma~\ref{lem:tone_distance}, can find their applications in the future.

\subsection{Preliminaries and mathematical facts}
\label{sec:convert_definitions}
%%%Isn't the entire paper about math? why do have a name ``mathematical facts''

In this part, we introduce some useful notations and mathematical facts. Recall Definition~\ref{def:rect_sinc} for the functions $\rect_{s_1}(\xi)$ and $\sin_{s_1}(\tau)$ in the single-dimensional setting (namely $\xi, \tau \in \R$). Below, we define in Definition~\ref{def:rect_sinc_multi} two counterpart functions (redenoted by $\rect_{s_1}(\xi)$ and $\sin_{s_1}(\tau)$ for convenience) when $\xi, \tau \in \R^{d}$ are $d$-dimensional vectors, and then show in Fact~\ref{fac:sinc_function_multi} several properties of these functions (which can be easily inferred from Fact~\ref{fac:sinc_function} or the previous literature like \cite{ckps16}).

\begin{definition}[Two basic functions]
\label{def:rect_sinc_multi}
Given any $s_1 > 0$, for all $\xi, \tau \in \R^{d}$, the $\rect_{s_1}(\xi)$ function and the $\sinc_{s_1}(\tau)$ function are defined as follows:
\begin{itemize}
    \item $\rect_{s_1}(\xi) = \prod_{r \in [d]} \rect_{s_1}(\xi_{r})$ for any $\xi \in \R^{d}$. When $s_1 = 1$, we shorthand it as $\rect(\xi)$.
    
    \item $\sinc_{s_1}(\tau) = \prod_{r \in [d]} \sinc_{s_1}(\tau_{r})$ for any $\tau \in \R^{d}$. When $s_1 = 1$, we shorthand it as  $\sinc(\tau)$.
\end{itemize}
\end{definition}

\begin{fact}[Facts about basic functions]
\label{fac:sinc_function_multi}
Given any $s_1 > 0$, the following hold for the functions $\sinc_{s_1}(\tau)$ and $\rect_{s_1}(\xi)$ in the $d$-dimensional setting, as Figure~\ref{fig:sinc_function_multi} suggests:
\begin{description}[labelindent = 1em]
    \item [Part~(a):]
    $|\sinc_{s_1}(\tau)| \leq \prod_{r \in [d]} \min \{ 1, \frac{1}{\pi \cdot s_1 \cdot |\tau_{r}|} \}$ for any $\tau \in \R^{d}$.
    
    \item [Part~(b):]
    $\sinc_{s_1}(\tau) = \hat{\rect_{s_1}}(\tau)$ for any $\tau \in \R^{d}$, and $\rect_{s_1}(\xi) = \hat{\sinc_{s_1}}(\xi)$ for any $\xi \in \R^{d}$.
    
    \item [Part~(c):]
    $\exp(-\frac{\pi^{2} \cdot s_{1}^{2}}{5} \cdot \|\tau\|_{2}^{2}) \leq \sinc_{s_1}(\tau) \leq \exp(-\frac{\pi^{2} \cdot s_{1}^{2}}{6} \cdot \|\tau\|_{2}^{2})$ for any $\tau \in \R^{d}$ that $\|\tau\|_{2} \leq \frac{2.05}{\pi s_1}$.
    
    \item [Part~(d):]
    $|\sinc_{s_1}(\tau)| \leq \exp(-\frac{2.05^2}{6}) < \frac{1}{2}$ for any $\tau \in \R^{d}$ that $\|\tau\|_{2} \geq \frac{2.05}{\pi s_1}$.
    
    \item [Part~(e):]
    $-\frac{1}{4} \leq \sinc_{s_1}(\tau) \leq 1$ and $\sinc_{s_1}(\tau) \geq 1 - \frac{\pi^2}{6} \cdot s_1^2 \cdot \|\tau\|_{2}^2$ for any $\tau \in \R^{d}$.
\end{description}
Also, in the single-dimensional setting:
\begin{description}[labelindent = 1em]
    \item [Part~(f):]
    $| \frac{\d}{\d \tau} \sinc_{s_1}(\tau)| = | \frac{\cos(\pi \cdot s_{1} \cdot \tau)}{\tau} - \frac{\sinc_{s_{1}}(\tau)}{\tau} | \leq \frac{7}{5} \cdot \min \{ s_{1}, \frac{1}{|\tau|} \}$ for any $\tau \in \R$.
\end{description}
\end{fact}

\begin{figure}
    \centering
    \includegraphics[scale=1.5]{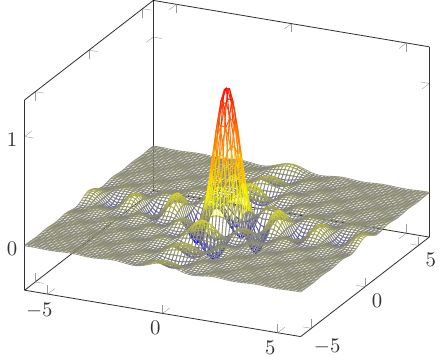}
    \caption{Demonstration for the two-dimensional sinc function.}
    \label{fig:sinc_function_multi}
\end{figure}

\subsection{Tone-wise errors and cross-tone errors}
\label{sec:convert_split}

The goal of this section is to prove Lemma~\ref{lem:tone_distance}. We first start with the following definitions.

\begin{definition}[Tone-wise errors in time domain]
\label{def:error_time_domain}
Given any pair of tones $(v_{i}, f_{i}) \in \C \times \R^d$ and $(v_{i}', f_{i}') \in \C \times \R^d$, where $i \in [k]$, the error is given by the complex-valued function $a_{i}(\tau) \in \C$:
\begin{itemize}
    \item Define $a_{i}(\tau) = v_{i} \cdot e^{2 \pi \i \cdot f_{i}^{\top} \tau} - v_{i}' \cdot e^{2 \pi \i \cdot f_{i}^{'\top} \tau}$ for all $\tau \in [-T / 2, T / 2]^{d}$ for notational brevity.
    
    \item The {\CFT} is given by $\hat{a_{i}(\xi)} = v_{i} \cdot \Dirac_{= f_{i}(\xi)} - v_{i}' \cdot \Dirac_{= f_{i}'}(\xi)$ for all $\xi \in \R^{d}$.
    
    \item Define the error $\| a_{i} \|_{T} = \sqrt{\E_{\tau}[| a_{i}(\tau) |^{2}]} = \sqrt{\E_{\tau}[ a_{i}(\tau) \cdot \bar{a_{i}(\tau)} ]}$, where $\tau \sim \unif[-T / 2, T / 2]^{d}$ is uniformly random, or equivalently,
    \begin{eqnarray*}
        \big\| a_{i} \big\|_{T}
        & = & \Big( \frac{1}{T^d} \cdot \int_{ \tau \in [-T / 2, T / 2]^d } \big| a_{i}(\tau) \big|^2 \cdot \d \tau \Big)^{1/2} \\
        & = & \Big( \frac{1}{T^d} \cdot \int_{ \tau \in [-T / 2, T / 2]^d } a_{i}(\tau) \cdot \bar{a_{i}(\tau)} \cdot \d \tau \Big)^{1/2}
    \end{eqnarray*}
\end{itemize}
\end{definition}

\begin{lemma}[Tone-wise and cross-tone errors in time domain]
\label{lem:tone_distance}
Respecting a pair of error functions $a_{i}(\tau) \in \C$ and $a_{j}(\tau) \in \C$ given in Definition~\ref{def:error_time_domain}, where $i, j \in [k]$, the following hold:
\begin{description}[labelindent = 1em]
    \item [Property~I:]
    When $\tau \sim \unif[-T / 2, T / 2]^{d}$ is uniformly random,
    \begin{eqnarray*}
        \E_{\tau} \Big[ a_{i}(\tau) \cdot \bar{a_{j}(\tau)} \Big]
        & = & v_{i} \cdot \bar{v_{j}} \cdot \sinc_{T}(f_{i} - f_{j})
        ~ - ~ v_{i} \cdot \bar{v_{j}'} \cdot \sinc_{T}(f_{i} - f_{j}') \\
        & & - ~ v_{i}' \cdot \bar{v_{j}} \cdot \sinc_{T}(f_{i}' - f_{j})
        ~ + ~ v_{i}' \cdot \bar{v_{j}'} \cdot \sinc_{T}(f_{i}' - f_{j}').
    \end{eqnarray*}
    
    \item [Property~II:]
    In the special case that $i = j$,
    \begin{eqnarray*}
        \big\| a_{i} \big\|_{T}^2
        ~ = ~ \E_{\tau} \Big[ a_{i}(\tau) \cdot \bar{a_{i}(\tau)} \Big]
        & = & | v_{i} |^{2} ~ + ~ | v_{i}' |^{2} ~ - ~ \big( v_{i} \cdot \bar{v_{i}'} + \bar{v_{i}} \cdot v_{i}' \big) \cdot \sinc_{T}(f_{i}' - f_{i}).
    \end{eqnarray*}
\end{description}
\end{lemma}

\begin{proof}
Assume Property~I to be true, then we can infer Property~II by elementary calculation.

Before proving Property~I, let us consider the following function $y_{i}(\tau)$ for all $\tau \in \R^{d}$:
\begin{eqnarray*}
    y_{i}(\tau)
    & = & a_{i}(\tau) \cdot \mathbb{I}\big\{\tau \in [-T / 2, T / 2]^{d}\big\} \\
    & = & T^{d} \cdot a_{i}(\tau) \cdot \rect_{T}(\tau),
\end{eqnarray*}
as well as its {\CFT} $\hat{y_{i}}(\xi)$ for all $\xi \in \R^{d}$:
\begin{eqnarray}
    \notag
    \hat{y_{i}}(\xi)
    & = & T^{d} \cdot \hat{a_{i} * \rect_{T}}(\xi) \\
    \notag
    & = & T^{d} \cdot \hat{a_{i}} * \hat{\rect_{T}}(\xi) \\
    \notag
    & = & T^{d} \cdot \hat{a_{i}} * \sinc_{T}(\xi) \\
    \label{eq:lem:convert_tone_to_signal:aux:2:1}
    & = & T^{d} \cdot v_{i} \cdot \sinc_{T}(f_{i} - \xi) ~ - ~ T^{d} \cdot v_{i}' \cdot \sinc_{T}(f_{i}' - \xi),
\end{eqnarray}
where the second step applies the convolution theorem; the third step is due to Part~(d) of Fact~\ref{fac:sinc_function}; and the last step follows from Definition~\ref{def:error_time_domain} that $\hat{a_{i}}(\xi) = v_{i} \cdot \Dirac_{= f_{i}}(\xi) - v_{i}' \cdot \Dirac_{= f_{i}'}(\xi)$ for $\xi \in \R^{d}$.

Similar to Equation~\eqref{eq:lem:convert_tone_to_signal:aux:2:1}, we also have
\begin{eqnarray}
    \label{eq:lem:convert_tone_to_signal:aux:2:1.5}
    \bar{\hat{y_{j}}(\xi)}
    & = & T^{d} \cdot \bar{v_{j}} \cdot \sinc_{T}(f_{j} - \xi) - T^{d} \cdot \bar{v_{j}'} \cdot \sinc_{T}(f_{j}' - \xi).
\end{eqnarray}

Based on the above arguments, we deduce that when $\tau \sim \unif[-T / 2, T / 2]^{d}$ is uniformly random,
\begin{eqnarray*}
    \E_{\tau} \Big[ a_{i}(\tau) \cdot \bar{a_{j}(\tau)} \Big]
    & = & \frac{1}{T^{d}} \cdot \int_{\tau \in [-T / 2, T / 2]^{d}} a_{i}(\tau) \cdot \bar{a_{j}(\tau)} \cdot \d \tau \\
    & = & \frac{1}{T^{d}} \cdot \int_{\tau \in \R^{d}} \Re\big(y_{i}(\tau) \cdot \bar{y_{j}(\tau)}\big) \cdot \d \tau \\
    & = & \frac{1}{T^{d}} \cdot \int_{\xi \in \R^{d}} \hat{y_{i}}(\xi) \cdot \bar{\hat{y_{j}}(\xi)} \cdot \d \xi \\
    & = & T^{d} \cdot v_{i} \cdot \bar{v_{j}} \cdot \underbrace{\int_{\xi \in \R^{d}} \sinc_{T}(f_{i} - \xi) \cdot \sinc_{T}(f_{j} - \xi) \cdot \d \xi}_{A_{1}} \\
    & & - ~ T^{d} \cdot v_{i} \cdot \bar{v_{j}}' \cdot \underbrace{\int_{\xi \in \R^{d}} \sinc_{T}(f_{i} - \xi) \cdot \sinc_{T}(f_{j}' - \xi) \cdot \d \xi}_{A_{2}} \\
    & & - ~ T^{d} \cdot v_{i}' \cdot \bar{v_{j}} \cdot \underbrace{\int_{\xi \in \R^{d}} \sinc_{T}(f_{i}' - \xi) \cdot \sinc_{T}(f_{j} - \xi) \cdot \d \xi}_{A_{3}} \\
    & & + ~ T^{d} \cdot v_{i}' \cdot \bar{v_{j}}' \cdot \underbrace{\int_{\xi \in \R^{d}} \sinc_{T}(f_{i}' - \xi) \cdot \sinc_{T}(f_{j}' - \xi) \cdot \d \xi}_{A_{4}},
\end{eqnarray*}
where the second step follows because $y_{i}(\tau) = T^{d} \cdot a_{i}(\tau) \cdot \rect_{T}(\tau) = 0$ for any $\tau \notin [-T / 2, T / 2]^{d}$; the third step applies Parseval's theorem; and the last step employs Equations~\eqref{eq:lem:convert_tone_to_signal:aux:2:1} and \eqref{eq:lem:convert_tone_to_signal:aux:2:1.5}.

We next give in Equation~\eqref{eq:lem:convert_tone_to_signal:aux:2:2} an explicit formula for $A_{1}$, and similar formulas respectively for $A_{2}$ and $A_{3}$ and $A_{4}$ can be obtained in the same way. Concretely, we have
\begin{eqnarray}
    \notag
    A_{1}
    & = & \int_{\xi \in \R^{d}} \sinc_{T}(f_{i} - \xi) \cdot \sinc_{T}(f_{j} - \xi) \cdot \d \xi \\
    \notag
    & = & \frac{1}{T^{d}} \cdot \int_{\xi \in \R^{d}} \sinc(T f_{i} - \xi) \cdot \sinc(T f_{j} - \xi) \cdot \d \xi \\
    \notag
    & = & \frac{1}{T^{d}} \cdot \int_{\xi \in \R^{d}} \sinc(\xi + T f_{i} - T f_{j}) \cdot \sinc(\xi) \cdot \d \xi \\
    \notag
    & = & \frac{1}{T^{d}} \cdot \sinc(T f_{i} - T f_{j}) \\
    \label{eq:lem:convert_tone_to_signal:aux:2:2}
    & = & \frac{1}{T^{d}} \cdot \sinc_{T}(f_{i} - f_{j}),
\end{eqnarray}
where the second step follows by substitution; the third step also follows by substitution; the fourth step follows from Part~(b) of Fact~\ref{fact:convert_tone_to_signal}; and the last step follows by substitution.

Applying Equation~\eqref{eq:lem:convert_tone_to_signal:aux:2:2} and the counterpart formulas for $A_{2}$ and $A_{3}$ and $A_{4}$, we conclude that when $\tau \sim \unif[-T / 2, T / 2]^{d}$ is uniformly random,
\begin{eqnarray*}
    \E_{\tau} \Big[ a_{i}(\tau) \cdot \bar{a_{j}(\tau)} \Big]
    & = & v_{i} \cdot \bar{v_{j}} \cdot \sinc_{T}(f_{i} - f_{j})
    ~ - ~ v_{i} \cdot \bar{v_{j}'} \cdot \sinc_{T}(f_{i} - f_{j}') \\
    & & - ~ v_{i}' \cdot \bar{v_{j}} \cdot \sinc_{T}(f_{i}' - f_{j})
    ~ + ~ v_{i}' \cdot \bar{v_{j}'} \cdot \sinc_{T}(f_{i}' - f_{j}').
\end{eqnarray*}

This completes the proof.
\end{proof}

\subsection{Upper bounding cross-tone errors}
\label{sec:convert_off_diagonal}

\begin{fact}
\label{fact:convert_tone_to_signal}
The following hold for the single-/multi-dimensional sinc function:
\begin{description}[labelindent = 1em]
    \item [Part~(a):]
    Single dimension. $\sinc(\Delta_{r}) = \int_{\xi_{r} \in \R} \sinc(\xi_{r} + \Delta_{r}) \cdot \sinc(\xi_{r}) \cdot \d \xi$ for any $\Delta_{r} \in \R$.
    
    \item [Part~(b):]
    Multi dimension. $\sinc(\Delta) = \int_{\xi \in \R^{d}} \sinc(\xi + \Delta) \cdot \sinc(\xi) \cdot \d \xi$ for any $\Delta \in \R^{d}$.
\end{description}
\end{fact}

\begin{proof}
Part~(b) can be easily inferred from Part~(a), since $\sinc(\Delta) = \prod_{r \in [d]} \sinc(\Delta_{r})$ is a product and we deal with all the coordinates $r \in [d]$ separately.

We deduce Part~(a) as follows:
\begin{eqnarray*}
    \int_{\xi_{r} \in \R} \sinc(\xi_{r} + \Delta_{r}) \cdot \sinc(\xi_{r}) \cdot \d \xi_r
    & = & \int_{\xi_{r} \in \R} \sinc(\Delta_{r} - \xi_{r}) \cdot \sinc(\xi_{r}) \cdot \d \xi_r \\
    & = & \sinc * \sinc(\Delta_r) \\
    & = & \int_{\tau \in \R} \hat{\sinc * \sinc}(\tau) \cdot e^{-2 \pi \i \cdot \Delta_{r} \cdot \tau} \cdot \d \tau \\
    & = & \int_{\tau \in \R} \rect^{2}(\tau) \cdot e^{-2 \pi \i \cdot \Delta_{r} \cdot \tau} \cdot \d \tau \\
    & = & \int_{-1/2}^{1/2} e^{-2 \pi \i \cdot \Delta_{r} \cdot \tau} \cdot \d \tau \\
    & = & \sinc(\Delta_{r}),
\end{eqnarray*}
where the first step follows by substitution and the fact that $\sinc(\xi_{r})$ is an even function; the second step follows from Definition~\ref{def:convolution}; the third step follows the definition of the {\CFT}; the fourth step applies the convolution Theorem as $\hat{\sinc * \sinc}(\tau) = \rect(\tau) \cdot \rect(\tau) = \rect^{2}(\tau)$; the fifth follows because $\rect(\tau) = \mathbb{I}\{|\tau| \leq 1 / 2\}$; and the last step can be seen via elementary calculation.

This completes the proof.
\end{proof}

\begin{lemma}[Upper bounds on the cross-tone errors]
\label{lem:convert_tone_to_signal:aux:2}
For any pair of indices $i < j \in [k]$. Assume the following for both $(v, f, v', f') = (v_{i}, f_{i}, v_{i}', f_{i}')$ and $(v, f, v', f') = (v_{j}, f_{j}, v_{j}', f_{j}')$:
\begin{itemize}
    \item $\| f  - f' \|_{2} \leq \Delta f_{i, j}$, where the distance $\Delta f_{i, j} \geq 0$ is given by
    \begin{eqnarray*}
        \Delta f_{i, j}
        & := & \min \Big\{ \big\| f_{i}'' - f_{j}'' \big\|_{2}: f_{i}'' \in \{f_{i}, f_{i}'\} \mbox{ and } f_{j}'' \in \{f_{j}, f_{j}'\} \Big\}.
    \end{eqnarray*}
\end{itemize}
Then for the functions $a_{i}(\tau) : \R^d \rightarrow \C$ and $a_{j}(\tau) : \R^d \rightarrow \C$ given in Definition~\ref{def:error_time_domain}, the cross-tone error satisfies the following when $T = \Omega(\frac{d}{\Delta f_{i, j}})$ is large enough:
\begin{eqnarray*}
    \big| \E_{\tau} [ a_{i}(\tau) \cdot \bar{a_{j}(\tau)} + \bar{a_{i}(\tau)} \cdot a_{j}(\tau) ] \big|
    & \lesssim & \frac{\sqrt{d}}{\Delta f_{i, j} \cdot T} \cdot \big\| a_{i} \big\|_{T} \cdot \big\| a_{j} \big\|_{T},
\end{eqnarray*}
where $\tau \sim \unif[-T / 2, T / 2]^{d}$ is uniformly random.
\end{lemma}

\begin{proof}
In this proof, we use $f_{i , s} \in \R$ to denote the $s$-th coordinate of the $i$-th frequency $f_{i} \in \R^{d}$.

For simplicity of notation, we define
\begin{align*}
    \mathrm{err}_{i, j} := \E_{\tau} [ a_{i}(\tau) \cdot \bar{a_{j}(\tau)} + \bar{a_{i}(\tau)} \cdot a_{j}(\tau) ].
\end{align*}

According to Property~II of Lemma~\ref{lem:tone_distance}, we can rewrite $\mathrm{err}_{i,j}$ as follows:
\begin{eqnarray}
    \notag
    \mathrm{err}_{i, j}
    & = & \E_{\tau} \Big[ a_{i}(\tau) \cdot \bar{a_{j}(\tau)} + \bar{a_{i}(\tau)} \cdot a_{j}(\tau) \Big] \\
    \notag
    & = & \big( v_{i} \cdot \bar{v_{j}} + \bar{v_{i}} \cdot v_{j} \big) \cdot \sinc_{T}(f_{i} - f_{j})
    ~ - ~ \big( v_{i} \cdot \bar{v_{j}'} + \bar{v_{i}} \cdot v_{j}' \big) \cdot \sinc_{T}(f_{i} - f_{j}') \\
    \label{eq:off_diagonal:case1:0}
    & & - ~ \big( v_{i}' \cdot \bar{v_{j}} + \bar{v_{i}'} \cdot v_{j} \big) \cdot \sinc_{T}(f_{i}' - f_{j})
    ~ + ~ \big( v_{i}' \cdot \bar{v_{j}'} + \bar{v_{i}'} \cdot v_{j}' \big) \cdot \sinc_{T}(f_{i}' - f_{j}').
\end{eqnarray}
Below, we would prove the lemme based on case analysis.

In total there are four cases:
\begin{itemize}
    \item $T \cdot \|f_{i}' - f_{i}\|_{2} \geq \frac{2.05}{\pi}$ and $T \cdot \|f_{j}' - f_{j}\|_{2} \geq \frac{2.05}{\pi}$ (see Claim~\ref{cla:case1}).
    
    \item $T \cdot \|f_{i}' - f_{i}\|_{2} < \frac{2.05}{\pi}$ and $T \cdot \|f_{j}' - f_{j}\|_{2} < \frac{2.05}{\pi}$ (see Claim~\ref{cla:case2}).
    
    \item $T \cdot \|f_{i}' - f_{i}\|_{2} \geq \frac{2.05}{\pi}$ and $T \cdot \|f_{j}' - f_{j}\|_{2} < \frac{2.05}{\pi}$ (see Claim~\ref{cla:case3}).
    
    \item $T \cdot \|f_{i}' - f_{i}\|_{2} < \frac{2.05}{\pi}$ and $T \cdot \|f_{j}' - f_{j}\|_{2} \geq \frac{2.05}{\pi}$ (see Claim~\ref{cla:case3}).
\end{itemize}
Combining all the four cases completes the proof.
\end{proof}

\subsubsection{Both pairs are far}

\begin{claim}[Case~(i) for Lemma~\ref{lem:convert_tone_to_signal:aux:2}]\label{cla:case1}
If $T \cdot \|f_{i}' - f_{i}\|_{2} \geq \frac{2.05}{\pi}$ and $T \cdot \|f_{j}' - f_{j}\|_{2} \geq \frac{2.05}{\pi}$, then we have
   \begin{align*}
       \big| \mathrm{err}_{i,j} \big| ~ \lesssim ~ \frac{1}{ \Delta f_{i,j} T } \cdot \| a_i \|_T \cdot \| a_j \|_T.
   \end{align*}
\end{claim}
    
\begin{proof}
Let us first bound the tone-wise errors $\| a_{i} \|_{T}^{2}$ and $\| a_{j} \|_{T}^{2}$ from below. Respecting the $i$-th pair of tones $(v_{i}, f_{i})$ and $(v_{i}', f_{i}')$, we know from Part~(d) of Fact~\ref{fac:sinc_function_multi} that
\begin{eqnarray*}
    | \sinc_{T}(f_{i}' - f_{i}) | & < & \frac{1}{2}.
\end{eqnarray*}
Then, the tone-wise error $\| a_{i} \|_{T}^{2}$ between $(v_{i}, f_{i})$ and $(v_{i}', f_{i}')$ admits the lower bound
\begin{eqnarray}
    \notag
        \| a_{i} \|_{T}^{2}
        & = & | v_{i} |^{2} ~ + ~ | v_{i}' |^{2} ~ - ~ \big( v_{i} \cdot \bar{v_{i}'} + \bar{v_{i}} \cdot v_{i}' \big) \cdot \sinc_{T}(f_{i}' - f_{i}) \\
        \notag
        & = & | v_{i} |^{2} ~ + ~ | v_{i}' |^{2} ~ - ~ 2 \cdot | v_{i} | \cdot | v_{i}' | \cdot \cos\big( \arg(v_{i}' / v_{i}) \big) \cdot \sinc_{T}(f_{i}' - f_{i}) \\
        \notag
        & \geq & | v_{i} |^{2} ~ + ~ | v_{i}' |^{2} ~ - ~ 2 \cdot | v_{i} | \cdot | v_{i}' | \cdot \Big| \cos\big( \arg(v_{i}' / v_{i}) \big) \Big| \cdot \Big| \sinc_{T}(f_{i}' - f_{i}) \Big| \\
        \notag
        & \geq & | v_{i} |^{2} ~ + ~ | v_{i}' |^{2} ~ - ~ 2 \cdot | v_{i} | \cdot | v_{i}' | \cdot 1 \cdot \frac{1}{2} \\
        \label{eq:off_diagonal:case1:1}
        & \geq & \frac{1}{4} \cdot \big( | v_{i} | + | v_{i}' | \big)^{2},
    \end{eqnarray}
    where the first step is by Property~II of Lemma~\ref{lem:tone_distance}; the fourth step follows since $| \cos( \arg(v_{i}' / v_{i}) ) | \leq 1$ and $| \sinc_{T}(f_{i}' - f_{i}) | < \frac{1}{2}$; and the last step applies the AM-GM inequality.
    
    Applying the same arguments to the $j$-th tone-wise error $\| a_{j} \|_{T}^{2}$, we also have
    \begin{eqnarray}
        \label{eq:off_diagonal:case1:2}
        \| a_{j} \|_{T}^{2}
        & \geq & \frac{1}{4} \cdot \big( | v_{j} | + | v_{j}' | \big)^{2}.
    \end{eqnarray}
    
    We next establish an upper bound on the cross-tone error $| \E_{\tau} [ a_{i}(\tau) \cdot \bar{a_{j}(\tau)} + \bar{a_{i}(\tau)} \cdot a_{j}(\tau) ] |$, where $\tau \sim \unif[-\frac{T}{2}, \frac{T}{2}]^{2}$ is uniformly random. For simplicity, we denote
    \begin{eqnarray}
        \label{eq:off_diagonal:case1:def}
        \sinc_{T, i, j}^{\max}
        & = & \max \Big\{ \big| \sinc_{T}(f_{i}'' - f_{j}'') \big|: f_{i}'' \in \{f_{i}, f_{i}'\} \mbox{ and } f_{j}'' \in \{f_{j}, f_{j}'\} \Big\}
        ~ \geq ~ 0.
    \end{eqnarray}
    
    Following Equation~\eqref{eq:off_diagonal:case1:0}, we deduce that
    \begin{eqnarray}
        \notag
        \big| \mathrm{err}_{i, j} \big|
        & = & \Big| \big( v_{i} \cdot \bar{v_{j}} + \bar{v_{i}} \cdot v_{j} \big) \cdot \sinc_{T}(f_{i} - f_{j})
        ~ - ~ \big( v_{i} \cdot \bar{v_{j}'} + \bar{v_{i}} \cdot v_{j}' \big) \cdot \sinc_{T}(f_{i} - f_{j}') \\
        \notag
        & & - ~ \big( v_{i}' \cdot \bar{v_{j}} + \bar{v_{i}'} \cdot v_{j} \big) \cdot \sinc_{T}(f_{i}' - f_{j})
        ~ + ~ \big( v_{i}' \cdot \bar{v_{j}'} + \bar{v_{i}'} \cdot v_{j}' \big) \cdot \sinc_{T}(f_{i}' - f_{j}') \Big| \\
        \notag
        & \leq & 2 \cdot | v_{i} | \cdot | v_{j} | \cdot \big| \sinc_{T}(f_{i} - f_{j}) \big|
        ~ + ~ 2 \cdot | v_{i} | \cdot | v_{j}' | \cdot \big| \sinc_{T}(f_{i} - f_{j}') \big| \\
        \notag
        & & + ~ 2 \cdot | v_{i}' | \cdot | v_{j} | \cdot \big| \sinc_{T}(f_{i}' - f_{j}) \big|
        ~ + ~ 2 \cdot | v_{i}' | \cdot | v_{j}' | \cdot \big| \sinc_{T}(f_{i}' - f_{j}') \big| \\
        \notag
        & \leq & 2 \cdot | v_{i} | \cdot | v_{j} | \cdot \sinc_{T, i, j}^{\max}
        ~ + ~ 2 \cdot | v_{i} | \cdot | v_{j}' | \cdot \sinc_{T, i, j}^{\max} \\
        \notag
        & & + ~ 2 \cdot | v_{i}' | \cdot | v_{j} | \cdot \sinc_{T, i, j}^{\max}
        ~ + ~ 2 \cdot | v_{i}' | \cdot | v_{j}' | \cdot \sinc_{T, i, j}^{\max} \\
        \notag
        & = & 2 \cdot \big( | v_{i} | + | v_{i}' | \big) \cdot \big( | v_{j} | + | v_{j}' | \big) \cdot \sinc_{T, i, j}^{\max} \\
        \notag
        & \leq & 2 \cdot \big( 2 \cdot \| a_{i} \|_{T} \big) \cdot \big( 2 \cdot \| a_{j} \|_{T} \big) \cdot \sinc_{T, i, j}^{\max} \\
        \label{eq:off_diagonal:case1:3}
        & = & 8 \cdot \| a_{i} \|_{T} \cdot \| a_{j} \|_{T} \cdot \sinc_{T, i, j}^{\max},
    \end{eqnarray}
    where the second step uses the triangle inequality; the third step follows from the definition of $\sinc_{T, i, j}^{\max}$ (see Equation~\eqref{eq:off_diagonal:case1:def}); the fifth step follows from Equations~\eqref{eq:off_diagonal:case1:1} and \eqref{eq:off_diagonal:case1:2}; and the last step follows from the AM-GM inequality.

    To accomplish Case~(i), given Equation~\eqref{eq:off_diagonal:case1:3}, we are left to justify that $\sinc_{T, i, j}^{\max} \geq 0$ diminishes to zero when $T > 0$ goes to the infinity (at the claimed rate). We safely assume $T = \Omega(\frac{d}{\Delta f_{i, j}})$ to be large enough, and consider a specific pair of frequencies $f_{i}'' \in \{f_{i}, f_{i}'\}$ and $f_{j}'' \in \{f_{j}, f_{j}'\}$. For simplicity, we denote $\delta_{r} = \max(0, \pi \cdot T \cdot |f_{i, r}'' - f_{j, r}''| - 1) \geq 0$ for each coordinate $r \in [d]$. Given these, one can easily see that
    \begin{align}
        \notag
        \sum_{r \in [d]} \delta_{r}
        & ~ \geq ~ \sum_{r \in [d]} \big( \pi \cdot T \cdot \big| f_{i, r}'' - f_{j, r}'' \big| - 1 \big) \\
        \notag
        & ~ = ~ \pi \cdot T \cdot \big\| f_{i}'' - f_{j}'' \big\|_{1} - d \\
        \notag
        & ~ = ~ \pi \cdot T \cdot \big\| f_{i}'' - f_{j}'' \big\|_{2} - d \\
        \label{eq:off_diagonal:case1:4}
        & ~ \geq ~ \pi \cdot T \cdot \Delta f_{i, j} - d ~ \geq ~ 0.
    \end{align}

    In addition, we have
    \begin{align}\label{eq:off_diagonal:case1:5}
        \big| \sinc_{T}(f_{i}'' - f_{j}'') \big|
        & ~ \leq ~ \prod_{r \in [d]} \min \Big\{ 1, \frac{1}{ \pi \cdot T \cdot | f_{i, r}'' - f_{j, r}'' | } \Big\} \notag \\
        & ~ = ~ \prod_{r \in [d]} \frac{1}{ 1 + \delta_{r} } \notag \\
        & ~ \leq ~ \frac{1}{ 1 + \sum_{r \in [d]} \delta_{r} } \notag\\
        & ~ \leq ~ \frac{1}{\pi \cdot T \cdot \Delta f_{i, j} - (d - 1)} \notag\\
        & ~ \lesssim ~ \frac{1}{T \cdot \Delta f_{i, j}},
    \end{align}
    where the first step applies Part~(a) of Fact~\ref{fac:sinc_function_multi}; the second step is due to the definition of $\delta_{r}$'s; the third step follows because $\delta_{r} \geq 0$ for each $r \in [d]$; the fourth step follows from Equation~\eqref{eq:off_diagonal:case1:4}; and the last step holds whenever $T = \Omega(\frac{d}{\Delta f_{i, j}})$ is large enough.

    We observe that Equation~\eqref{eq:off_diagonal:case1:5} holds for any pair of frequencies $f_{i}'' \in \{f_{i}, f_{i}'\}$ and $f_{j}'' \in \{f_{j}, f_{j}'\}$. In other words,
    \[
        \sinc_{T, i, j}^{\max}
        ~ \lesssim ~ \frac{1}{T \cdot \Delta f_{i, j}}
    \]
    
    Combining the above equation and Equation~\eqref{eq:off_diagonal:case1:3} together completes the proof.
\end{proof}
    
%%%%%%%%%%%%%%%%%%%%%%%%%%%%%%%%%%%%%%%%%%%%%%%%%%%%%%%%%%%%%%%%%%%%%%%%%%%%%%%%%%%%%%%%%%%%%%%%%%%%%%%%%%%%%%%%%%%%%%%%%%%%%%%%%%%%%%%%%%%%%%%%%%%%%%%%%%%%%%%%%%%%%%%%%%%%%%%%%%%%%%%%%%%%%%%%%%%%%%%%%%%%%%%%%%%%%%%%%%%%

\subsubsection{Both pairs are close}

\begin{claim}[Case~(ii) for Lemma~\ref{lem:convert_tone_to_signal:aux:2}]\label{cla:case2}
 If $T \cdot \|f_{i}' - f_{i}\|_{2} < \frac{2.05}{\pi}$ and $T \cdot \|f_{j}' - f_{j}\|_{2} < \frac{2.05}{\pi}$, then we have
     \begin{align*}
         \mathrm{err}_{i,j} ~ \lesssim ~ \frac{ \sqrt{d} }{ \Delta f_{i,j} T } \cdot \| a_i \|_T \cdot \| a_j \|_T.
     \end{align*}
\end{claim}

\begin{proof}
Let us first bound the tone-wise errors $\| a_{i} \|_{T}^{2}$ and $\| a_{j} \|_{T}^{2}$ from below. Respecting the $i$-th pair of tones $(v_{i}, f_{i})$ and $(v_{i}', f_{i}')$, we know from Part~(d) of Fact~\ref{fac:sinc_function_multi} that
\begin{align*}
    \exp\Big(-\frac{\pi^{2}}{5} \cdot T^{2} \cdot \| f_{i}' - f_{i} \|_{2}^{2}\Big)
    ~ \leq ~ \sinc_{T}(f_{i}' - f_{i})
    ~ \leq ~ \exp\Big(-\frac{\pi^{2}}{6} \cdot T^{2} \cdot \| f_{i}' - f_{i} \|_{2}^{2}\Big).
\end{align*}
Then, the tone-wise error $\| a_{i} \|_{T}^{2}$ between $(v_{i}, f_{i})$ and $(v_{i}', f_{i}')$ admits the lower bound
\begin{eqnarray}
    \| a_{i} \|_{T}^{2}
    & = & | v_{i} |^{2} ~ + ~ | v_{i}' |^{2} ~ - ~ \big( v_{i} \cdot \bar{v_{i}'} + \bar{v_{i}} \cdot v_{i}' \big) \cdot \sinc_{T}(f_{i}' - f_{i})
    \notag \\
    & = & | v_{i} |^{2} ~ + ~ | v_{i}' |^{2} ~ - ~ 2 \cdot | v_{i} | \cdot | v_{i}' | \cdot \cos\big( \arg(v_{i}' / v_{i}) \big) \cdot \sinc_{T}(f_{i}' - f_{i}),
    \label{eq:off_diagonal:case2:1}
\end{eqnarray}
where the first step applies Property~II of Lemma~\ref{lem:tone_distance}.

Given Equation~\eqref{eq:off_diagonal:case2:1}, we would prove that $\| a_{i} \|_{T}^{2}$ is lower bounded by
\begin{eqnarray}
    \| a_{i} \|_{T}^{2}
    & \geq & \frac{3}{13} \cdot \Big( | v_{i} - v_{i}' |^{2} ~ + ~ \big( | v_{i} |^{2} + | v_{i}' |^{2} \big) \cdot \big(1 - \sinc_{T}(f_{i}' - f_{i})\big) \Big).
    \label{eq:off_diagonal:case2:1.1}
\end{eqnarray}
To see so, we denote $w_{1} \cdot e^{\i \cdot \theta} = v_{i} / v_{i}'$ for some norm $w_{1} \geq 0$ and some phase $\theta \in [0, 2 \pi)$, and denote $w_{2} = \sinc_{T}(f_{i}' - f_{i}) \in [-\frac{1}{4}, 1]$ (see Part~(e) of Fact~\ref{fac:sinc_function_multi}). We notice that the $\RHS$ of Equation~\eqref{eq:off_diagonal:case2:1.1} is non-negative. Thus, it suffices to show that the following function $L(w_{1}, w_{2}, \theta) \geq \frac{3}{11}$, for any $w_{1} \geq 0$, any $w_{2} \in [-\frac{1}{4}, 1]$ and any $\theta \in [0, 2 \pi)$:
\begin{eqnarray*}
    L(w_{1}, w_{2}, \theta)
    & := & \frac{\RHS \mbox{ of } \eqref{eq:off_diagonal:case2:1}}{\RHS \mbox{ of } \eqref{eq:off_diagonal:case2:1.1}} \\
    & = & \frac{|v_{i}' \cdot w_{1} \cdot e^{\i \cdot \theta}|^{2} + |v_{i}'|^{2} - 2 \cdot |v_{i}' \cdot w_{1} \cdot e^{\i \cdot \theta}| \cdot |v_{i}'| \cdot \cos(\theta) \cdot w_{2}}{|v_{i}' \cdot w_{1} \cdot e^{\i \cdot \theta} - v_{i}'|^{2} + (|v_{i}' \cdot w_{1} \cdot e^{\i \cdot \theta}|^{2} + |v_{i}'|^{2}) \cdot (1 - w_{2})} \\
    & = & \frac{w_{1}^{2} + 1 - 2 \cdot w_{1} \cdot w_{2} \cdot \cos(\theta)}{|w_{1} \cdot e^{\i \cdot \theta} - 1|^{2} + (w_{1}^{2} + 1) \cdot (1 - w_{2})} \\
    & = & \frac{w_{1}^{2} + 1 - 2 \cdot w_{1} \cdot w_{2} \cdot \cos(\theta)}{(w_{1}^{2} + 1) \cdot (2 - w_{2}) - 2 \cdot w_{1} \cdot \cos(\theta)}
\end{eqnarray*}
where the second step is by the definition of $w_{1}$, $w_{2}$ and $\theta$; the third step divides both the numerator and the denominator by $|v_{i}'|^{2}$; and the last step can be seen via elementary calculation.

Let us investigate the partial derivative $\frac{\partial L}{\partial \theta}$ in $\theta \in [0, 2 \pi)$:
\begin{eqnarray*}
    \frac{\partial L}{\partial \theta}
    & = & \frac{2 \cdot w_{1} \cdot w_{2} \cdot \sin(\theta)}{(w_{1}^{2} + 1) \cdot (2 - w_{2}) - 2 \cdot w_{1} \cdot \cos(\theta)}
    - \frac{\big(w_{1}^{2} + 1 - 2 \cdot w_{1} \cdot w_{2} \cdot \cos(\theta)\big) \cdot \big(2 \cdot w_{1} \cdot \sin(\theta)\big)}{\big((w_{1}^{2} + 1) \cdot (2 - w_{2}) - \cos(\theta)\big)^{2}} \\
    & = & -\sin(\theta) \cdot \underbrace{\frac{2 \cdot w_{1} \cdot (w_{1}^2 + 1) \cdot (w_{2} + 1)^{2}}{\big((w_{1}^{2} + 1) \cdot (2 - w_{2}) - 2 \cdot w_{1} \cdot \cos(\theta)\big)^{2}}}_{A_{5}},
\end{eqnarray*}
where the second step can be seen via elementary calculation.

Because $w_{1} \geq 0$ and $w_{2} \in [-\frac{1}{4}, 1]$, we must have $A_{5} \geq 0$. Hence, for any fixed $w_{1}$ and $w_{2}$, the function $L(w_{1}, w_{2}, \theta)$ is non-increasing when $\theta \in [0, \pi]$, and is non-decreasing when $\theta \in [\pi, 2 \pi)$. Then we conclude that the function $L_{1}(w_{1}, w_{2}) := \min_{\theta \in [0, 2 \pi)} L(w_{1}, w_{2}, \theta)$ for any $w_{1} \geq 0$ and any $w_{2} \in [-\frac{1}{4}, 1]$ is given by
\begin{eqnarray*}
    L_{1}(w_{1}, w_{2})
    & = & L(w_{1}, w_{2}, \pi)
    ~ = ~ \frac{A_{6}(w_{1}, w_{2})}{A_{7}(w_{1}, w_{2})}, \\
    A_{6}(w_{1}, w_{2})
    & := & w_{1}^2 + 1 + 2 \cdot w_{1} \cdot w_{2}, \\
    A_{7}(w_{1}, w_{2})
    & := & (w_{1}^2 + 1) \cdot (2 - w_{2}) + 2 \cdot w_{1}.
\end{eqnarray*}
Clearly, for any fixed $w_{1} \geq 1$, the numerator $A_{6}(w_{1}, w_{2})$ is a non-decreasing function in $w_{2} \in [-\frac{1}{4}, 1]$, while the denominator $A_{7}(w_{1}, w_{2})$ is a non-increasing {\em non-negative} function in $w_{2} \in [-\frac{1}{4}, 1]$. Given these, we deduce that
\begin{eqnarray*}
    \min_{w_{1} \in [0, 1]}
    \min_{w_{2} \in [-\frac{1}{4}, 1]}
    L_{1}(w_{1}, w_{2})
    & = & \min_{w_{1} \in [0, 1]}
    L_{1}(w_{1}, -1 / 4) \\
    & = & \min_{w_{1} \in [0, 1]}
    \frac{w_{1}^{2} + 1 - (1 / 2) \cdot w_{1}}{(9 / 4) \cdot (w_{1}^{2} + 1) + 2 \cdot w_{1}} \\
    & = & \min_{w_{1} \in [0, 1]}
    \Big(\frac{4}{9} - \frac{50 / 81}{w_{1} + (1 / w_{1}) + 8 / 9}\Big) \\
    & = & \frac{4}{9} - \frac{50 / 81}{1 + 1 + 8 / 9} \\
    & = & \frac{3}{13},
\end{eqnarray*}
which implies Equation~\eqref{eq:off_diagonal:case2:1.1} immediately.

Following Equation~\eqref{eq:off_diagonal:case2:1.1}, we further have
\begin{eqnarray}
    \| a_{i} \|_{T}^{2}
    & \gtrsim & | v_{i} - v_{i}' |^{2} ~ + ~ \big( | v_{i} |^{2} + | v_{i}' |^{2} \big) \cdot \big(1 - \sinc_{T}(f_{i}' - f_{i})\big)
    \notag \\
    & \geq & | v_{i} - v_{i}' |^{2} ~ + ~ \big( | v_{i} |^{2} + | v_{i}' |^{2} \big) \cdot \Big(1 - \exp\big(-\frac{\pi^{2}}{6} \cdot T^{2} \cdot \| f_{i}' - f_{i} \|_{2}^{2}\big)\Big)
    \notag \\
    & \geq & | v_{i} - v_{i}' |^{2} ~ + ~ \big( | v_{i} |^{2} + | v_{i}' |^{2} \big) \cdot T^{2} \cdot \| f_{i}' - f_{i} \|_{2}^{2}
    \notag \\
    & \geq & \frac{1}{2} \cdot \Big(| v_{i} - v_{i}' | ~ + ~ \sqrt{| v_{i} |^{2} + | v_{i}' |^{2}} \cdot T \cdot \| f_{i}' - f_{i} \|_{2}\Big)^{2}
    \notag \\
    & \geq & \frac{1}{2} \cdot \Big(| v_{i} - v_{i}' | ~ + ~ \frac{\sqrt{2}}{2} \cdot \big( |v_{i}| + |v_{i}'| \big) \cdot T \cdot \| f_{i}' - f_{i} \|_{2}\Big)^{2}
    \notag \\
    & \geq & \frac{1}{4} \cdot \Big(| v_{i} - v_{i}' | ~ + ~ \big( |v_{i}| + |v_{i}'| \big) \cdot T \cdot \| f_{i}' - f_{i} \|_{2}\Big)^{2},
    \label{eq:off_diagonal:case2:1.5}
\end{eqnarray}
where the first step applies Equation~\eqref{eq:off_diagonal:case2:1.1}; the second step applies Part~(d) of Fact~\ref{fac:sinc_function_multi}; the third step follows from the premise that $T^{2} \cdot \|f_{i}' - f_{i}\|_{2}^{2} < (\frac{2.05}{\pi})^{2}$, together with the fact that, for any $0 \leq z < (\frac{2.05}{\pi})^{2} \approx 0.4258$, we have $\exp(-\frac{\pi^{2}}{6} \cdot z) \leq 1 - z$; and both of the fourth step and the fifth step apply the AM-GM inequality.

Applying the same arguments to the $j$-th tone-wise error $\| a_{j} \|_{T}^{2}$, we also have
\begin{eqnarray}
    \label{eq:off_diagonal:case2:2}
    \| a_{j} \|_{T}^{2}
    & \geq & \frac{1}{9} \cdot \Big(| v_{j} - v_{j}' | ~ + ~ \big( |v_{j}| + |v_{j}'| \big) \cdot T \cdot \| f_{j}' - f_{j} \|_{2}\Big)^{2}.
\end{eqnarray}

Following Equation~\eqref{eq:off_diagonal:case1:0}, we deduce that
\begin{eqnarray}
    \notag
    \big| \mathrm{err}_{i, j} \big|
    & = & \Big| \big( v_{i} \cdot \bar{v_{j}} + \bar{v_{i}} \cdot v_{j} \big) \cdot \sinc_{T}(f_{i} - f_{j})
    ~ - ~ \big( v_{i} \cdot \bar{v_{j}'} + \bar{v_{i}} \cdot v_{j}' \big) \cdot \sinc_{T}(f_{i} - f_{j}') \\
    \notag
    & & - ~ \big( v_{i}' \cdot \bar{v_{j}} + \bar{v_{i}'} \cdot v_{j} \big) \cdot \sinc_{T}(f_{i}' - f_{j})
    ~ + ~ \big( v_{i}' \cdot \bar{v_{j}'} + \bar{v_{i}'} \cdot v_{j}' \big) \cdot \sinc_{T}(f_{i}' - f_{j}') \Big| \\
    \notag
    & = & \Big| \big( (v_{i} - v_{i}') \cdot \bar{(v_{j} - v_{j}')} + \bar{(v_{i} - v_{i}')} \cdot (v_{j} - v_{j}') \big) \cdot \sinc_{T}(f_{i} - f_{j}) \\
    \notag
    & & + ~ \big( v_{i} \cdot \bar{v_{j}'} + \bar{v_{i}} \cdot v_{j}' \big) \cdot \big(\sinc_{T}(f_{i} - f_{j}) - \sinc_{T}(f_{i} - f_{j}')\big) \\
    \notag
    & & + ~ \big( v_{i}' \cdot \bar{v_{j}} + \bar{v_{i}'} \cdot v_{j} \big) \cdot \big(\sinc_{T}(f_{i} - f_{j}) - \sinc_{T}(f_{i}' - f_{j})\big) \\
    \notag
    & & - ~ \big( v_{i}' \cdot \bar{v_{j}'} + \bar{v_{i}'} \cdot v_{j}' \big) \cdot \big(\sinc_{T}(f_{i} - f_{j}) - \sinc_{T}(f_{i}' - f_{j}')\big) \Big| \\
    \notag
    & \leq & \Big| \big( (v_{i} - v_{i}') \cdot \bar{(v_{j} - v_{j}')} + \bar{(v_{i} - v_{i}')} \cdot (v_{j} - v_{j}') \big) \cdot \sinc_{T}(f_{i} - f_{j}) \Big| \\
    \notag
    & & + ~ \Big| \big( v_{i} \cdot \bar{v_{j}'} + \bar{v_{i}} \cdot v_{j}' \big) \cdot \big(\sinc_{T}(f_{i} - f_{j}) - \sinc_{T}(f_{i} - f_{j}')\big) \Big| \\
    \notag
    & & + ~ \Big| \big( v_{i}' \cdot \bar{v_{j}} + \bar{v_{i}'} \cdot v_{j} \big) \cdot \big(\sinc_{T}(f_{i} - f_{j}) - \sinc_{T}(f_{i}' - f_{j})\big) \Big| \\
    \notag
    & & + ~ \Big| \big( v_{i}' \cdot \bar{v_{j}'} + \bar{v_{i}'} \cdot v_{j}' \big) \cdot \big(\sinc_{T}(f_{i} - f_{j}) - \sinc_{T}(f_{i}' - f_{j}')\big) \Big| \\
    \notag
    & \leq & 2 \cdot | v_{i} - v_{i}' | \cdot | v_{j} - v_{j}' | \cdot \Big| \underbrace{\sinc_{T}(f_{i} - f_{j})}_{A_{8}} \Big| \\
    \notag
    & & + ~ 2 \cdot | v_{i} | \cdot | v_{j}' | \cdot \Big| \underbrace{\sinc_{T}(f_{i} - f_{j}) - \sinc_{T}(f_{i} - f_{j}')}_{A_{9}} \Big| \\
    \notag
    & & + ~ 2 \cdot | v_{i}' | \cdot | v_{j} | \cdot \Big| \underbrace{\sinc_{T}(f_{i} - f_{j}) - \sinc_{T}(f_{i}' - f_{j})}_{A_{10}} \Big| \\
    \label{eq:off_diagonal:case2:2.5}
    & & + ~ 2 \cdot | v_{i}' | \cdot | v_{j}' | \cdot \Big| \underbrace{\sinc_{T}(f_{i} - f_{j}) - \sinc_{T}(f_{i}' - f_{j}')}_{A_{11}} \Big|,
\end{eqnarray}
where the second step follows by elementary calculation; and the third step follows from the triangle inequality.

In what follows, we safely assume $T = \Omega(\frac{d}{\Delta f_{i, j}})$ to be large enough, and upper bound the terms $|A_{8}|$ and $|A_{9}|$ and $|A_{10}|$ and $|A_{11}|$ one by one.

Bound on $|A_{8}|$. Recall the quantity $\sinc_{T, i, j}^{\max}$ defined in Equation~\eqref{eq:off_diagonal:case1:def}. We have
\begin{align}
    \notag
    \big| A_{8} \big|
    & ~ \leq ~ \sinc_{T, i, j}^{\max} \\
    \label{eq:off_diagonal:case2:9}
    & ~ \lesssim ~ \frac{1}{T \cdot \Delta f_{i, j}}
\end{align}
where the last step is by Equation~\eqref{eq:off_diagonal:case1:5}, and holds whenever $T = \Omega(\frac{d}{\Delta f_{i, j}})$ is large enough.

Bound on $|A_{9}|$. Under the premises $\|f_{j}' - f_{j}\|_{2} \leq \Delta_{i, j}$ and $\| f_{i} - f_{j} \|_{2} \geq \Delta_{i, j}$ and $\| f_{i} - f_{j}' \|_{2} \geq \Delta_{i, j}$ (see the statement of Lemma~\ref{lem:convert_tone_to_signal:aux:2}), via a standard geometric argument, we know that the next equation holds for any $\lambda \in [0, 1]$:
\begin{align}
    \label{eq:off_diagonal:case2:3}
    \big\| f_{i} - f_{j}''(\lambda) \big\|_{2} \geq \frac{\sqrt{3}}{2} \cdot \Delta_{i, j},
\end{align}
where $f_{j}''(\lambda) = \lambda \cdot f_{j} + (1 - \lambda) \cdot f_{j}'$.

Due to the mean value theorem, there exists a particular $\lambda \in [0, 1]$ such that
\begin{align}
    \label{eq:off_diagonal:case2:4}
    A_{9}
    ~ = ~ \sinc_{T}(f_{i} - f_{j}) - \sinc_{T}(f_{i} - f_{j}')
    ~ = ~ \big(\nabla(f_{i} - f_{j}'')\big)^{\top} (f_{j}' - f_{j}),
\end{align}
where the gradient $\nabla(f_{i} - f_{j}'') \in \R^{d}$ is given by $\nabla_{l}(f_{i} - f_{j}'') = (\frac{\partial}{\partial \xi_{l}} \sinc_{T}(\xi))|_{\xi = f_{i} - f_{j}''}$ for each coordinate $l \in [d]$.

Consider a specific coordinate $l \in [d]$. The corresponding partial derivative is
\begin{align}\label{eq:off_diagonal:case2:5}
    \Big| \nabla_{l}(f_{i} - f_{j}'') \Big|
    & ~ = ~ \bigg| \Big(\frac{\d}{\d \xi_{l}} \sinc_{T}(\xi_{l})\Big)\Big|_{\xi_{l} = f_{i, l} - f_{j, l}''}\bigg| \cdot \prod_{r \in [d] \setminus \{l\}} \bigg| \sinc_{T}(f_{i, r} - f_{j, r}'') \bigg| \notag \\
    & ~ \leq ~ \frac{7}{5} \cdot T \cdot \min \Big\{ 1, \frac{1}{T \cdot |f_{i, l} - f_{j, l}''|} \Big\} \cdot \prod_{r \in [d] \setminus \{l\}} \bigg| \sinc_{T}(f_{i, r} - f_{j, r}'') \bigg| \notag \\
    & ~ \leq ~ \frac{7}{5} \cdot T \cdot \min \Big\{ 1, \frac{1}{T \cdot |f_{i, l} - f_{j, l}''|} \Big\} \cdot \prod_{r \in [d] \setminus \{l\}} \min \Big\{ 1, \frac{1}{\pi \cdot T \cdot |f_{i, r} - f_{j, r}''|}\Big\} \notag \\
    & ~ \lesssim ~ T \cdot \min \Big\{ 1, \frac{1}{T \cdot |f_{i, l} - f_{j, l}''|} \Big\} \cdot \prod_{r \in [d] \setminus \{l\}} \min \Big\{ 1, \frac{1}{\pi \cdot T \cdot |f_{i, r} - f_{j, r}''|} \Big\} \notag \\
    & ~ \lesssim ~ T \cdot \frac{1}{T \cdot \|f_{i} - f_{j}''\|_{2}} \notag \\
    & ~ \lesssim ~ \frac{1}{\Delta f_{i, j}},
\end{align}
where the second step uses Part~(f) of Fact~\ref{fac:sinc_function_multi}; the third step uses Part~(a) of Fact~\ref{fac:sinc_function_multi}; the fifth step holds whenever $T = \Omega(\frac{d}{\Delta f_{i, j}})$ is large enough, and can be seen by reusing the arguments for Equation~\eqref{eq:off_diagonal:case1:5}; and the last step follows from Equation~\eqref{eq:off_diagonal:case2:3}.

We emphasize that Equation~\eqref{eq:off_diagonal:case2:5} holds for any coordinate $l \in [d]$, and therefore holds for the $\ell_{\infty}$-norm $\| \nabla(f_{i} - f_{j}'') \|_{\infty}$ as well. Putting everything together,
\begin{align}\label{eq:off_diagonal:case2:6}
    | A_{9} |
    & ~ = ~ | (\nabla(f_{i} - f_{j}'') )^{\top} (f_{j}' - f_{j}) | \notag \\ 
    & ~ \leq ~ \| \nabla(f_{i} - f_{j}'') \|_{\infty} \cdot \| f_{j}' - f_{j} \|_{1} \notag \\   
    & ~ \lesssim ~ \frac{1}{\Delta f_{i, j}} \cdot \| f_{j}' - f_{j} \|_{1} \notag \\
    & ~ \lesssim ~ \frac{\sqrt{d}}{\Delta f_{i, j}} \cdot \| f_{j}' - f_{j} \|_{2},
\end{align}
where the first step is by Equation~\eqref{eq:off_diagonal:case2:4}; the third step is by Equation~\eqref{eq:off_diagonal:case2:5}; and the last step follows because $\sqrt{d} \cdot \|f_{j}' - f_{j}\|_{2} \geq \|f_{j}' - f_{j}\|_{1}$.

Reapplying the above arguments for $|A_{9}|$, we also have
\begin{align}
    \label{eq:off_diagonal:case2:7}
    | A_{10} |
    & ~ \lesssim ~ \frac{\sqrt{d}}{\Delta f_{i, j}} \cdot \| f_{i}' - f_{i} \|_{2}, \\
    \label{eq:off_diagonal:case2:8}
    | A_{11} |
    & ~ \lesssim ~ \frac{\sqrt{d}}{\Delta f_{i, j}} \cdot ( \| f_{i}' - f_{i} \|_{2} + \| f_{j}' - f_{j} \|_{2} ).
\end{align}

Plugging Equations~\eqref{eq:off_diagonal:case2:9} and \eqref{eq:off_diagonal:case2:6} and \eqref{eq:off_diagonal:case2:7} and \eqref{eq:off_diagonal:case2:8} into Equation~\eqref{eq:off_diagonal:case2:2.5} results in
\begin{eqnarray*}
    | \mathrm{err}_{i, j} |
    & \lesssim & | v_{i} - v_{i}' | \cdot | v_{j} - v_{j}' | \cdot | A_{8} | \\
    \notag
    & & ~ + ~ | v_{i} | \cdot | v_{j}' | \cdot | A_{9} |
    ~ + ~ | v_{i}' | \cdot | v_{j} | \cdot | A_{10} |
    ~ + ~ | v_{i}' | \cdot | v_{j}' | \cdot | A_{11} | \\
    & \lesssim & | v_{i} - v_{i}' | \cdot | v_{j} - v_{j}' | \cdot \frac{1}{T \cdot \Delta f_{i, j}} \\
    \notag
    & & ~ + ~ (| v_{i} | + | v_{i}' |) \cdot (| v_{j} | + | v_{j}' |) \cdot \frac{\sqrt{d}}{\Delta f_{i, j}} \cdot ( \| f_{i}' - f_{i} \|_{2} + \| f_{j}' - f_{j} \|_{2} ) \\
    & \lesssim & \| a_{i} \|_{T} \cdot \| a_{j} \|_{T} \cdot \frac{\sqrt{d}}{T \cdot \Delta f_{i, j}},
\end{eqnarray*}
where the second step uses Equations~\eqref{eq:off_diagonal:case2:9} and \eqref{eq:off_diagonal:case2:6} and \eqref{eq:off_diagonal:case2:7} and \eqref{eq:off_diagonal:case2:8}; and the last step uses Equations~\eqref{eq:off_diagonal:case2:1.5} and \eqref{eq:off_diagonal:case2:2}.

This completes the proof.
\end{proof}

%%%%%%%%%%%%%%%%%%%%%%%%%%%%%%%%%%%%%%%%%%%%%%%%%%%%%%%%%%%%%%%%%%%%%%%%%%%%%%%%%%%%%%%%%%%%%%%%%%%%%%%%%%%%%%%%%%%%%%%%%%%%%%%%%%%%%%%%%%%%%%%%%%%%%%%%%%%%%%%%%%%%%%%%%%%%%%%%%%%%%%%%%%%%%%%%%%%%%%%%%%%%%%%%%%%%%%%%%%%%%%%%%%%%%%%%%%%%%%%%%%%%%%%%%%%%%%%%%%%%%%%%%%%%%%%%%%%%%%%%%%%%%%%%%%%%%%%%%%%%%%%%%%%%%%%%%%%%%%%%%%%%%%%%%%%%%%%%%%%%%%%%%%%%%%%%%%%%%%%%%%%%%%%%%%%%%%%%%%%%%%%%%%%%%%%%%%%%%%%%%%%%%%%%%%%%%%%%%%%%%%%%%%%%%%%%%%%%%%%%%%%%%%%%%%%%%%%%%%%%%%%%%%%%%%%%%%%%%%%%%%%%%%%%%%%%%%%%%%%%%%%%%%%%%%%%%%%%%%%%%%%%%%%%%%%%%%%%%%%%%%%%%%%%%%%%%%%%%%%%%%%%%%%%%%%%%%%%%%%%%%%%%%%%%%%%%%%%%%%%%%%%%%%%%%%%%%%%%%%%%%%%%%%%%%%%%%%%%%     

\subsubsection{One pair is far and one pair is close}
    
\begin{claim}[Case~(iii) for Lemma~\ref{lem:convert_tone_to_signal:aux:2}]\label{cla:case3}
If $T \cdot \|f_{i}' - f_{i}\|_{2} \geq \frac{2.05}{\pi}$ and $T \cdot \|f_{j}' - f_{j}\|_{2} < \frac{2.05}{\pi}$, then we have
\begin{align*}
    | \mathrm{err}_{i,j} | ~ \lesssim ~ \frac{ \sqrt{d} }{ \Delta f_{i,j} T } \cdot \| a_i \|_T \cdot \| a_j \|_T.
\end{align*}
\end{claim}

\begin{proof}
We have shown in Equation~\eqref{eq:off_diagonal:case1:1} that
\begin{eqnarray}
    \label{eq:off_diagonal:case3:1}
    \| a_{i} \|_{T}
    & \gtrsim & | v_{i} | + | v_{i}' |,
\end{eqnarray}
and have shown in Equation~\eqref{eq:off_diagonal:case2:2} that
\begin{eqnarray}
    \label{eq:off_diagonal:case3:2}
    \| a_{j} \|_{T}
    & \gtrsim & | v_{j} - v_{j}' | ~ + ~ \big( |v_{j}| + |v_{j}'| \big) \cdot T \cdot \| f_{j}' - f_{j} \|_{2}.
\end{eqnarray}

Following Equation~\eqref{eq:off_diagonal:case2:2.5}, we deduce that
\begin{eqnarray}\label{eq:off_diagonal:case3:3}
    \big| \mathrm{err}_{i, j} \big|
    & \lesssim & | v_{i} - v_{i}' | \cdot | v_{j} - v_{j}' | \cdot \Big| \sinc_{T}(f_{i} - f_{j}) \Big| \notag\\
    & & + ~ | v_{i} | \cdot | v_{j}' | \cdot \Big| \sinc_{T}(f_{i} - f_{j}) - \sinc_{T}(f_{i} - f_{j}') \Big| \notag \\
    & & + ~ | v_{i}' | \cdot | v_{j} | \cdot \Big| \sinc_{T}(f_{i} - f_{j}) - \sinc_{T}(f_{i}' - f_{j}) \Big| \notag \\
    \notag
    & & + ~ | v_{i}' | \cdot | v_{j}' | \cdot \Big| \sinc_{T}(f_{i} - f_{j}) - \sinc_{T}(f_{i}' - f_{j}') \Big| \\
    & \lesssim & | v_{i} - v_{i}' | \cdot | v_{j} - v_{j}' | \cdot \Big| \underbrace{\sinc_{T}(f_{i} - f_{j})}_{A_{12}} \Big| \notag \\
    & & + ~ | v_{i} | \cdot | v_{j}' | \cdot \Big| \underbrace{\sinc_{T}(f_{i} - f_{j}) - \sinc_{T}(f_{i} - f_{j}')}_{A_{13}} \Big| \notag \\
    & & + ~ | v_{i}' | \cdot | v_{j} | \cdot \bigg(\underbrace{\Big| \sinc_{T}(f_{i} - f_{j}) \Big| + \Big| \sinc_{T}(f_{i}' - f_{j}) \Big|}_{A_{14}}\bigg) \notag \\
    \notag
    & & + ~ | v_{i}' | \cdot | v_{j}' | \cdot \bigg(\underbrace{\Big| \sinc_{T}(f_{i} - f_{j}) \Big| + \Big| \sinc_{T}(f_{i}' - f_{j}') \Big|}_{A_{15}}\bigg) \notag \\
    & \lesssim & | v_{i} - v_{i}' | \cdot | v_{j} - v_{j}' | \cdot \frac{1}{T \cdot \Delta f_{i, j}} \notag \\ 
    & & + ~ | v_{i} | \cdot | v_{j}' | \cdot \frac{\sqrt{d}}{\Delta f_{i, j}} \cdot \| f_{j}' - f_{j} \|_{2} \notag \\
    & & + ~ | v_{i}' | \cdot | v_{j} | \cdot \frac{1}{T \cdot \Delta f_{i, j}} \notag \\
    & & + ~ | v_{i}' | \cdot | v_{j}' | \cdot \frac{1}{T \cdot \Delta f_{i, j}},
\end{eqnarray}
where the first step applies the triangle inequality; the second step applies Equation~\eqref{eq:off_diagonal:case1:5} to $A_{12}$, applies Equation~\eqref{eq:off_diagonal:case2:6} to $A_{13}$, applies Equation~\eqref{eq:off_diagonal:case1:5} to $A_{14}$, and applies Equation~\eqref{eq:off_diagonal:case1:5} to $A_{15}$.

Combining Equations~\eqref{eq:off_diagonal:case3:1} and \eqref{eq:off_diagonal:case3:2} and \eqref{eq:off_diagonal:case3:3} together, it can be easily seen that
\begin{eqnarray}
    \notag
    \big| \mathrm{err}_{i, j} \big|
    & \lesssim & \frac{ \sqrt{d} }{ T \cdot \Delta f_{i,j} } \cdot \| a_i \|_T \cdot \| a_j \|_T.
\end{eqnarray}

This completes the proof.
\end{proof}

\subsection{Combining tone-wise errors and cross-tone errors}
\label{sec:convert_combine}

Let $\Re(z) \in \R$ denote the real part of a complex number $z \in \C$.

%The formal version of Theorem~\ref{thm:intro_signal}
\begin{lemma}
\label{lm:convert_tone_to_signal}
Let $\{(v_{i}, f_{i})\}_{i \in [k]}$ and $\{(v_{i}', f_{i}')\}_{i \in [k]}$ be two sets of $k \in \mathbb{N}_{\geq 1}$ tones, for which
\begin{align*}
    & \min_{i \neq j} \| f_i - f_j \|_1 \geq \eta
    && \mbox{and}
    && \min_{i \neq j} \| f_i' - f_j' \|_1 \geq \eta
    && \mbox{and}
    && \min_{i\in[k]} \|f_i - f_i' \|_1 \leq \eta /100
\end{align*}
%Suppose that $T = \Omega(???)$ is large enough, 
Then these two sets can be reindexed such that
\begin{eqnarray}
    \label{eq:lem:convert_tone_to_signal:0}
    \frac{1}{T^{d}} \cdot \int_{\tau \in [-T / 2, T / 2]^{d}} \Big| \sum_{i \in [k]} a_{i}(\tau) \Big|^{2} \cdot \d \tau
    & \leq &
    (1+ \alpha ) \cdot \sum_{i \in [k]} \frac{1}{T^d} \int_{ \tau \in [-T / 2, T / 2]^{d} } | a_i ( \tau ) |^2 \d \tau.
\end{eqnarray}
where 
\begin{align*}
    \alpha
    & ~ := ~ O(\eta^{-1} \cdot T^{-1}) \cdot \sqrt{d} \cdot \min \Big\{ k, ~~ \sum_{j = 1}^{k - 1} \sqrt{d} \cdot j^{-1 / d} \Big\},
\end{align*}
which further implies
\begin{align*}
    \alpha = 
    \left\{
    \begin{aligned}
    & O(\eta^{-1} \cdot T^{-1}) \cdot \log k,
    && \mathrm{~if~} d = 1; \\
    & O(\eta^{-1} \cdot T^{-1}) \cdot \sqrt{d} \cdot \min \{k, ~ \sqrt{d} \cdot k^{1 - 1 / d} \},
    && \mathrm{~if~} d \geq 2.
    \end{aligned}
    \right.
\end{align*}
\end{lemma}

\begin{proof}
It follows that
\begin{eqnarray*}
    \LHS \mbox{ of } \eqref{eq:lem:convert_tone_to_signal:0}
    & = & \frac{1}{T^{d}} \cdot \int_{\tau \in [-T / 2, T / 2]^{d}} \sum_{i \in [k]} a_{i}(\tau) \cdot \sum_{i \in [k]} \bar{a_{i}(\tau)} \cdot \d \tau \\
    & = & \mbox{diagonal~terms} \quad + \quad \mbox{off-diagonal~terms} 
\end{eqnarray*}
where 
\begin{align*}
    \mbox{diagonal~terms}
    & ~ = ~  \sum_{i \in [k]} \frac{1}{T^{d}} \cdot\int_{\tau \in [-T / 2, T / 2]^{d}} \big| a_{i}(\tau) \big|^{2} \cdot \d \tau
    ~ = ~ \sum_{i \in [k]} \big\| a_i \big\|_T^2, \\
    \mbox{off-diagonal~terms}
    & ~ = ~ \sum_{i < j} \frac{1}{T^{d}} \cdot \int_{\tau \in [-T / 2, T / 2]^{d}} \big(a_{i}(\tau) \cdot \bar{a_{j}(\tau)} + \bar{a_{i}(\tau)} \cdot a_{j}(\tau)\big) \cdot \d \tau \\
    & ~ = ~ \sum_{i < j} \E_{\tau \sim \unif[-T / 2, T / 2]^{d}} \Big[ a_{i}(\tau) \cdot \bar{a_{j}(\tau)} + \bar{a_{i}(\tau)} \cdot a_{j}(\tau) \Big] 
\end{align*}

%Note that if $T \geq C/\eta$ for a sufficiently large constant $C$, this means $\Delta f_{i,j} T \gtrsim \eta T \geq e$. Since $\frac{\log x}{x}$ is decreasing on the region, this implies
%\begin{align*}
%    \frac{\log(\Delta f_{i, j} \cdot T)}{\Delta f_{i, j} \cdot T} \leq \frac{\log (\eta T)}{ \eta T}
%\end{align*}

%Whenever $T = \Omega(\frac{kd}{\eta})$ is large enough, we can upper bound the off-diagonal term by employing Lemma~\ref{lem:convert_tone_to_signal:aux:2} as follows.
First we can simplify the off-diagonal terms in the following sense:
\begin{align*}
    \big| \mbox{off-diagonal~terms} \big|
     = & ~ \bigg| \sum_{i < j} \E_{\tau \sim \unif[-T / 2, T / 2]^{d}} \Big[ a_{i}(\tau) \cdot \bar{a_{j}(\tau)} + \bar{a_{i}(\tau)} \cdot a_{j}(\tau) \Big] \bigg| \\
    \leq & ~ \sum_{i < j} \bigg| \E_{\tau \sim \unif[-T / 2, T / 2]^{d}} \Big[ a_{i}(\tau) \cdot \bar{a_{j}(\tau)} + \bar{a_{i}(\tau)} \cdot a_{j}(\tau) \Big] \bigg| \\
    \lesssim & ~ \sum_{i < j} \frac{ \sqrt{d} }{ T \cdot \Delta f_{i,j} } \cdot \| a_i \|_T \cdot \| a_j \|_T \\
    \leq & ~ \sum_{i < j} \frac{ \sqrt{d} }{ T \cdot \Delta f_{i,j} } \cdot ( \| a_i \|_T^2 + \| a_j \|_T^2 )
\end{align*}
where the second step uses the triangle inequality, and the third step applies Lemma~\ref{lem:convert_tone_to_signal:aux:2}. %the fourth step follows because $\Delta f_{i, j} \geq \eta > 0$ for any $i \neq j$; and the fifth step applies the AM-GM inequality.

We consider two cases. Case 1. $d=1$. Case $d\geq 2$. The reason we consider $d=1$ separately because, for $d=1$ we can get a much better bound than general $d$.

{\bf Case 1.} $d=1$.

We have
\begin{align*}
    \big| \mbox{off-diagonal~terms} \big| 
    \lesssim & ~ \frac{ 1 }{ T  \eta } \cdot \sum_{i < j} \frac{1}{|i-j|} \cdot ( \| a_i \|_T^2 + \| a_j \|_T^2 ) \\
    \leq & ~ \frac{ 1 }{ T \eta } \cdot \sum_{i = 1}^k \| a_i \|_T^2 \sum_{j=1}^k \frac{1}{j} \\
    \leq & ~ \frac{1}{T \eta} \cdot \log k \cdot \sum_{i = 1}^k \| a_i \|_T^2.
\end{align*}

{\bf Case 2.} $d \geq 2$. We give two bounds which are not comparative.

{\bf Case 2a.}

We have
\begin{align*}
    \big| \mbox{off-diagonal~terms} \big| 
    \lesssim & ~ \frac{ \sqrt{d} }{ T  \eta } \cdot \sum_{i < j}  ( \| a_i \|_T^2 + \| a_j \|_T^2 ) \\
    \leq & ~ \frac{ \sqrt{d} }{ T \eta } \cdot k \cdot \sum_{i = 1}^k \| a_i \|_T^2 .
\end{align*}

{\bf Case 2b.}

We have
\begin{align*}
    \big| \mbox{off-diagonal~terms} \big| 
    \lesssim & ~ \frac{ \sqrt{d} }{ T  \eta } \cdot \sum_{i < j} \frac{ \sqrt{d} }{|i-j|^{1/d}} \cdot ( \| a_i \|_T^2 + \| a_j \|_T^2 ) \\
    \leq & ~ \frac{ \sqrt{d} }{ T \eta } \cdot \sqrt{d} \cdot \sum_{i = 1}^k \| a_i \|_T^2 \sum_{j=1}^k \frac{1}{ j^{1/d}} \\
    \leq & ~ \frac{ d }{T \eta} \cdot k^{1-1/d} \cdot \sum_{i = 1}^k \| a_i \|_T^2.
\end{align*}
where the first step follows from  Lemma~\ref{lem:converte_geometry_reindex}.

This completes the proof.
\end{proof}

\subsection{Geometric property}\label{sec:convert_geometry}

\begin{lemma}[Geometric property]
\label{lem:converte_geometry_reindex}
Given a set $\{f_{j}\}_{j \in [k]} \subseteq \R^{d}$ of $k \in \mathbb{N}_{\geq 1}$ many $d$-dimensional frequencies with the minimum $\ell_{2}$-norm separation $\eta := \min_{i \neq j \in [k]} \| f_{i} - f_{j} \|_2 > 0$. Consider any particular frequency $f$ in the set,  then these frequencies can be reindexed such that $ f_1=f$ and
\begin{align*}
    \| f_{1} - f_{j} \|_2 ~ \gtrsim ~ j^{1/d} \cdot \eta / \sqrt{d}, ~~~ \forall j \in [k]
\end{align*}
which further implies
\begin{align*}
    \sum_{j \in [2: k]} \frac{1}{\|f_1-f_j\|_2} \lesssim k^{1-1/d} \cdot \sqrt{d}/\eta.
\end{align*}
\end{lemma}

\begin{figure}
    \centering
    \includegraphics[scale=1]{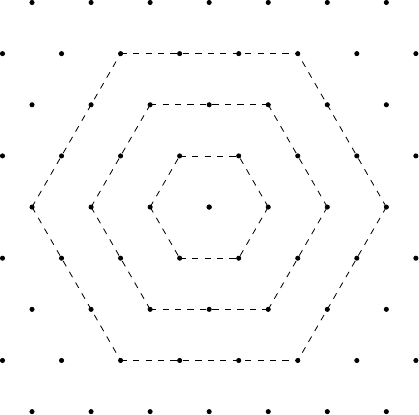}
    \caption{Demonstration for the proof of Lemma~\ref{lem:converte_geometry_reindex}.}
    \label{fig:converte_geometry_reindex}
\end{figure}

\begin{proof}
Fix an arbitrary $f_1 = f$ from the set, and w.l.o.g.\ reindex the frequencies $\{f_{j}\}_{j \in [k]}$ such that
\[
    \|f_2 - f_1\|_2 \leq \cdots \leq \|f_j - f_1\|_2 \leq \cdots \leq \|f_1 - f_k\|_2.
\]
We would prove that $\| f_{j} - f_{1} \|_2 ~ \gtrsim ~ ( |j - 1|^{1/d} / \sqrt{d}) \cdot \eta$ for each $j \in [2: k]$.

Consider the $\ell_{2}$-balls that are centered at the frequencies $\{f_{j}\}_{j \in [k]}$ and have the radius $\eta / 2$ each. Those $\ell_{2}$-balls are disjoint, because the distance of every two frequencies is at least $\eta > 0$.

For a specific $j \in [2: k]$, let us denote $L := \|f_j - f_1\|_2 \geq \eta$. Clearly (as Figure~\ref{fig:converte_geometry_reindex} suggests), all the radius-$(\eta / 2)$ balls centered at $f_{1}, \cdots, f_{j}$ are contained within the bigger $\ell_{2}$-ball that is centered at $f_{1}$ and has the radius $(L + \eta / 2)$. Now consider another geometric question:
\begin{quote}
    How many disjoint radius-$(\eta / 2)$ balls we can pack within a bigger radius-$(L + \eta / 2)$ ball? Let $m \geq 1$ denote this number, and it is easy to see $m \geq j$.
\end{quote}
Indeed, the $m \geq 1$ is call the {\em packing number}. According to \cite[Page~337]{sb14}, we have
\begin{align*}
    m \leq ( 2 \sqrt{d} \cdot (L + \eta / 2) / (\eta / 2))^d,
\end{align*}
which after being rearranged gives $L \gtrsim m^{1/d} \cdot \eta / \sqrt{d}$ and further
\begin{align*}
    & \|f_j - f_1\|_2 \gtrsim j^{1/d} \cdot \eta / \sqrt{d},
    && \forall j \in [2: k].
\end{align*}

Then it is easy to get 
\begin{align*}
    \sum_{j \in [2: k]} \frac{1}{\|f_1-f_j\|_2} \lesssim \sum_{j \in [2: k]} \frac{1}{j^{1/d} \cdot \eta / \sqrt{d}} \lesssim k^{1-1/d} \cdot \sqrt{d}/\eta.
\end{align*}
This finishes the proof.
\end{proof}

\begin{comment}

all the frequencies  falls within the $\ell_{2}$-ball that is centered at $f_{1}$ and has the radius $L$.  Now consider the smaller 

The proof relies on the notion of covering number and packing number. %Given the $\ell_{1}$-norm separation $\eta = \min_{i \neq j \in [k]} \| f_i - f_j \|_1$, the $\ell_{2}$-norm separation $\delta := \min_{i\neq j \in [k]} \| f_i - f_j \|_2$ satisfies that
%\begin{align*}
%    \eta / \sqrt{d} ~ \leq ~ \delta ~ \leq ~ \eta.
%\end{align*}
Let $L $ denote the length of maximum vectors in the set $\in \R^d$. Let $m$ denote the number of points. Due to the geometry properties of covering number and packing number ,

the number of disjoint balls that fit in a space

which implies that

Replacing $m = |1-j|$ and $L = \| f_1 - f_j \|_2$, we obtain the bound that
\begin{align*}
    \|f_1 - f_j \|_2 \gtrsim |1-j|^{1/d} \eta / \sqrt{d}.
\end{align*}

\end{comment}

\subsection{Main result}\label{sec:convert_main}

\begin{theorem}[Signal reconstruction. Formal version of Theorem~\ref{thm:intro_signal}]
\label{thm:signal_formal}
Let ${\cal C}$ be some universal constant and $C_*=d^2$.
When {\RecoveryStage} observes the signal $x(t)$ over a duration
$
    \mbox{$T ~ \gtrsim ~ \eta^{-1} \cdot (d^{4.5}\cdot \log(d) + d^{1.5} \cdot k^{1 - 1 / d}) \cdot \log(k d / \delta)$},
$
the signal estimation error of the $k$-Fourier-sparse recovered signal $x'(t) := \sum_{i \in [k]} x_{i}'(t)$ against the observed signal $x(t) = x^*(t) + g(t)$ is bounded as follows:
\begin{align}\label{eq:formal_intro_result:2}
     \frac{1}{T^d}\int_{[0,T]^d} | x'(t) - x(t) |^2 \cdot \d t
    ~ \leq~ (1+\alpha)({\cal C}^2 + d C_*^2) \N^2,
\end{align}
where 
\begin{align*}
    \alpha = 
    \left\{
    \begin{aligned}
    & O(\eta^{-1} \cdot T^{-1}) \cdot \log k,
    && \mathrm{~if~} d = 1; \\
    & O(\eta^{-1} \cdot T^{-1}) \cdot \sqrt{d} \cdot \min \{k, ~ \sqrt{d} \cdot k^{1 - 1 / d} \},
    && \mathrm{~if~} d \geq 2.
    \end{aligned}
    \right.
\end{align*}

\end{theorem}

\begin{proof}
This result follows directly by Property III of Theorem~\ref{thm:recovery_stage} and Lemma~\ref{lm:convert_tone_to_signal}.
\end{proof} %%% convert tone estimation to signal estimation
%  \newpage
 \section{Conclusion, future directions, other related work}\label{sec:conclusion}

In this paper, we designed a randomized non-adaptive algorithm for the multi-dimensional continuous sparse Fourier transform problem, which achieves a constant approximation under the $\ell_{2} / \ell_{2}$ guarantee and, in any constant dimension, takes sublinear samples and running time. Many attractive directions deserve exploring in the future, for which we give a short discussion below.

\subsection{Future directions}

\noindent
{\bf Approximation ratio.} 
First, whether we can improve the approximation guarantees to $O(\mathcal{N}^2)$ (i.e., making it independent of the dimension $d \geq 1$) or even $(1 + \eps) \cdot \mathcal{N}^2$? In the discrete settings, what enables the $(1 + \eps)$ approximation algorithms is that the noise spectrum $\hat{g} = \hat{x} - \hat{x}^*$ is on the grid (i.e., the whole spectrum except the top-$k$ frequencies) and the noise $g(t)$ is the sum of the sine/cosine functions with given frequencies. The past works like \cite{ik14,k16} use $\Theta_{d}(k/\eps)$ bins, so the average noise in most of bins is $\Theta_{d}(\eps{\cal N}^2/k)$ each.

But in the continuous model, achieving an $(1+\eps)$ approximation seems difficult, and the hurdles come from the current hashing and sampling methods. First, because the noise $g(t)$ is arbitrary, we cannot hope the noise energy to distribute almost uniformly among the bins after the hashing. Second, as mentioned in Section~\ref{sec:our_technique:permutation_hashing}, the sampling range of the time points $a \in [0, T]^{d}$ cannot be too large. Namely, we can only hope $|\supp(a)| = c \cdot T^d$ for some $0 < c < 1$, limiting the precision to which the frequencies $f_{i}' \approx f_{i}$ can be recovered. These are the two main reasons why, even in the one-dimensional case $d = 1$, the past work \cite{ps15} can only get a constant approximation rather than an $(1 + \eps)$ approximation. (See \cite[Lemma~3.3]{ps15} for more details.)

% First, whether we can improve the approximation guarantees to $O(\mathcal{N}^2)$ (i.e., making it independent of the dimension $d \geq 1$) or even $(1 + \eps) \cdot \mathcal{N}^2$? In the discrete settings, what enables the $(1 + \eps)$-approximation algorithms is that the noise spectrum $\hat{g} = \hat{x} - \hat{x}^*$ is on the grid (i.e., the whole spectrum except the top-$k$ frequencies) and the noise $g(t)$ is the sum of the sine/cosine functions with given frequencies. But in the continuous model, we consider an arbitrary noise $g(t)$, which makes it much harder (or even impossible) to achieve an $(1 + \eps)$-approximation. To make the problem interesting, maybe we need to modify the model itself.

% Achieving $(1+\eps)$ approximation ratio in the continuous case seems very hard or even impossible with current sampling and hashing methods, even in one dimension. Previous works on discrete setting use $O(k/\eps)$ bins, so the average noise in each bin is $O(\eps{\cal N}^2/k)$. However, in our setting, the noise is arbitrarily distributed and we cannot hash the noise independently of the signal, and this is also one of the reasons why \cite{ps15} can only achieve constant ratio. Another loss comes from the sampling. We can only hope $|\supp(a)|$ be $\Theta(T^d)$, but can not be arbitrarily close to $T^d$ even for the much simpler case when $d=1$. This also leads to a loss of ratio in \cite{ps15}, see the Lemma 3.3 therein.

\vspace{.1in}
\noindent
{\bf Deterministic algorithm.}
Actually, no deterministic  sublinear-sample algorithm can achieve the $\ell_{\infty} / \ell_2$-guarantee or the $\ell_2 / \ell_2$-guarantee \cite{dipw10}. But under the (weaker) $\ell_{\infty} / \ell_1$-guarantee, the past works \cite{mzic19,ln19} design an $\wt{O}(k^2)$-sample deterministic algorithm for the discrete Fourier transform; both works reply on the tools from functional analysis. It would be interesting to see a deterministic algorithm for the continuous Fourier transform (even in one dimension $d = 1$).

\vspace{.1in}
\noindent
{\bf $\Tilde{O}(N)$-time algorithms.}
For the discrete model (i.e., recover top-$k$ frequencies out of $N = n^{d}$ ones), several past works improve the sample complexity or other performance guarantees by allowing an $\Tilde{O}(N)$-time Fourier transform (instead of a sublinear-time one). In the multi-dimensional case, the past works \cite{ik14,nsw19} implement the ``point-query'' idea (which originates from the sparse recovery/heavy hitter literature) in a clever way, and thus optimize the sample complexity. Can we obtain such results in the continuous model? The main difficulty is that, different from the discrete cases where the ``on-the-grid'' frequencies can be checked coordinate by coordinate, the ``continuous'' frequencies have infinitely many possibilities.

\vspace{.1in}
\noindent
{\bf Sample complexity.}
As mentioned in the introduction, another potential direction is to reduce the sample complexity. Up to the iterated logarithmic factors, our algorithm {\RecoveryStage} takes $k \cdot (\log k)^{d+O(1)} \cdot \log( F / \eta ) \cdot 2^{O(d \log d)}$ samples/running time. Here the term $\log^{d} k$ is a consequence of our ``{\em precise}'' filter function. As quoted:
\begin{quote}
     \cite{k16,k17} {\em ``in the discrete settings ...\ the price to pay for the precision of the filter, however, is that each hashing becomes a $\log^{d} k$ factor more costly in terms of sample complexity and running time than in the idealized case ...''}
\end{quote}
To shave the $\log^{d} k$ term in the discrete model, the past works \cite{ik14,k16} randomize the noise by using the ``crude'' filters. However, randomizing the noise does not work in the continuous model, since two noise frequencies $f, f' \in \textsc{Tail}$ can be arbitrarily close and, no matter how we randomized the noise, the errors can accumulate in the estimation. 
The exponential dependence on dimension seems to be intrinsic to the current
sampling methods, and avoiding it seems need completely different methods.

\vspace{.1in}
\noindent
{\bf Set query.}
A problem in the ``sparse recovery'' paradigm has two primary tasks: (i)~to recover the heavy locations; and (ii)~to pin down the masses/densities in those locations. Price \cite{p11} pulls the second task out from the sparse recovery literature and defines the ``set query'' problem. Kapralov \cite{k17} introduces and studies the Fourier set query problem in the discrete settings. It would be interesting to explore such problems in the continuous settings.

\subsection{Further related works}
Over the last two decades, the Sparse FT problem has been investigated and extended in various directions. By now we can even say that it constitutes a ``subarea'' within sublinear algorithms.
These former works can be classified into two lines: (i) those in the one-/multi-dimensional discrete settings \cite{hikp12a,hikp12b,ikp14,ik14,k16,k17,nsw19,kvz19,bkm+21} and follow-ups.
(ii) those in the one-dimensional continuous setting \cite{bcgls12,m15,ps15,ckps16} and follow-ups.

Compressed sensing is initiated by \cite{ct06,don06}. Since then, there is a long line of works exploring and extending it in various directions \cite{glps10,glps10,ip11,ipw11,ir13,pw13,price13,agr16,lnnt16,birw16,kp19_soda,ns19}. Compressed sensing allows us to design the sensing matrices, which is the main difference between it and the Sparse FT problem.

Apart from the one-/multi-dimensional discrete/continuous Sparse FT problems that we have considered thus far, where the sampling is carried out in an arbitrary yet {\em non-adaptive} way, there are other meaningful adjustments to the model.

For example, there is (i)~a line of works studying the model where the sampling is conducted in a (more restricted) {\em uniform} way \cite[and the references therein]{rv08,bd08,cgv13,bou14,hr16,bllmr19}; and (ii)~another line of works studying the model that allows an algorithm to {\em adaptively} take the samples and recover the Fourier spectrum \cite[and the references therein]{pw13,cksz17}.

Within theoretical compute science (TCS), Fourier Transform also finds an abundance of applications: integer multiplication \cite{fur09}, \textsc{Subset Sum} and \textsc{3SUM} \cite{clrs09}, %,bri17,kx17},
linear programming \cite{lsz19,blss20,jswz21}, %distributional learning \cite{dks16a,dks16b,dks16c}, 
learning mixture of regressions \cite{cls20}, and fast Johnson-Lindenstrauss transform \cite{ldfu13} etc.

% \newpage

\appendix
\section{Building-block function \texorpdfstring{$(G(t), \wh{G}(f))$}{} in a single dimension}
\label{sec:filter_function_single}

This appendix presents the construction of a basic function $(G(t), \wh{G}(f))$ as well as its properties, which serves as the building block of our single-dimensional filter function (see Appendix~\ref{sec:HashToBins_single}) and multi-dimensional filter function (see Appendix~\ref{sec:HashToBins_multi}).

\subsection{Construction of function \texorpdfstring{$(G(t), \wh{G}(f))$}{}}
\label{subsec:filter_function_single:construction}

To introduce the building-block function $(G(t), \wh{G}(f))$, we will employ the rectangular function $\rect_{s_1}(f)$ and the sinc function $\sinc_{s_1}(t)$. Both functions are widely used in the previous literature, and we shall be familiar with their properties given in Fact~\ref{fac:sinc_function} (e.g.\ see \cite{ckps16}).

\begin{definition}[Two basic functions]
\label{def:rect_sinc}
Given any $s_1 > 0$, the $\rect_{s_1}(f)$ function and the $\sinc_{s_1}(t)$ function are defined as follows:
\begin{itemize}
    \item $\rect_{s_1}(f) = 1 / s_1 \cdot \mathbb{I}\{|f| \leq s_1 / 2\}$ for any $f \in \R$. When $s_1 = 1$, we shorthand it as $\rect(f)$.
    
    \item $\sinc_{s_1}(t) = \frac{\sin(\pi s_1 t)}{\pi s_1 t}$ for any $t \neq 0$ and $\sinc_{s_1}(0) = 1$. When $s_1 = 1$, we shorthand it as $\sinc(t)$.
\end{itemize}
\end{definition}

\begin{fact}[Facts about basic functions {\cite[Appendix~C]{ckps16}}]
\label{fac:sinc_function}
Given any $s_1 > 0$, the following hold for the functions $\sinc_{s_1}(t)$ and $\rect_{s_1}(f)$:
\begin{description}[labelindent = 1em]
    \item [Part~(a):]
    $1 - \frac{\pi^2}{6} \cdot (s_1 t)^2 \leq |\sinc_{s_1}(t)| \leq 1$ for any $t \in \R$.
    
    \item [Part~(b):]
    $|\sinc_{s_1}(t)| \leq 1 - \frac{\pi^2}{8} \cdot (s_1 t)^2$ for any $|t| \leq \frac{2.3}{\pi s_1}$.
    
    \item [Part~(c):]
    $|\sinc_{s_1}(t)| \leq \min (1, \frac{1}{\pi \cdot |s_1 t|})$ for any $t \in \R$.
    
    \item [Part~(d):]
    $\sinc_{s_1}(t) = \hat{\rect_{s_1}}(t)$ for any $t \in \R$, and $\rect_{s_1}(f) = \hat{\sinc_{s_1}}(f)$ for any $f \in \R$.
\end{description}
\end{fact}

Our building-block function $(G(t), \hat{G}(f))$ is constructed in the following Definition~\ref{def:filter_function_single}. This construction is similar to \cite[Definition~C.11]{ckps16}, and we carefully modify the involved parameters for our later use. We present several important properties of $(G(t), \hat{G}(f))$ in Section~\ref{subsec:filter_function_single:property}, and then prove these properties in Section~\ref{subsec:filter_function_single:proof}.

\begin{definition}[Building-block function in a single dimension]
\label{def:filter_function_single}
We set the parameters as follows:
\begin{itemize}
    \item The number of bins in a single dimension $B = \Theta(d \cdot k^{1 / d})$ is a certain multiple of $d \in \mathbb{N}_{\geq 1}$.
    
    \item The noise level parameter $\delta \in (0, 1)$.
    
    \item $\alpha = \Theta(1 / d)$ is chosen such that $\frac{1}{100 \cdot (d + 1) \cdot \alpha} \in \mathbb{N}_{\geq 1}$ is an integer; clearly $\alpha \leq \frac{1}{100 \cdot (d + 1)} \leq \frac{1}{200}$.
    
    \item $s_1 = \frac{2 B}{\alpha}$ and $s_2 = \frac{1}{B + B / d}$.
    
    \item $\ell = \Theta(\log(k d / \delta))$ is an even integer. We safely assume $\ell \geq 1000$.
\end{itemize}
Then for any $t, f \in \R$ the building-block function $(G(t), \hat{G}(f))$ is given by
\begin{eqnarray*}
    G(t)
    & = & s_0 \cdot \rect_{s_1}^{* \ell}(t) \cdot \sinc_{s_2}(t) \\
    & = & s_0 \cdot \rect_{2 B / \alpha}^{* \ell}(t) \cdot \sinc_{1 / (B + B / d)}(t), \\
    \hat{G}(f)
    & = & s_0 \cdot (\sinc_{s_1}(f))^{\cdot \ell} * \rect_{s_2}(f) \\
    & = & s_0 \cdot (\sinc_{2 B / \alpha}(f))^{\cdot \ell} * \rect_{1 / (B + B / d)}(f),
\end{eqnarray*}
where the scalar $s_0 > 0$ achieves the normalization $\hat{G}(0) = 1$. Notice that both $G(t)$ and $\hat{G}(f)$ take real values, and are even functions.
\end{definition}

\subsection{Properties of function \texorpdfstring{$(G(t), \wh{G}(f))$}{}}
\label{subsec:filter_function_single:property}

\begin{lemma}[Building-block function in a single dimension]
\label{lem:filter_function_single}
The function $(G(t), \hat{G}(f))[B, \delta, \alpha, \ell]$ given in Definition~\ref{def:filter_function_single} satisfies the following (as Figure~\ref{fig:filter_function_single} illustrates):
\begin{description}[labelindent = 1em]
    \item [Property~I:]
    The scalar $s_0 \eqsim s_1 s_2 \sqrt{\ell} \eqsim \sqrt{\ell} / \alpha$.
    
    \item [Property~II:]
    $1 - \frac{\delta}{\poly(k, d)} \leq \hat{G}(f) \leq 1$ when $|f| \leq \frac{1 - \alpha}{2 B}$.
    
    \item [Property~III:]
    $\hat{G}(f) \in [0, 1]$ when $\frac{1 - \alpha}{2 B} \leq |f| \leq \frac{1}{2 B}$.
    
    \item [Property~IV:]
    $0 \leq \hat{G}(f) \leq (\pi B f)^{-\ell} \leq \frac{\delta}{\poly(k, d)}$ when $|f| \geq \frac{1}{2 B}$.
    
    \item [Property~V:]
    $\supp(G) \subseteq [-\ell \cdot \frac{B}{\alpha}, \ell \cdot \frac{B}{\alpha}]$.
    
    \item [Property~VI:]
    $\max_{t \in \R} |G(t)| = G(0) \in [\frac{1 - \alpha - \delta / (4 k d)}{B}, \frac{1 + \delta / (4 k d)}{B}]$.
    
    \item [Property~VII:]
    $\sum_{i \in \mathbb{Z}} G(i + 1 / 2)^2 \leq (1 + \frac{\delta}{4 k d})^2 \cdot (1 + \frac{1}{d}) \cdot B^{-1} \lesssim B^{-1}$.
\end{description}
\end{lemma}

\begin{figure}
    \centering
    \includegraphics[width = 0.8 \textwidth]{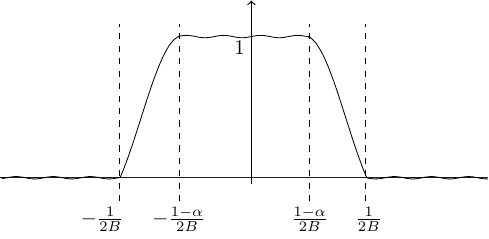}
    \caption{Demonstration for the function $\hat{G}(f)$ in Lemma~\ref{lem:filter_function_single}. }
    \label{fig:filter_function_single}
\end{figure}

\subsection{Proof of properties}
\label{subsec:filter_function_single:proof}

\begin{claim}[Property~I of Lemma~\ref{lem:filter_function_single}]
\label{cla:filter_function:property_1}
The scalar $s_0 \eqsim s_1 s_2 \sqrt{\ell} \eqsim \sqrt{\ell} / \alpha$.
\end{claim}

\begin{proof}
Recall that the scalar $s_0 > 0$ achieves the normalization $\hat{G}(0) = 1$. By definition,
\begin{eqnarray*}
    \hat{G}(0)
    & = & s_0 \cdot \int_{-\infty}^{+\infty} \big(\sinc_{s_1}(\xi)\big)^{\ell} \cdot \rect_{s_2}(0 - \xi) \cdot \d \xi \\
    & = & \frac{s_0}{s_2} \cdot \int_{-s_2 / 2}^{+s_2 / 2} \big(\sinc_{s_1}(\xi)\big)^{\ell} \cdot \d \xi \\
    & = & \frac{2 s_0}{s_2} \cdot \int_{0}^{s_2 / 2} \big(\sinc_{s_1}(\xi)\big)^{\ell} \cdot \d \xi \\
    & = & \frac{2 s_0}{s_2} \cdot \int_{0}^{s_2 / 2} \big|\sinc_{s_1}(\xi)\big|^{\ell} \cdot \d \xi,
\end{eqnarray*}
where the second step follows because $\rect_{s_2}(\xi) = \frac{1}{s_2} \cdot \mathbb{I}\{|\xi| \leq \frac{s_2}{2}\}$ for any $\xi \in \R$ (see Definition~\ref{def:rect_sinc}); the third step follows because $\sinc_{s_1}(\xi)$ is an even function in $\xi \in \R$; and the last step is because $\ell \in \mathbb{N}_{\geq 1}$ is an even integer (see Definition~\ref{def:filter_function_single}).

Given that $s_1 = \frac{2 B}{\alpha}$ and $s_2 = \frac{1}{B + B / d}$ and $0 < \alpha \leq \frac{1}{100 \cdot (d + 1)}$ (see Definition~\ref{def:filter_function_single}), one can easily check that $\frac{2}{\pi s_1} \leq \frac{s_2}{2}$. Accordingly, we know from the additivity of integration that
\begin{eqnarray*}
    \hat{G}(0)
    & = & \frac{2 s_0}{s_2} \cdot \int_{0}^{2 / (\pi s_1)} |\sinc_{s_1}(\xi)|^{\ell} \cdot \d \xi \quad+\quad \frac{2 s_0}{s_2} \cdot \int_{2 / (\pi s_1)}^{s_2 / 2} |\sinc_{s_1}(\xi)|^{\ell} \cdot \d \xi \\
    & = & \frac{2 s_0}{s_2} \cdot \underbrace{\int_{0}^{2 / (\pi s_1)} |\sinc_{s_1}(\xi)|^{\ell} \cdot \d \xi}_{A_1} \quad\pm\quad \frac{2 s_0}{s_2} \cdot \underbrace{\int_{2 / (\pi s_1)}^{+\infty} |\sinc_{s_1}(\xi)|^{\ell} \cdot \d \xi}_{A_2},
\end{eqnarray*}
where the second step follows because $(\sinc_{s_1}(\xi))^{\ell} \geq 0$ for any $\xi \in \R$ (note that $\ell$ is an even integer; see Definition~\ref{def:filter_function_single}), and thus $\hat{G}(0) \cdot \frac{s_2}{2 s_0}$ is bounded between $(A_1 - A_2)$ and $(A_1 + A_2)$.

We will verify respectively in Claims~\ref{cla:filter_function:1} and \ref{cla:filter_function:2} (see Section~\ref{subsec:filter_function_single:facts} for the proofs of both claims) that $A_1 \eqsim \frac{1}{s_1} \cdot \ell^{-1 / 2}$ and $A_2 = O(\frac{1}{s_1} \cdot 2^{-\ell})$. Under our choice of $\ell = \Theta(\log(k d / \delta))$, it follows that $A_1 \gg A_2$ and thus, that
\begin{align*}
    1 = \hat{G}(0) \eqsim \frac{2 s_0}{s_2} \cdot A_1 \eqsim \frac{s_0}{s_1 s_2} \cdot \ell^{-1 / 2},
\end{align*}
which implies $s_0 \eqsim s_1 s_2 \sqrt{\ell} \eqsim \sqrt{\ell} / \alpha$ (since $s_1 = \frac{2 B}{\alpha}$ and $s_2 = \frac{1}{B + B / d}$).

This completes the proof of Claim~\ref{cla:filter_function:property_1}.
\end{proof}

\begin{claim}[Property~II of Lemma~\ref{lem:filter_function_single}]
\label{cla:filter_function:property_2}
$1 - \frac{\delta}{\poly(k, d)} \leq \hat{G}(f) \leq 1$ when $|f| \leq \frac{1 - \alpha}{2 B}$.
\end{claim}

\begin{proof}
We first prove the upper-bound part that $\hat{G}(f) \leq \hat{G}(0) = 1$ for any $f \in \R$. Since $\hat{G}(f)$ is an even function (see Definition~\ref{def:filter_function_single}), it suffices to deal with the case that $f \geq 0$. By definition,
\begin{eqnarray}
    \notag
    \hat{G}(f) - \hat{G}(0)
    & = & s_0 \cdot \int_{-\infty}^{+\infty} (\sinc_{s_1}(\xi))^{\ell} \cdot \big(\rect_{s_2}(f - \xi) - \rect_{s_2}(0 - \xi)\big) \cdot \d \xi \\
    \notag
    & = & \frac{s_0}{s_2} \cdot \int_{-s_2 / 2 + f}^{s_2 / 2 + f} (\sinc_{s_1}(\xi))^{\ell} \cdot \d \xi
    - \frac{s_0}{s_2} \cdot \int_{-s_2 / 2}^{s_2 / 2} (\sinc_{s_1}(\xi))^{\ell} \cdot \d \xi \\
    \notag
    & = & \frac{s_0}{s_1 s_2} \cdot \int_{-s_1 s_2 / 2 + s_1 f}^{s_1 s_2 / 2 + s_1 f} (\sinc(\xi))^{\ell} \cdot \d \xi
    - \frac{s_0}{s_1 s_2} \cdot \int_{-s_1 s_2 / 2}^{s_1 s_2 / 2} (\sinc(\xi))^{\ell} \cdot \d \xi \\
    \notag
    & = & \frac{s_0}{s_1 s_2} \cdot \left(\int_{s_1 s_2 / 2}^{s_1 s_2 / 2 + s_1 f} - \int_{-s_1 s_2 / 2}^{-s_1 s_2 / 2 + s_1 f}\right) (\sinc(\xi))^{\ell} \cdot \d \xi \\
    \notag
    & = & \frac{s_0}{s_1 s_2} \cdot \left(\int_{s_1 s_2 / 2}^{s_1 s_2 / 2 + s_1 f} - \int_{s_1 s_2 / 2 - s_1 f}^{s_1 s_2 / 2}\right) (\sinc(\xi))^{\ell} \cdot \d \xi \\
    \label{eq:filter_function:property:1}
    & = & \frac{s_0}{s_1 s_2} \cdot \int_{0}^{s_1 f} \underbrace{\left((\sinc(\frac{s_1 s_2}{2} + \xi))^{\ell} - (\sinc(\frac{s_1 s_2}{2} - \xi))^{\ell}\right)}_{A_3} \cdot \d \xi,
\end{eqnarray}
where the second step follows because $\rect_{s_2}(\xi) = \frac{1}{s_2} \cdot \mathbb{I}\{|\xi| \leq \frac{s_2}{2}\}$ for any $\xi \in \R$ (see Definition~\ref{def:rect_sinc}); the third step is by substitution; the fourth step applies the additivity of integration; the fifth step follows because $\sinc(\xi)$ is an even function in $\xi \in \R$; and the last step is by substitution.

Given Equation~\eqref{eq:filter_function:property:1}, it suffices to show that $A_3 \leq 0$ when $\xi \in [0, s_1 f]$. Recall that $\ell \in \mathbb{N}_{\geq 1}$ is an even integer, and $\frac{s_1 s_2}{2} = (\frac{2 B}{\alpha}) \cdot (\frac{1}{B + B / d}) \cdot \frac{1}{2} = \frac{1}{100 \cdot (d + 1) \cdot \alpha} \cdot (100 \cdot d) \in \mathbb{N}_{\geq 1}$ is an integer (see Definition~\ref{def:filter_function_single}). Therefore, for any $\xi \in [0, s_1 f]$ we have
\begin{eqnarray*}
    A_3
    & = & \left|\sinc(\frac{s_1 s_2}{2} + \xi)\right|^{\ell} - \left|\sinc(\frac{s_1 s_2}{2} - \xi)\right|^{\ell} \\
    & = & \frac{\left|\sin(\pi \cdot \frac{s_1 s_2}{2} + \pi \cdot \xi)\right|^{\ell}}{\left|\pi \cdot \frac{s_1 s_2}{2} + \pi \cdot \xi\right|^{\ell}} - \frac{\left|\sin(\pi \cdot \frac{s_1 s_2}{2} - \pi \cdot \xi)\right|^{\ell}}{\left|\pi \cdot \frac{s_1 s_2}{2} - \pi \cdot \xi\right|^{\ell}} \\
    & = & \frac{\left|\sin(\pi \cdot \xi)\right|^{\ell}}{\left|\pi \cdot \frac{s_1 s_2}{2} + \pi \cdot \xi\right|^{\ell}} - \frac{\left|\sin(-\pi \cdot \xi)\right|^{\ell}}{\left|\pi \cdot \frac{s_1 s_2}{2} - \pi \cdot \xi\right|^{\ell}} \\
    & = & \left|\sin(\pi \cdot \xi)\right|^{\ell} \cdot \left(\left|\pi \cdot \frac{s_1 s_2}{2} + \pi \cdot \xi\right|^{-\ell} - \left|\pi \cdot \frac{s_1 s_2}{2} - \pi \cdot \xi\right|^{-\ell}\right) \\
    & \leq & 0,
\end{eqnarray*}
where the first step follows because $\ell \in \mathbb{N}_{\geq 1}$ is an even integer; the third step follows because $|\sin(\pi \cdot \xi)|$ is a periodic function in $\xi \in \R$ and its basic period is $1$ (notice that $\frac{s_1 s_2}{2}$ is an integer); and the fourth step follows because $|\sin(\pi \cdot \xi)|$ is a even function.

To see the lower-bound part, due to the normalization $\hat{G}(0) = 1$, we have
\begin{eqnarray*}
    1 - \hat{G}(f)
    & = & \hat{G}(0) - \hat{G}(f) \\
    & = & \frac{s_0}{s_1 s_2} \cdot \int_{0}^{s_1 f} \left(\left(\sinc\left(\frac{s_1 s_2}{2} + \xi\right)\right)^{\ell} - \left(\sinc\left(\frac{s_1 s_2}{2} - \xi\right)\right)^{\ell}\right) \cdot \d \xi \\
    & \leq & \frac{s_0}{s_1 s_2} \cdot \int_{0}^{s_1 f} \left|\sinc\left(\frac{s_1 s_2}{2} + \xi\right)\right|^{\ell} \cdot \d \xi \\
    & \leq & \frac{s_0}{s_1 s_2} \cdot \int_{0}^{s_1 f} \frac{1}{\pi^{\ell} \cdot |s_1 s_2 / 2 + \xi|^{\ell}} \cdot \d \xi \\
    & \leq & \frac{s_0}{s_1 s_2} \cdot s_1 f \cdot \frac{1}{\pi^{\ell} \cdot (s_1 s_2 / 2)^{\ell}} \\
    & = & s_0 \cdot (B + B / d) \cdot f \cdot \left(\frac{(d + 1) \cdot \alpha}{\pi \cdot d}\right)^{\ell} \\
    & \leq & s_0 \cdot 2 B \cdot f \cdot \left(\frac{(d + 1) \cdot \alpha}{\pi \cdot d}\right)^{\ell},
\end{eqnarray*}
where the second step follows from Equation~\eqref{eq:filter_function:property:1}; the third step follows because $\ell$ is an even integer (see Definition~\ref{def:filter_function_single}), namely $(\sinc(\xi))^{\ell} \geq 0$ for any $\xi \in \R$; the fourth step is by Part~(c) of Fact~\ref{fac:sinc_function}; and the sixth step is by $s_1 = \frac{2 B}{\alpha}$ and $s_2 = \frac{1}{B + B / d}$ (see Definition~\ref{def:filter_function_single}).

According to Claim~\ref{cla:filter_function:property_1}, for some universal constant $C_0 > 0$, we have $s_0 \leq C_0 \cdot \sqrt{\ell} / \alpha$. Also, as promised by the concerning claim, $|f| \leq \frac{1 - \alpha}{2 B} \leq \frac{1}{2 B}$. Plugging these into the above inequality:
\begin{eqnarray*}
    1 - \hat{G}(f)
    & \leq & \left(C_0 \cdot \sqrt{\ell} / \alpha\right) \cdot 2 B \cdot \frac{1}{2 B} \cdot \left(\frac{(d + 1) \cdot \alpha}{\pi \cdot d}\right)^{\ell} \\
    & = & C_0 \cdot \sqrt{\ell} \cdot \frac{d + 1}{\pi \cdot d} \cdot \left(\frac{(d + 1) \cdot \alpha}{\pi \cdot d}\right)^{\ell - 1} \\
    & \leq & C_0 \cdot \sqrt{\ell} \cdot \left(\frac{(d + 1) \cdot \alpha}{\pi \cdot d}\right)^{\ell - 1} \\
    & \leq & C_0 \cdot \sqrt{\ell} \cdot \left(\frac{1}{100 \pi \cdot d}\right)^{\ell - 1} \\
    & \leq & \frac{\delta}{\poly(k, d)},
\end{eqnarray*}
where the third step follows because $\frac{d + 1}{\pi \cdot d} \leq \frac{2}{\pi} \leq 1$; the fourth step follows because $0 < \alpha \leq \frac{1}{100 \cdot (d + 1)}$ (see Definition~\ref{def:filter_function_single}); and the last step holds for any large enough $\ell = \Theta(\log(k d / \delta))$.

This completes the proof of Claim~\ref{cla:filter_function:property_2}.
\end{proof}

\begin{claim}[Property~III of Lemma~\ref{lem:filter_function_single}]
\label{cla:filter_function:property_3}
$\hat{G}(f) \in [0, 1]$ when $\frac{1 - \alpha}{2 B} \leq |f| \leq \frac{1}{2 B}$.
\end{claim}

\begin{proof}
The upper-bound part has been shown in the proof of Claim~\ref{cla:filter_function:property_2}, namely $\hat{G}(f) \leq \hat{G}(0) = 1$ for any $f \in \R$. The lower-bound part is trivial, since both functions $(\sinc_{s_1}(f))^{\ell}$ and $\rect_{s_2}(f)$ are nonnegative (note that $\ell \in \mathbb{N}_{\geq 1}$ is an even integer; see Definition~\ref{def:filter_function_single}).

This completes the proof of Claim~\ref{cla:filter_function:property_3}.
\end{proof}

\begin{claim}[Property~IV of Lemma~\ref{lem:filter_function_single}]
\label{cla:filter_function:property_4}
$0 \leq \hat{G}(f) \leq (\pi B f)^{-\ell} \leq \frac{\delta}{\poly(k, d)}$ when $|f| \geq \frac{1}{2 B}$.
\end{claim}

\begin{proof}
The lower-bound part has been shown in the proof of Claim~\ref{cla:filter_function:property_3}, namely $\ell \in \mathbb{N}_{\geq 1}$ is an even integer (see Definition~\ref{def:filter_function_single}) and thus both functions $(\sinc_{s_1}(f))^{\ell}$ and $\rect_{s_2}(f)$ are nonnegative.

For the upper-bound part, since $\hat{G}(f)$ is an even function, it suffices to handle the case $f \geq \frac{1}{2 B}$. By definition,
\begin{eqnarray*}
    \hat{G}(f)
    & = & s_0 \cdot \int_{-\infty}^{+\infty} (\sinc_{s_1}(\xi))^{\ell} \cdot \rect_{s_2}(f - \xi) \cdot \d \xi \\
    & = & \frac{s_0}{s_2} \cdot \int_{f - s_2 / 2}^{f + s_2 / 2} (\sinc_{s_1}(\xi))^{\ell} \cdot \d \xi \\
    & = & \frac{s_0}{s_2} \cdot \int_{f - s_2 / 2}^{f + s_2 / 2} |\sinc_{s_1}(\xi)|^{\ell} \cdot \d \xi \\
    & \leq & \frac{s_0}{s_2} \cdot \int_{f - s_2 / 2}^{f + s_2 / 2} \frac{1}{\pi^{\ell} \cdot |s_1 \xi|^{\ell}} \cdot \d \xi,
\end{eqnarray*}
where the second step follows because $\rect_{s_2}(\xi) = \frac{1}{s_2} \cdot \mathbb{I}\{|\xi| \leq \frac{s_2}{2}\}$ for any $\xi \in \R$ (see Definition~\ref{def:rect_sinc}); the third step follows because $\ell \in \mathbb{N}_{\geq 1}$ is an even integer (see Definition~\ref{def:filter_function_single}); and the last step is by Part~(c) of Fact~\ref{fac:sinc_function}.

Recall Definition~\ref{def:filter_function_single} that $s_2 = \frac{1}{B + B / d}$. Given this and since we assume $f \geq \frac{1}{2 B}$, one can easily check that the above interval of integral is lower bounded by $f - s_2 / 2 \geq f / (d + 1)$. Hence,
\begin{eqnarray*}
    \hat{G}(f)
    & \leq & \frac{s_0}{s_2} \cdot s_2 \cdot \left.\frac{1}{\pi^{\ell} \cdot |s_1 \xi|^{\ell}}\right|_{\xi = f / (d + 1)} \\
    & = & s_0 \cdot \left(\frac{(d + 1) \cdot \alpha}{2 \pi B f}\right)^{\ell},
\end{eqnarray*}
where the second step follows because $s_1 = \frac{2 B}{\alpha}$ and $s_2 = \frac{1}{B + B / d}$. According to Claim~\ref{cla:filter_function:property_1}, for some universal constant $C_0 > 0$, we have $s_0 \leq C_0 \sqrt{\ell} / \alpha$. As a consequence,
\begin{eqnarray*}
    \hat{G}(f)
    & \leq & C_0 \cdot \frac{\sqrt{\ell}}{\alpha} \cdot \left(\frac{(d + 1) \cdot \alpha}{2 \pi B f}\right)^{\ell} \\
    & \leq & C_0 \cdot 100 \cdot (d + 1) \cdot \sqrt{\ell} \cdot (200 \pi B f)^{-\ell} \\
    & \leq & (\pi B f)^{-\ell},
\end{eqnarray*}
where the second step follows because, given that $\ell \geq 1000$, the concerning formula $C_0 \cdot \frac{\sqrt{\ell}}{\alpha} \cdot (\frac{(d + 1) \cdot \alpha}{2 \pi B f})^{\ell}$ is an increasing function when $0 < \alpha \leq \frac{1}{100 \cdot (d + 1)}$ (see Definition~\ref{def:filter_function_single}); and the last step, which is equivalent to $\frac{C_0 \cdot 100 \cdot (d + 1) \cdot \sqrt{\ell}}{200^{\ell}} \leq 1$, holds for any large enough $\ell = \Theta(\log(k d / \delta))$.

Following the above calculation, for any $f \geq \frac{1}{2 B}$ we have
\begin{align*}
    \hat{G}(f)
    \leq (\pi B f)^{-\ell}
    \leq (\pi / 2)^{-\ell}
    \leq \frac{\delta}{\poly(k, d)},
\end{align*}
where the last step holds for any large enough $\ell = \Theta(\log(k d / \delta))$.

This completes the proof of Claim~\ref{cla:filter_function:property_4}.
\end{proof}

\begin{claim}[Property~V of Lemma~\ref{lem:filter_function_single}]
\label{cla:filter_function:property_5}
$\supp(G) \subseteq [-\ell \cdot \frac{B}{\alpha}, \ell \cdot \frac{B}{\alpha}]$.
\end{claim}

\begin{proof}
Recall that $s_1 = \frac{2 B}{\alpha}$. By definition, the function $\rect_{s_1}(t) = \frac{1}{s_1} \cdot \mathbb{I}\{|t| \leq \frac{s_1}{2}\}$ is supported on the interval $t \in [-\frac{s_1}{2}, \frac{s_1}{2}]$, and thus $\rect_{s_1}^{* \ell}(t)$ is supported on $t \in [-\ell \cdot \frac{s_1}{2}, \ell \cdot \frac{s_1}{2}] = [-\ell \cdot \frac{B}{\alpha}, \ell \cdot \frac{B}{\alpha}]$. Clearly, the later interval contains the support of the function $G(t) = s_0 \cdot \rect_{s_1}^{* \ell}(t) \cdot \sinc_{s_2}(t)$.

This completes the proof of Claim~\ref{cla:filter_function:property_5}.
\end{proof}

\begin{claim}[Property~VI of Lemma~\ref{lem:filter_function_single}]
\label{cla:filter_function:property_6}
$\max_{t \in \R} |G(t)| = G(0) \in [\frac{1 - \alpha - \delta / (4 k d)}{B}, \frac{1 + \delta / (4 k d)}{B}]$.
\end{claim}

\begin{proof}
Observe that $\hat{G}(\xi) = s_0 \cdot (\sinc_{s_1}(f))^{\cdot \ell} * \rect_{s_2}(f)$ is an even function in $\xi \in \R$, since both $\sinc_{s_1}(\xi)$ and $\rect_{s_2}(\xi)$ are even functions.

We first prove that $\max_{t \in \R} |G(t)| = G(0)$. By the definition of the inverse {\CFT},
\begin{align*}
\int_{-\infty}^{+\infty} \hat{G}(\xi) \cdot \d \xi = G(0) \leq \max_{t \in \R} G(t) \leq \max_{t \in \R} \left|G(t)\right|.
\end{align*}

Also, for any $t \in \R$ we can derive $G(t)$ from $\hat{G}(f)$ via the inverse {\CFT}:
\begin{eqnarray*}
    \left|G(t)\right|
    & = & \left| \int_{-\infty}^{+\infty} \hat{G}(\xi) \cdot e^{2 \pi \i t \cdot \xi} \cdot \d \xi \right| \\
    & = & \left| \int_{-\infty}^{+\infty} \hat{G}(\xi) \cdot \big(\cos(2 \pi t \cdot \xi) + \i \cdot \sin(2 \pi t \cdot \xi)\big) \cdot \d \xi \right| \\
    & = & \left| \int_{-\infty}^{+\infty} \hat{G}(\xi) \cdot \cos(2 \pi t \cdot \xi) \cdot \d \xi \right| \\
    & \leq & \int_{-\infty}^{+\infty} \left| \hat{G}(\xi) \right| \cdot \big| \cos(2 \pi t \cdot \xi) \big| \cdot \d \xi \\
    & \leq & \int_{-\infty}^{+\infty} \left| \hat{G}(\xi) \right| \cdot \d \xi \\
    & = & \int_{-\infty}^{+\infty} \hat{G}(\xi) \cdot \d \xi
\end{eqnarray*}
where the third step follows because $\hat{G}(\xi)$ is an even function in $\xi \in \R$ (see Definition~\ref{def:filter_function_single}), whereas $\sin(2 \pi t \cdot \xi)$ is an odd function; the fifth step is because $|\cos(2 \pi t \cdot \xi)| \leq 1$ for any $\xi \in \R$; and the last step follows as $\hat{G}(\xi) \geq 0$ for any $\xi \in \R$ (see Claims~\ref{cla:filter_function:property_2}, \ref{cla:filter_function:property_3} and \ref{cla:filter_function:property_4}).

We conclude from the above that
\begin{align}
    \label{eq:cla:filter_function:property_6}
    \max_{t \in \R} |G(t)|
    = \int_{-\infty}^{+\infty} \hat{G}(\xi) \cdot \d \xi
    = 2 \cdot \int_{0}^{+\infty} \hat{G}(\xi) \cdot \d \xi
    = 2 \cdot (A_4 + A_5 + A_6),
\end{align}
where the second step follows as $\hat{G}(\xi)$ is an even function in $\xi \in \R$; and for the third step we denote the terms $A_4$ and $A_5$ and $A_6$ as follows:
\begin{eqnarray*}
    A_4 & = & \int_{0}^{(1 - \alpha) / (2 B)} \hat{G}(\xi) \cdot \d \xi, \\
    A_5 & = & \int_{(1 - \alpha) / (2 B)}^{1 / (2 B)} \hat{G}(\xi) \cdot \d \xi, \\
    A_6 & = & \int_{1 / (2 B)}^{+\infty} \hat{G}(\xi) \cdot \d \xi.
\end{eqnarray*}

Let us quantify the three terms $A_4$ and $A_5$ and $A_6$ respectively:
\begin{itemize}
    \item $A_4 \in [\frac{1 - \alpha - \delta / (4 k d)}{2 B}, \frac{1 - \alpha}{2 B}]$. This is because $1 - \frac{\delta}{4 k d} \leq \hat{G}(\xi) \leq 1$ for any $\xi \in [0, \frac{1 - \alpha}{2 B}]$ (Claim~\ref{cla:filter_function:property_2}); we shall notice that $0 < \alpha \leq \frac{1}{100 \cdot (d + 1)} < 1$ and that $B > 1$ (see Definition~\ref{def:filter_function_single}).
    
    \item $A_5 \in [0, \frac{\alpha}{2 B}]$. This is because $\hat{G}(\xi) \in [0, 1]$ for any $\xi \in [\frac{1 - \alpha}{2 B}, \frac{1}{2 B}]$ (see Claim~\ref{cla:filter_function:property_3}).
    
    \item $A_6 \in [0, \frac{\delta / (4 k d)}{2 B}]$. Based on Claim~\ref{cla:filter_function:property_4}, we have $0 \leq \hat{G}(\xi) \leq (\pi B \xi)^{-\ell}$ for any $\xi \geq \frac{1}{2 B}$. Then the lower-bound part $A_6 \geq 0$ follows immediately. For the upper-bound part, we have
    \begin{eqnarray*}
        A_6
        & \leq & \int_{1 / (2 B)}^{+\infty} (\pi B \xi)^{-\ell} \cdot \d \xi \\
        & = & \frac{1}{\pi B} \cdot \int_{\pi / 2}^{+\infty} \xi^{-\ell} \cdot \d \xi \\
        & = & \frac{1}{2 B} \cdot \frac{1}{\ell - 1} \cdot (\pi / 2)^{-\ell} \\
        & \leq & \frac{\delta / (4 k d)}{2 B},
    \end{eqnarray*}
    where the second step is by substitution; the third step is by elementary calculation; and the last step, given Definition~\ref{def:filter_function_single}, holds for any large enough $\ell = \Theta(\log(k d / \delta))$.
\end{itemize}

Applying the above bounds to Equation~\eqref{eq:cla:filter_function:property_6} completes the proof of Claim~\ref{cla:filter_function:property_6}.
\end{proof}

\begin{claim}[Property~VII of Lemma~\ref{lem:filter_function_single}]
\label{cla:filter_function:property_7}
$\sum_{i \in \mathbb{Z}} G(i + 1 / 2)^2 \leq (1 + \frac{\delta}{4 k d})^2 \cdot (1 + \frac{1}{d}) \cdot B^{-1} \lesssim B^{-1}$.
\end{claim}

\begin{proof}
We first prove by induction that $\rect_{s_1}^{* \ell}(t)$ is an even function and is non-increasing for any $t \geq 0$. Obviously, $\rect_{s_1}(t)$ itself meets the both properties. Given any $\ell' < \ell$, w.l.o.g.\ we assume $\rect_{s_1}^{* \ell'}(t)$ to satisfy the two properties as well. Then for any $t \in \R$, it follows that
\begin{eqnarray*}
    \rect_{s_1}^{* \ell' + 1}(-t)
    & = & \int_{-\infty}^{+\infty} \rect_{s_1}(-t - \tau) \cdot \rect_{s_1}^{* \ell'}(\tau) \cdot \d \tau \\
    & = & \int_{-\infty}^{+\infty} \rect_{s_1}(t + \tau) \cdot \rect_{s_1}^{* \ell'}(-\tau) \cdot \d \tau \\
    & = & \rect_{s_1}^{* \ell' + 1}(t),
\end{eqnarray*}
namely $\rect_{s_1}^{* \ell' + 1}(t)$ is also an even function. In addition, for any $t' \geq t \geq 0$ we have
\begin{eqnarray*}
    \rect_{s_1}^{* \ell' + 1}(t') - \rect_{s_1}^{* \ell' + 1}(t)
    & = & \int_{-\infty}^{+\infty} \rect_{s_1}(t' - \tau) \cdot \rect_{s_1}^{* \ell'}(\tau) \cdot \d \tau \\
    & & \hspace{1cm} - \int_{-\infty}^{+\infty} \rect_{s_1}(t - \tau) \cdot \rect_{s_1}^{* \ell'}(\tau) \cdot \d \tau \\
    & = & 1 / s_1 \cdot \int_{t' - s_2 / 2}^{t' + s_2 / 2} \rect_{s_1}^{* \ell'}(\tau) \cdot \d \tau
    - 1 / s_1 \cdot \int_{t - s_2 / 2}^{t + s_2 / 2} \rect_{s_1}^{* \ell'}(\tau) \cdot \d \tau \\
    & = & 1 / s_1 \cdot \int_{t + s_2 / 2}^{t' + s_2 / 2} \rect_{s_1}^{* \ell'}(\tau) \cdot \d \tau
    - 1 / s_1 \cdot \int_{t - s_2 / 2}^{t' - s_2 / 2} \rect_{s_1}^{* \ell'}(\tau) \cdot \d \tau \\
    & = & 1 / s_1 \cdot \int_{t}^{t'} \left(\rect_{s_1}^{* \ell'}(\tau + s_2 / 2) - \rect_{s_1}^{* \ell'}(\tau - s_2 / 2)\right) \cdot \d \tau \\
    & \leq & 0,
\end{eqnarray*}
where the second step follows since $\rect_{s_1}(\xi) = \frac{1}{s_1} \cdot \mathbb{I}\{|\xi| \leq \frac{s_1}{2}\}$ for any $\xi \in \R$ (see Definition~\ref{def:rect_sinc}); the third step is by the additivity of integration; the fourth step is by substitution; and the last step uses our induction hypotheses that $\rect_{s_1}^{* \ell'}(t)$ is an even function and is non-increasing when $t \geq 0$. Thus, $\rect_{s_1}^{* \ell' + 1}(t)$ also meets the properties, and our claim follows by induction.

Further, it is easy to see that $\rect_{s_1}^{* \ell}(t)$ is a non-negative function. Put everything together:
\begin{eqnarray*}
    \max_{\tau \in \R} \big\{ s_0^2 \cdot \rect_{s_1}^{* \ell}(\tau)^2 \big\}
    & = & s_0^2 \cdot \rect_{s_1}^{* \ell}(0)^2 \\
    & = & G(0)^2 / \sinc_{s_2}(0)^2 \\
    & = & G(0)^2 \\
    & \leq & \big(1 + \delta / (4 k d)\big)^2 \cdot B^{-2},
\end{eqnarray*}
where the third step follows because $\sinc_{s_1}(0) = 1$ (see Definition~\ref{def:rect_sinc}); and the fourth step follows from Claim~\ref{cla:filter_function:property_6}.

For any $t \in \R$, we infer from the above that
\begin{eqnarray*}
    G(t)^2
    & = & s_0^2 \cdot \rect_{s_1}^{* \ell}(t)^2 \cdot \sinc_{s_2}(t)^2 \\
    & \leq & \max_{\tau \in \R} \big\{ s_0^2 \cdot \rect_{s_1}^{* \ell}(\tau)^2 \big\} \cdot \sinc_{s_2}(t)^2 \\
    & \leq & \big(1 + \delta / (4 k d)\big)^2 \cdot B^{-2} \cdot \sinc_{s_2}(t)^2
\end{eqnarray*}

Further, given that $s_2 = \frac{1}{B + B / d} < 1$ (see Definition~\ref{def:filter_function_single}), we have
\begin{eqnarray*}
    \sum_{i \in \mathbb{Z}} G(i + 1 / 2)^2
    & \leq & \big(1 + \delta / (4 k d)\big)^2 \cdot B^{-2} \cdot \sum_{i \in \mathbb{Z}} \sinc_{1 / (B + B / d)}(i + 1 / 2)^2 \\
    & = & \big(1 + \delta / (4 k d)\big)^2 \cdot B^{-2} \cdot (B + B / d) \\
    & = & \big(1 + \delta / (4 k d)\big)^2 \cdot (1 + 1 / d) \cdot B^{-1},
\end{eqnarray*}
where the second step, which is equivalent to $\sum_{i = 0}^{+\infty} \sinc_{1 / (B + B / d)}(i + 1 / 2)^2 = B + B / d$, can be directly inferred from \cite[Equation~(1)]{boas1973continuous}.

This completes the proof of Claim~\ref{cla:filter_function:property_7}.
\end{proof}

\subsection{Construction and properties of standard window function \texorpdfstring{$(G'(t), \hat{G'}(f))$}{}}
\label{subsec:filter_function_single:window}

We associate the building-block function with the standard window function $(G'(t), \hat{G'}(f))[B, \delta, \alpha, \ell]$ (similar to the ones used in \cite{hikp12a,hikp12b}), which is more convenient for our later use.

\begin{lemma}[Standard window function in a single dimension]
\label{lem:window_function_single}
Consider the building-block function $(G(t), \hat{G}(f))[B, \delta, \alpha, \ell]$ given in Definition~\ref{def:filter_function_single}, there exists another function $(G'(t), \hat{G'}(f))$ such that:
\begin{description}[labelindent = 1em]
    \item [Property~I:]
    $\hat{G'}(f) = 1$ when $|f| \leq \frac{1 - \alpha}{2 B}$.
    
    \item [Property~II:]
    $\hat{G'}(f) \in [0, 1]$ when $\frac{1 - \alpha}{2 B} \leq |f| \leq \frac{1}{2 B}$.
    
    \item [Property~III:]
    $\hat{G'}(f) = 0$ when $|f| \geq \frac{1}{2 B}$.
    
    \item [Property~IV:]
    $\|\hat{G'} - \hat{G}\|_{\infty} = \max_{f \in \R} |\hat{G'}(f) - \hat{G}(f)| \leq \frac{\delta}{\poly(k, d)}$.
\end{description}
\end{lemma}

\begin{proof}
We define $\hat{G'}(f)$ as follows; note that, similar to $\hat{G}(f)$, this is also an even function:
\begin{eqnarray*}
    \hat{G'}(f) & = & 
    \begin{cases}
    1, & \forall |f| \leq \frac{1 - \alpha}{2 B}; \\
    \hat{G}(f), & \forall |f| \in \left(\frac{1 - \alpha}{2 B}, \frac{1}{2 B}\right]; \\
    0, & \forall |f| > \frac{1}{2 B}.
    \end{cases}
\end{eqnarray*}
By construction, Properties~I and III follows directly. Further, Property~II follows from Property~III of Lemma~\ref{lem:filter_function_single}, and Property~IV follows from Properties~II to IV of Lemma~\ref{lem:filter_function_single}.

This completes the proof of Lemma~\ref{lem:window_function_single}.
\end{proof}

\subsection{Facts}
\label{subsec:filter_function_single:facts}

The following facts are helpful in proving Claim~\ref{cla:filter_function:property_1}.

\begin{claim}
\label{cla:filter_function:1}
$\displaystyle{\int}_{0}^{2 / (\pi s_1)} (\sinc_{s_1}(\xi))^{\ell} \cdot \d \xi \eqsim \frac{1}{s_1} \cdot \ell^{-1 / 2}$.
\end{claim}

\begin{proof}
Let $i^* = \lceil2 / \pi \cdot \sqrt{\ell / 8}\rceil - 1$. We safely that assume $\ell = \Theta(\log (k d / \delta))$ is an integer larger than $1000$ (see Definition~\ref{def:filter_function_single}), which guarantees the following facts:
\begin{enumerate}[label = (\alph*):, font = {\bfseries}]
    \item The integrand $(\sinc_{s_1}(\xi))^{\ell} \geq 0$ for any $\xi \in \R$;
    
    \item $i^* \geq 2 / \pi \cdot \sqrt{\ell / 8} - 1 \geq 2 / \pi \cdot \sqrt{1000 / 8} - 1 \approx 6.118 \geq 6$.
    
    \item $i^* \cdot \sqrt{8 / \ell} \leq 2 / \pi \cdot \sqrt{\ell / 8} \cdot \sqrt{8 / \ell} = 2 / \pi$;
    
    \item $(i^* + 1) \cdot \sqrt{8 / \ell} \geq 2 / \pi \cdot \sqrt{\ell / 8} \cdot \sqrt{8 / \ell} = 2 / \pi$; and
    
    \item $(i^* + 1) \cdot \sqrt{8 / \ell} \leq (2 / \pi \cdot \sqrt{\ell / 8} + 1) \cdot \sqrt{8 / \ell} \leq 2 / \pi + \sqrt{8 / 1000} \approx \frac{2.281}{\pi} \leq \frac{2.3}{\pi}$.
\end{enumerate}
These facts are useful in proving the current claim.

For the upper-bound part, we have
\begin{eqnarray}
    \notag
    \int_{0}^{2 / (\pi s_1)} (\sinc_{s_1}(\xi))^{\ell} \cdot \d \xi
    & = & \frac{1}{s_1} \cdot \int_{0}^{2 / \pi} (\sinc(\xi))^{\ell} \cdot \d \xi \\
    \notag
    & = & \frac{1}{s_1} \cdot \int_{0}^{2 / \pi} |\sinc(\xi)|^{\ell} \cdot \d \xi \\
    \notag
    & \leq & \frac{1}{s_1} \cdot \int_{0}^{(i^* + 1) \cdot \sqrt{8 / \ell}} |\sinc(\xi)|^{\ell} \cdot \d \xi \\
    \label{eq:cla:filter_function:1:1}
    & = & \frac{1}{s_1} \cdot \sum_{i = 0}^{i^*} \int_{i \cdot \sqrt{8 / \ell}}^{(i + 1) \cdot \sqrt{8 / \ell}} |\sinc(\xi)|^{\ell} \cdot \d \xi,
\end{eqnarray}
where the first step is by substitution; the second step by because $\ell \in \NP$ is an even integer (see Definition~\ref{def:filter_function_single}); the third step follows from the above Fact~(d); and the last step follows from the additivity of integration.

Given the above Fact~(e), the whole interval of integral $\xi \in [0, (i^* + 1) \cdot \sqrt{8 / \ell}]$ is a subset of $\xi \in [0, \frac{2.3}{\pi}]$, namely Parts~(b) of Fact~\ref{fac:sinc_function} is applicable here. In particular, each $i$-th summand in Equation~\eqref{eq:cla:filter_function:1:1} equals
\begin{eqnarray}
    \notag
    \int_{i \cdot \sqrt{8 / \ell}}^{(i + 1) \cdot \sqrt{8 / \ell}} |\sinc(\xi)|^{\ell} \cdot \d \xi
    & \leq & \int_{i \cdot \sqrt{8 / \ell}}^{(i + 1) \cdot \sqrt{8 / \ell}} \left(1 - \pi^2 / 8 \cdot \xi^2\right)^{\ell} \cdot \d \xi \\
    \notag
    & \leq & \int_{i \cdot \sqrt{8 / \ell}}^{(i + 1) \cdot \sqrt{8 / \ell}} \left(1 - \pi^2 i^2 / \ell\right)^{\ell} \cdot \d \xi \\
    \notag
    & = & \sqrt{8 / \ell} \cdot \left(1 - \pi^2 i^2 / \ell\right)^{\ell} \\
    \label{eq:cla:filter_function:1:2}
    & \leq & \sqrt{8 / \ell} \cdot e^{-\pi^2 \cdot i^2},
\end{eqnarray}
where the first step is follows from Parts~(b) of Fact~\ref{fac:sinc_function}; the second step is because $1 - \pi^2 / 8 \cdot t^2 \leq (1 - \pi^2 / 8 \cdot t^2)|_{t = i \cdot \sqrt{8 / \ell}} = 1 - \pi^2 i^2 / \ell$; and the last step is by $0 \leq 1 - \frac{1}{x} \leq e^{-x}$ for any $x \in (0, 1)$.

Applying Equation~\eqref{eq:cla:filter_function:1:2} to Equation~\eqref{eq:cla:filter_function:1:1} over all $i \in [0: i^*]$ results in
\begin{eqnarray*}
    \int_{0}^{2 / (\pi s_1)} (\sinc_{s_1}(\xi))^{\ell} \cdot \d \xi
    & \leq & \frac{1}{s_1} \cdot \sum_{i = 0}^{i^*} \sqrt{8 / \ell} \cdot e^{-\pi^2 \cdot i^2} \\
    & \leq & \frac{1}{s_1} \cdot \sqrt{8 / \ell} \cdot \sum_{i = 0}^{+\infty} e^{-\pi^2 \cdot i^2} \\
    & \leq & \frac{1}{s_1} \cdot \sqrt{8 / \ell} \cdot \sum_{i = 0}^{+\infty} \frac{1}{(1 + i)^2} \\
    & = & \frac{1}{s_1} \cdot \sqrt{8 / \ell} \cdot \frac{\pi^2}{6} \\
    & \leq & 5 / s_1 \cdot \ell^{-1 / 2},
\end{eqnarray*}
where the third step is because $e^{-\pi^2 \cdot i^2} \leq e^{-2i} \leq e^{-2\ln(1 + i)} = (1 + i)^{-2}$ for each $i \in \mathbb{N}_{\geq 0}$; and the last step is because $\sqrt{8} \cdot \pi^2 / 6 \approx 4.6526 < 5$.

Further, we can infer the lower-bound part as follows:
\begin{eqnarray}
    \notag
    \int_{0}^{2 / (\pi s_1)} (\sinc_{s_1}(\xi))^{\ell} \cdot \d \xi
    & = & \frac{1}{s_1} \cdot \int_{0}^{2 / \pi} (\sinc(\xi))^{\ell} \cdot \d \xi \\
    \notag
    & = & \frac{1}{s_1} \cdot \int_{0}^{2 / \pi} |\sinc(\xi)|^{\ell} \cdot \d \xi \\
    \notag
    & \geq & \frac{1}{s_1} \cdot \int_{0}^{i^* \cdot \sqrt{8 / \ell}} |\sinc(\xi)|^{\ell} \cdot \d \xi \\
    \label{eq:cla:filter_function:1:3}
    & \geq & \frac{1}{s_1} \cdot \int_{0}^{\sqrt{8 / \ell}} |\sinc(\xi)|^{\ell} \cdot \d \xi,
\end{eqnarray}
where the first step is by substitution; the second step by because $\ell \in \NP$ is an even integer (see Definition~\ref{def:filter_function_single}); the third step follows from the Fact~(c) given in the beginning of this proof; and the last step is due to the above Fact~(b) that $i^* \geq 6 > 1$.

Under the assumption $\ell \geq 1000$, we have $0 < \pi^2 / 6 \cdot \xi^2 \leq \pi^2 / 6 \cdot \frac{8}{\ell} \leq \pi^2 / 6 \cdot \frac{8}{1000} \approx 0.013 \leq 1$ for any $\xi \in [0, \sqrt{8 / \ell}]$. Then for any $\xi \in [0, \sqrt{8 / \ell}]$ we have
\begin{eqnarray}
    \notag
    |\sinc(\xi)|^{\ell}
    & \geq & \left(1 - \pi^2 / 6 \cdot \xi^2\right)^{\ell} \\
    \notag
    & \geq & \left(1 - \pi^2 / 6 \cdot \frac{8}{\ell}\right)^{\ell} \\
    \notag
    & \geq & \left(1 - \pi^2 / 6 \cdot \frac{8}{1000}\right)^{1000} \\
    \label{eq:cla:filter_function:1:4}
    & = & \left(1 - \pi^2 / 750\right)^{1000},
\end{eqnarray}
where the first step is by Part~(a) of Fact~\ref{fac:sinc_function}; and the third step is because $0 \leq \pi^2 / 6 \cdot \frac{8}{\ell} \leq 1$ and that $y = (1 - z)^{1 / z}$ is a decreasing function for any $z \in (0, 1)$.

Plugging Equation~\eqref{eq:cla:filter_function:1:4} back into Equation~\eqref{eq:cla:filter_function:1:3} results in
\begin{align*}
    \int_{0}^{2 / (\pi s_1)} (\sinc_{s_1}(\xi))^{\ell} \cdot \d \xi
    \geq \frac{1}{s_1} \cdot \sqrt{8 / \ell} \cdot \left(1 - \pi^2 / 750\right)^{1000}
    \geq 1 / 2^{18} \cdot \frac{1}{s_1} \cdot \ell^{-1 / 2},
\end{align*}
where the last step follows from elementary calculation.

This completes the proof of Claim~\ref{cla:filter_function:1}.
\end{proof}

\begin{claim}
\label{cla:filter_function:2}
$\displaystyle{\int}_{2 / (\pi s_1)}^{+\infty} (\sinc_{s_1}(\xi))^{\ell} \cdot \d \xi = O(\frac{1}{s_1} \cdot 2^{-\ell})$.
\end{claim}

\begin{proof}
We check the claim as follows:
\begin{eqnarray*}
    \int_{2 / (\pi s_1)}^{+\infty} (\sinc_{s_1}(\xi))^{\ell} \cdot \d \xi
    & = & \int_{2 / (\pi s_1)}^{+\infty} \left|\sinc_{s_1}(\xi)\right|^{\ell} \cdot \d \xi \\
    & \leq & \int_{2 / (\pi s_1)}^{+\infty} \frac{1}{\pi^{\ell} \cdot |s_1 \xi|^{\ell}} \cdot \d \xi \\
    & = & \frac{1}{\pi s_1} \cdot \int_{2}^{+\infty} \xi^{-\ell} \cdot \d \xi \\
    & = & \frac{1}{\pi s_1} \cdot \frac{2}{\ell - 1} \cdot 2^{-\ell} \\
    & = & O\big(\frac{1}{s_1} \cdot 2^{-\ell}\big),
\end{eqnarray*}
where the first step is because $\ell \in \NP$ is an even integer (see Definition~\ref{def:filter_function_single}); the second step is by Part~(c) of Fact~\ref{fac:sinc_function}; and the third step is by substitution.

This completes the proof of Claim~\ref{cla:filter_function:2}.
\end{proof}

\newpage
\section{Filter, permutation and hashing in a single dimension}
\label{sec:HashToBins_single}

In this section, we first construct our single-dimensional filter function $(\mathsf{G}(t), \hat{\mathsf{G}}(f))$ and investigate several properties of it, based on the building-block function $(G(t), \hat{G}(f))$ introduced in Section~\ref{sec:filter_function_single}. In particular:
\begin{itemize}
    \item Sections~\ref{subsec:HashToBins_single:filter} to \ref{subsec:HashToBins_single:filter_proof}. We first leverage the function $(G(t), \hat{G}(f))$ introduced in Definition~\ref{def:filter_function_single} to construct our ultimate single-dimensional filter $(\mathsf{G}(t), \hat{\mathsf{G}}(f))$, and then prove the properties of this filter (by applying Lemma~\ref{lem:filter_function_single} and extra arguments).
    
    \item Section~\ref{subsec:HashToBins_single:window}. We associate the filter $(\mathsf{G}(t), \hat{\mathsf{G}}(f))$ with another standard window $(\mathsf{G}'(t), \hat{\mathsf{G}'}(f))$ (in a manner similar to Lemma~\ref{lem:window_function_single}), which is more convenient for our later use.
\end{itemize}
%The remainder of this section is organized as follows:
%\begin{itemize}
%    \item Section~\ref{subsec:HashToBins_single:prelim} introduces our {\em one-dimensional} permutation scheme and hashing scheme, as well as other requisite definitions for our sampling algorithm.
    
%    \item Sections~\ref{subsec:HashToBins_single:alg} to \ref{subsec:HashToBins_single:probabilities} present our {\em one-dimensional} sampling algorithm {\HashToBins} (Algorithm~\ref{alg:HashToBins_single}) and the proofs of its performance guarantees. Particularly, we show in Section~\ref{subsec:HashToBins_single:probabilities} that {\HashToBins} meets certain desirable conditions with some good probability, and show in Sections~\ref{subsec:HashToBins_single:error_by_noise} and \ref{subsec:HashToBins_single:error_by_bad_events} that (once the conditions are met) {\HashToBins} works well.
%\end{itemize}

\subsection{Construction of filter \texorpdfstring{$(\mathsf{G}(t), \hat{\mathsf{G}}(f))$}{}}
\label{subsec:HashToBins_single:filter}

\begin{definition}[The single-dimensional filter]
\label{def:shift_filter_function_single}
Recall the parameters defined in Definition~\ref{def:filter_function_single}:
\begin{itemize}
    \item The number of bins in a single dimension $B = \Theta(d \cdot k^{1 / d})$ is a certain multiple of $d \in \mathbb{N}_{\geq 1}$.
    
    \item The noise level parameter $\delta \in (0, 1)$.
    
    \item $\alpha = \Theta(1 / d)$ is chosen such that $\frac{1}{100 \cdot (d + 1) \cdot \alpha} \in \mathbb{N}_{\geq 1}$ is an integer; clearly $\alpha \leq \frac{1}{100 \cdot (d + 1)} \leq \frac{1}{200}$.
    
    \item $s_1 = \frac{2 B}{\alpha}$ and $s_2 = \frac{1}{B + B / d}$.
    
    \item $\ell = \Theta(\log(k d / \delta))$ is an even integer. We safely assume $\ell \geq 1000$.
\end{itemize}
Further, the {\em width parameter} $W = \Theta(\frac{F}{B \eta})$ is chosen to be a sufficiently large integer. Based on the building-block function $(G(t), \hat{G}(f))$ given in Definition~\ref{def:filter_function_single}, for $i \in \Z$, define the shifted function
\begin{align*}
    \hat{G}_{i}(f) & ~ := ~ \hat{G}(f + i), \\
    G_{i}(t) & ~ := ~ \int_{-\infty}^{+\infty} \hat{G}_i(\xi) \cdot e^{2 \pi \i t \cdot \xi} \cdot \d \xi.
\end{align*}
Then for any $t, f \in \R$ the single-dimensional filter $(\mathsf{G}(t), \hat{\mathsf{G}}(f))$ is given by
\begin{eqnarray*}
    \hat{\mathsf{G}}(f)
    & = & e^{-\frac{\delta}{\poly(k, d)}} \cdot \sum_{i \in [-W: W]} \hat{G}_{i}(f), \\
    \mathsf{G}(t)
    & = & \int_{-\infty}^{+\infty} \hat{\mathsf{G}}(\xi) \cdot e^{2 \pi \i t \cdot \xi} \cdot \d \xi.
\end{eqnarray*}
\end{definition}

\subsection{Properties of filter \texorpdfstring{$(\mathsf{G}(t), \hat{\mathsf{G}}(f))$}{}}

Later we will employ another slightly different filter $(\mathsf{G}(t), \hat{\mathsf{G}}(f))$, just by shifting the one given in Definition~\ref{def:shift_filter_function_single}. For ease of presentation, the following Lemma~\ref{lem:shifted_filter_function_single} is stated for the shifted filter, but we show in Section~\ref{subsec:HashToBins_single:filter_proof} the counterpart claims for the unshifted filter.

\begin{lemma}[The single-dimensional filter]
\label{lem:shifted_filter_function_single}
The filter $(\mathsf{G}(t), \hat{\mathsf{G}}(f))[B, \delta, \alpha, \ell, W]$ given in Definition~\ref{def:shift_filter_function_single} satisfies the following (as Figure~\ref{fig:shifted_filter_function_single} illustrates):
\begin{description}[labelindent = 1em]
    \item [Property~I:]
    $e^{-\frac{\delta}{\poly(k, d)}} \leq \hat{\mathsf{G}}(f) \leq 1$ when $|f - i| \leq \frac{1 - \alpha}{2 B}$ for some integer $|i| \leq W$.
    
    \item [Property~II:]
    $0 \leq \hat{\mathsf{G}}(f) \leq 1$ when $\frac{1 - \alpha}{2 B} \leq |f - i| \leq \frac{1}{2 B}$ for some integer $|i| \leq W$.
    
    \item [Property~III:]
    $0 \leq \hat{\mathsf{G}}(f) \leq \frac{\delta}{\poly(k, d)}$ when $|f - i| \geq \frac{1}{2 B}$ for any integer $|i| \leq W$.
    
    \item [Property~IV:]
    $\mathsf{G}(t) = (2 W + 1) \cdot e^{-\frac{\delta}{\poly(k, d)}} \cdot G(t) \cdot \frac{\sinc_{(2 W + 1)}(t + 1 / 2)}{\sinc(t + 1 / 2)}$ for any $t \in \R$.
    
    \item [Property~V:]
    $\supp(\mathsf{G}) \subseteq \supp(G) \subseteq [-\ell \cdot \frac{B}{\alpha}, \ell \cdot \frac{B}{\alpha}]$.
    
    \item [Property~VI:]
    $\sum_{i \in \mathbb{Z}} \mathsf{G}(i)^2 = e^{-\frac{\delta}{\poly(k, d)}} \cdot \sum_{i \in \mathbb{Z}} G(i)^2 \leq (1 + \frac{2}{d}) \cdot B^{-1}$.
\end{description}
\end{lemma}

\begin{remark}
\label{rem:shifted_filter_function_single}
The function values $\mathsf{G}(i)$ are the scaling coefficients for our samples in the time domain. In fact, we just need the values $\mathsf{G}(i)$ at a few {\em fixed} points. Thus, we can calculate and store those effective coefficients before the sampling process.
\end{remark}

\begin{figure}
    \centering
    \includegraphics[width = 0.9 \textwidth]{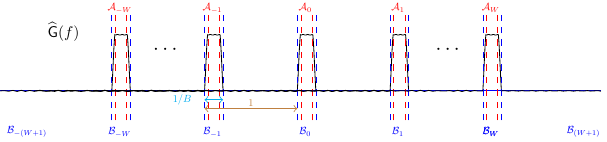}
    \caption{Demonstration for Lemma~\ref{lem:shifted_filter_function_single}. ${\cal B}_i=[i-1/(2B),i+1/(2B)]$ and ${\cal A}_i=[i-(1-\alpha)/(2B),i+(1-\alpha)/(2B)]$. Intuitively, $\hat{\mathsf{G}}(f)$ has $2W+1$ ``peaks'' near the $2W+1$ integers, where its value is very close 1. And its value drops very quickly and oscillates near 0 outside the ``peaks''.}
    \label{fig:shifted_filter_function_single}
\end{figure}

\subsection{Proof of properties}
\label{subsec:HashToBins_single:filter_proof}

Recall that $e^{\delta / (4 k d)} \cdot \hat{\mathsf{G}}(f) = \sum_{i \in [-W: W]} \hat{G}_{i}(f)$ (see Definition~\ref{def:shift_filter_function_single}) and, that every summand function $\hat{G}_{i}(f)$ has an axis of symmetry $f = -i$ (see Definition~\ref{def:filter_function_single}). The next Claim~\ref{cla:shifted_filter_function_single_1} suggests that, at any $f \in \R$, the value $e^{\delta / (4 k d)} \cdot \hat{\mathsf{G}}(f)$ is dominated by one particular summand $\hat{G}_{i^*}(f)$, and the other $2 W$ summands are negligibly small. Given this observation, we can easily conclude Properties~I to III of Lemma~\ref{lem:shifted_filter_function_single} from Properties~II to IV of Lemma~\ref{lem:filter_function_single}.

\begin{claim}[Auxiliary result for Lemma~\ref{lem:shifted_filter_function_single}]
\label{cla:shifted_filter_function_single_1}
Given any $f' \in \R$, let $i^* = \mathrm{argmin}_{i \in [-W: W]} |f' + i|$ be the particular summand function $\hat{G}_{i}(f)$ with the closest-to-$f'$ axis of symmetry, then
\begin{align*}
    0 \leq e^{\frac{\delta}{\poly(k, d)}} \cdot \hat{\mathsf{G}}(f') - \hat{G}_{i^*}(f') = \sum_{i \in [-W: W] \setminus \{i^*\}} \hat{G}_{i}(f') \leq \frac{\delta}{\poly(k, d)}.
\end{align*}
\end{claim}

\begin{proof}
The first part $\sum_{i \in [-W: W] \setminus \{i^*\}} \hat{G}_{i}(f') \geq 0$ follows because each summand function $\hat{G}_{i}(f')$ is non-negative (see Properties~II to IV of Lemma~\ref{lem:filter_function_single}).

We now show the second part $\sum_{i \in [-W: W] \setminus \{i^*\}} \hat{G}_{i}(f) \leq \delta / (4 k d)$. Clearly, the summand functions $\hat{G}_{i}(f)$ have axes of symmetry $f = -i \in [-W: W]$ (see Definition~\ref{def:shift_filter_function_single}). Since $f = -i^*$ is the axis closet to $f'$, the distances between $f'$ and either the left-hand-side axes (i.e.\ $f = -i \leq -i^* - 1$) or the right-hand-side axes (i.e.\ $f = -i \geq -i^* + 1$) are at least $(1 - \frac{1}{2}),\; (2 - \frac{1}{2}),\; (3 - \frac{1}{2}),\; \cdots$. Further, when the distance between $f'$ and an axis $f = -i$ is at least $j - \frac{1}{2} \geq \frac{j}{2} \geq \frac{1}{2 B}$ (for some $j \in \mathbb{N}_{\geq 1}$; recall Definition~\ref{def:shift_filter_function_single} that $B = \Theta(d \cdot k^{1 / d})$ is the number of bins in a single dimension), it follows from Property IV of Lemma~\ref{lem:filter_function_single} that
\begin{eqnarray}
    \notag
    \hat{G}_{i}(f')
    & \leq & (\pi B \cdot (j - 1 / 2)\big)^{-\ell} \\
    \label{eq:cla:shifted_filter_function_single_1:1}
    & \leq & \big(\pi B \cdot j / 2)^{-\ell}.
\end{eqnarray}

Take all of the $2 W$ remaining summands $i \in [-W: W] \setminus \{i^*\}$ into account:
\begin{eqnarray*}
    \sum_{i \in [-W: W] \setminus \{i^*\}} \hat{G}_{i}(f')
    & \leq & \sum_{i \in \mathbb{Z} \setminus \{i^*\}} \hat{G}_{i}(f') \\
    & \leq & 2 \cdot \sum_{j \in \mathbb{N}_{\geq 1}} (\pi B \cdot j / 2)^{-\ell} \\
    & \leq & \frac{\delta}{\poly(k, d)},
\end{eqnarray*}
where the first step follows because each summand function $\hat{G}_{i}(f')$ is non-negative (see Properties~II to IV of Lemma~\ref{lem:filter_function_single}); the second step follows from Inequality~\ref{eq:cla:shifted_filter_function_single_1:1}; and the last step holds whenever $B = \Theta(d \cdot k^{1 / d})$ and $\ell = \Theta(\log(k d / \delta))$ are large enough.

This completes the proof of Claim~\ref{cla:shifted_filter_function_single_1}.
\end{proof}

\begin{claim}[Property~I of Lemma~\ref{lem:shifted_filter_function_single}]
\label{cla:shifted_filter_function_single_2}
$e^{-\frac{\delta}{\poly(k, d)}} \leq \hat{\mathsf{G}}(f) \leq 1$ when $|f - i| \leq \frac{1 - \alpha}{2 B}$ for some integer $|i| \leq W$.
\end{claim}

\begin{proof}
We let $i^* = \mathrm{argmin}_{i \in [-W: W]} |f + i|$ index the summand function with the closest-to-$f$ axis of symmetry. For the lower-bound part, we observe that
\begin{eqnarray*}
    \hat{\mathsf{G}}(f)
    & = & e^{-\frac{\delta}{\poly(k, d)}} \cdot \sum_{i \in [-W: W]} \hat{G}_{i}(f) \\
    & \geq & e^{-\frac{\delta}{\poly(k, d)}} \cdot \hat{G}_{i^*}(f) \\
    & \geq & e^{-\frac{\delta}{\poly(k, d)}} \cdot \big(1 - \frac{\delta}{\poly(k, d)}\big) \\
    & \geq & e^{-\frac{\delta}{\poly(k, d)}},
\end{eqnarray*}
where the first step follows from Claim~\ref{cla:shifted_filter_function_single_1}; the third step applies Property~II of Lemma~\ref{lem:filter_function_single}; and the last step is because $1 - z / 4 \geq e^{-3 z / 4}$ for any $z \in [0, 1]$ (recall Definition~\ref{def:shift_filter_function_single} that the noise level parameter $0 < \delta < 1$).

In addition, for the upper-bound part we have
\begin{eqnarray*}
    \hat{\mathsf{G}}(f)
    & = & e^{-\frac{\delta}{\poly(k, d)}} \cdot \hat{G}_{i^*}(f) + e^{-\frac{\delta}{\poly(k, d)}} \cdot \sum_{i \in [-W: W] \setminus \{i^*\}} \hat{G}_{i}(f) \\
    & \leq & e^{-\frac{\delta}{\poly(k, d)}} \cdot 1 + e^{-\frac{\delta}{\poly(k, d)}} \cdot \frac{\delta}{\poly(k, d)} \\
    & \leq & 1,
\end{eqnarray*}
where the second step applies Property~II of Lemma~\ref{lem:filter_function_single} (to the first term) and Claim~\ref{cla:shifted_filter_function_single_1} (to the second term); and the last step is because $1 + z \leq e^{z}$ for any $z \geq 0$.

This completes the proof of Claim~\ref{cla:shifted_filter_function_single_2}.
\end{proof}

\begin{claim}[Property~II of Lemma~\ref{lem:shifted_filter_function_single}]
\label{cla:shifted_filter_function_single_3}
$0 \leq \hat{\mathsf{G}}(f) \leq 1$ when $\frac{1 - \alpha}{2 B} \leq |f - i| \leq \frac{1}{2 B}$ for some integer $|i| \leq W$.
\end{claim}

\begin{proof}
The first part $\hat{\mathsf{G}}(f) \geq 0$ follows because each summand function $\hat{G}_{i}(f)$ is non-negative (see Properties~II to IV of Lemma~\ref{lem:filter_function_single}). For the upper-bound part, by definition we have
\begin{eqnarray*}
    \hat{\mathsf{G}}(f)
    & = & e^{-\frac{\delta}{\poly(k, d)}} \cdot \hat{G}_{i^*}(f) + e^{-\frac{\delta}{\poly(k, d)}} \cdot \sum_{i \in [-W: W] \setminus \{i^*\}} \hat{G}_{i}(f) \\
    & \leq & e^{-\frac{\delta}{\poly(k, d)}} \cdot 1 + e^{-\frac{\delta}{\poly(k, d)}} \cdot \frac{\delta}{\poly(k, d)} \\
    & \leq & 1,
\end{eqnarray*}
where the second step applies Property~III of Lemma~\ref{lem:filter_function_single} (to the first term) and Claim~\ref{cla:shifted_filter_function_single_1} (to the second term); and the last step is because $1 + z \leq e^{z}$ for any $z \geq 0$.

This completes the proof of Claim~\ref{cla:shifted_filter_function_single_3}.
\end{proof}

\begin{claim}[Property~III of Lemma~\ref{lem:shifted_filter_function_single}]
\label{cla:shifted_filter_function_single_4}
$0 \leq \hat{\mathsf{G}}(f) \leq \frac{\delta}{\poly(k, d)}$ when $|f - i| \geq \frac{1}{2 B}$ for any integer $|i| \leq W$.
\end{claim}

\begin{proof}
The first part $\hat{\mathsf{G}}(f) \geq 0$ has been justified in the proof of Claim~\ref{cla:shifted_filter_function_single_3}. For the upper-bound part, by definition we have
\begin{eqnarray*}
    \hat{\mathsf{G}}(f)
    & = & e^{-\frac{\delta}{\poly(k, d)}} \cdot \hat{G}_{i^*}(f) + e^{-\frac{\delta}{\poly(k, d)}} \cdot \sum_{i \in [-W: W] \setminus \{i^*\}} \hat{G}_{i}(f) \\
    & \leq & e^{-\frac{\delta}{\poly(k, d)}} \cdot \frac{\delta}{\poly(k, d)} + e^{-\frac{\delta}{\poly(k, d)}} \cdot \frac{\delta}{\poly(k, d)} \\
    & \leq & \frac{\delta}{\poly(k, d)},
\end{eqnarray*}
where the second step applies Property~IV of Lemma~\ref{lem:filter_function_single} (to the first term) and Claim~\ref{cla:shifted_filter_function_single_1} (to the second term); and the last step follows from elementary calculation.

This completes the proof of Claim~\ref{cla:shifted_filter_function_single_4}.
\end{proof}

\begin{claim}[Property~IV of Lemma~\ref{lem:shifted_filter_function_single}]
\label{cla:shifted_filter_function_single_5}
$\mathsf{G}(t) = (2 W + 1) \cdot e^{-\frac{\delta}{\poly(k, d)}} \cdot G(t) \cdot \frac{\sinc_{(2 W + 1)}(t)}{\sinc(t)}$ for any $t \in \R$.
\end{claim}

\begin{proof}
For convenient, in this proof we ignore the $\frac{0}{0}$ issue; this can be easily remedied by applying L'Hospital's rule. According to the definition of the inverse {\CFT},
\begin{eqnarray*}
    \mathsf{G}(t)
    & = & \int_{-\infty}^{+\infty} \hat{\mathsf{G}}(\xi) \cdot e^{2 \pi \i t \cdot \xi} \cdot \d \xi \\
    & = & e^{-\frac{\delta}{\poly(k, d)}} \cdot \sum_{i \in [-W: W]} \int_{-\infty}^{+\infty} \hat{G}(\xi + i) \cdot e^{2 \pi \i t \cdot \xi} \cdot \d \xi \\
    & = & e^{-\frac{\delta}{\poly(k, d)}} \cdot \sum_{i \in [-W: W]} e^{-2 \pi \i t \cdot i} \cdot \int_{-\infty}^{+\infty} \hat{G}(\xi + i) \cdot e^{2 \pi \i t \cdot (\xi + i)} \cdot \d \xi \\
    & = & e^{-\frac{\delta}{\poly(k, d)}} \cdot \sum_{i \in [-W: W]} e^{-2 \pi \i t \cdot i} \cdot \int_{-\infty}^{+\infty} \hat{G}(\xi) \cdot e^{2 \pi \i t \cdot \xi} \cdot \d \xi \\
    & = & e^{-\frac{\delta}{\poly(k, d)}} \cdot G(t) \cdot \sum_{i \in [-W: W]} e^{-2 \pi \i t \cdot i},
\end{eqnarray*}
where the second step is by Definition~\ref{def:shift_filter_function_single}; the fourth step is by substitution; and the last step is by the definition of the {\CFT}.

It remains to calculate the sum of the geometric sequence $\sum_{i \in [-W: W]} e^{-2 \pi \i \cdot (t + 1 / 2) \cdot i}$. Concretely, for any $\tau \in \R$ we have
\begin{eqnarray*}
    \sum_{i \in [-W: W]} e^{-2 \pi \i t \cdot i}
    & = & e^{2 \pi \i t \cdot W} \cdot \left(\frac{1 - e^{-2 \pi \i t \cdot (2 W + 1)}}{1 - e^{-2 \pi \i t}}\right) \\
    & = & e^{2 \pi \i t \cdot W} \cdot \left(\frac{e^{-\pi \i t \cdot (2 W + 1)}}{e^{-\pi \i t}} \cdot \frac{e^{\pi \i t \cdot (2 W + 1)} - e^{-\pi \i t \cdot (2 W + 1)}}{e^{\pi \i t} - e^{-\pi \i t}}\right) \\
    & = & \frac{e^{\pi \i t \cdot (2 W + 1)} - e^{-\pi \i t \cdot (2 W + 1)}}{e^{\pi \i t} - e^{-\pi \i t}} \\
    & = & \frac{2 \i \cdot \sin\big(\pi t \cdot (2 W + 1)\big)}{2 \i \cdot \sin(\pi t)} \\
    & = & \frac{\sinc_{(2 W + 1)}(t)}{\sinc(t)} \cdot (2 W + 1),
\end{eqnarray*}
where the fourth step is by Euler's formula (note that $\cos(z)$ is an even function while $\sin(z)$ is an odd function); and the last step follows from Definition~\ref{def:rect_sinc}.

Combining everything together completes the proof of Claim~\ref{cla:shifted_filter_function_single_5}.
\end{proof}

\begin{claim}[Property~V of Lemma~\ref{lem:shifted_filter_function_single}]
\label{cla:shifted_filter_function_single_6}
$\supp(\mathsf{G}) \subseteq \supp(G) \subseteq [-\ell \cdot \frac{B}{\alpha}, \ell \cdot \frac{B}{\alpha}]$.
\end{claim}

\begin{proof}
By Claim~\ref{cla:shifted_filter_function_single_5}, the filter $\mathsf{G}(t) = (2 W + 1) \cdot e^{-\frac{\delta}{\poly(k, d)}} \cdot G(t) \cdot \frac{\sinc_{(2 W + 1)}(t)}{\sinc(t)}$ has the same support as the function $G(t)$. Thus we immediately infer this claim from Property~V of Lemma~\ref{lem:filter_function_single}.
\end{proof}

\begin{claim}[Property~VI of Lemma~\ref{lem:shifted_filter_function_single}]
\label{cla:shifted_filter_function_single_7}
It follows that
\begin{align*}
    \sum_{i \in \mathbb{Z}} \mathsf{G}(i + 1 / 2)^2 = e^{-\frac{\delta}{\poly(k, d)}} \cdot \sum_{i \in \mathbb{Z}} G(i + 1 / 2)^2 \leq (1 + 2 / d) \cdot B^{-1}.
\end{align*}
\end{claim}

\begin{proof}
The second part of the claim is a direct follow-up to Property~VII of Lemma~\ref{lem:filter_function_single}. That is,
\begin{eqnarray*}
    e^{-\frac{\delta}{\poly(k, d)}} \cdot \sum_{i \in \mathbb{Z}} G(i + 1 / 2)^2
    & \leq & e^{-\frac{\delta}{\poly(k, d)}} \cdot \left(1 + \frac{\delta}{4 k d}\right)^2 \cdot \left(1 + \frac{1}{d}\right) \cdot B^{-1} \\
    & \leq & (1 + 2 / d) \cdot B^{-1},
\end{eqnarray*}
where the first step follows from Property~VII of Lemma~\ref{lem:filter_function_single}; and the last step applies the fact that $e^{-z} \cdot (1 + z) \leq 1$ for any $z \in \R_{\geq 0}$.

To see the first part, it suffices to show that 
\begin{align*}
|\mathsf{G}(i + 1 / 2)| = e^{-\frac{\delta}{\poly(k, d)}} \cdot |G(i + 1 / 2)| , \forall i \in \mathbb{Z}.
\end{align*} 
Based on Claim~\ref{cla:shifted_filter_function_single_5}, this equation is equivalent to 
\begin{align*}
\Big| (2 W + 1) \cdot \frac{\sinc_{(2 W + 1)}(i + 1 / 2)}{\sinc(i + 1 / 2)} \Big| = 1.
\end{align*}
According to Definition~\ref{def:rect_sinc}, we have
\begin{eqnarray*}
    \left|(2 W + 1) \cdot \frac{\sinc_{(2 W + 1)}(i + 1 / 2)}{\sinc(i + 1 / 2)}\right|
    & = & \left|\frac{\sin\big(\pi \cdot (2 W + 1) \cdot (i + 1 / 2)\big)}{\sin\big(\pi \cdot (i + 1 / 2)\big)}\right| \\
    & = & \left|\frac{\sin\big(\pi \cdot (2 i \cdot W + W + i) + \pi / 2\big)}{\sin(\pi \cdot i + \pi / 2)}\right| \\
    & = & \left|\frac{(-1)^{2 i \cdot W + i + W}}{(-1)^{i}}\right| \\
    & = & 1,
\end{eqnarray*}
where the third step follows because both $(2 i \cdot W + W + i)$ and $i$ are integers; thus we can apply certain properties of the $\sin(z)$ function.

This completes the proof of Claim~\ref{cla:shifted_filter_function_single_7}.
\end{proof}

\subsection{Construction and properties of standard window \texorpdfstring{$(\mathsf{G}'(t), \hat{\mathsf{G}'}(f))$}{}}
\label{subsec:HashToBins_single:window}

Now we associate the filter $(\mathsf{G}(t), \hat{\mathsf{G}}(f))$ introduced in Definition~\ref{def:shift_filter_function_single} with another standard window $(\mathsf{G}'(t), \hat{\mathsf{G}'}(f))$ (in a manner similar to Lemma~\ref{lem:window_function_single}), which is more convenient for our later use.

\begin{figure}
    \centering
    \includegraphics[width = 0.9 \textwidth]{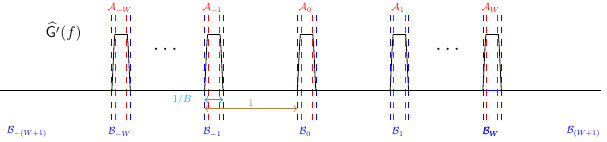}
    \caption{Demonstration for Lemma~\ref{lem:shifted_window_function_single}. ${\cal B}_i=[i-1/(2B),i+1/(2B)]$ and ${\cal A}_i=[i-(1-\alpha)/(2B),i+(1-\alpha)/(2B)]$. Intuitively, $\hat{\mathsf{G}'}(f)$ has $2W+1$ ``peaks'', where its value is exactly 1. And its value drops to 0 very quickly outside the ``peaks''.}
    \label{fig:shifted_window_function_single}
\end{figure}

\begin{lemma}[The single-dimensional standard window]
\label{lem:shifted_window_function_single}
For the filter $(\mathsf{G}(t), \hat{\mathsf{G}}(f))$ given in Definition~\ref{def:filter_function_single}, there exists another function $(\mathsf{G}'(t), \hat{\mathsf{G}'}(f))$ such that(as Figure~\ref{fig:shifted_window_function_single} illustrates):
\begin{description}[labelindent = 1em]
    \item [Property~I:]
    $\hat{\mathsf{G}'}(f) = 1$ when $|f - i| \leq \frac{1 - \alpha}{2 B}$ for some integer $|i| \leq W$.
    
    \item [Property~II:]
    $\hat{\mathsf{G}'}(f) \in [0, 1]$ when $\frac{1 - \alpha}{2 B} \leq |f - i| \leq \frac{1}{2 B}$ for some integer $|i| \leq W$.
    
    \item [Property~III:]
    $\hat{\mathsf{G}'}(f) = 0$ when $|f - i| \geq \frac{1}{2 B}$ for any integer $|i| \leq W$.
    
    \item [Property~IV:]
    $\|\hat{\mathsf{G}'} - \hat{\mathsf{G}}\|_{\infty} = \max_{f \in \R} |\hat{\mathsf{G}'}(f) - \hat{\mathsf{G}}(f)| \leq \frac{\delta}{\poly(k, d)}$.
\end{description}
\end{lemma}

\begin{proof}
Recall Definition~\ref{def:shift_filter_function_single} for the parameters $B$, $\delta$, $\alpha$, $\ell$ and $W$. We define the single-dimensional standard window $\hat{\mathsf{G}'}(f)$ as follows:
\begin{itemize}
    \item $\hat{\mathsf{G}'}(f) = 1$ when $|f - i| \leq \frac{1 - \alpha}{2 B}$ for some integer $|i| \leq W$;
    
    \item $\hat{\mathsf{G}'}(f) = \hat{\mathsf{G}}(f)$ when $\frac{1 - \alpha}{2 B} \leq |f - i| \leq \frac{1}{2 B}$ for some integer $|i| \leq W$; and
    
    \item $\hat{\mathsf{G}'}(f) = 0$ when $|f - i| \geq \frac{1}{2 B}$ for any integer $|i| \leq W$.
\end{itemize}
By construction, Properties~I and III follows directly. Further, Property~II follows from Property~II of Lemma~\ref{lem:shifted_filter_function_single}, and Property~IV can be inferred from Properties~I to III of Lemma~\ref{lem:shifted_filter_function_single}.

This completes the proof of Lemma~\ref{lem:shifted_window_function_single}.
\end{proof}

\newpage

%\iffalse
%\paragraph{Acknowledgment.}
%We are grateful to Lijie Chen, Sitan Chen, Xi Chen, Yin Tat Lee, Jerry Li, Ankur Moitra, Rocco Servedio, Ruoqi Shen, Huacheng Yu, Peilin Zhong, and Hengjie Zhang, for their helpful discussions at various stages of this work.
%\fi

\addcontentsline{toc}{section}{References}
\bibliographystyle{alpha}
\bibliography{ref}

\end{document}